\documentclass[a4paper,10pt]{article}
\usepackage{geometry}                
\usepackage[active]{srcltx}
\usepackage{hyperref}
\usepackage[normalem]{ulem}

\usepackage{amsmath}
\usepackage{amssymb}
\usepackage{parskip}
\usepackage{mathtools}
\usepackage{fullpage}
\usepackage{mathrsfs}
\usepackage{amsthm}
\usepackage{physics}
\usepackage{bbm}
\usepackage[utf8]{inputenc}
\usepackage{caption}
\usepackage{subcaption}
\usepackage{graphicx}
\usepackage{attachfile}
\usepackage{navigator}
\usepackage{theoremref}
\usepackage{faktor}
\usepackage{musicography}
\usepackage{accents}
\usepackage{tikz}
\usetikzlibrary{angles,quotes}

\newtheorem{theorem}{Theorem}[section]

\newtheorem{definition}[theorem]{Definition}
\newtheorem{proposition}[theorem]{Proposition}
\newtheorem{corollary}[theorem]{Corollary}
\newtheorem{lemma}[theorem]{Lemma}
\newtheorem{remark}[theorem]{Remark}

\numberwithin{equation}{section}

\newcommand{\R}{\mathbb{R}}

\newcommand{\C}{\mathbb{C}}
\newcommand{\N}{\mathbb{N}}

\newcommand{\e}{\varepsilon}

\newcommand{\supp}{\mathrm{supp}}

\newcommand{\dist}{\mathrm{dist}}
\newcommand{\Ric}{\mathrm{Ric}}

\newcommand{\vol}{\mathrm{vol}}

\newcommand{\tcd}{\mathsf{TCD}}
\newcommand{\wtcd}{\mathsf{wTCD}}

\newcommand{\sfd}{\mathsf d}
\newcommand{\di}{{\rm d}}
\newcommand{\cz}{\Bar{Z}}
\newcommand{\mm}{\mathfrak{m}}
\newcommand{\Ent}{\mathrm{Ent}}
\newcommand{\Dom}{\mathrm{Dom}}
\newcommand{\hh}{\mathcal{H}}
\newcommand{\uu}{\mathcal{U}}
\newcommand{\vv}{\mathcal{V}}
\newcommand{\zz}{\mathcal{Z}}
\newcommand{\kk}{\mathcal{K}}
\newcommand{\ccz}{\Bar{\mathcal{Z}}}
\newcommand{\domlor}{\mathrm{Dom}_{\mathrm{Lor}}}
\newcommand{\ee}{{\rm e}}

\usepackage{todonotes}

\newtheorem{ithm}{Theorem}[section]

\newcommand{\LM}[1]{\hbox{\vrule width.2pt \vbox to#1pt{\vfill \hrule
width#1pt height.2pt}}}
\newcommand{\LL}{{\mathchoice
{\,\LM7\,}{\,\LM7\,}{\,\LM5\,}{\,\LM{3.35}\,}}}

\title{Stability of Synthetic Timelike Ricci Bounds under 
$C^0$-Limits and Applications to Impulsive Gravitational Waves}
\author{Andrea Mondino\thanks{{\tt andrea.mondino@maths.ox.ac.uk} Mathematical Institute, University of Oxford, UK.}, Vanessa Ryborz\thanks{{\tt vanessa.ryborz@chch.ox.ac.uk} Mathematical Institute, University of Oxford, UK.}, and Clemens S\"amann\thanks{{\tt clemens.saemann@univie.ac.at} Faculty of Mathematics, University of Vienna, Austria.}}

\begin{document}

\maketitle

\begin{abstract}
We investigate the stability of timelike Ricci curvature lower bounds under low-regularity limits of Lorentzian metrics. Specifically, we prove that the synthetic curvature-dimension condition \( \tcd^e_p(K,N) \), which provides an optimal transport formulation of the Hawking--Penrose strong energy condition, is stable under locally uniform convergence of smooth Lorentzian metrics, provided a uniform global hyperbolicity assumption holds. As a consequence, smooth locally uniform limits of vacuum spacetimes satisfy the strong energy condition, even though curvature is not controlled a priori.

As a main application, we study impulsive gravitational waves --- spacetimes with Lipschitz continuous metrics --- and show that large classes of such waves satisfy synthetic timelike Ricci curvature lower bounds. In the case of Minkowski background, we further establish synthetic upper Ricci curvature bounds. Our approach relies on constructing suitable smooth approximations with lower bounds on the timelike Ricci, and analyzing the limiting behavior via Lorentzian optimal transport.

We also  apply the aforementioned stability theorem to weak solutions of the Einstein equations arising from the nonlinear interaction of impulsive gravitational waves.

These results yield new geometric insights into low-regularity solutions of the Einstein equations.

\vskip 1em

\medskip

\end{abstract}

\tableofcontents
\section{Introduction}

\subsection*{Locally uniform limits of solutions to the vacuum Einstein equations satisfy the Hawking--Penrose strong energy condition}
Let $(g_j)$ be a sequence of smooth Lorentzian metrics defined on the smooth manifold $M$, solving the vacuum Einstein equations
\begin{equation}\label{eq:vacuum}
\mathrm{Ric}(g_j)\equiv 0,
\end{equation}
and assume that $g_j \to g$ locally uniformly, where $g$ is a smooth Lorentzian metric. A basic problem is to determine which geometric or analytic features of the sequence are stable under this limiting procedure. Owing to the quasilinear, second-order nature of \eqref{eq:vacuum}, the passage to the limit under mere $C^0$ convergence is highly nontrivial and, a priori, does not control curvature. Indeed, at a schematic level, one may write
\begin{equation}\label{eq:RicQ}
\mathrm{Ric}(g) = \partial^2 g + \mathcal{Q}(\partial g),
\end{equation}
where $\mathcal{Q}(\partial g)$ is a non-linear expression involving   quadratic terms in the first order derivatives of the metric $g$. In view of the non-linear term in \eqref{eq:RicQ},  locally uniform limits of $(g_j)$ may retain nontrivial contributions from quadratic expressions in $\partial g_j$. From the perspective of nonlinear PDE, this places the problem within the general paradigm in which lack of strong compactness gives rise to defect measures and effective source terms; see, e.g., the compensated compactness theory of Tartar and Murat \cite{Tartar1979,Murat1978} and its measure-valued refinements \cite{DiPernaMajda1987}. In particular, there is no reason to expect stability of the vacuum condition under $C^0$ convergence in the absence of additional structure.

A guiding example is provided by the high-frequency analysis of Einstein's equations initiated by Isaacson \cite{Isaacson1968a,Isaacson1968b}. There, rapidly oscillating families of vacuum solutions give rise, in an appropriate weak sense, to an effective stress--energy tensor, leading formally to a limit equation of the form
\begin{equation}\label{eq:effective}
\mathrm{Ric}(g) = T_{\mathrm{eff}},
\end{equation}
where $T_{\mathrm{eff}}$ is generated by quadratic oscillations of $\partial g_j$. This phenomenon may be interpreted as a geometric instance of compensated compactness for quasilinear wave equations.

Concrete realizations of \eqref{eq:effective} arise from families of vacuum plane waves. In suitable regimes, one obtains limits satisfying
\begin{equation}\label{eq:null_dust_Ricci}
\mathrm{Ric}_{\mu\nu}(g) = \tfrac{1}{2} \alpha(u)^2 \, \partial_\mu u \, \partial_\nu u,
\end{equation}
corresponding to null dust stress--energy tensors of the form
\begin{equation*}
T_{\mu\nu} = \rho \, k_\mu k_\nu, \qquad k^\mu k_\mu = 0, \qquad \rho \ge 0.
\end{equation*}
In particular, for any timelike vector $X$ one has
\begin{equation*}
\mathrm{Ric}(X,X) = \rho (k \cdot X)^2 \ge 0,
\end{equation*}
so that the Hawking--Penrose strong energy condition holds in these examples. More generally, such limits are consistent with the expectation that effective matter arising from high-frequency gravitational fields obeys natural positivity conditions.

A far-reaching conjectural framework was proposed by Burnett \cite{Burnett1989}, asserting that suitable weak limits of high-frequency vacuum solutions converge to solutions of the Einstein--massless Vlasov system. This essentially means that high-frequency gravitational waves behave collectively like a gas of massless particles (photons or gravitons). In this setting, the effective stress--energy tensor takes the form
\begin{equation*}
T_{\mu\nu}(x) = \int p_\mu p_\nu \, f(x,p) \, dp, \qquad f \ge 0,
\end{equation*}
which automatically satisfies the strong energy condition.
The high frequency assumption in Burnett's conjecture can be stated precisely as follows:
\begin{enumerate}
\item \emph{$C^0$-convergence:} There exists a $C^\infty$-Lorentzian metric $g$ and a decreasing sequence $(\lambda_j)_{j\in \N}$ with $\lim_{j\to \infty} \lambda_j=0$, such that
\begin{equation}\label{eq:gjlaj}
  |g_j-g|\leq \lambda_j;
\end{equation}
\item \emph{High frequency:} There exists a constant $C>0$ --- possibly depending on $M$ and on the sequence $g_j$, but independent of $j$ --- such that
\begin{equation}\label{eq:degjlajC}
|\partial g_j|\leq  C,\quad   |\partial^2 g_j|\leq \lambda_j^{-1}.
\end{equation}
\end{enumerate}
This conjecture connects the problem of stability --- under $C^0$-convergence --- of the vacuum Einstein equations  to kinetic theory and suggests a robust mechanism enforcing positivity in the limit. Substantial progress toward this conjecture has been obtained in recent years \cite{HL-AENS, Guerra, HL:2024}, for a survey see \cite{HL:survey}

\medskip

Parallel to these developments, there has been substantial progress on weak and low-regularity formulations of the Einstein equations. The foundational work of Geroch and Traschen \cite{GT1987} identifies a maximal stable class of metrics with distributional curvature, namely metrics whose first derivatives lie in $L^2_{\mathrm{loc}}$, thereby providing a coordinate-invariant setting in which the Einstein tensor can be interpreted as a distribution. Subsequent contributions (see for instance the work by LeFloch and Mardare \cite{LM:2007}) develop weak formulations and compactness results for $C^0\cap W^{1,2}$-Lorentzian metrics, yielding weak notions of the Einstein equations compatible with geometric limits.

In a complementary direction, a substantial body of work has been devoted to the study of the Einstein vacuum equations at low regularity. Building on earlier formulations in harmonic (wave) coordinates, Klainerman and Rodnianski \cite{KR:2005} established local well-posedness for rough initial data in Sobolev spaces close to the scaling-critical threshold, relying on a refined analysis of the null structure of the equations and bilinear spacetime estimates. This line of research was further developed using microlocal and Fourier-analytic techniques, notably by Smith and Tataru \cite{ST:2005}, who proved low-regularity well-posedness results for quasilinear wave equations, including the Einstein equations in wave coordinates, at nearly optimal regularity. Subsequent works have refined these methods and extended the range of admissible data, combining paradifferential calculus, Strichartz estimates, and geometric analysis. These results show that the Einstein equations admit a robust well-posedness theory below classical smoothness, provided one retains quantitative control of first derivatives in suitable Sobolev spaces. At the same time,  these approaches remain fundamentally tied to Sobolev regularity and do not extend to the $C^0$ setting, where no intrinsic formulation of the equations is currently available.

In conclusion, the aforementioned literature indicates that the passage to the limit in \eqref{eq:vacuum} is intimately related to general principles of weak convergence and rigidity for nonlinear geometric PDEs. In particular, it highlights the persistence of oscillation and concentration phenomena in curvature, when the metric converges merely in $C^0$.  Beyond the high-frequency or the Sobolev regimes described above, the behavior of arbitrary $C^0$ limits of vacuum metrics remains largely open. In particular, there is currently no general mechanism -- at the PDE level -- ensuring that positivity properties of curvature persist under $C^0$-limits \eqref{eq:gjlaj}, dropping the high frequency assumption \eqref{eq:degjlajC}.

This leads naturally to the question of whether intrinsic positivity properties of curvature can be recovered directly from the vacuum structure of the approximating sequence, without additional assumptions. The first  main result of the present work answers this question affirmatively by showing that the limit Lorentzian metric necessarily satisfies the Hawking--Penrose strong energy condition:
\begin{equation}
\mathrm{Ric}_g(X,X) \ge 0 \quad \text{for all timelike vectors } X.
\end{equation}
Before stating the result, we introduce the following terminology.

\begin{definition}[Uniform global hyperbolicity]\label{uniform_global_hyperbolicity_def-Intro}
    A sequence of smooth Lorentzian metrics $(g_j)_{j\in \N}$ on $M$ is said to be \emph{uniformly globally hyperbolic}, if for any two compact sets $K_1, K_2 \subset M$,
$$\bigcup_{j=1}^\infty (J^-_{g_j}(K_1) \cap J^+_{g_j}(K_2)) \Subset M.$$
\end{definition}

In essence, this means that each Lorentzian metric $g_j$ is globally hyperbolic and that the compactness of the causal emeralds $J^-_{g_j}(K_1) \cap J^+_{g_j}(K_2)$ holds uniformly along the sequence.

We refer to Section \ref{SS:SuffUGH} for a sufficient condition ensuring uniform global hyperbolicity, and to Appendix \ref{subsec:counterexample_uniform_hyperbolicity} for an example of a sequence of  globally hyperbolic Lorentzian metrics $(g_j)_{j\in \mathbb{N}}$ on $\mathbb{R}^{d+1}$, $d \ge 2$, converging locally uniformly to a globally hyperbolic limit $g_\infty$, which fails to be uniformly globally hyperbolic.

We are now in a position to state the first main result.
\smallskip

\begin{theorem}\label{thm:VECC0}
  Let $M$ be a smooth $n$-dimensional manifold and $g$ be a smooth globally hyperbolic Lorentzian metric. Let $(g_j)_{j\in \N}$ be a sequence of smooth Lorentzian metrics on $M$ such that
   \begin{enumerate}
\item    $g_j \to g$ locally uniformly;
\item $(M, g_j)$ are uniformly globally hyperbolic;
\item $(M,g_j)$ solve the Einstein vacuum equations with cosmological constant $\Lambda\in \R$, i.e., $\Ric(g_j)\equiv  \Lambda g_j$.
\end{enumerate}

Then $(M, g)$ satisfies the strong energy condition (with cosmological constant $\Lambda$):
\begin{equation}\label{eq:SECIntro}
\mathrm{Ric}_{g}(X,X) \ge \Lambda \, |g(X,X)| \quad \text{for all timelike vectors } X.
\end{equation}
\end{theorem}

Theorem~\ref{thm:VECC0} is a special case of Theorem~\ref{stability_tcd}, which establishes the stability of the synthetic, optimal-transport formulation of the strong energy condition — namely the \( \tcd^e_p(K,N) \) and \( \wtcd^e_p(K,N) \) conditions — under \( C^0 \)-convergence.  Introduced in \cite{cavalletti2020optimal} as a Lorentzian counterpart of the Lott–Sturm–Villani \( \mathsf{CD}(K,N) \) condition \cite{sturm2006geometryI, sturm2006geometryII, lott2009ricci}, these synthetic curvature-dimension conditions are equivalent to the classical ones in the smooth Lorentzian setting \cite{mccann2018displacement, mondino2022optimal}. At the same time, they remain stable under low-regularity limits, thereby providing a robust framework for capturing curvature positivity beyond the smooth category.
 We illustrate this point on a classical toy example. Let $(f_j)$ be a sequence of convex $C^2$ functions $f_j:[0,1]\to \mathbb{R}$ converging uniformly to a function $f:[0,1]\to \mathbb{R}$. The convexity of $f_j$ can be expressed either in the \emph{differential form}
\begin{equation}\label{eq:fj''}
f_j'' \ge 0 \quad \text{on } [0,1],
\end{equation}
or in the \emph{synthetic form}
\begin{equation}\label{eq:fjconvex}
f_j(t x + (1-t) y) \le t\,f_j(x) + (1-t)\,f_j(y) \quad \text{for all } t,x,y \in [0,1].
\end{equation}
While the stability of \eqref{eq:fjconvex} under $C^0$ convergence is immediate, the corresponding statement for \eqref{eq:fj''} is not. One may view \eqref{eq:fj''} as a toy model for the classical strong energy condition \eqref{eq:SECIntro}, and \eqref{eq:fjconvex} as a toy model for its synthetic, optimal-transport formulation $\tcd^e_p(K,N)$, thereby illustrating why the latter is better suited to low-regularity limits.

\subsection*{Penrose's impulsive gravitational waves}

Originally introduced by Penrose \cite{Pen:68, Pen:72}, \emph{impulsive gravitational waves} provide a model for highly energetic yet extremely brief bursts of gravitational radiation. These idealized models are remarkable both from the physical and from the mathematical point-of-view: First, they consider a spacetime background of constant curvature (flat or (anti-)de Sitter) and assume a source of gravitational radiation located at infinity (or effectively “outside” the model). From this external source, a gravitational wave propagates through the spacetime, localized on a null hypersurface. Off this null hypersurface the spacetime is smooth and has constant curvature. Moreover, the passing of the wave gives rise to the so-called memory effect, see \cite{BG:85, Chr:91, Ste:19}. 
From a mathematical point of view, impulsive gravitational waves are particularly significant because they provide models of spacetimes endowed with Lorentzian metrics of low regularity; see \cite{PS:22}. In one representation --- the so-called \emph{Rosen form} \eqref{nonsmooth_metric} --- the metric is locally Lipschitz continuous. There is, however, a physically equivalent \emph{Brinkmann form}, in which the metric is distributional and explicitly contains a Dirac delta term. The Brinkmann form has the advantage that, both before and after the impulse, the metric manifestly coincides with that of the constant-curvature background. This can be vividly and geometrically visualized via the so-called ``cut-and-paste" method, cf.\ \cite[Ch.\ 20]{GP:09}. These two forms are related by a so-called “discontinuous coordinate transformation” \cite{PSSS:19}. Nevertheless, this equivalence can be established in a mathematically rigorous way; see \cite{SSSS:24}. Various methods of constructing the different forms of impulsive gravitational waves are reviewed in \cite{podolsky2014global, PS:22}.

So far, impulsive gravitational waves have been studied mathematically exclusively by analytical methods --- for example, through the use of Filippov geodesics and regularization techniques \cite{podolsky2014global, SSLP:16}. In this work, we approach them for the first time from a \emph{synthetic} perspective, that is, by analyzing their geometry --- and in particular their curvature --- within the framework of (Lorentzian) metric measure geometry, which neither relies on a manifold structure nor on differential calculus. Such an approach proved to be extremely fruitful, see \cite{CM:22, Sae:24, McC:25, Bra:26a} for recent reviews on this topic.  In particular, within the framework of Lorentzian (pre-)length spaces \cite{kunzinger2018lorentzian}, Cavalletti and the first author introduced synthetic timelike lower Ricci curvature bounds \cite{cavalletti2020optimal}. Their construction extends to the non-smooth setting an optimal transport characterization of timelike Ricci curvature bounds for smooth spacetimes, obtained independently by McCann \cite{mccann2018displacement} and by the first author together with Suhr \cite{mondino2022optimal}.

Impulsive gravitational waves are exact solutions to the vacuum Einstein equations (if the profile function $\hh$ satisfies $\Delta \hh=0$) as the Ricci curvature can be computed explicitly (even though the metric is singular on a null hypersurface), see \cite[Ch.\ 20]{GP:09} for more details. In this work we establish for the first time that impulsive gravitational waves have timelike Ricci curvature bounded from below, i.e., satisfy the $\tcd$-condition (if $\Delta \hh\leq 0$). Moreover, we also study timelike upper curvature bounds and the synthetic vacuum Einstein equations \cite{mondino2022optimal}.

We study impulsive gravitational waves in their (Lipschitz) continuous formulation, after first establishing the $C^0$-stability of the
$\tcd$-condition. The stability of the $\tcd$-condition under measured convergence of the underlying Lorentzian length spaces --- via Lorentzian isometric embeddings --- was proved in \cite[Thm.\ 3.15]{cavalletti2020optimal}. Here, by contrast, we establish its stability with respect to $C^0$-convergence, a result that is of considerable independent interest. More precisely, we prove that the $\tcd$-condition holds for continuous Lorentzian metrics arising as $C^0$-limits of suitable approximating sequences of smooth Lorentzian metrics. We then apply this result to impulsive gravitational waves, for which we will construct explicit smooth approximations. 

The outlined strategy naturally splits the article into two parts. The first part studies the stability of synthetic timelike Ricci curvature bounds under the locally uniform convergence of Lorentzian metrics on smooth manifolds. In the second part, these results are applied to prove synthetic timelike Ricci curvature bounds for impulsive gravitational waves. 

The main result of the first part is:

\setcounter{section}{2}
\setcounter{ithm}{27}
\begin{ithm}[Stability of the $\tcd$ condition under $C^0$-convergence of the metric]
    Let $M$ be a smooth $n$-dimensional manifold and $g_\infty$ be a continuous Lorentzian metric on $M$ that turns $(M,g_\infty)$ into a globally hyperbolic space such that every  length-maximising causal curve has a causal character. Let $(g_j)_{j\in \N}$ be a sequence of smooth Lorentzian metrics on $M$ such that
   \begin{enumerate}
\item    $g_j \to g_\infty$ locally uniformly;
\item $(M, g_j)$ are uniformly globally hyperbolic;
\item $(M,g_j)$ satisfy $\Ric_{g_j}(v,v)\geq K |g_j(v,v)|$ for all $v\in TM$ timelike.
\end{enumerate}

Then $(M, g_\infty)$ satisfies the $\wtcd^{e}_p(K, N)$-condition,  for all $p\in (0,1)$ and $N\geq n$.

In particular, if the limit (in $C^0_{loc}$-sense) metric $g_\infty$ is smooth, then it satisfied the Hawking--Penrose strong energy condition (with cosmological constant $K\in \R$):  $\Ric_{g_\infty}(v,v)\geq K |g_\infty(v,v)|$, for all $v\in TM$ timelike.
\end{ithm}
\setcounter{section}{1}

To prove Theorem \ref{stability_tcd},  we first analyse the behaviour of the causal relations $\ll$ and $\leq$, as well as the time-separation function $\tau$ for the continuous limit of smooth Lorentzian metrics. This is achieved in Section~\ref{Sec:2}. A key step in our argument is to show that if a sequence of smooth Lorentzian metrics \( g_j \) converges locally uniformly to a continuous Lorentzian metric \( g_\infty \), and if $(g_j)_{j\in \N}$ is uniformly globally hyperbolic  (see Definition \ref{uniform_global_hyperbolicity_def-Intro}), then the associated time-separation functions \( \tau_j \) converge locally uniformly to \( \tau_\infty \) as \( j \to \infty \) (see Proposition~\ref{locally_uniform_convergence_tau}). This result constitutes a crucial ingredient in the proof of Theorem~\ref{stability_tcd}, as well as in the proof of a corresponding stability statement for synthetic timelike upper Ricci curvature bounds (see Theorem~\ref{stability_upper_bounds}).  

Section~\ref{Sec:3} is devoted to the application of our main result to impulsive gravitational waves. 
The corresponding Lorentzian metric \( g_{\hh,\Lambda} \) is Lipschitz continuous; it depends on a profile function 
\( \hh \colon \mathbb{R}^4 \to \mathbb{R} \), which encodes the geometry of the propagating wave, and on the cosmological constant \( \Lambda \in \mathbb{R} \), which characterizes the constant-curvature background spacetime. 
For the explicit form of the metric, see \eqref{nonsmooth_metric} and \eqref{explicit_metric_coefficients}.
\newpage

The second main result is:
\setcounter{section}{3}
\setcounter{ithm}{11}
\begin{ithm}\label{thm:3.10Intro}
    Let $(\R^4, g_{\hh,\Lambda} )$ be an impulsive gravitational wave with profile function $\hh: \R^4 \to \R$ and cosmological constant $\Lambda \in \R$. Assume that 
    \begin{equation}\label{eq:DeltahURB.Intro}
     -2\partial^2_{Z, \cz} \hh  -\frac{\Lambda(\hh - Z\partial_Z \hh -\cz \partial_{\cz}\hh)}{3(1+\frac{\Lambda}{6}|Z|^2)} \geq 0 \quad \text{on } \left\{1+\frac{\Lambda}{6}|Z|^2 \neq 0 \right\},
     \end{equation}
     that $D^2 \hh$ is bounded and $\hh$ has subquadratic growth (see \eqref{eq:hhSubQuad}). Then:
    \begin{itemize}
    \item  If $\Lambda \geq 0$, then  $(\R^4, g_{\hh,\Lambda} )$ satisfies the $\tcd^e_p(\Lambda, 4)$-condition, for all $p \in (0,1)$. 
     \item If $\Lambda<0$, then $(\R^4, g_{\hh,\Lambda} )$ satisfies the $\tcd^e_p(\Lambda, 4)$-condition locally, for all $p \in (0,1)$; more precisely, it satisfies the $\tcd^e_p(\Lambda, 4)$-condition  on diamonds of compact sets satisfying the assumptions of Proposition \ref{lambda_negative_coord_shift}.
    \end{itemize}
\end{ithm}
\setcounter{section}{1}
Note that, in case of zero cosmological constant (i.e., $\Lambda=0$), the assumption \eqref{eq:DeltahURB.Intro} reads simply as $\Delta \hh\leq 0$.

We briefly outline the strategy of the proof. 

The Lipschitz continuous Rosen form \( g_{\hh,\Lambda} \) of the impulsive gravitational wave enables us to compute the (measure-valued) distributional Ricci curvature tensor and to identify the regions in which it admits lower or upper bounds. 
Exploiting this information, we construct smooth approximations \( g_\varepsilon \) such that \( g_\varepsilon \to g_{\hh,\Lambda} \) locally uniformly and \( \Ric[g_\varepsilon] \) is bounded from below. 
More precisely, the measure-valued Ricci curvature of \( g_{\hh,\Lambda} \) admits the decomposition into absolutely continuous and singular parts
\[
\Ric[g_{\hh,\Lambda}] = \Ric[g_{\hh,\Lambda}]^{(a)} + \Ric[g_{\hh,\Lambda}]^{(s)}.
\]
Likewise, we will decompose the (now smooth) Ricci curvature of 
$g_\varepsilon$  as 
\[
\Ric[g_\varepsilon] = \Ric[g_\varepsilon]^{(a)} + \Ric[g_\varepsilon]^{(s)},
\]
with $\Ric[g_\varepsilon]^{(a)}$ converging locally uniformly to  $\Ric[ g_{\hh,\Lambda}]^{(a)}$  and $\Ric[g_\varepsilon]^{(s)}$ converging distributionally to $\Ric[ g_{\hh,\Lambda}]^{(s)}$, \emph{while preserving its sign}. 

The construction proceeds by mollifying the non-smooth profile $U^+$ in the expression \eqref{nonsmooth_metric} of $g_{\hh,\Lambda}$, and analyzing the resulting curvature through explicit coordinate expressions. While most terms behave well under this regularization, certain nonlinear interactions fail to converge uniformly and must be isolated. To compensate for these defects, we introduce suitable correction terms in the metric, designed to cancel the non-vanishing contributions arising from the mollification. This yields a refined approximation scheme in which all problematic terms are either controlled or canceled, ensuring convergence of the Ricci tensor in the above sense. For the details on the construction of the smooth approximations $g_\varepsilon$, see Section \ref{SS:SmoothApproxIPP}.

As a consequence, the approximating metrics $g_\varepsilon$ fall within the scope of the stability Theorem~\ref{stability_tcd}, which allows us to establish the $\wtcd^e_p(\Lambda,4)$ condition for a large class of impulsive gravitational waves on backgrounds of non-negative constant curvature $\Lambda \ge 0$, as well as a local version of the same condition for waves propagating on an anti-de~Sitter background. Since $(\mathbb{R}^4, g_{\hh,\Lambda})$ is timelike non-branching, we conclude by invoking \cite[Thm.~3.35]{Bra:23}, which yields the equivalence of the $\wtcd^e_p(\Lambda,4)$ and $\tcd^e_p(\Lambda,4)$ conditions in this setting.

Regarding recent developments related to Theorem~\ref{thm:3.10Intro}, we highlight the work \cite{Braun-Lipschitz2026} by Braun and S\'alamo Candal, which shows that Lipschitz spacetimes with non-negative timelike Ricci curvature in a distributional sense satisfy the \( \mathsf{TMCP}_p(0,N) \) condition. 
Recall from \cite{cavalletti2020optimal} that both the \( \wtcd^e_p(0,N) \) and \( \tcd^e_p(0,N) \) conditions imply \( \mathsf{TMCP}_p(0,N) \).
\medskip

For impulsive gravitational waves on Minkowski background, we also show validity of synthetic upper bounds on the Ricci curvature (in the sense of \cite{mondino2022optimal}). This is the content of the third main result stated below.
\smallskip

\setcounter{section}{3}
\setcounter{ithm}{16}
\begin{ithm}\label{thm:IPP-UB-Intro}
    If $\Lambda =0$, the impulsive gravitational wave $(\R^4, g_{\hh,\Lambda} )$ with profile function $\hh$ has synthetic timelike Ricci curvature bounded from above by $0$, provided that  $\Delta \hh \geq c |D^2 \hh|$, for some $c > 0$.
\end{ithm}
\setcounter{section}{1}

The proof of  Theorem \ref{upper_bounds_waves} proceeds by performing a suitable change of coordinates, which may be interpreted as a Lorentz transformation adapted to impulsive gravitational waves, and by explicitly constructing optimal transport plans that directly verify the synthetic upper timelike Ricci curvature bounds.

As a further application of the stability result established in Theorem~\ref{stability_tcd}, we show in Section~\ref{Sec:4} that the weak solutions of the Einstein equations constructed by Luk and Rodnianski in \cite{luk2013nonlinear} fit naturally into the framework of \( \tcd^e_p(0,4) \) Lorentzian length spaces. 
These weak solutions include, in particular, the nonlinear interaction of weak impulsive gravitational waves on a flat background.

Finally, in the appendices, we present an example demonstrating that global hyperbolicity is strictly weaker than uniform global hyperbolicity, and we provide additional technical details concerning the construction of the approximations employed in Section~\ref{Sec:3}.

\subsubsection*{Acknowledgments}
We are grateful to Tobias Beran, Argam Ohanyan, Ji\v{r}\'{\i} Podolsk\'y, Jona Röhrig, Roland Steinbauer, and Inés Vega-González for useful discussions. We also thank Nicola Gigli, Marco Picerni, Zhefeng Xu, and Xuewei Wang for carefully reading a preliminary version of the manuscript and for identifying an incorrect claim concerning the failure of a non-smooth Lorentzian splitting theorem.

A.\,M.\;acknowledges support from the European Research Council (ERC) under the European Union's Horizon 2020 research and innovation programme, grant agreement No.\;802689 ``CURVATURE''. 

C.\,S.\;acknowledges: This research was funded in whole or in part by the Austrian Science Fund (FWF) [Grant DOI \href{https://doi.org/10.55776/STA32}{10.55776/STA32}].

For the purpose of Open Access, the authors applied a CC BY public copyright license to any Author Accepted Manuscript (AAM) version arising from this submission. 

\section{Stability of synthetic curvature dimension conditions for continuous spacetimes}\label{Sec:2}
\subsection{Lorentzian pre-length spaces and timelike curvature dimension conditions}

 In this section, we briefly recall some basic notions from the theory of Lorentzian pre-length spaces and synthetic timelike lower Ricci curvature bounds thereof. 
 
 A Lorentzian pre-length space shall be seen as the Lorentzian analogue of a metric spaces (while Lorentzian length spaces are the analogues of metric length spaces). These notions were introduced in \cite{kunzinger2018lorentzian} by following earlier works of Busemann \cite{Bus:67} and Kronheimer--Penrose \cite{KP:67}. At present, several variants of the basic axiomatization of these spaces have been developed; see \cite{McC:24, BMcC:23, BBCGMORS:24, MS:25}. 
We also mention related alternative frameworks, such as bounded Lorentzian metric spaces \cite{MS:24, BMS:25}.  We have chosen to work within the original framework of \cite{kunzinger2018lorentzian}, since our setting is that of smooth manifolds, where any auxiliary Riemannian background metric \( h \) induces the given manifold topology through its associated length metric \( \sfd^h \).
\medskip

\begin{definition}[\cite{kunzinger2018lorentzian} Def.\ 2.8]
    A \emph{causal space} is a set $X$ endowed with a preorder $\leq$ and a transitive relation $\ll$ contained in $\leq$. \\
    A \emph{Lorentzian pre-length space} $(X, \sfd, \ll, \leq, \tau)$ is a causal space $(X, \ll, \leq)$ equipped with a proper metric $\sfd$ (i.e., closed and bounded subsets are compact) and a lower semi-continuous function $\tau: X \times X \to [0, \infty]$, called \emph{time-separation function}, that satisfies
    \begin{align}
        &\tau(x, y) + \tau(y, z) \leq \tau(x, z), \ \forall \ x \leq y \leq z\ \mathrm{(reverse\ triangle\ inequality)} \nonumber \\
        &\tau(x, y)>0 \iff x \ll y, \ \tau(x, y) = 0 \ \mathrm{if}\ x \not \leq y. 
    \end{align}
\end{definition}
We define the \emph{chronological} (resp.\ \emph{causal}) future of a subset $A \subset X$ as 
\begin{align*}
    I^+(A) := \{y \in X: \ \exists x \in A, x \ll y\}, \\
    J^+(A) := \{y \in X: \ \exists x \in A, x \leq y\},
\end{align*}
respectively. Analogously we define the chronological and causal pasts $I^-(A), J^-(A)$. In case $A = \{x\}$, we write $I^\pm(x), J^\pm(x)$ with a slight abuse of notation. 
\medskip

\begin{definition}
    A non-constant curve $\gamma:I \to X$ is called (future-directed) \emph{timelike} (resp.\ \emph{causal}) if $\gamma$ is locally Lipschitz continuous with respect to $\sfd$ and for each $s, t \in I$, $s < t$, it holds $\gamma_s \ll \gamma_t$, (resp. $\gamma_s \leq \gamma_t$). We say that $\gamma$ is a \emph{null} curve if it is causal and no two points in $\gamma(I)$ are related by $\ll$. 
\end{definition}
\medskip

\begin{definition}
    For $\gamma:[a,b] \to X$ future-directed causal, we set 
    \begin{align*}
        L_\tau(\gamma):= \inf \Bigg\{ \sum_{i=1}^N \tau(\gamma_{t_{i-1}}, \gamma_{t_i}): a=t_0<t_1<\ldots <t_N =b\Bigg\}.
    \end{align*}
    If \( I \) is half-open, say \( I = [a,b) \), we take the infimum over all partitions of the form \( a = t_0 < t_1 < \dots < t_N < b \); the definition for other types of intervals is treated analogously.
\end{definition}
Given a $g$-causal curve $\gamma$ in a smooth Lorentzian spacetime, we a priori have that $L_\tau(\gamma) \geq L_g(\gamma)$. However, if the spacetime is strongly causal, Proposition 2.32 in \cite{kunzinger2018lorentzian} shows that $L_\tau = L_g$. 
\medskip

\begin{definition}
    A causal curve $\gamma$ is called a \emph{geodesic} if it is maximal and continuous when parametrised by $\tau$-arclength, i.e., the set of causal geodesics is given by 
    \begin{align*}
        \mathrm{Geo}(X):= \{\gamma \in C([0,1], X): \tau(\gamma_s, \gamma_t)= (t-s)\tau(\gamma_0, \gamma_1)\}.
    \end{align*}
    The set of timelike geodesics is defined as follows:
    \begin{align*}
        \mathrm{TGeo}(X):= \{\gamma \in \mathrm{Geo}(X): \tau(\gamma_0, \gamma_1) >0\}.
    \end{align*}
\end{definition}
A Lorentzian pre-length space $(X, \sfd, \ll, \leq, \tau)$ is called
\begin{itemize}
    \item \emph{non-totally imprisoning} if for every compact set $K\subset X$ there exists a $C>0$ such that every causal curve contained in $K$ has $\sfd$-length at most $C$.  
    \item \emph{globally hyperbolic} if it is non-totally imprisoning and for every $x, y \in X$, the causal diamond $J^+(x) \cap J^-(y)$ is compact. 
    \item \emph{geodesic} if for all $x, y \in X$, $x \leq y$ there exists a curve $\gamma \in \mathrm{Geo}(X)$ connecting them. 
\end{itemize}
By \cite[Theorem 3.28]{kunzinger2018lorentzian}, the time-separation function $\tau$ in a globally hyperbolic Lorentzian geodesic space is finite and continuous.

The next proposition is due to Minguzzi \cite[Corollary 3.8]{minguzzi2023further}.
\medskip

\begin{proposition}
    Let $(X, \sfd, \ll, \leq, \tau)$ be a Lorentzian geodesic space. It is globally hyperbolic if and only if for any two compact sets $K_1, K_2 \subset X$, the causal diamond $J^+(K_1) \cap J^-(K_2)$ is compact and the relation $\{x \leq y\} \subset X^2$ is closed in $X^2$. 
\end{proposition}
\medskip

\begin{definition}
    A measured Lorentzian pre-length space $(X, \sfd, \ll, \leq, \tau, \mathfrak{m})$ is a Lorentzian pre-length space $(X, \sfd, \ll, \leq, \tau)$ endowed with a non-negative Radon measure $\mathfrak{m}$ with $\supp \, \mathfrak{m} =X$. 
\end{definition}
Next, we recall some basics of the theory of optimal transport on Lorentzian pre-length spaces. Denote by $\mathcal{P}(X)$ the set of Borel probability measures on $X$. For $\mu, \nu \in \mathcal{P}(X)$, define 
\begin{align*}
    \Pi_{\leq}(\mu, \nu):= \{\pi \in \mathcal{P}(X^2): (P_1)_\# \pi = \mu, (P_2)_\# \pi = \nu, \pi(X^2_\leq)=1 \}, \\
    \Pi_{\ll}(\mu, \nu):= \{\pi \in \mathcal{P}(X^2): (P_1)_\# \pi = \mu, (P_2)_\# \pi = \nu, \pi(X^2_\ll)=1\}.
\end{align*}
Here we denote $X^2_\ll:= \{(x,y) \in X^2: x \ll y\}$ and  $X^2_\leq:= \{(x,y) \in X^2: x \leq y\}$. Moreover, \( P_i \colon X \times X \to X \) denotes the projection onto the \( i \)-th factor, for \( i = 1,2 \), and \( (P_i)_\sharp \colon \mathcal{P}(X^2) \to \mathcal{P}(X) \) is the induced push-forward map on probability measures.
\medskip

\begin{definition}
    Let $(X, \sfd, \ll, \leq, \tau)$ be a Lorentzian pre-length space and let $p\in (0,1]$. Given $\mu, \nu \in \mathcal{P}(X)$, the $p$-Lorentz Wasserstein distance is defined by 
    \begin{align}\label{eq:defellp}
        \ell_p(\mu, \nu):= \sup_{\pi \in \Pi_{\leq}(\mu, \nu)} \Big( \int_{X^2} \tau(x,y)^p \di \pi(x, y)\Big)^{1/p}. 
    \end{align}
    When $\Pi_{\leq}(\mu, \nu)= \emptyset$, we set $\ell_p(\mu, \nu) = - \infty$. A measure $\pi\in \Pi_{\leq}(\mu, \nu)$ achieving the supremum in \eqref{eq:defellp} is called \emph{a $p$-optimal coupling from $\mu$ to $\nu$}. The set of $p$-optimal couplings from $\mu$ to $\nu$ is denoted by $\Pi_{\leq}^{p\text{-opt}}(\mu, \nu)$. 
\end{definition}
Furthermore, define $\ell^p: X^2 \to \{- \infty\} \cup [0, \infty]$ by 
\begin{equation*}
\ell^p(x, y)=
\begin{cases}
\tau(x, y)^p   &\text{if } x \leq y, \\
-\infty &\text{if } x \not \leq y.
\end{cases}
\end{equation*}
\begin{definition}[\cite{cavalletti2020optimal}, Def.\ 2.6]
    Fix $p \in (0,1]$ and let $(X, \sfd, \mm, \ll, \leq, \tau)$ be a measured Lorentzian pre-length space. A set $\Gamma \subset X^2_{\leq}$ is called $\tau^p$-cyclically monotone if for any finite family of points $(x_1, y_1), \ldots, (x_N, y_N) \in \Gamma$ the following inequality holds:
    \begin{align}\label{cyclical_monotonicity}
        \sum_{i=1}^N \tau^p(x_i, y_i) \geq \sum_{i=1}^N \tau^p(x_{i+1}, y_i),
    \end{align}
    with the convention that $x_{N+1}= x_1$. 
\end{definition}
\medskip

\begin{definition}[\cite{cavalletti2020optimal}, Def.\ 2.18 and 2.27]
    Let $(X, \sfd, \ll, \leq, \tau)$ be a Lorentzian pre-length space and let $p\in (0,1]$. We say that $(\mu, \nu) \in (\mathcal{P}(X))^2$ is \emph{timelike $p$-dualisable} (by $\pi \in \Pi_{\ll}(\mu, \nu)$), if 
    \begin{itemize}
        \item[1.] $\ell_p(\mu, \nu) \in (0, \infty)$,
        \item[2.] $\pi \in \Pi_{\leq}^{p\text{-opt}}(\mu, \nu)$ and $\pi(X^2_\ll)=1$,
        \item[3.] there exist measurable functions $a, b:X \to \R$, with $a \oplus b \in L^1(\mu \oplus \nu)$ such that $\ell^p \leq a \oplus b$ on $\supp\, \mu \times \supp\, \nu$. 
    \end{itemize}
    We say that $(\mu, \nu)$ is \emph{strongly timelike $p$-dualisable}, if it is timelike $p$-dualisable and if there exists a measurable $\ell^p$-cyclically monotone set $\Gamma \subset X^2_\ll \cap (\supp\, \mu \times \supp\, \nu)$ such that a coupling $\pi \in \Pi_\leq (\mu, \nu)$ is optimal if and only if $\pi$ is concentrated on $\Gamma$. 
\end{definition}
The evaluation map is defined by
\begin{equation}\label{def:eet}
\ee_{t}: C([0,1], X) \to X, 
\qquad 
\gamma \mapsto \ee_{t}(\gamma):=\gamma_{t}, 
\qquad \forall t\in [0,1].
\end{equation}
\medskip

\begin{definition}
    Let $(X, \sfd, \ll, \leq, \tau)$ be a Lorentzian pre-length space and let $p\in (0,1]$. We say that $\eta \in \mathcal{P}(\mathrm{Geo}(X))$ is an \emph{$\ell_p$-optimal dynamical plan from $\mu_0 \in \mathcal{P}(X)$ to $\mu_1 \in \mathcal{P}(X)$} if $(\ee_0)_\# \eta = \mu_0$, $(\ee_1)_\# \eta = \mu_1$ and 
    \begin{align*}
        (\ee_0, \ee_1)_\# \eta \in \Pi^{p\text{-opt}}_\leq (\mu_0, \mu_1). 
    \end{align*}
    We denote the set of $\ell_p$-optimal dynamical plans from $\mu_0$ to $\mu_1$ by $\mathrm{OptGeo}_{\ell_p}(\mu_0, \mu_1)$. 
    We say that a curve $[0,1] \ni t \mapsto \mu_t \in \mathcal{P}(X)$ is an \emph{$\ell_p$-geodesic} if there exists an optimal dynamical plan $\eta$ from $\mu_0$ to $\mu_1$ such that $\mu_t = (\ee_t)_\# \eta$.  
\end{definition} 

Given a measured Lorentzian pre-length space $(X, \sfd, \ll, \leq, \tau, \mathfrak{m})$ and a Borel probability measure $\mu = \rho \mathfrak{m} \in \mathcal{P}(X)$, such that $(\rho \log \rho)_+ \in L^1(\mathfrak{m})$, consider the Boltzmann--Shannon entropy
\begin{align*}
    \mathrm{Ent}(\mu | \mathfrak{m}):= \int_X \rho \log \rho\, \di \mathfrak{m}. 
\end{align*}
For all other measures $\mu$,  the entropy is set to $+\infty$.  Denote by $\Dom(\Ent(\cdot|\mm))$ the set of measures $\mu \in \mathcal{P}(X)$ with finite entropy. For $N \in (0, \infty)$, consider the so called \emph{Shannon entropy power} 
\begin{align*}
    U_N(\mu| \mathfrak{m}):= \exp\left(-\frac{\mathrm{Ent}(\mu| \mathfrak{m})}{N}\right).
\end{align*}
Set 
\begin{align*}
     \mathfrak{s}_\kappa(\vartheta) = \left\{ \begin{array}{ll}
           \frac{1}{\sqrt{\kappa}}\sin(\sqrt{\kappa}\vartheta) & \kappa >0,  \\
          \vartheta &\kappa =0,\\
          \frac{1}{\sqrt{-\kappa}}\sinh(\sqrt{-\kappa}\vartheta) & \kappa <0,
        \end{array}\right.
\end{align*}
and 
\begin{align*}
     \sigma^{(t)}_\kappa(\vartheta) = \left\{ \begin{array}{ll}
           \frac{\mathfrak{s}_k(t \vartheta)}{\mathfrak{s}_k( \vartheta)} & \kappa \vartheta^2 \neq 0, \kappa \vartheta^2 < \pi^2,  \\
          t &\kappa \vartheta^2 =0,\\
          +\infty & \kappa \vartheta^2 \geq \pi^2.
        \end{array}\right. 
\end{align*}

We now have all the ingredients needed to recall the definition of synthetic timelike curvature-dimension condition from \cite{cavalletti2020optimal}.
\medskip

\begin{definition}
    Let $p \in (0,1)$, $K \in \R$ and $N \in (0, \infty)$. We say that the measured Lorentzian pre-length space $(X, \sfd, \ll, \leq, \tau, \mathfrak{m})$ satisfies the \emph{(resp. weak) timelike curvature dimension condition} $\tcd^{e}_p(K, N)$ (resp. $\wtcd^{e}_p(K, N)$) if for any pair $(\mu_0, \mu_1) \in [\Dom(\mathrm{Ent}(\cdot|\mathfrak{m}))]^2$ that is timelike $p$-dualisable (resp. compactly supported and strongly timelike $p$-dualisable) by some $\pi \in \Pi_{\ll}^{p\text{-opt}}(\mu_0, \mu_1)$, there exists an $\ell_p$-geodesic $(\mu_t)_{t \in [0,1]}$ such that the function $[0,1] \ni t\mapsto u_N(t):= U_N(\mu_t | \mathfrak{m})$ satisfies
    \begin{align}
        u_N(t) \geq \sigma_{K/N}^{(1-t)}(\norm{\tau}_{L^2(\pi)})u_N(0) + \sigma_{K/N}^{(t)}(\norm{\tau}_{L^2(\pi)})u_N(1).
    \end{align}
\end{definition}
In the case of a smooth, globally hyperbolic, \( N \)-dimensional Lorentzian manifold, McCann showed that these conditions are equivalent to the pointwise bound 
\[
\Ric(v,v) \ge -K\, g(v,v)
\]
for all timelike vectors \( v \in TM \) \cite{mccann2018displacement} (see also \cite{mondino2022optimal}). 
Subsequently, Braun \cite{Bra:23} introduced alternative variants of the axiomatization, denoted by \( \tcd_p(K,N) \) and \( \tcd^{*}_p(K,N) \). 
Under the assumption of timelike non-branching, he established the equivalence of the conditions \( \tcd^{e}_p(K,N) \), \( \wtcd^{e}_p(K,N) \), and \( \tcd^{*}_p(K,N) \). 
Moreover, \( \tcd_p(K,N) \) always implies \( \tcd^{*}_p(K,N) \), and it is conjectured that the two are equivalent in the timelike non-branching setting.

\subsection{Spacetimes with locally uniformly converging Lorentzian metrics}
\textbf{Standing assumptions.} Throughout this section,  \( M \) is a smooth \( (d+1) \)-dimensional manifold,  \( (g_j)_{j \in \mathbb{N}} \) is a sequence of smooth Lorentzian metrics on \( M \), and  \( g_\infty \) is a continuous Lorentzian metric such that \( g_j \to g_\infty \) locally uniformly on \( M \). We assume, moreover, that for every \( j \in \mathbb{N} \cup \{\infty\} \), the Lorentzian manifold \( (M,g_j) \) is globally hyperbolic and that there exists a continuous vector field \( X \) which provides a common time orientation for all Lorentzian metrics \( g_j \).

Given an auxiliary complete Riemannian metric \( h \) on \( M \), it induces a (length) distance function \( \sfd_h \) on \( M \).
 For $(0,2)$-tensors $g, g'$  on $M$ and $S \subset M$, define  (see \cite{CG:12}) 
\begin{align}\label{norm_for_lorentz_metric}
    \sfd_S(g, g') := \sup_{X, Y \in TS, X, Y \neq 0} \frac{|g(X,Y)-g'(X,Y)|}{|X|_h|Y|_h}.
\end{align}

Recall that a causally plain (see \cite[Def.~1.16]{CG:12}) continuous Lorentzian metric $g$ on $M$ gives rise to a measured Lorentzian pre-length space as follows: A locally Lipschitz continuous curve $\gamma: I \to M$ is called \emph{future directed and causal} if $\dot{\gamma}_t$ is future directed and causal for almost all $t \in I$. The \emph{$g$-length} of a causal curve is then given by $$L_{g}(\gamma):= \int_I \sqrt{-g(\dot{\gamma}_t, \dot{\gamma}_t)}\, \di t, $$ and the time-separation function $\tau$ is defined by
\begin{align*}
    \tau_g(x, y) := \sup \{L_g(\gamma)\mid \gamma:[0,1] \to M \ \mathrm{causal\ and \ future\ directed\ }, \gamma_0=x, \gamma_1 =y\} \cup \{0\}.
\end{align*}
The causal relations $\leq$, $\ll$ are defined via $x \ll y \iff \tau_{g}(x, y) >0$ and $x \leq y$ if there exists a causal and future directed curve $\gamma:[0,1] \to M$ such that $\gamma_0=x, \gamma_1 =y$. The volume measure is locally given by $\di \vol_g = \sqrt{|\det g|}\, \di\mathcal{L}^{d+1}$. Then $(M, \sfd_h, \ll, \leq, \tau_g, \vol_g)$ is a Lorentzian pre-length space (see \cite{kunzinger2018lorentzian}).

Denote by $\ll_{j}, \leq_{j}$ the timelike and causal relations induced by $g_j, \; j \in \mathbb{N} \cup \{\infty\}$. Moreover, denote by $\tau_j$ the induced time-separation functions. Note that $\tau_j$ is continuous for all $j \in \N$. If, in addition to global hyperbolicity, we assume that the continuous Lorentzian metric \( g_\infty \) is causally plain, then the associated time-separation function \( \tau_\infty \) is continuous.

A $g_\infty$-causal curve $\gamma:I\to M$   is said to have a \emph{causal character} in $(M,g_\infty)$ if either $g_\infty(\dot\gamma_t, \dot\gamma_t)<0$ for a.e.\ $t\in I$, or $g_\infty(\dot\gamma_t, \dot\gamma_t)=0$ for a.e.\ $t\in I$.

\subsubsection{Uniform global hyperbolicity and locally uniform convergence of the time-separation functions}

We first establish a uniform variant of \cite[Lem.\ 2.6.6]{chrusciel2011elements}.
\begin{proposition}\label{uniform_bound_prosuct_with_time_or}
    For every compact subset $K \subset M$ there exists a constant $c=c(K)>0$ such that for every $j \in \N \cup \{\infty\}$, any $x \in K$ and any $g_j$-causal $v \in T_xM$, it holds
    \begin{align}
        |g_j(X,v)| \geq c|v|_h.
    \end{align}
\end{proposition}
\begin{proof}
    We will establish the statement along the lines of the proof of \cite[Lem.\ 2.6.6]{chrusciel2011elements}. We may assume that $K \subset \R^{d+1}$ and that $h = |\cdot|_{euc}$. Define
    \begin{align*}
        B:= \overline{\{g_j, j \in \N\}\cup \{g_\infty\}}^{\Gamma^0((T^*K)^{\otimes2})} \subset \Gamma^0((T^*K)^{\otimes2}). 
    \end{align*}
    As $g_j \to g_\infty$ in $C^0$, $B$ is compact in $C^0(K, \R^{(d+1) \times (d+1)}_{sym})$ and only contains Lorentzian metrics. 
    It follows that the set $$\Tilde{S}:= B \times K \times \partial B^{\R^{d+1}}_1(0) \subset C^{0}(K, \R^{{(d+1)} \times {(d+1)}}_{sym}) \times \R^{d+1} \times \R^{d+1}$$ is compact. 
    The map 
    \begin{align*}
        m: \Tilde{S} \ni (g, x, v) \mapsto g(x)(v,v) \in \R,
    \end{align*}
    is continuous. It follows that  $$S:= m^{-1}((-\infty, 0]) \subset \Tilde{S}$$ is closed and hence compact. Now note that for $(g, x, v) \in S$, we have that $v$ is $g$-causal when viewed as a vector in $T_xM$ and $|v|_{euc} = 1$. 
    Then the continuous map
    \begin{align*}
        f : S \ni (g,x,v) \mapsto |g(x)(X,v)| \in \R
    \end{align*}
    is strictly positive on $S$. The compactness of $S$ implies that there exists $c>0$ such that $f(g,x,v)\geq c$, for all $(g,x,v)\in S$ . 
\end{proof}

The next proposition in inspired by  \cite[Prop.\ 1.2]{CG:12}. 
\begin{proposition}\label{appx_majorising_metrics}
    For any compact subset $K \subset M$, it is possible to construct a sequence of smooth Lorentzian metrics $\Tilde{g}_k$ and an increasing sequence of integers $J_k$, with the following properties:
    \begin{itemize}
    \item  ${g}_\infty \prec \Tilde{g}_{k+1} \prec \Tilde{g}_k$;
    \item $g_{j} \prec \Tilde{g}_k $ on $K$, for all $j \geq J_k$;
    \item $\Tilde{g}_k \to g_\infty$ locally uniformly;
    \item $\Tilde{g}_k(v,v) < g_j(v,v)$, for all $j \geq J_k$ and all $v \in TK$ $g_j$-timelike.
    \end{itemize}
\end{proposition}
\begin{proof}
     We will first construct a sequence of continuous Lorentzian metrics $\hat{g}_k$ and an increasing sequence of integers $J_k$ such that:
     \begin{itemize}
    \item  ${g}_\infty \prec \hat{g}_{k+1} \prec \hat{g}_k$;
    \item $g_{j} \prec \hat{g}_k $ on $K$, for all $j \geq J_k$;
    \item $\hat{g}_k \to g_\infty$  locally uniformly.
      \end{itemize}
    Using Proposition \ref{uniform_bound_prosuct_with_time_or}, we get that there exists a constant $c \in (0,1)$ such that 
    $$
    |g_j(X,v)| \geq c|v|_h,
    $$
    for all $j \in \N \cup \{\infty\}$ and all $v \in T_xM$ with $g_j(v,v) \leq 0$ and  $x \in K$. 
    For each $k \in \N$, there exists a $J_k \geq 0$ such that, for all $j \geq J_k$, it holds:
    \begin{align*}
        &|(g_\infty(X,v))^2-(g_j(X,v))^2| \leq \frac{c^2}{4k} |v|^2_h \ \mathrm{on\ } TK,\\
        &|g_\infty(v,v)-g_j(v,v)| \leq \frac{c^2}{4k} |v|^2_h \ \mathrm{on\ } TK.
    \end{align*}
Denote, locally, $X^{\flat, \infty}_m := g^\infty_{ml}X^l$, and define
\begin{align}
    \hat{g}_k|_K := g_\infty -\frac{1}{k} X^{\flat, \infty} \otimes X^{\flat, \infty}.
\end{align}
The compactness of $K$ implies that there exists $k_0 \in \N$ such that  $\hat{g}_k$ is a Lorentzian metric on a neighbourhood of $K$, for all $k \geq k_0$. 
 It follows that 
\begin{align*}
    \hat{g}_k(v,v) =  g_\infty(v,v) - \frac{1}{k} g_\infty(X,v)^2, 
\end{align*}
for any vector $v \in T_xM$, $x \in K$.
Now assume $g_j(v,v) \leq 0$ for some $j\in \N_{\geq {J_k}} \cup \{\infty\}$.
Then 
\begin{align}\label{g_j_versus_g_hat_k}
    \hat{g}_k(v,v) &=  g_\infty(v,v) - \frac{1}{k}g_\infty(X,v)^2 \leq g_j(v,v)  -\frac{1}{k}g_j(X,v)^2 + \frac{c^2}{2k}|v|_h^2 \nonumber \\
    &\leq g_j(v,v) -\frac{c^2}{2k}|v|_h^2 <0. 
\end{align}
Moreover, for any vector $v \in TK$, it holds 
\begin{align*}
    \hat{g}_k(v,v) =  g_\infty(v,v) - \frac{1}{k} g_\infty(X,v)^2 \leq   g_\infty(v,v) - \frac{1}{k+1} g_\infty(X,v)^2 =\hat{g}_{k+1}(v,v), 
\end{align*}
with the inequality being strict if $g_\infty(v, X) \neq 0$.  
We can continuously extend the metrics $\hat{g}_k$ outside $K$ such that they are Lorentzian everywhere and such that $g_\infty \prec \hat{g}_{k+1} \prec \hat{g}_k$.
Thanks to \cite[Prop.\ 2.3]{kunzinger2015penrose}, for each $k$, we can  find a smooth Lorentzian metric $\Tilde{g}_k$ on $M$ such that $\hat{g}_{k+1} \prec \Tilde{g}_k \prec \hat{g}_k$ and  
\begin{align*}
    \sfd_{M}(\Tilde{g}_k, \hat{g}_k) \leq \frac{c^2}{8k}.
\end{align*}
Then note that for $j \geq J_k$ and any $v \in TK$ $g_j$-causal, \eqref{g_j_versus_g_hat_k} implies  
\begin{align*}
    \Tilde{g}_k(v,v) \leq \hat{g}_k(v,v) +  \frac{c^2}{8k}|v|_h^2 \leq g_j(v,v) - \frac{c^2}{2k}|v|^2_h +  \frac{c^2}{8k}|v|_h^2 < g_j(v,v).
\end{align*}
Observing that $\Tilde{g}_k$ converge to $g_\infty$ locally uniformly, the proof is complete.
\end{proof}
The next proposition is inspired by  \cite[Thm.\ 1.5 and Lem.\ 2.7]{samann2016global}.
\begin{proposition}\label{Lipschitz_bd_causal_curves}
    Let $\Tilde{g}_k$ be constructed in Proposition \ref{appx_majorising_metrics}.  Then there exist  $J_0, k_0 \in \N$ and a constant $C >0$ such that:
    \begin{itemize}
   \item  Any $\Tilde{g}_{k_0}$-causal curve that is contained in $K$ has  $h$-length at most $C$;
   \item Any $g_j$-causal curve $\gamma$ contained in $K$ has $h$-length at most $C$, for all $j \geq J_0$. 
   \end{itemize}
\end{proposition}
\begin{proof}
Let $\Tilde{g}_k$ be as in Proposition \ref{appx_majorising_metrics}. 
We will prove the proposition by contradiction.
Assume that for all $k \geq 1$, we can find $\Tilde{g}_k$-causal curves $\gamma_k:[0, k] \to K$ that are parameterised by $h$-arclength. As $K$ is compact, up to passing to a subsequence, we may assume that $\gamma_k(0) \to x \in K$ as $k \to \infty$.
Note that all the $\gamma_k$ are $\Tilde{g}_1$-causal. By \cite[Thm.\ 3.1 (1)]{minguzzi2008limit}, there exists a $\Tilde{g}_1$-causal curve $\gamma :[0,\infty) \to K$ such that $\gamma(0) = x$ and (a subsequence of) $\gamma_k$ converges locally uniformly to $\gamma$. We will now prove that $\gamma$ is $g_\infty$-causal. 
Since this is a local property, we may restrict to some interval $[a,b]$, where $0 \leq a < b < \infty$. Let $g'$ be any continuous metric that satisfies $g_\infty \prec g'$. By \cite[Lem.\ 1.4]{samann2016global}, there exists $k' \in \N$ such that  $\Tilde{g}_k \prec g'$ on $K$, for all $k \geq k$'.  
Now  $(\gamma_k|_{[a,b]})_{k \geq k'}$ is a sequence of $g'$-causal curves, that has again a subsequence converging uniformly to a $g'$-causal limit curve $\gamma':[a,b] \to K$. By the uniqueness of limits, it follows that $\gamma|_{[a,b]} = \gamma'$, hence $\gamma$ is $g'$-causal on $[a,b]$. This holds for any continuous $g' \succ g_\infty$ and all intervals $[a,b] \subset [0,\infty)$. 
By  \cite[Prop.\ 1.5]{CG:12}, it follows that $\gamma$ is $g_\infty$-causal. 
Hence, there exists a $g_\infty$-causal curve $\gamma$ contained in $K$ whose $h$-length is infinite. This contradicts  \cite[Lem.\ 2.7]{samann2016global}, since  $(M,g_\infty)$ was assumed to be globally hyperbolic. This concludes the proof of the first claim. In order to prove the second claim, set $J_0:=J_{k_0}$.  We get that any $g_j$-causal curve contained in $K$ is $\Tilde{g}_{k_0}$-causal, hence it has $h$-length at most $C>0$. 
\end{proof}

Throughout this section, we assume that the sequence $(g_j)_{j\in \N}$ satisfies the following uniform form of global hyperbolicity.
\smallskip

\begin{definition}\label{uniform_global_hyperbolicity_def}
    We say that the metrics $(g_j)_{j\in \N}$ are \emph{uniformly globally hyperbolic}, if for any two compact sets $K_1, K_2 \subset M$,
    $$\bigcup_{j=1}^\infty (J^-_{g_j}(K_1) \cap J^+_{g_j}(K_2)) \Subset M.$$ 
\end{definition}
In essence, this means that each Lorentzian metric $g_j$ is globally hyperbolic and that the compactness of the causal emeralds $J^-_{g_j}(K_1) \cap J^+_{g_j}(K_2)$ holds uniformly along the sequence.

We refer to Section \ref{SS:SuffUGH} for a sufficient condition ensuring uniform global hyperbolicity, and to Appendix \ref{subsec:counterexample_uniform_hyperbolicity} for an example of a sequence of smooth globally hyperbolic Lorentzian metrics $(g_j)_{j\in \mathbb{N}}$ on $\mathbb{R}^{d+1}$, $d \ge 2$, converging locally uniformly to a smooth globally hyperbolic limit $g_\infty$, which fails to be uniformly globally hyperbolic.
\smallskip

\begin{lemma}\label{separable_basis}
    Let $M$ be a smooth manifold and let $K \subset M$ be a compact subset. Then for each Borel measure $\mu$ with $\supp\, \mu \subset K$, there exists a countable basis for the $w^*$-open topology on $\mathcal{P}(K)$ containing $\mu$. 
\end{lemma}
\begin{proof}
Since every smooth manifold admits a Riemannian metric, and hence a corresponding Riemannian distance function, it follows that \( M \) is metrizable and therefore so is \( K \). A classical result from functional analysis then implies that \( C(K) \) is separable.  Moreover, it is well known that for any separable Banach space \( X \), the unit ball of the dual space, \( (B_{X^*}, w^*) \), is metrizable with respect to the weak\(^*\) topology. In particular, there exists a countable basis of \( w^* \)-open neighbourhoods of \( 0 \). Taking \( X = C(K) \) and recalling that its dual space \( C(K)^* \) coincides with the space of finite regular Borel measures, the claim follows by translation.
\end{proof}
\begin{lemma}\label{tau_uniform_equicontinuity}
    Assume the metrics $(g_j)_{j \in \N}$ are smooth and uniformly globally hyperbolic and converge locally uniformly to $g_\infty$, a continuous Lorentzian metric.
 Moreover,   assume that any two points $(x, y) \in M^2$ with $\tau_\infty(x, y) >0$ can be joined by a timelike and Lipschitz continuous $g_\infty$-maximizer (when parametrized with respect to arclength). Then the family $(\tau_j)_{j\in \N}$ is locally equicontinuous. 
\end{lemma}
\begin{proof}
    Assume by contradiction that the family $\{\tau_j\}_{j\in \N}$ is not locally equicontinuous. Then, there exists a compact set $K \subset M$, an $\e > 0$, and a sequence $(x_j, y_j, z_j, w_j) \in K^4$ such that
    \begin{align*}
        &\tau_j(x_j, y_j) > \tau_j(z_j, w_j) + \e, \\
        &\sfd_h^{M^2}((x_j, y_j), (z_j, w_j)) \to 0 \ \mathrm{as\ } j \to \infty.
    \end{align*}
 It follows that $\tau_j(x_j, y_j) > \e > 0$, hence there exists a smooth $g_j$-geodesic $\gamma_j:[0, c_j] \to M$ from $x_j$ to $y_j$.
 Denote 
 $$E:= \overline{\bigcup_{j\in \N \cup \{\infty\}} ( J^+_{{g}_j}(K) \cap J^-_{{g}_j}(K))}.$$ 
 The assumption on the uniform global hyperbolicity of $(g_j)_{j\in \N}$ implies the compactness of $E$.
  Thus, the local uniform $C^0$-bounds on $(g_j)_{j\in \N}$ and Proposition \ref{Lipschitz_bd_causal_curves} yield that  $(\tau_j)_{j\in \N}$ are uniformly bounded on $K^2$. Hence, by potentially passing to a subsequence, we may assume that  
 \begin{align}\label{convergence_to_length}
 \tau_j(x_j, y_j) \to \xi \geq \e > 0 \quad\text{ and } \quad     \tau_j(x_j, y_j) \geq \xi - \frac{1}{j}.
 \end{align}
 Note that all $\gamma_j$ are contained in the compact set $E$. We may reparametrise so that $\gamma_j$ has unit $|\cdot|_h$-speed. 
 Take $\Tilde{g}_k$ and $J_k$ as in Proposition \ref{appx_majorising_metrics} with the compact set $E$. It then follows that we can pick a subsequence $\gamma_{j_k} =:\gamma_k$ such that $\gamma_k$ is a $\Tilde{g}_k$-causal curve from $x_{j_k}=: x_k$ to $y_{j_k}=: y_k$ and $L_{\Tilde{g}_k}(\gamma_k) > L_{g_{j_k}}(\gamma_k)> \e$. 
 Using the arguments from the proof of Proposition \ref{Lipschitz_bd_causal_curves}, we get that there exists a $g_\infty$-causal limit curve $\gamma:[0,c]\to M$ such that (a subsequence of) $\gamma_k$ converges uniformly to $\gamma$.
 Note that $\liminf_j c_j \geq \frac{\xi}{\sup_j\norm{g_j}_K}$, hence $c >0$.
 Denote $\gamma(0)=:x \in K$ and $\gamma(c)=:y \in K$. Then $x_j \to x$, $y_j \to y$, $z_j \to x$ and $w_j \to y$.  \smallskip
 
 \textbf{Step 1:} We show that $\tau_\infty(x,y) \geq \xi$. \\  
From the proof of   \cite[Thm.\ 4.5]{samann2016global}, it follows that there exists a partition of unity $(\chi_k)_{k \in \N}$ of $M$ such that 
 \begin{align*}
     g' := \sum_k \chi_k\Tilde{g}_k
 \end{align*}
 satisfies that:
 \begin{itemize}
 \item $g' \succ g_\infty$; 
 \item $(M, g')$ is globally hyperbolic; 
 \item for $k \in \N$, the Lorentzian metrics
 \begin{align*}
     g'_k := \sum_{l \leq k} \chi_l \Tilde{g}_k +\sum_{l >k} \chi_l\Tilde{g}_l
 \end{align*}
 satisfy that $(M, g'_k)$ is globally hyperbolic;
 \item  for each compact set $B \subset M$, it holds $|\{k: \supp \, \chi_k \cap B \neq \emptyset\}| < \infty$. 
 \end{itemize}
It follows that $g'_{k+1} \prec g'_{k} \prec g'$.
 Denote $$E' := E \cup (J^+_{g'}(K)\cap J^-_{g'}(K)).$$
There exists  $k_1 \in \N$ such that for $k \geq k_1$, $\supp \, \chi_k \cap E' = \emptyset$.
Now we have that for each $k \geq k_1$, it holds $\gamma_k([0, c_k]) \subset E$, and $\Tilde{g}_k = g'_k$ on $E'$, hence $L_{g'_k}(\gamma_k) = L_{\Tilde{g}_k}(\gamma_k)$. 
Moreover, as $g'_k$ is globally hyperbolic, we get that the induced time-separation function $\tau'_k$ is continuous, by \cite[Thm.\ 3.28]{kunzinger2018lorentzian}. 
Moreover, by \cite[Lem.\ 2.32]{kunzinger2018lorentzian}, we have that for any $k \geq k_1$ and any $g'_k$-causal curve $\lambda$, it holds
\begin{align*}
    L_{\tau_k'}(\lambda) = L_{g'_k}(\lambda). 
\end{align*}
Hence, by our construction of $\Tilde{g}_k$ and $g'_k$, we get that, for each $l \geq k \geq k_1$:  $\gamma_l$ is $g'_k$-causal and 
\begin{align*}
    L_{\tau_k'}(\gamma_l) = L_{g'_k}(\gamma_l) = L_{\Tilde{g}_k}(\gamma_l) \geq L_{{g}_{j_l}}(\gamma_l) \geq \xi - \frac{1}{j_l}.
\end{align*}
In order to estimate $L_{\tau'_k}(\gamma)$, we pick a partition $0= a_0 < a_1 < \ldots < a_m  = c$ such that 
\begin{align*}
    \sum_{i=0}^m \tau'_k(\gamma(a_i), \gamma(a_{i+1})) \leq L_{\tau'_k}(\gamma) + \frac{1}{k} = L_{g'_k}(\gamma) + \frac{1}{k}.
\end{align*}
Since $\tau'_k$ is continuous and  $\gamma_l(a_i) \to \gamma(a_i)$, as $l \to \infty$ for all $i\in \N$ , we get that 
\begin{align*}
    \tau_k'(\gamma_l(a_i), \gamma_l(a_{i+1})) \to \tau'_k(\gamma(a_i), \gamma(a_{i+1})), \quad \text{as $l \to \infty$. }
\end{align*}
Hence,
\begin{align*}
    \sum_{i=0}^m  \tau_k'(\gamma_l(a_i), \gamma_l(a_{i+1})) \to \sum_{i=0}^m \tau'_k(\gamma(a_i), \gamma(a_{i+1})),\ \mathrm{as\ } l \to\infty.
\end{align*}
Note that 
\begin{align*}
     \sum_{i=0}^m  \tau_k'(\gamma_l(a_i), \gamma_l(a_{i+1})) \geq L_{\tau'_k}(\gamma_l) = L_{g'_k}(\gamma_l) \geq  L_{{g}_{j_l}}(\gamma_{j_l}), 
\end{align*}
which implies 
\begin{align*}
    L_{\tau'_k}(\gamma) +\frac{1}{k} \geq \sum_{i=0}^m \tau'_k(\gamma(a_i), \gamma(a_{i+1})) \geq \liminf_{l \to \infty}  L_{{g}_{j_l}}(\gamma_{j_l}) =\xi \geq \e >0. 
\end{align*}
Therefore, for all $k \geq k_0$:
\begin{align}
    \xi-\frac{1}{k} \leq L_{\tau'_k}(\gamma)=L_{g'_k}(\gamma) = \int_0^c \sqrt{g'_k(\dot{\gamma}(t),\dot{\gamma}(t))} \, dt = \int_0^c \sqrt{\Tilde{g}_k(\dot{\gamma}(t),\dot{\gamma}(t))} \, dt.
\end{align}
 Now, since $\Tilde{g}_k \to g_\infty$ uniformly on $E$, and $|\dot{\gamma}|_h \leq 1$, we get that 
\begin{align*}
    |\Tilde{g}_k(\dot{\gamma}(t),\dot{\gamma}(t)) - g_\infty(\dot{\gamma}(t),\dot{\gamma}(t))| \leq \norm{\Tilde{g}_k-g_\infty}|\dot{\gamma}(t)|^2_h \leq \norm{\Tilde{g}_k-g_\infty},
\end{align*}
for almost all $t \in [0,c]$. As the metrics $\Tilde{g}_k$ and $g_\infty$ are uniformly bounded on $E$, we can use the uniform continuity of $\sqrt{\cdot}$ together with the dominated convergence theorem, to conclude that 
\begin{align}\label{limit_curve_length}
    \lim_{k\to \infty} L_{g'_k}(\gamma) = L_{g_\infty}(\gamma) \geq \xi.
\end{align}
Since $\tau_\infty(x,y)\geq L_{g_\infty}(\gamma)$, the proof of  the first step is complete.
\smallskip

\textbf{Step 2:} We prove that there exists an open neighbourhood $U$ of $(x,y)$ and constants $J, L>0$ such that for any $(z, w)\in U$, there exists a $L$-Lipschitz continuous curve $\gamma_{z, w}$ from $z$ to $w$ with the following property.  For every $j \geq J$: $\gamma_{z, w}$ is $g_j$-causal  future directed, and satisfies $L_{g_j}(\gamma_{z, w}) \geq \xi -\frac{2\e}{3}$. \\

From step 1, we know that $\tau_\infty(x,y) \geq \xi$. From our assumptions, there exists a Lipschitz continuous curve $\lambda:[0,1] \to M$ from $x$ to $y$ that is $g_\infty$-maximising and parametrised by $g_\infty$-arclength. We get that $\lambda([0,1]) \subset E$ and that there exists a constant $L>0$ such that $\lambda$ is $L$-Lipschitz. Moreover, we get that almost everywhere on $[0,1]$, it holds
\begin{align*}
    -g_\infty(\dot{\lambda}, \dot{\lambda}) = \tau^2_\infty(x,y) \geq \xi^2.
\end{align*}
It  follows that 
\begin{align*}
    |\dot{\lambda}|_h^2 \geq \frac{\xi^2}{\sfd_M(0, g_\infty)} > 0,
\end{align*}
almost everywhere on $[0,1]$. Using Proposition \ref{uniform_bound_prosuct_with_time_or}, we get that there exists a constant $1 > \zeta > 0$ such that
\begin{align*}
    g_\infty(X, \dot{\lambda}) < -\frac{\zeta \xi}{\sqrt{\sfd_M(0, g_\infty)}},
\end{align*}
almost everywhere on $[0,1]$. Denote $B := \overline{B_1^h(E)}$ and let
\begin{align}
    G:= \sup_{j \in \N \cup \{\infty\}} \sfd_B(0,g_j) < \infty.
\end{align}
For $\kappa \in (0,1)$ small enough,   there exist coordinate patches $(U_x, \psi_x), (U_y, \psi_y)$ with $U_x \ni x, U_y \ni y$, such that $\lambda([0, \kappa]) \subset U_x$, $\lambda([1-\kappa, 1]) \subset U_y$. 
Up to  choosing smaller $U_x$, $U_y$ and $\kappa$, there exists   $\beta >1$ such that $(|\cdot|_{euc, U_x}, \sfd_{euc, U_x})$, (resp. $(|\cdot|_{euc, U_y}, \sfd_{euc, U_y}$)), and $|\cdot|_h$ are $\beta$-equivalent on $U_x$ (resp. $U_y$). Fix  $\delta \in (0, \xi)$.
By potentially further shrinking, we can achieve that there is an $\alpha \in (0,\frac{1}{\beta^2})$ such that for $p, q \in U_x$ (respectively both $p,q \in U_y$) with $\sfd_h(p,q) < \beta^2\alpha$, we have that 
\begin{align}\label{uniform_continuity_choice_alpha}
    \norm{g_\infty(p)-g_\infty(q)}_{euc} \leq \frac{\delta^2\zeta^2}{16\beta^4(1+G^2+L^2+ \sup_E\norm{X}^2_h)^2}.
\end{align}
Let $r \in (0, 1)$ such that $B^{euc}_r(\lambda([0, \kappa])) \subset U_x$ and $B^{euc}_r(\lambda([1-\kappa, 1])) \subset U_y$ and define 
\begin{align*}
    \theta := \min\Big(r^2, \frac{L}{2},\frac{\zeta^2\delta^2 \alpha^2}{16\beta^2(1+G^2+L^2+ \sup_E\norm{X}^2_h)^2}\Big).
\end{align*}
For $v \in B_\theta^{euc}(\psi_x(x))$ and $u \in B_\theta^{euc}(\psi_y(y))$, define $\lambda_{v,u}:[0,1] \to M$ via 
\begin{align*}
 \lambda_{v,u}(t) = \left\{ \begin{array}{ll}
           \psi_x^{-1}(\psi_x(\lambda(t)) + (\frac{\kappa}{3}-t)(v-\psi_x(x))) &\mathrm{for}\ t \leq \frac{1}{3}\kappa,  \\
          \lambda(t) &\mathrm{for}\ t \in [\frac{\kappa}{3}, 1-\frac{\kappa}{3}],\\
          \psi_y^{-1}(\psi_y(\lambda(t)) +(t- (1-\frac{\kappa}{3}))(u-\psi_y(y))) &\mathrm{for}\ t \geq 1-\frac{\kappa}{3}.
        \end{array}\right.
\end{align*}
Then $\lambda_{v,u}([0,1]) \subset B$ and $\lambda_{v,u}([0,\kappa]) \subset U_x$, $\lambda_{v,u}([1-\kappa,1]) \subset U_y$. Now for almost every $t \in [0,1]$, we have that 
\begin{align}\label{difference_perturbed_curve}
    &\sfd_{h}(\lambda(t), \lambda_{v,u}(t)) \leq \beta\theta < \alpha, \nonumber\\
    &|\dot{\lambda}(t) - \dot{\lambda}_{v,u}(t)|_{h} \leq \beta\theta.
\end{align}
Using \eqref{uniform_continuity_choice_alpha} and \eqref{difference_perturbed_curve}, we get that
\begin{align}\label{scalar_product_difference}
    &|g(\lambda(t))(\dot{\lambda}(t), \dot{\lambda}(t))- g(\lambda_{v,u}(t))(\dot{\lambda}_{v,u}(t), \dot{\lambda}_{v,u}(t))| \nonumber\\
    & \leq  |g(\lambda(t))(\dot{\lambda}(t), \dot{\lambda}(t))- g(\lambda(t))(\dot{\lambda}(t), \dot{\lambda}_{v,u}(t))| \\
    & \quad + |g(\lambda(t))(\dot{\lambda}(t), \dot{\lambda}_{v,u}(t))- g(\lambda(t))(\dot{\lambda}_{v,u}(t), \dot{\lambda}_{v,u}(t))| \nonumber \\
    & \quad +|g(\lambda(t))(\dot{\lambda}_{v,u}(t), \dot{\lambda}_{v,u}(t))- g(\lambda_{v,u}(t))(\dot{\lambda}_{v,u}(t), \dot{\lambda}_{v,u}(t))| 
 \nonumber\\
    & \leq G|\dot{\lambda}(t)|_h |\dot{\lambda}(t) - \dot{\lambda}_{v,u}(t)|_h + G|\dot{\lambda}_{v,u}(t)|_h |\dot{\lambda}(t) - \dot{\lambda}_{v,u}(t)|_h + \frac{\delta^2}{4L^2} |\dot{\lambda}_{v,u}(t)|_h^2 \leq \delta^2.
\end{align}
Similarly
\begin{align*}
    |g(\lambda(t))(X,\dot{\lambda}(t))- g(\lambda_{v,u}(t))(X,\dot{\lambda}_{v,u}(t))| \leq \frac{\zeta^2\delta^2}{4\beta^2(1+G^2+L^2+ \sup_E\norm{X}^2_h)}.
\end{align*}
Hence, if $\delta < \min(1, \xi)$, then, for almost every $t$: 
\begin{align*}
    g(\lambda_{v,u}(t))(X,\dot{\lambda}_{v,u}(t)) < -\frac{\zeta \xi}{2G} < 0.
\end{align*}
This gives that $\lambda_{v,u}$ is timelike and future directed. Moreover, using that $\sqrt{|a-b|} \geq |\sqrt{a}-\sqrt{b}|$ for all $a, b \geq 0$, we infer that \eqref{scalar_product_difference} implies:
\begin{align*}
    |L_{g_\infty}(\lambda) - L_{g_\infty}(\lambda_{v,u})| \leq \delta. 
\end{align*}
Choose 
\begin{align*}
    \delta \leq \frac{\e}{3}.
\end{align*}
Let  $J \in \N$ such that, for $j \geq J$:
\begin{align*}
    \norm{g_j-g_\infty}_{E} \leq \frac{\theta}{10}.
\end{align*}
By similar arguments as above, we get that each $\lambda_{v,u}$ is $g_j$-timelike, future directed and 
\begin{align*}
     |L_{g_j}(\lambda_{v,u}) - L_{g_\infty}(\lambda_{v,u})| \leq \delta. 
\end{align*}
Hence, we found an open neighbourhood $U:= \psi_x^{-1}(B^{euc}_{\kappa\theta/3}(\psi_x(x))) \times \psi_y^{-1}(B^{euc}_{\kappa\theta/3}(\psi_y(y)))$ of $(x,y)$ with the following property: for any $(z,w) \in U$ and $j \geq J$, there exists a $g_j$-timelike, future directed curve $\hat{\lambda}= \lambda_{v,u}$ with $\kappa v = 3(\psi_x(z)-\psi_x(x))$ and $\kappa u = 3(\psi_y(w)-\psi_y(y))$ such that 
$L_{g_j}(\hat{\lambda}) \geq \xi - \frac{2\e}{3}$. This proves the second claim.

Now for $j$, large enough, we have that $(z_j, w_j) \in U$ and $\tau_j(z_j, w_j) \leq \xi - \frac{5\e}{6}$. 
Step 2 gives:
\begin{align*}
   \xi -\frac{2\e}{3}\leq  \tau_j(z_j, w_j) \leq  \xi -\frac{5\e}{6}.
\end{align*}
That contradiction proves the lemma. 
\end{proof}
\begin{remark}\label{facts_about_limit_curves}
    The proof of Lemma \ref{tau_uniform_equicontinuity} implies that under the above assumptions: \begin{itemize}
        \item[(i)] For any sequence of $g_j$-timelike, $\tau_j$-distance maximising curves $\gamma_j$ with bounded Lipschitz constants, with $\tau_j(\gamma_j(0), \gamma_j(1))$ bounded away from zero and for which there exists a compact set $K$ such that \\ 
        $\bigcup_j\gamma_j([0,1]) \subset K$, there exists a subsequence $j_k$ such that the $\gamma_{j_k}$ converge uniformly to a $g_\infty$-causal curve $\gamma$ with $L_{g_\infty}(\gamma) = \lim_{k \to \infty} L_{g_{j_k}}(\gamma_{j_k})$.
        \item[(ii)] For any two points $(x, y)$ with $\tau_\infty(x,y) > 0$ and $\e > 0$, there exists a neighbourhood $U$ of $x, y$ and a $J_\e > 0$ such that for $j \geq J_\e$ and any $(z,w) \in U$, it holds $\tau_j(z,w) \geq \tau_\infty(x,y) - \e$. 
    \end{itemize}
\end{remark}
\medskip

\begin{proposition}
    \label{locally_uniform_convergence_tau}
Assume the metrics $(g_j)_{j \in \N}$ are smooth, uniformly globally hyperbolic and converge locally uniformly to $g_\infty$, a continuous Lorentzian metric.
 Moreover,   assume that any two points $(x, y) \in M^2$ with $\tau_\infty(x, y) >0$ can be joined by a timelike and Lipschitz continuous  $g_\infty$-maximizer (when parametrized with respect to arclength).
   Then, there is a subsequence of  $(\tau_j)_{j\in \N}$ that  converges to $\tau_\infty$ locally uniformly. In particular, $\tau_\infty$ is continuous. 
\end{proposition}

\begin{proof}
    Fix a compact set $K\subset M$. Lemma \ref{tau_uniform_equicontinuity} implies the equicontinuity of the family $(\tau_j)_{j\in \N}$ on $K^2$. 
    Moreover, the $\tau_j$ are uniformly bounded on $K^2$. Thus,  by Arzela--Ascoli's Theorem, there exists a subsequence of the $\tau_j$ that converges locally uniformly to a continuous function $\tau'$. We will next prove that $\tau_\infty = \tau'$. \\
    \textbf{Step 1:} $\tau_\infty \geq \tau'$. This was established in the first step of the proof of Lemma \ref{tau_uniform_equicontinuity}, cf.\ Remark \ref{facts_about_limit_curves} (i). \\
    \textbf{Step 2:} $\tau_\infty \leq \tau'$. Assume the contrary: then there exist $\e >0$ and $x,y \in K$ such that 
    \begin{align*}
        \tau_\infty(x,y) > \tau'(x,y) + \e \geq \e >0.
    \end{align*}
    By assumption, there exists a Lipschitz continuous timelike, future directed curve $\lambda$ from $x$ to $y$ with uniform $g_\infty$-speed. By similar arguments as in the second step of the proof of Lemma \ref{tau_uniform_equicontinuity}, we get that for $j$ large enough, $\lambda$ is $g_j$-timelike, future directed and
    \begin{align*}
        L_{g_j}(\lambda) \geq \tau_\infty(x,y) - \frac{\e}{2}. 
    \end{align*}
    Then $\tau_j(x,y) \geq \tau_\infty(x,y) - \frac{\e}{2}$, contradicting that $\tau'(x,y) = \lim_{j \to \infty} \tau_j(x,y) < \tau_\infty(x,y) -\e$.
\end{proof}

\begin{corollary} Under the same assumptions of Proposition \ref{locally_uniform_convergence_tau},
    for each compact set $K\subset M$ and each $\e> 0$ there exists  $J_\e \in \N$ such that:  
    \begin{align*}
        (\{\tau_\infty > 3\e\} \cap K^2)\subset (\{ \tau_j > 2\e \}\cap K^2)\subset (\{ \tau_\infty > \e\}\cap K^2), \quad \forall j > J_\e. 
    \end{align*} 
\end{corollary}

\subsubsection{Convergence of transport plans}

We begin with a basic lemma. Although its proof is standard, we include it for the sake of completeness.
\smallskip

\begin{lemma}\label{countable_cover}
    Let $(X, \sfd)$ be a separable metric space and $\Omega \subset X$ be an open subset. Then there exists a sequence of open sets $(\Omega_n)_n$ such that, for all $n \in \N$:  $\overline{\Omega_n} \subset \Omega$, $\dist(\overline{\Omega_n}, \partial \Omega) > 0$ and $\Omega = \bigcup_n \Omega_n$. 
\end{lemma}
\begin{proof}
    Choose a sequence $(q_n)_{\in \N}$ which is dense in $X$. Take a subsequence consisting only of those $q_n$'s  contained in $\Omega$.  As $q_n \in \Omega$ and $\Omega$ is open, $d_n := \dist(q_n, \partial \Omega)> 0$. Define
    \begin{align*}
        \Omega_n := B_{\frac{2d_n}{3}}(q_n). 
    \end{align*}
    Clearly $\dist(\overline{\Omega_n}, \partial \Omega) \geq \frac{d_n}{3}>0$. In order to see that $\Omega = \bigcup_n \Omega_n$, consider an arbitrary $p \in \Omega$. Since $\Omega$ is open, there exists an $R > 0$ such that $B_R(p) \subset \Omega$.
    By density  $(q_n)$ in $X$, we can find an $n$ such that $\sfd(q_n, p) < \frac{R}{6}$. It  follows that $B_{\frac{5R}{6}}(q_n) \subset \Omega$ and hence $d_n \geq \frac{5R}{6}$. 
    Then
    \begin{align*}
       p \in B_{\frac{R}{6}}(q_n) \subset B_{\frac{2}{3}\cdot\frac{5R}{6}}(q_n) = B_{\frac{5R}{9}}(q_n) \subset \Omega_n.
    \end{align*}
\end{proof}

\begin{lemma}\label{causal_limit_plan}
   Let $E \subset M$ be a compact subset and let $\pi_j \in \mathcal{P}(E^2)$ be such that $\pi_j(M^2_{\leq_j}) =1$. Assume that there exists   $\pi \in \mathcal{P}(E^2)$ such that $\pi_j \to \pi$ narrowly as $j \to \infty$. Then $\pi(M^2_{\leq_\infty})=1$. 
\end{lemma}
\begin{proof}
    We argue by contradiction. By the causal closedness, we have that $\Omega:= E^2 \setminus M^2_{\leq_\infty}$ is an open set --- in the subspace topology on $E^2$. As $M$ is separable, we can use Lemma \ref{countable_cover} to find a countable open cover $(B_m)_{m \in \N}$ of $\Omega$ such that $$\overline{B}_m \subset \Omega \text{ and } \dist(B_m, \partial \Omega)>0, \quad \forall m\in \N.$$
    We  argue by contradiction. Suppose that there exists  $\delta>0$ such that $\pi(M^2_{\leq_{\infty}}) = 1- \delta < 1.$  Then 
 \begin{align*}
     \pi(E \cap (M^2_{\leq_{\infty}})^c) = \delta. 
 \end{align*}
 Hence, there exists $m \in \N$ such that, setting $B:=B_m$, it holds
 \begin{align*}
     \pi(B) > 0.
 \end{align*}
 By lower semi-continuity of the measure of open sets under narrow convergence, we infer that there exists $j_0>0$ such that: 
 \begin{align*}
     \pi_j(B) > 0, \quad \forall j\geq j_0. 
 \end{align*}
 As $\pi_j$ is supported on $M^2_{\leq_{j}}$, we get that  $M^2_{\leq_{j}} \cap B \neq \emptyset$, for all $j\geq j_0$. Pick $(x_j, y_j) \in M^2_{\leq_{j}} \cap B$ and let $\gamma_j$ be a $g_{j}$-causal curve  from $x_j$ to $y_j$ with bounded $h$-speed. 
 By Proposition \ref{Lipschitz_bd_causal_curves} and \cite[Thm.\ 3.1]{minguzzi2008limit}, we obtain a subsequence $(\gamma_{j_k})_{k\in \N}$ that  converges uniformly to the a $g_\infty$-causal curve $\gamma$.
 Denote $x:= \gamma(0)$ and $y := \gamma(1)$. Then $(x, y) \in \overline{B} \subset (M^2_{\leq_\infty})^c$, which contradicts $x \leq_\infty y$. 
\end{proof}

\begin{proposition}\label{dynamical_plans_precompactness}
Let $g_j$ and $g_\infty$ satisfy the assumptions of Proposition \ref{locally_uniform_convergence_tau}, and that any length-maximising curve for $g_\infty$ has a causal character.
    Let $\mu_0, \mu_1 \in \mathcal{P}_c(M)$ be $(M, g_\infty)$-strongly timelike $p$-dualisable. 
Assume that there exist a compact set $K\subset M$ and a sequence of measures $\mu^j_0, \mu^j_1 \in \mathcal{P}_c(K)$, with $\Pi_{\leq_j}(\mu^j_0, \mu^j_1) \neq \emptyset$ that narrowly converge to $\mu_0, \mu_1$ and such that $\ell_p^j(\mu^j_0, \mu^j_1) \to \ell_p^\infty(\mu_0, \mu_1)$.
Let $\eta_j\in \mathrm{OptGeo}_{\ell_{p,j}}(\mu^j_0, \mu^j_1)$ be $p$-optimal dynamical plans from $\mu^j_0$ to $\mu^j_1$ with respect to $\tau_j$.
Then the sequence $(\eta_j)_j$ has a subsequence that narrowly converges to a measure $\eta \in \mathcal{P}(C([0,1], K))$ with $(\ee_0, \ee_1)_\#\eta \subset \{\tau_\infty > 0\}$.
\end{proposition}
\begin{proof}
Let $K \subset M$ be a compact subset such that  $\supp\, \mu^j_0, \supp\, \mu^j_1, \supp\, \mu_0, \supp\, \mu_1 \subset K$, for all $j\in \N$. 
Combining the uniform global hyperbolicity with  \cite[Lem.\ 1.15]{cavalletti2020optimal} and Prokhorov's theorem, we infer that a subsequence of $$\pi_j:=(\ee_0, \ee_1)_\# \eta_j$$ converges narrowly to  some measure $\pi\in \mathcal{P}(M^2)$.  Lemma \ref{causal_limit_plan} together with Proposition \ref{locally_uniform_convergence_tau}, imply that $\pi \in \Pi^{p\text{-opt}}_{\leq_\infty}(\mu_0, \mu_1)$. As $\mu_0, \mu_1$ is strongly timelike $p$-dualisable, we get that $\pi(\{\tau_\infty > 0\})= 1$. 
\\The lower semicontinuity of the measure of open sets under narrow convergence implies that $$
    1 \geq \liminf_{j \to \infty} \pi_j(\{\tau_\infty > 0\}) \geq, 1$$ forcing
    $$
    \lim_{j\to \infty}  \pi_j(\{\tau_\infty > 0\}) = 1.
    $$
    We next show that the sequence $(\eta_j)_j\subset \mathcal{P}(C([0,1], K))$ is precompact in the narrow topology. To this aim, for all $\e >0$ and $k\in \N$ define the sets 
\begin{align}\label{def:eps_steep_geodesics}
    T^k_\e := \{ \gamma \in \mathrm{TGeo}^k \mid \gamma(0), \gamma(1) \in K^2, \tau_k(\gamma(0),\gamma(1))\geq \e\},
\end{align}
and 
$$
T_\e := \bigcup_{k<\infty}T_\e^k.
$$
We will show that the measures \( \eta_j \) are almost concentrated on \( T_\varepsilon \) (see Claim~1), and that \( T_\varepsilon \subset C([0,1],K) \) is precompact (see Claim~2). 
These properties imply that the sequence \( (\eta_j)_j \subset \mathcal{P}(C([0,1],K)) \) is tight and therefore precompact by Prokhorov's theorem.
\smallskip

\textbf{Claim 1:} $\limsup_j \eta_j(T_\e) \to 1$, as $\e \to 0$.  \\
Assume the contrary. Then there exists  $\delta > 0$ and a sequence $\e_j \to 0$, as $j \to \infty$, such that:
\begin{equation}\label{eq:etajTepsjDelta}
    \eta_j\Big( C([0,1], M) \setminus T_{\e_j} \Big) > \delta. 
\end{equation}
Since $\eta_j(\mathrm{TGeo}^j) =1$, \eqref{eq:etajTepsjDelta} implies that 
\begin{align*}
    \pi_j\Big(\bigcap_{k < \infty} \{ \tau_k < \e_j\}\Big) > \delta. 
\end{align*}
Since $\pi(\{\tau_\infty >0\}) =1$, there exists  $\e_0>0$ such that 
\begin{align}\label{small_tau_small_set}
    \pi(\{\tau_\infty\leq \e_0\}) \leq \frac{\delta}{2}.
\end{align}
For $k_0>0$ large enough, Proposition \ref{locally_uniform_convergence_tau} implies that
\begin{align*}
    |\tau_\infty-\tau_k| \leq \frac{\e_0}{3}, \quad \text{ on } K, \quad \forall  k\geq k_0. 
\end{align*}
It follows that 
\begin{align*}
    \{\tau_\infty < \e_0\} \supset \bigcap_{k_0\leq k< \infty} \Big\{ \tau_k < \frac{\e_0}{2} \Big\}.
\end{align*}
For $j_0>0$ large enough, it holds $\e_j \leq \frac{\e_0}{2},$ for all $j \geq j_0$. Thus, for all $j \geq j_0$: 
\begin{align*}
    \delta <  \pi_j\Big(\bigcap_{k < \infty} \{ \tau_k < \e_j\}\Big) \leq  \pi_j\Big(\bigcap_{k < \infty} \Big\{ \tau_k < \frac{\e_0}{2}\Big\}\Big) \leq \pi_j\Big(\bigcap_{k_0 \leq k < \infty} \Big\{ \tau_k < \frac{\e_0}{2}\Big\} \Big) \leq \pi_j( \{\tau_\infty \leq \e_0\}).
\end{align*}
By \eqref{small_tau_small_set}, the continuity of $\tau_\infty$ (see Proposition \ref{locally_uniform_convergence_tau}) and the upper semicontinuity of the measure of closed sets under narrow convergence, we get
\begin{align*}
    \delta \leq \limsup_{j \to \infty} \pi_j( \{\tau_\infty \leq \e_0\}) \leq \pi(\{\tau_\infty \leq \e_0\}) \leq \frac{\delta}{2}.
\end{align*}
This contradiction and  proves the first claim. \smallskip

\textbf{Claim 2:} $T_\e$ is precompact in $C([0,1], M)$. \\
Fix $\e>0$ and take a sequence $\gamma_k, j_k$ such that $\gamma_k \in T^{j_k}_\e$. If $j_k$ has an accumulation point $j < \infty$, we get a sequence of geodesics contained inside a compact set $K$ with respect to the smooth metric $g_j$. In that case set $j_\infty :=j$. Otherwise we can find a subsequence $j_k$ that is unbounded and strictly increasing, in which case we set $j_\infty = \infty$.
We may assume that $\xi := \lim_{k \to \infty} \tau_{j_k}(\gamma_k(0), \gamma_k(1)) \geq \e$ exists. 
Write $\Tilde{\gamma}_k$ for the curve obtained by reparametrising $\gamma_k$ by $h$-arclength. 
From the arguments of the first step of the proof of Lemma \ref{tau_uniform_equicontinuity}, we get a $g_{j_\infty}$-causal limit curve $\Tilde{\gamma}$ such that 
$$L_{g_{j_\infty}}(\Tilde{\gamma}) \geq \lim_{k \to \infty} \tau_{j_k}(\gamma_k(0), \gamma_k(1))= \xi.$$ 
From the locally uniform convergence of $\tau_{j_k}$ to $\tau_{j_\infty}$ (Proposition \ref{locally_uniform_convergence_tau}), we infer that $\Tilde{\gamma}$ is  maximising length. 

By our assumption on $g_\infty$, it follows that $\Tilde{\gamma}$ is timelike.
By slightly scaling the parametrization, we may assume that the $\Tilde{\gamma}_k$'s and $\Tilde{\gamma}$ are all defined on the same interval $[0,c]$.
For $k \in \N \cup \{\infty\}$, define $\psi_k:[0, c] \to [0, \infty)$ via
\begin{align*}
    \psi_k(t) := \int_0^t \sqrt{-g_{j_k}(\dot{\Tilde{\gamma}}_{k}, \dot{\Tilde{\gamma}}_{k})(s)} \, ds.
\end{align*}
Each $\psi_k$ is strictly increasing, and
\begin{align*}
    \psi_k(t) = \tau_{j_k}(\Tilde{\gamma}_k(0),\Tilde{\gamma}_k(t)) \to \tau_{j_\infty}(\Tilde{\gamma}(0),\Tilde{\gamma}(t)),
\end{align*}
as $k \to \infty$. Since $\Tilde{\gamma}$ is maximising, we have that 
\begin{align*}
    \tau_{j_\infty}(\Tilde{\gamma}(0),\Tilde{\gamma}(t)) = L_{g_{j_\infty}}(\Tilde{\gamma}|_{[0,t]}).
\end{align*}
Proposition \ref{locally_uniform_convergence_tau} with the uniform convergence $\Tilde{\gamma}_k\to \Tilde{\gamma}$,  gives that 
$$\psi_k \to \psi_\infty=: \psi, \quad  \text{uniformly as $k \to \infty$}.
$$
In particular,  $(\psi_k)_k$ are uniformly bounded by a constant $B>0$.
Note that $\gamma_k = \Tilde{\gamma}_k \circ \psi^{-1}_k$. Similarly, $\gamma := \Tilde{\gamma} \circ \psi^{-1}$ is parametrised by $\tau_{j_\infty}$-length. In order to prove that  $\gamma_k \to \gamma$ uniformly, it suffices to show that $\psi_k^{-1} \to \psi^{-1}$ uniformly. 
\smallskip

\textbf{Uniform convergence on inner intervals:} We start with intervals  $[\theta, \max \psi - \theta]$ for some $\theta >0$. 
If  uniform convergence fails, then there exist $\epsilon >0$ and a sequence $x_k \in [\theta, \max \psi - \theta]$ such that
\begin{align*}
    |\psi_k^{-1}(x_k)-\psi^{-1}(x_k)| > \epsilon.
\end{align*}
By compactness of $[\theta, \max \psi - \theta]$ and the continuity of $\psi^{-1}$, we get that there exists $x \in [\theta, \max \psi - \theta]$ and a non-relabeled subsequence such that $x_k \to x$ and
\begin{align*}
    |\psi_k^{-1}(x_k)-\psi^{-1}(x)| > \frac{\epsilon}{2}.
\end{align*}
We may assume that $\epsilon < \theta$.
Note that 
\begin{align*}
\psi_k^{-1}(x_k) =\inf \{y : \psi_k(y) \geq x_k \}, \quad y^* := \psi^{-1}(x)= \inf \{ y: \psi(y) \geq x \}.
\end{align*}
\underline{Case 1}: 
\begin{align*}
     \psi_k^{-1}(x_k)-\psi^{-1}(x) \geq \frac{\epsilon}{2}.
\end{align*}
Then $\psi_k \leq x_k$ on $[y^*, y^*+\frac{\epsilon}{3}]$. Since $\psi$ is strictly increasing, there exists  $\delta >0$ such that $\psi( y^*+\frac{\epsilon}{3}) > x+ 2\delta$. Moreover,  for $k$ large enough, we have that $x_k \leq x + \delta$. Then: 
\begin{align*}
   x+ 2\delta< \psi\Big( y^*+\frac{\epsilon}{3}\Big) = \lim_{k \to \infty} \psi_k\Big(y^*+\frac{\epsilon}{3}\Big) \leq x+ \delta,
\end{align*}
a contradiction. \\
\underline{Case 2}:
\begin{align*}
     \psi^{-1}(x) > \frac{\epsilon}{2} + \psi_k^{-1}(x_k).
\end{align*}
Similarly to case 1, we infer that $\psi(y^* - \frac{\epsilon}{3}) < x- 2\delta$, for some $\delta >0$. However, for $k$ large, we have that $x_k > x- \delta$. Then
\begin{align*}
    x- \delta \leq \lim_{k \to \infty} \psi_k\left(y^* - \frac{\epsilon}{3}\right) = \psi\left(y^* - \frac{\epsilon}{3}\right) <x-2\delta,
\end{align*}
again a contradiction. \\
Hence, $\psi^{-1}_k \to \psi^{-1}$ uniformly on $[\theta, \max \psi - \theta]$ for any $\theta>0$.\smallskip

\textbf{Convergence around zero:} In order to prove convergence around zero, fix an $\epsilon >0$ and define $$p := p_\epsilon := \psi(2\epsilon) > 0.$$ By the discussion above, we know that  $\psi_j^{-1}\to \psi^{-1}$ uniformly on $[\frac{p}{2}, \max \psi - \frac{p}{2}]$. Thus there exists  $J_\epsilon>0$ such that 
$$\psi_j^{-1}(p) \in (\epsilon, 3\epsilon), \quad \forall j \geq J_\epsilon.$$ 
The strict monotonicity of $\psi_j, \psi$ and the fact that $\psi_j(0)= \psi(0)=0$ for all $j$, imply that 
$$|\psi^{-1}(x)-\psi^{-1}_j(x)| \leq 3 \epsilon, \quad \forall j \geq J_\epsilon, \; \forall x\in [0,p].$$ \smallskip

\textbf{Convergence around max $\boldsymbol{\psi}$:}
Finally, we prove uniform convergence in a neighbourhood of $m := \max \psi$. Fix  $\epsilon >0$. Let $J_\epsilon^1>0$ be such that  
$$|\max \psi_j - m| \leq \epsilon, \quad \forall j \geq J_\epsilon^1.$$ 
Denote $q_\epsilon:= \psi(c-2\epsilon) < m$. 
By the above discussion, there exists  $J_\epsilon^2>0$ such that
$$|\psi^{-1}(q_\epsilon)-\psi_j^{-1}(q_\epsilon)| \leq \epsilon, \quad \forall j \geq J_\epsilon^2.$$ 
The monotonicity of $\psi_j, \psi$ imply that 
$$|\psi^{-1}(x)-\psi^{-1}_j(x)| \leq 3 \epsilon,\quad \forall j \geq \max(J_\epsilon^1, J_\epsilon^2), \quad \forall x \in [q_\epsilon,m].$$ \smallskip
The three observations above imply that $\psi_j^{-1}\to \psi^{-1}$ uniformly, thus completing the proof of claim 2. \\
To conclude note that claim 1 yields tightness of the sequence $(\eta_j)_j$: Fix $m \in \N$ and choose $\e =\e_m > 0$ such that $\limsup_j \eta_j(T_\e) \geq 1-\frac{1}{2m}$. Then there exists  $j_m \in \N$ such that  
\begin{align*}
    \eta_j(C([0,1],M)\backslash T_\e) \leq \frac{1}{m}, \quad \forall j \geq j_m. 
\end{align*}
Since the set $\{ \eta_j, j \leq j_m\}$ is finite and hence tight, there exists a precompact set $S_m \subset C([0,1], M)$ such that  
\begin{align*}
    \eta_j(C([0,1], M) \setminus S_m) \leq \frac{1}{m}, \quad \forall j \leq j_m.
\end{align*}
It follows that 
\begin{align*}
     \eta_j(C([0,1], M) \setminus (S_m \cup T_{\e_m})) \leq \frac{1}{m}, \quad \forall j \in \N,
\end{align*}
yielding the tightness of $(\eta_j)_j$.  
The lemma now follows from Prokhorov's theorem. 
\end{proof}
\begin{remark}\label{limit_of_geodesics_is_geodesic}
    Together with Remark \ref{facts_about_limit_curves} (i), the above proof gives that for any $\e > 0$ and any sequence of $g_j$-geodesics $\gamma_j$ with $\tau_j(\gamma_j(0), \gamma_j(1)) \geq \e$ there exists a subsequence $j_k$ such that  $(\gamma_{j_k})_k$ uniformly converge to a $g_\infty$-maximising curve that is parametrised by unit $g_\infty$-speed and of length $\lim_{k \to \infty} L_{g_{j_k}}(\gamma_{j_k})$. 
\end{remark}
\smallskip

\begin{lemma}\label{openness_of_complement_of_tl_geos}
    Let $\e >0$.  The set $(\mathrm{TGeo}^\infty)^c  \cap (\ee_0, \ee_1)^{-1}(\{\tau_\infty > \e\}) \subset C([0,1], M)$ is open. 
\end{lemma}
\begin{proof}
    Pick $\gamma \in (\mathrm{TGeo}^\infty)^c \cap (\ee_0, \ee_1)^{-1}(\{\tau_\infty > \e\})$ and denote $K := B_1(\gamma([0,1])) \subset M$. Since $K$ is precompact,  $\tau_\infty$ is uniformly continuous on $K^2$. Denote 
    \begin{align*}
        \tau_\infty (\gamma_0, \gamma_1)- \e = \alpha > 0.
    \end{align*}
    As $\gamma$ is not a timelike geodesic, there exist $0 \leq s < t \leq 1$ and a $\beta >0$ such that 
    \begin{align*}
       |\tau_\infty(\gamma_s, \gamma_t)-(t-s)\tau_\infty (\gamma_0, \gamma_1)|= |\tau_\infty(\gamma_s, \gamma_t)-(t-s)(\alpha+\e)| = \beta.
    \end{align*}
    Since $\tau_\infty$ is uniformly continuous on $K^2$, there exists $\delta > 0$ such that $$|\tau_\infty(x,y)-\tau_\infty(z,w)| \leq \min\Big(\frac{\beta}{3}, \frac{\alpha}{2}, \frac{\e}{2}\Big),$$
    for all $x, y, z, w \in K$ with $\sfd(x,z)+\sfd(y,w) \leq \delta$.
    \\ Next, let $\lambda \in  C([0,1], M)$ be such that $\sfd_{C^0}(\lambda, \gamma) \leq \frac{\delta}{2}$.
    Then
    \begin{align*}
        &|\tau_\infty(\lambda_0, \lambda_1) - \alpha - \e| \leq \min\Big(\frac{\beta}{3}, \frac{\alpha}{2}\Big),\ \mathrm{and} \\
        &||\tau_\infty(\lambda_s, \lambda_t) - (t-s)(\alpha+\e)| - \beta| \leq \frac{\beta}{3}.
    \end{align*}
    Then
    \begin{align*}
        |\tau_\infty(\lambda_s, \lambda_t)-(t-s)\tau_\infty(\lambda_0, \lambda_1)| \geq  \frac{\beta}{3}. 
    \end{align*}
    Hence, $B^{C^0}_\frac{\delta}{2}(\gamma) \subset (\mathrm{TGeo}^\infty)^c  \cap (\ee_0, \ee_1)^{-1}(\{\tau_\infty > \e\})$.
\end{proof}
\begin{lemma}\label{limit_plan_is_dynamical}Let $g_j, g_\infty, \eta_j$ and $\eta$ be as in Proposition \ref{dynamical_plans_precompactness}. Then $\eta$  is concentrated on $\mathrm{TGeo}^\infty$.
\end{lemma}
\begin{proof}
    As we have seen in the proof of Proposition \ref{dynamical_plans_precompactness}, we have that $\eta((\ee_0, \ee_1)^{-1}(\{\tau_\infty > 0\})) = 1$. 
    Suppose by contradiction that there exists  $\delta > 0$ such that 
    \begin{align*}
        \eta (\mathrm{TGeo}^\infty) \leq 1- 2\delta.
    \end{align*}
    As in the beginning of the proof of Proposition \ref{dynamical_plans_precompactness},  there exists $\e > 0$ such that 
    \begin{align*}
        \eta((\ee_0, \ee_1)^{-1}(\{\tau_\infty > \e\})) \geq 1- \delta.
    \end{align*}
    Then 
     \begin{align*}
        \eta((\mathrm{TGeo}^\infty)^c  \cap (\ee_0, \ee_1)^{-1}(\{\tau_\infty > \e\})) \geq \delta.
    \end{align*}
    By Lemma \ref{openness_of_complement_of_tl_geos}, we know that $\Omega:= (\mathrm{TGeo}^\infty)^c  \cap (\ee_0, \ee_1)^{-1}(\{\tau_\infty > \e\})$ is open and by Lemma \ref{countable_cover}, we can find a countable cover $\overline{B_n} \subset \Omega$ of $\Omega$ such that for each $n$, $\overline{B_n}$ has positive distance to $\partial \Omega$.
    By the $\sigma$-subadditivity of measures,  there exists $n\in \N$ such that $B_n =:B$ satisfies $\eta(B) >0$. 
    By the lower semicontinuity of the measure on open sets under narrow convergence, we get that 
    \begin{align*}
        \liminf_{j \to \infty} \eta_j(B) \geq \eta(B) >0. 
    \end{align*}
    As $\eta_j(\mathrm{TGeo}^j) = 1$, then $\mathrm{TGeo}^j \cap B \neq \emptyset$, for $j$ large enough. Recall also that there exists a compact set $E \subset M$ such that $\eta_j(C([0,1], E)) = 1$,  for all $j$. 
    Hence we can pick $\gamma_j \in C([0,1], E) \cap \mathrm{TGeo}^j \cap B $.
    Now, as in the proof of Proposition \ref{dynamical_plans_precompactness} (see also Remark \ref{limit_of_geodesics_is_geodesic}), we get that a subsequence of $(\gamma_j)_j$ uniformly converges to a timelike $g_\infty$-geodesic $\gamma$. But then $\gamma \in \overline{B} \subset \Omega$,  contradicting that  $\Omega \subset (\mathrm{TGeo}^\infty)^c$. 
\end{proof}

\subsection{Stability of the $\tcd$-condition under locally uniform convergence of Lorentzian metrics}

In this subsection, we will use the previous results to prove the stability of the timelike curvature dimension condition. This requires one more preparatory lemma. 

\begin{lemma}\label{restriction_of_plans_converge}
    Let $\pi_j, \pi \in \mathcal{P}(M^2)$ such that $\pi_j \to \pi$ narrowly as $j\to \infty$ and $\pi(\{ \tau_\infty > 0\})=1$. For $R> 0$ define $\Omega_R:= \{ \tau_\infty > \frac{1}{R}\}$. Then there exists a sequence $R_k \in (0, \infty)$ with $\lim_{k \to \infty} R_k = \infty$ with the following properties:
    \begin{itemize}
    \item    There exists  $k_0 \geq 0$ such that $\pi(\Omega_{R_k})>0$, for all $k \geq k_0$;
    \item $\pi_j(\Omega_{R_k}) \to  \pi(\Omega_{R_k})\ \mathrm{as} \ j \to \infty$, for all $k\in \N$;
    \item $\pi_j(\Omega_{R_k})>0$, for all $j\geq J_{k_0}$, and 
 
   \begin{align}\label{restriction_convergence_lemma} 
    \pi_j(\Omega_{R_k})^{-1} \pi_j \LL \Omega_{R_k} \to  \pi(\Omega_{R_k})^{-1} \pi \LL \Omega_{R_k} \ \mathrm{narrowly\ as\ } j \to \infty, \mathrm{ for\ all\ }k\geq k_0.
    \end{align}
    \end{itemize}
\end{lemma}
\begin{proof}
    \textbf{Step 1:}
    We first notice that for any measure $\mu \in \mathcal{P}(M^2)$, the set 
    $$Q^\mu:=\{R>0, \mu(\partial \Omega_R)>0\}$$ 
    is countable. Indeed, $\partial \Omega_R= \{\tau_\infty = \frac{1}{R}\}$, so for $R \neq R'$, it holds $\partial \Omega_R \cap \partial \Omega_{R'} = \emptyset$.    Thus
    \begin{align*}
        1 \geq \mu(\{\tau_\infty>0\}) = \mu\left(\bigcup_{R>0} \partial \Omega_R\right) = \sum_{R>0} \mu(\partial \Omega_R) \geq 0,
    \end{align*}
    so  only countably many $R$ can satisfy $\mu(\partial \Omega_R)>0$.
    \smallskip
    
    \textbf{Step 2:} There exists a sequence $R_K \in (0, \infty)$ with $\lim_{k \to \infty} R_k = \infty$ and such that  $\Omega_{R_k}$ is a continuity set of $\pi_j$, for all $j\in \N$, and of $\pi$. 

    This follows from Step 1, by observing that $\mathcal{L}^1(Q^{\pi_j}) = 0$, $\mathcal{L}^1(Q^\pi) = 0$ and hence 
    \begin{align*}
        \mathcal{L}^1 \Big( Q^\pi \cup \bigcup_{j \in \N} Q^{\pi_j} \Big)=0. 
    \end{align*}
    Thus, there exists a strictly increasing sequence 
    $R_k \in (0, \infty) \setminus \Big( Q^\pi \cup \bigcup_{j \in \N} Q^{\pi_j} \Big)$.

    \textbf{Step 3:} Conclusion.

    As $\pi(\{ \tau_\infty > 0\})=1$, and $\{ \tau_\infty > 0\}= \bigcup_{k\in \N} \Omega_{R_k}$, we get that 
 \begin{align*}
     &  \lim_{k \to \infty} \pi(\Omega_{R_k}) =  \pi(\{\tau_\infty > 0\}) = 1.
 \end{align*}
As $\Omega_{R_k}$ is a continuity set for $\pi$,  we get that
 \begin{align}\label{setwise_conv_omega}
  & \pi_{j}(\Omega_{R_k}) \to \pi(\Omega_{R_k}) \quad \mathrm{as} \ j \to \infty, \quad \forall k\in \N.  
 \end{align} 
 Moreover, for every open set $E \subset M^2$:
 \begin{align}\label{lower_semi_1}
   \liminf_{j \to \infty} \pi_j \LL \Omega_{R_k} (E) = \liminf_{j \to \infty} \pi_j(E \cap \Omega_{R_k}) \geq \pi(E \cap \Omega_{R_k}) = \pi \LL \Omega_{R_k} (E).
 \end{align}
The combination of   \eqref{setwise_conv_omega} and \eqref{lower_semi_1} yields \eqref{restriction_convergence_lemma}.
\end{proof}
\smallskip

The following result is inspired by \cite[Thm.~1.5]{cavalletti2020optimal}, where the stability of synthetic timelike Ricci lower bounds was established under a Lorentzian version of pointed measured Gromov convergence (see \cite{GMS-PLMS} for the corresponding metric framework), relying on isometric embeddings into a common Lorentzian pre-length space.  In contrast, our approach replaces the assumption of such isometric embeddings by requiring locally uniform convergence of the Lorentzian metrics. This perspective will be crucial in the subsequent sections, where we show that impulsive gravitational waves satisfy synthetic timelike Ricci lower bounds.
\smallskip

\begin{theorem}\label{stability_tcd}
    Let $M$ be a smooth manifold and $g_\infty$ be a continuous globally hyperbolic Lorentzian metric on $M$  such that every length-maximising causal curve has a causal character. For all $j \in \N$, let $g_j$ be a smooth Lorentzian metric on $M$ satisfying the following assumptions:
    \begin{itemize}
   \item  $g_j \to g_\infty$ locally uniformly;
    \item  $(M, g_j)$ are uniformly globally hyperbolic (see Definition \ref{uniform_global_hyperbolicity_def});
    \item there exist $p \in (0,1)$, $N \in (0, \infty)$ and $K \in \R$ such that the measured Lorentzian length spaces induced by $(M, g_j)$ satisfy the $\tcd^{e}_p(K, N)$-condition. 
   \end{itemize}
   Then $(M, g_\infty)$ satisfies the $\wtcd^{e}_p(K, N)$-condition.
\end{theorem}
\begin{proof}
    For $j \in \N \cup \{\infty\}$, denote by $\mm_j$ the volume measure on $M$ induced by $g_j$.
    Fix 
    \begin{equation}\label{eq:defmuiinf}
    \mu_0^\infty, \mu_1^\infty \in \mathrm{Dom}(\Ent)(\cdot |\mm_\infty) \cap \mathcal{P}_c(M),
    \end{equation}
    strongly timelike $p$-dualisable and let $\pi_\infty \in \Pi^{p\text{-opt}}_{\ll_\infty
    }(\mu_0^\infty, \mu_1^\infty)$. Since $\mu_0^\infty$ and $\mu_1^\infty$ have compact support and the metrics $(g_j)_{j\in \N}$ are uniformly globally hyperbolic,  there exists a compact set $E \subset M$ such that:
    \begin{align*}
        B_1^h\Big(\bigcup_{j \in \N \cup \{\infty\}} (J_j^+(\supp\, \mu_0^\infty) \cap J_j^-(\supp\, \mu_1^\infty))\Big)\subseteq E.
    \end{align*}
 By Proposition \ref{locally_uniform_convergence_tau}, up to a subsequence, we have that
 $$
 \tau_j \to \tau_\infty\quad  \text{uniformly on } E.
 $$
 Up to replacing $E$ with a larger compact set so that $\mm_\infty(\partial E)=0$, we may assume that  $$\mm_\infty(E) = \lim_{j \to \infty} \mm_j(E).$$ Then the probability measures $\Tilde{\mm}_j:= \mm_j(E)^{-1} {\mm_j} \LL E$ narrowly converge to $\Tilde{\mm}_\infty:= \mm_\infty(E)^{-1} {\mm_\infty} \LL E$. For simplicity, we will write $\mm_j$ for $\Tilde{\mm}_j$ and $\mm_\infty$ for $\Tilde{\mm}_\infty$. \smallskip

 \textbf{Step 1:} We prove that, up to a subsequence, for every $j \in \N$ there exists $(\mu^{j}_0, \mu^j_1) \in \mathcal{P}(M)^2$ $g_j$-timelike $p$-dualisable, such that
 \begin{align}\label{step1}
     \mu^j_0 \to \mu^\infty_0, \ \mu^j_1 \to \mu^\infty_1 \ \mathrm{narrowly \ in \ } M \ \mathrm{and\ } \ell_p(\mu_0^j, \mu_1^j) \to \ell_p(\mu_0^\infty, \mu_1^\infty), \ \mathrm{as}\ j \to \infty.
 \end{align}
 Since $E$ is compact, the narrow convergence $\mm_j \to \mm_\infty$ implies $\mm_j \to \mm_\infty$ in $W_2^{(M, \sfd_h)}$. Let
 \begin{align}
     \boldsymbol{\gamma}_j \in \Pi(\mm_\infty, \mm_j) \ \mathrm{be \ a }\ W_2^{(M, \sfd_h)}\mathrm{-optimal\ coupling.}
 \end{align}
 Using \cite[Lemma 3.16]{cavalletti2020optimal},  we can approximate $\pi_\infty$ by 
 \begin{align}\label{pi_n_infty_all_convergences}
     &\pi_{\infty, n} = \rho_{\infty, n} \mm_\infty \otimes \mm_{\infty},\ \rho_{\infty, n} \in L^{\infty}(\mm_\infty \otimes \mm_{\infty}), \ \pi_{\infty, n}(\{\tau_\infty >0\})=1, \ \pi_{\infty, n} \to \pi_\infty \ \mathrm{narrowly}, \nonumber \\
     &\lim_{n \to \infty} \Ent((P_1)_\#\pi_{\infty, n}| \mm_\infty) = \Ent(\mu_0^\infty|\mm_\infty), \ \lim_{n \to \infty} \Ent((P_2)_\#\pi_{\infty, n}| \mm_\infty) = \Ent(\mu_1^\infty|\mm_\infty).
 \end{align}
 Define  
 \begin{align}\label{def:tilde_pi}
     \Tilde{\pi}_{j,n}(\di x_1\di x_2\di x_3\di x_4):= \rho_{\infty, n}(x_1, x_3) \boldsymbol{\gamma}_j(\di x_1 \di x_2) \otimes  \boldsymbol{\gamma}_j(\di x_3 \di x_4), \ \pi_{j,n} := (P_{24})_\# \Tilde{\pi}_{j,n}.
 \end{align}
 Note that $\pi_{j,n} \ll \mm_j \otimes \mm_j$ and that $\pi_{j,n} \to \pi_{\infty, n}$ narrowly as $j \to \infty$. For $R > 0$, define the open set $$\Omega_R := \Big\{\tau_\infty > \frac{1}{R}\Big\} \subset M^2.$$ For any sequence  $R_k \to \infty$, we get that $\{\tau_\infty > 0\} = \bigcup_{k \in \N} \Omega_{R_k}$. It follows from \eqref{pi_n_infty_all_convergences} that for any such sequence $R_k$, it holds
 \begin{align}\label{eq:piinfnOmegaR}
     & 1 \geq \lim_{k \to \infty} \pi_{\infty, n}(\Omega_{R_k}) =  \pi_{\infty, n}(\{\tau_\infty > 0\}) = 1. 
 \end{align}
 Fix  $n \in \N$.    By Lemma \ref{restriction_of_plans_converge}, there exists a sequence  $R^n_k \to \infty$ such that  
 \begin{align}\label{narrow_k}
 &\pi_{\infty, n}(\Omega_{R^n_k})>0,  \quad \text{for all }k \geq k_0=k_0(n)>0,\\
    & c_{\infty, k, n}:= \pi_{\infty, n}(\Omega_{R^n_k}) \to 1 \ \mathrm{as} \ k \to \infty \ \mathrm{and} \nonumber\\
      & \pi_{\infty, k, n}:= \pi_{\infty, n}(\Omega_{R^n_k})^{-1} \pi_{\infty, n} \LL \Omega_{R^n_k} \to \pi_{\infty, n} \ \mathrm{narrowly\ as} \ k \to \infty.
 \end{align}
 Moreover, $\pi_{j, n}(\Omega_{R^n_k})>0$  for $j \geq J_{k,n}$, and
 \begin{align}\label{narrow:j:k}
  & c_{j, k, n}:= \pi_{j, n}(\Omega_{R^n_k}) \to c_{\infty, k, n} \ \mathrm{as} \ j \to \infty,   \nonumber\\
    &\pi_{j, k, n}:= \pi_{j, n}(\Omega_{R^n_k})^{-1}\pi_{j,n} \LL \Omega_{R^n_k} \to \pi_{\infty, k, n} \ \mathrm{narrowly\ as} \ j \to \infty. 
 \end{align}
 As $E$ is a Polish space, so is $\mathcal{P}(E)$, hence there exists a sequence of open neighbourhoods $(O_l)_{l \in \N} \subset \mathcal{P}(E)$ of $\pi_{\infty, n}$ with $O_{l+1} \subset O_l$ for all $l$ and such that for any open $U$ containing  $\pi_{\infty, n}$ there exists an $l_0$ such that for each $l \geq l_0$, it holds $O_l \subset U$.
 Fix an $l \in \N$. 
 By \eqref{narrow_k}, there exists  $k_l \in \N$ with $k_0\leq k_{l-1} \leq k_l$, such that for $k \geq k_l$ it holds $\pi_{\infty, k, n} \in O_l$ and $c_{\infty, k, n} \geq 1-\frac{1}{2l}$. By the locally uniform convergence of the $\tau_j$ to $\tau_\infty$, we get that there exists  $J^1_l \in \N$ such that 
 \begin{align*}
     \Omega_{R^n_k} \subset \Big\{\tau_j > \frac{1}{2{R^n_{k_l}}}\Big\} \subset \Big\{\tau_\infty > \frac{1}{4{R^n_{k_l}}}\Big\}, \quad \text{for all } j \geq J^1_l.
 \end{align*}
 Since $\pi_{\infty, k_l, n} \in O_l$, we can use \eqref{narrow:j:k}, to find  $J^2_l \in \N$ such that $\pi_{j,k_l,n} \in O_l$, for all $j \geq J^2_l$. Choose  $j_l \geq \max(J^1_l,J^2_l, j_{l-1}+1)$ such that 
 \begin{align}\label{c_l_convergence}
     c_{j_l, k_l, n} \geq 1-\frac{1}{l}.
 \end{align}
Then the plans
\begin{align}\label{pi_l_n_prime:narrow_convergence}
    \pi'_{l,n} := \pi_{j_l, k_l, n} \to \pi_{\infty, n} \ \mathrm{narrowly\ as} \ l \to \infty.
\end{align}

Moreover, by  \cite[Lem.\ 1.17]{cavalletti2020optimal}, Proposition \ref{locally_uniform_convergence_tau} and the uniform boundedness of the $\tau_j$ guaranteed by Lemma \ref{Lipschitz_bd_causal_curves}, we get that
\begin{align}\label{pi_prime_convergence_of_cost}
    \lim_{l \to \infty} \int \tau_{j_l}^p \pi'_{l,n} = \int \tau_{\infty}^p \pi_{\infty, n}. 
\end{align}
Let 
\begin{align}\label{def:mu_l_n_prime}
    (\mu')^{l,n}_0 := (P_1)_\# \pi'_{l,n}, \ (\mu')^{l,n}_1 := (P_2)_\# \pi'_{l,n}.
\end{align}
 Then $\Pi_{\ll_{j_l}}((\mu')^{l,n}_0, (\mu')^{l,n}_1) \neq \emptyset$ and  $\ell^{j_l}_p((\mu')^{l,n}_0, (\mu')^{l,n}_1) \in (0, \infty)$. By \cite[Prop.\ 2.3]{cavalletti2020optimal}, there exist $\ell_p^{j_l}$-optimal plans 
 \begin{align}\label{def:pi_prime_prime}
     \pi''_{l,n} \in \Pi^{p\text{-opt}}_{\leq_{j_l}}((\mu')^{l,n}_0, (\mu')^{l,n}_1).
 \end{align}
 Combining \cite[Lemma 1.15]{cavalletti2020optimal} with Prokhorov's Theorem, we infer that there exists  $\hat{\pi}_{\infty, n} \in \mathcal{P}(E^2)$ such that $\pi''_{l,n} \to \hat{\pi}_{\infty, n}$ narrowly as $l \to \infty$, up to a subsequence. Lemma \ref{causal_limit_plan} yields \begin{align}\label{limit_of_transport_plans_is_transport_plan}
     \hat{\pi}_{\infty, n}(\leq_{\infty}) = 1.
 \end{align}
Moreover, \eqref{pi_l_n_prime:narrow_convergence} and \eqref{def:mu_l_n_prime}  imply that 
 \begin{align*}
     \hat{\pi}_{\infty, n} \in \Pi_{\leq_\infty}((P_1)_\#(\pi_{\infty, n}),(P_2)_\#(\pi_{\infty, n})).
 \end{align*}
 With an analogous argument using \eqref{pi_n_infty_all_convergences},  \cite[Lemma 1.15]{cavalletti2020optimal}, Prokhorov's Theorem and the closedness of $\leq_\infty$, we get that there exists a $\hat{\pi}_\infty \in \Pi_{\leq_\infty}(\mu_0^\infty, \mu_1^\infty)$ such that, up to a subsequence, $\hat{\pi}_{n, \infty} \to \hat{\pi}_\infty$ narrowly as $n \to \infty$. 
 Recalling \eqref{pi_prime_convergence_of_cost}, a diagonal argument gives that there exist strictly increasing sequences $n_q, l_q$ such that 
\begin{align}\label{pi_prime_prime_narrow_convergence}
     \pi''_{l_q, n_q} \to \hat{\pi}_\infty \ \mathrm{narrowly\ as\ } q \to \infty,
 \end{align}
 and
 \begin{align}\label{choice_of_q}
     \int \tau^p_{j_{l_q}} \pi'_{l_q, n_q} \geq \int \tau^p_\infty \pi_{\infty, n_q} - \frac{1}{q}, \quad \text{ for all }q\in \N.
 \end{align}
 The combination of \eqref{pi_n_infty_all_convergences},   \eqref{pi_prime_convergence_of_cost} and \eqref{choice_of_q},   together with the optimality \eqref{def:pi_prime_prime} of $\pi''_{l,n}$ gives that 
 \begin{align}\label{pi_prime_prime_covergence_cost}
     \int \tau^p_\infty \hat{\pi}_\infty = \lim_{q \to \infty} \int \tau_{j_{l_q}}^p \pi''_{l_q, n_q} \geq \lim_{q \to \infty} \int \tau^p_{j_{l_q}} \pi'_{l_q, n_q} = \int \tau^p_\infty \pi_\infty. 
 \end{align}
 Recalling that $\pi_\infty \in \Pi^{p\text{-opt}}_{\leq_\infty}(\mu_0^\infty, \mu_1^\infty)$, we infer that $\hat{\pi}_\infty \in \Pi^{p\text{-opt}}_{\leq_\infty}(\mu_0^\infty, \mu_1^\infty)$. Since by assumption $(\mu_0^\infty, \mu_1^\infty)$ is strongly timelike $p$-dualisable, we get that $\hat{\pi}_\infty(\{\tau_\infty >0\})=1$. 
Recall that we defined $\Omega_R := \{\tau_\infty > \frac{1}{R}\}$. There exists an $R_0> 0$ large enough, such that $\hat{\pi}_\infty(\Omega_R) >0$,  for all $R\geq R_0$.
 Then by Lemma \ref{restriction_of_plans_converge}, we get that there exists a sequence $R_m \to \infty$ such that  
 \begin{align}
    &c_m := \hat{\pi}_\infty(\Omega_{R_m}) \to 1 \ \mathrm{as} \ m \to \infty \ \mathrm{and} \nonumber \\
     &\hat{\pi}_\infty(\Omega_{R_m})^{-1}\hat{\pi}_\infty \LL \Omega_{R_m} \to \hat{\pi}_\infty \ \mathrm{narrowly\ as} \ m \to \infty.
 \end{align}
Moreover,  $\pi''_{l_q, n_q} (\Omega_{R_m}) > 0$ for $q \geq Q_m^1$ sufficiently large, and 
 \begin{align}\label{restriction_pi_hat_narrow}
     & c_{q,m} := \pi''_{l_q, n_q} (\Omega_{R_m}) \to c_m \ \mathrm{as\ }  q \to \infty, \nonumber \\
     &\pi''_{l_q, n_q} (\Omega_{R_m})^{-1}  \pi''_{l_q, n_q} \LL \Omega_{R_m} \to \hat{\pi}_\infty \LL \Omega_{R_m} \ \mathrm{narrowly\ as} \ q \to \infty.
 \end{align}
 Moreover, from Proposition \ref{locally_uniform_convergence_tau}, we get that for each $m$ there exists a $J_m$ such that $\Omega_{R_m} \subset \bigcap_{j \geq J_m}\{\tau_j > 0\}$. Hence, using a diagonal argument we can choose an increasing sequence $q_m$ such that $q_m \geq Q^1_m$, $j_{l_{q_m}} \geq J_m$, $c_{q_m,m} \to 1$, and 
 \begin{align}\label{def_pi_m_final_opt_plan}
     \pi_m := \pi''_{l_{q_m}, n_{q_m}} (\Omega_{R_m})^{-1}  \pi''_{l_{q_m}, n_{q_m}} \LL \Omega_{R_m} \to \hat{\pi}_\infty \ \mathrm{narrowly \ as} \ m \to \infty.
 \end{align}
 Set 
 \begin{align*}
    \mu_0^m:= (P_1)_\# \pi_m, \  \mu_1^m:= (P_2)_\# \pi_m,
 \end{align*}
 and note that 
 \begin{align*}
     \mu_0^m \to \mu_0^\infty,\  \mu_1^m \to \mu_1^\infty \ \mathrm{narrowly \ as} \ m \to \infty.
 \end{align*}
 Denote $j_m :=j_{l_{q_m}}$,
 By the choice of $q_m$ and  \cite[Lemma 2.10]{cavalletti2020optimal}, we get that $\pi_m \in \Pi_{\ll_{j_m}}^{p\text{-opt}}(\mu_0^m, \mu_1^m)$; moreover, by construction, $\pi_m(\{ \tau_\infty>0\}) = 1$. Then $(\mu_0^m, \mu_1^m)$ is $g_{j_m}$-timelike $p$-dualisable by $\pi_m$ and $\ell^p_{j_m}(\mu_0^m, \mu_1^m) \to  \ell^p_\infty(\mu_0^\infty, \mu_1^\infty)$.
 This proves the first step, up to renaming indices. 
\smallskip

 \textbf{Step 2:} We prove that the sequences $(\mu^j_0), (\mu^j_1)$ constructed in Step 1 satisfy:
 \begin{align}\label{endpoint_entropy_upper_semicont}
     \limsup_{j \to \infty} \Ent(\mu_0^j| \mm_j) \leq \Ent(\mu^\infty_0| \mm_\infty), \ \limsup_{j \to \infty} \Ent(\mu_1^j| \mm_j) \leq \Ent(\mu^\infty_1| \mm_\infty).
 \end{align}
 First, recall the definition \eqref{def:tilde_pi} of $\Tilde{\pi}_{j,n}$ and set
 \begin{align*}
     \mu^{j,n}_0= (P_2)_\# \Tilde{\pi}_{j,n}, \ \mu^{j,n}_1= (P_4)_\# \Tilde{\pi}_{j,n}.
 \end{align*}
 \textbf{Step 2a:} We first prove that 
 \begin{align}\label{step_2a}
     \Ent(\mu_0^{j,n}| \mm_j) \leq \Ent((P_1)_\#\pi_{\infty, n}|\mm_\infty), \ \Ent(\mu_1^{j,n}| \mm_j) \leq \Ent((P_2)_\#\pi_{\infty, n}|\mm_\infty).
 \end{align}
 We will give the argument for the former, the latter being completely analogous. The expression \eqref{def:tilde_pi}, together with \eqref{pi_n_infty_all_convergences} and Fubini's theorem allow us to write
\begin{align}
    &(P_1)_\# \pi_{\infty, n} = \rho_0^{\infty, n}\mm_\infty;\ \rho_0^{\infty, n}(x_1) = \int_M \rho_{\infty, n}(x_1, x_3) \mm_\infty(\di x_3), \  (P_1)_\# \pi_{\infty, n}\mathrm{-a.e.}\ x_1 \in M; \nonumber \\
    &\mu^{j,n}_0 = \rho_0^{j, n}\mm_j;\ \rho_0^{j, n}(x_2)= \int_M \Bigg( \int_{M^2} \rho_{\infty, n}(x_1, x_3) \boldsymbol{\gamma}_j(\di x_3 \di x_4) \Bigg) (\boldsymbol{\gamma}_j)_{x_2}(\di x_1),\ \mu_0^{j,n}\mathrm{-a.e.}\ x_2 \in M.
\end{align}
Here $\{(\boldsymbol{\gamma}_j)_{x_2}\}$ denotes the disintegration of $\boldsymbol{\gamma}_j$ with respect to $P_2$. Since $u(t)= t \log t$ is convex on $[0, \infty)$, Jensen's inequality gives
\begin{align*}
    \Ent(\mu_0^{j,n}| \mm_j) &= \int_M u(\rho_0^{j,n}(x_2)) \mm_j(\di x_2) \\
    &\leq \int_M \int_M u \Bigg( \int_{M^2} \rho_{\infty, n}(x_1, x_3) \boldsymbol{\gamma}_j(\di x_3 \di x_4) \Bigg) (\boldsymbol{\gamma}_j)_{x_2}(\di x_1) \mm_j(\di x_2) \\
    &= \int_{M^2} u(\rho_0^{\infty, n}(x_1)) \boldsymbol{\gamma}_j(\di x_1 \di x_2) = \int_M u(\rho_0^{\infty, n}(x_1)) \mm_\infty(\di x_1) \\
    &= \Ent((P_1)_\# \pi_{\infty, n}| \mm_\infty).
\end{align*}
\textbf{Step 2b}. We prove that the sequences $(\mu^j_0), (\mu^j_1)$ constructed in Step 1 satisfy:
\begin{align}\label{step_2b}
    \limsup_{j \to \infty} \Ent(\mu_i^j| \mm_j) \leq \limsup_{m \to \infty} \Ent(\mu^{l_{q_m}, n_{q_m}}_i | \mm_{j_{l_{q_m}}}) = \limsup_{m \to \infty} \Ent(\mu^{j_m, n_{q_m}}_i | \mm_{j_m}) ,\ i=0,1,
\end{align}
where $(l_q, n_q)$ are as in \eqref{pi_prime_prime_narrow_convergence} (but note that  neither $\pi'_{l,n}$ nor $\pi''_{l,n}$ were used to define $\mu^{l_{q_m}, n_{q_m}}_i$). We will give the argument for $i=0$, as the argument for $i=1$ is completely analogous. 
From \eqref{def:tilde_pi}, it follows that $\mu_0^{j,n} = (P_1)_{\#} \pi_{j,n}$. Recalling \eqref{narrow_k} and \eqref{narrow:j:k}, we get that for $j$ and $k$ large enough:
\begin{align*}
    \pi_{j,k,n}(B) \leq \frac{1}{c_{j,k,n}} \pi_{j, n}(B), \quad \text{for every  Borel set } B \subset M^2.
\end{align*}
By our choice of $k_l, j_l$ (see \eqref{c_l_convergence} and \eqref{pi_l_n_prime:narrow_convergence}), and from the definition \eqref{def:mu_l_n_prime} of $(\mu')^{l,n}_0 := (P_1)_\# \pi'_{l,n}$, we get:
\begin{align*}
    (\mu')^{l,n}_0 \ll \mm_{j_l},  \quad  (\rho')^{l,n}_0:= \frac{\di (\mu')^{l,n}_0} {\di \mm_{j_l}},
\end{align*}
where
\begin{align}\label{rho_prime_estimate}
    {(\rho')^{l,n}_0} \leq \frac{1}{c_{{j_l},{k_l},n}} \rho^{j_l, n}_0.
\end{align}
Recall that, by \eqref{def:pi_prime_prime}:
\begin{align*}
    (\mu')^{l,n}_0 = (P_1)_\# \pi''_{l, n}.
\end{align*}
The combination \eqref{restriction_pi_hat_narrow} and \eqref{def_pi_m_final_opt_plan} yields 
\begin{align*}
    \pi_{m}(B) \leq \frac{1}{c_{{q_m}, m}} \pi''_{l_{q_m}, n_{q_m}}(B),   \quad \text{for every  Borel set } B \subset M^2.
\end{align*}
Recalling the notation  $j_m := j_{l_{q_m}}$ and that  $\mu^m_0 = (P_1)_\# \pi_m \ll \mm_{j_m}$,  we get that there exists $\rho^m_0\in L^1(\mm_{j_m})$ such that $\mu^m_0 = \rho^m_0 \mm_{j_m}$ and 
\begin{align}\label{rho_m_estimate}
  {\rho^m_0} \leq \frac{1}{c_{q_m, m}} \cdot {(\rho')^{l_{q_m},n_{q_m}}_0}.
\end{align}
Combining  \eqref{rho_prime_estimate} and \eqref{rho_m_estimate}, and recalling \eqref{c_l_convergence}, \eqref{def_pi_m_final_opt_plan}, we infer that 
\begin{align}\label{final_estimate_rho}
   0 \leq  \rho^m_0 \leq c_m \, \rho_0^{j_m, n_{q_m}} \leq c_m \, \norm{\rho_{\infty, n_{q_m}}}_{L^\infty(\mm_\infty\otimes \mm_\infty)} 
\end{align}
where $\rho_{\infty,n}$ was defined in \eqref{pi_n_infty_all_convergences},   $c_m \to 1$, and the last inequality follows from the local uniform convergence of the volume forms $\vol_j \to \vol_\infty$. 
As $u(t):= t \log t$ is convex on $[0, \infty)$ and $u(0)=0$, it holds
\begin{align*}
    u(t+h)-u(t) \geq u(h), \ \forall t, h \geq 0.
\end{align*}
Then, Jensen's inequality together with \eqref{final_estimate_rho}, give that 
\begin{align}\label{rho_entropy_comparison}
    \int u(c_m \rho_0^{j_m, n_{q_m}}) \mm_{j_m} - \int u(\rho^m_0)\, \mm_{j_m} &\geq \int u(c_m \rho_0^{j_m, n_{q_m}}-\rho^m_0)\,\mm_{j_m} \nonumber\\
    & \geq u \Bigg(\int (c_m \rho_0^{j_m, n_{q_m}}-\rho^m_0)\,\mm_{j_m} \Bigg) \nonumber \\
    &= u(c_m -1) \to 0 \ \mathrm{as\ } m \to \infty.
\end{align}
Note that $\int u(\rho^m_0) \mm_{j_m} = \Ent(\mu^m_0| \mm_{j_m})$ and that 
\begin{align}\label{entropy_under_constant_factor}
   &\Bigg|\int u(c_m \rho_0^{j_m, n_{q_m}}) \mm_{j_m} - \Ent(\mu^{j_m, n_{q_m}}_0 | \mm_{j_m})\Bigg| \nonumber\\
   &= \Bigg|\int ((c_m-1)( \rho_0^{j_m, n_{q_m}} \log \rho_0^{j_m, n_{q_m}}) + c_m (\log c_m) \rho_0^{j_m, n_{q_m}})  \mm_{j_m}\Bigg| \nonumber\\
   &\leq (c_m-1) |\Ent(\mu^{j_m, n_{q_m}}_0 | \mm_{j_m}) |+ c_m \log c_m. 
\end{align}
Using \eqref{step_2a} and \eqref{rho_prime_estimate}, we get that \eqref{entropy_under_constant_factor}  converges to $0$ as $m \to \infty$. Combining  this fact with \eqref{rho_entropy_comparison} imply \eqref{step_2b}. 
The combination of \eqref{step_2a} and \eqref{step_2b}, together with \eqref{pi_n_infty_all_convergences} imply \eqref{endpoint_entropy_upper_semicont}.
\smallskip

\textbf{Step 3:} Passing to the limit in the $\tcd$ condition. \\
For simplicity, we present the argument only for the \( \tcd^{e}_p(0,N) \) condition. 
The case \( K \neq 0 \) can be treated analogously, albeit with more involved coefficients. 
As $(\mu_0^j, \mu_1^j)\in \Dom(\Ent(\cdot | \mm_j))^2\subset \mathcal{P}(M)^2$ are $g_j$-timelike $p$-dualisable, the assumption that $(M, g_j)$ has timelike Ricci curvature bounded below by $0$ yields the existence of a $\ell^j_p$-geodesic $(\mu^j_t)_{t \in [0,1]}$ such that 
\begin{align}\label{tcd0n_for_sequence}
    U_N(\mu^j_t|\mm_j) \geq (1-t)U_N(\mu_0^j|\mm_j) + tU_N(\mu_1^j|\mm_j), \quad \forall t \in [0,1],\ \forall j \in \N.
\end{align}
Since by construction 
\begin{align*}
    \bigcup_{j \in \N, i \in \{0, 1\}} \supp\,\mu_i^j \subset E,\quad \text{with } E\Subset M,
\end{align*}
and $(M, g_j)$ are uniformly globally hyperbolic, we get that there exists a compact subset $\mathcal{K} \subset M$, such that 
\begin{align*}
    \bigcup_{j \in \N, t \in [0,1]} \supp \, \mu_t^j \subset \mathcal{K}. 
\end{align*}
Let $\eta_j$ be the $\ell_p^j$-optimal dynamical plan representing the $\ell_p^j$-geodesic $(\mu^j_t)_{t \in [0,1]}$. By Proposition \ref{dynamical_plans_precompactness} and Lemma \ref{limit_plan_is_dynamical}, we get that up to a subsequence, the plans $\eta_j$ narrowly converge to a plan $\eta_\infty$ that is concentrated on $\mathrm{TGeo}^\infty$. 
By the continuity of $\ee_t: C([0,1], M)\to M$, it follows that 
\begin{align}
    &\mu_t^j  = (\ee_t)_\# \eta_j \to (\ee_t)_\# \eta_\infty =: \mu_t^\infty \ \mathrm{narrowly\ as\ } j \to \infty, \ \forall t \in [0,1] \nonumber \\
    &\check{\pi}_j:= (\ee_0, \ee_1)_\# \eta_j \to (\ee_0, \ee_1)_\# \eta_\infty \ \mathrm{narrowly\ as\ } j \to \infty. \label{eq:checkpijtopi}
\end{align}
By \eqref{step1} and the uniqueness of the narrow limit, we infer that $(\ee_i)_\# \eta_\infty = \mu^\infty_i$ for $i = 0, 1$, the measures fixed at the very start of the proof, see \eqref{eq:defmuiinf}. 
 Lemma \ref{limit_plan_is_dynamical} ensures also that 
\begin{align*}
    \check{\pi}_\infty := (\ee_0, \ee_1)_\# \eta_\infty \in \Pi_{\ll_\infty}(\mu^\infty_0, \mu^\infty_1). 
\end{align*}
As in the proof of Proposition \ref{dynamical_plans_precompactness},  the narrow convergence \eqref{eq:checkpijtopi} together with Proposition \ref{locally_uniform_convergence_tau} yields
\begin{align}
    \lim_{j \to \infty} \ell_p(\mu^j_0, \mu^j_1) = \lim_{j \to \infty} \int_{M^2} \tau_j(x, y)^p \check{\pi}_j = \int_{M^2} \tau_\infty(x, y)^p \check{\pi}_\infty. 
\end{align}
Together with \eqref{step1}, this gives $\ell^\infty_p$-optimality of $\check{\pi}_\infty$ and therewith shows that $\eta_\infty$ is an optimal dynamical plan. 
The upper semicontinuity of $U_N$ under joint narrow/weak convergence (see e.g. \cite[Eq.\ (1.15)]{cavalletti2020optimal}) implies that 
\begin{align}\label{convergence_of_u_n_inside_geo}
    U_N(\mu_t^\infty| \mm_\infty) \geq \limsup_{j \in \N} U_N(\mu_t^j| \mm_j), \quad \forall t\in [0,1]. 
\end{align}
The combination of \eqref{tcd0n_for_sequence}, \eqref{endpoint_entropy_upper_semicont}, and  \eqref{convergence_of_u_n_inside_geo} gives
\begin{align*}
    U_N(\mu_t^\infty|\mm_\infty) \geq (1-t) U_N(\mu^\infty_0| \mm_\infty) +tU_N(\mu^\infty_1| \mm_\infty), \quad \forall t\in [0,1], 
\end{align*}
as desired. 
\end{proof}

\subsection{Stability of synthetic upper Ricci bounds}
We start this subsection by analysing stability properties of cyclical monotonicity. Then we recall synthetic upper Ricci curvature bounds from \cite{mondino2022optimal} and prove its stability for locally uniformly converging Lorentzian metrics $g_j\to g_\infty$ under suitable uniform controls on $g_j$. 
\begin{lemma}
    Assume $X^2_{\leq}$ is closed and $\tau$ is continuous. If $\Gamma$ is $\tau^p$-cyclically monotone, then $\overline{\Gamma}$ is $\tau^p$-cyclically monotone.
\end{lemma}
\begin{proof}
    Note that $\Gamma \subset X^2_{\leq}$, hence $\overline{\Gamma} \subset X^2_{\leq}$. Fix $(x_1, y_1), \ldots, (x_N, y_N) \in \overline{\Gamma}$. Then for $i =1, \ldots, N$, we can find sequences $(x^k_i, y^k_i) \in \Gamma$ such that $\lim_{k \to \infty} (x^k_i, y^k_i) = (x_i, y_i)$, for all $i=1,\ldots, N$. Now, for each $k$, \eqref{cyclical_monotonicity} holds for $(x^k_1, y^k_1), \ldots, (x^k_N, y^k_N)$, so we can use the continuity of $\tau$ for passing to the limit. 
\end{proof}
\begin{lemma}\label{cyclical_monotonicity_limit}
Let $g_j$, $j\in \N$, and $g_\infty$ satisfy the assumptions of Proposition \ref{locally_uniform_convergence_tau},  and fix  $p \in (0,1)$. Moreover, assume that
\begin{itemize}
\item $\pi_j$ are $\tau_j$-optimal plans, satisfying that $\pi_j(M^2_{\ll_j})=1$ and $\supp \, (P_1)_\# \pi_j \times \supp \, (P_2)_\# \pi_j \subset M^2_{\leq_j}$;
\item there exists  $\pi_\infty\in \mathcal{P}(M^2)$ such that $\pi_j \to \pi_\infty$ narrowly.
\end{itemize}
 Then $\pi_\infty$ is $\tau^p_\infty$-cyclically monotone. 
\end{lemma}
\begin{proof}
    Denote $\mu^j_0 := (P_1)_\# \pi_j$ and $\mu^j_1 := (P_2)_\# \pi_j$ for $j \in \N \cup \{\infty\}$. 
    By  \cite[Prop.\ 2.8 (1)]{cavalletti2020optimal}, we have that $\supp\, \pi_j$ is $\ell^j_p$-cyclically monotone. It is a general fact that for $i = 1, 2$, and any $\pi \in \mathcal{P}(M^2)$, 
    \begin{align}\label{projection_support_inclusion}
        P_i(\supp\, \pi) \subset \supp \, ((P_i)_\#\pi).
    \end{align}
    Indeed,
    \begin{align*}
        &x \in P_1(\supp\, \pi) \Rightarrow \exists y \in M:\ (x,y) \in \supp\, \pi \Rightarrow \forall r>0, \, \pi(B_r((x, y))) > 0 \Rightarrow \forall r> 0, \  \pi(B_r(x) \times M) > 0 \\
        & \Rightarrow \forall r> 0, \ (P_1)_\#\pi(B_r(x)) \geq 0 \Rightarrow x \in \supp \, ((P_1)_\#\pi).
    \end{align*}
    The case of $i=2$ is completely analogous. 
    Together with the fact that  $\supp \, (P_1)_\# \pi_j \times \supp \, (P_2)_\# \pi_j \subset M^2_{\leq_j}$, we get that $P_1(\supp\, \pi_j) \times P_2(\supp\, \pi_j) \subset M^2_{\leq_j}$. 
    By  \cite[Rem.\ 2.7]{cavalletti2020optimal}, this gives that $\supp\, \pi_j$ is $\tau^p_j$-cyclically monotone. 
    We will show that $\supp \, \pi_\infty$ is $\tau^p_\infty$-cyclically monotone.
    We argue by contradiction. Suppose there exist $(x_1, y_1), \ldots, (x_N, y_N) \in \supp\, \pi_\infty$, and  $\e > 0$ such that 
    \begin{align*}
         \sum_{i=1}^N \tau_\infty^p(x_i, y_i) \leq \e + \sum_{i=1}^N \tau^p_\infty(x_{i+1}, y_i).
    \end{align*}
  As $(x_1, y_1), \ldots, (x_N, y_N) \in \supp\, \pi_\infty$, we get that for each $r > 0$ and each $i \in \{1, \ldots, N\}$, it holds 
  \begin{align*}
      \pi_\infty(B_r((x_i, y_i))) >0. 
  \end{align*}
  By Proposition \ref{locally_uniform_convergence_tau}, we know that $\tau_\infty$ is continuous;  hence, there exists  $r>0$ such that
  \begin{align}
      \max(|\tau^p_\infty(x_i, y_i)-\tau_\infty^p(z_i, w_i)|, |\tau^p_\infty(x_{i+1}, y_i)-\tau_\infty^p(z_{i+1}, w_i)|) \leq \frac{\e}{16N},
  \end{align}
   for all $i \in \{1, \ldots, N\}$ and all $(z_i, w_i) \in B_r((x_i, y_i))$.
  From the lower semi-continuity on open sets under narrow convergence, we get that there exists  $J_1 \in \N$ such that for $j \geq J_1$ and all $i \in \{1, \ldots, N\}$, it holds
  \begin{align*}
      \pi_j(B_r((x_i, y_i))) >0.
  \end{align*}
  Moreover, by Proposition \ref{locally_uniform_convergence_tau}, there exists  $J_2 \in \N$ such that for $j \geq J_2$, it holds that for all $i$ and $z_i \in B_r(x_i)$, $w_i \in B_r(y_i)$
  \begin{align}
      |\tau^p_\infty(z_i, w_i)-\tau_j^p(z_i, w_i)|+ |\tau^p_\infty(z_{i+1}, w_i)-\tau_j^p(z_{i+1}, w_i)| \leq \frac{\e}{16N}.
  \end{align}
  For $j \geq J_1+J_2$, we get that for all $i\in \{1, \ldots, N\}$, it holds $\supp\, \pi_j \cap B_r((x_i, y_i)) \neq \emptyset$. Hence, we can find $(z_i, w_i) \in \supp\, \pi_j \cap B_r((x_i, y_i))$ such that 
  \begin{align*}
         \sum_{i=1}^N \tau_j^p(z_i, w_i) \leq \frac{\e}{2} + \sum_{i=1}^N \tau^p_j(z_{i+1}, w_i).
    \end{align*}
    This contradicts that $\supp \, \pi_j$ is $\tau^p_j$-cyclically monotone. 
\end{proof}
\begin{lemma}\label{time_order_support_preserved}
Assume the metrics $(g_j)_{j \in \N}$ are smooth and uniformly globally hyperbolic and converge locally uniformly to $g_\infty$, a continuous Lorentzian metric.
  Let $K \subset M$ be compact.
    For $j \in \N$ let $\mu^j_0, \mu^j_1 \in \mathcal{P}(K)$ and let $\mu^\infty_0, \mu^\infty_1 \in \mathcal{P}(K)$ such that $\mu^j_0 \to \mu^\infty_0$ and $\mu^j_1 \to \mu^\infty_1$ narrowly as $j \to \infty$. Assume that 
    \begin{equation}\label{eq:muj0muj1Mleq}
        \supp \, \mu^j_0 \times \supp \, \mu^j_1 \subset M^2_{\leq_j}, \quad \text{for all } j \in \N.
    \end{equation}
    Then \begin{align*}
         \supp \, \mu^\infty_0 \times \supp \, \mu^\infty_1 \subset M^2_{\leq_\infty}.
    \end{align*}
\end{lemma}
\begin{proof}
    We argue by contradiction. Assume that 
    \begin{align*}
        Z:= ( \supp \, \mu^\infty_0 \times \supp \, \mu^\infty_1) \cap (M^2_{\leq_\infty})^c \neq \emptyset,
    \end{align*}
    and pick $(x, y) \in Z$. The causal closedness ensures that there exists an $r > 0$ such that $B_{4r}(x, y) \subset (M^2_{\leq_\infty})^c)$.
    The properties of the support give that 
    \begin{align*}
        \mu^\infty_0(B_r(x)) >0, \ \mu^\infty_1(B_r(y)) >0.
    \end{align*}
    The lower semicontinuity on open sets under narrow convergence yields that there exists a $J \in \N$ such that 
    \begin{align*}
        \mu^j_0(B_r(x)) >0, \ \mu^j_1(B_r(y)) >0, \quad \text{for all } j \geq J. 
    \end{align*}
    Hence,
    \begin{align*}
        B_r(x) \cap \supp\, \mu^j_0 \neq \emptyset, \ B_r(y) \cap \supp\, \mu^j_0 \neq \emptyset.
    \end{align*}
    We may pick $x_j \in B_r(x) \cap \supp\, \mu^j_0$ and $y_j \in B_r(y) \cap \supp\, \mu^j_1$. By assumption \eqref{eq:muj0muj1Mleq} we have that $x_j \leq_j y_j$. Hence, Proposition \ref{Lipschitz_bd_causal_curves} ensures that there exists an $L>0$ and $g_j$-causal, $L$-Lipschitz-continuous curves $\gamma_j \in C([0,1], M)$ such that $\gamma_j(0) = x_j$ and $\gamma_j(1) = y_j$. 
    Then, using uniform global hyperbolicity and the limit curve theorem \cite[Thm.\ 3.1]{minguzzi2008limit}  (similarly to the proof of Proposition \ref{Lipschitz_bd_causal_curves}), we get that there exists a $g_\infty$-causal limit curve $\gamma \in C([0,1], M)$ such that $\gamma_j \to \gamma$ uniformly as $\gamma \to \infty$. But then $\gamma(0) \in B_{2r}(x)$ and $\gamma(1) \in  B_{2r}(y)$, and $\gamma(0) \leq_\infty \gamma(1)$. This contradicts $(\gamma(0), \gamma(1)) \in B_{4r}((x,y)) \subset (M_{\leq_\infty})^c$. 
\end{proof}

Next, we recall the synthetic notion of timelike Ricci curvature upper bound from \cite{mondino2022optimal}. These shall be seen as the Lorentzian counterpart of \cite{sturm:UB}.
\smallskip

\begin{definition}[\cite{mondino2022optimal}, Def.\ B.5]
    Fix  $p \in (0, 1)$, $K \in \R$. We say that a measured Lorentzian pre-length space $(X, \sfd, \mm, \ll, \leq, \tau)$ has timelike Ricci curvature bounded above by $K$ in the synthetic sense if there exists $r_0 > 0$ and a function $\omega: [0, r_0) \to [0, \infty)$ with $\lim_{r \to 0} \omega(r)=0$ such that for every $r \in (0, r_0)$ the following holds. 
    For any $x, y \in X$ with $\sfd(x, y) = r> 0$ and such that $B^\sfd_{r^4}(x) \times B^\sfd_{r^2}(y) \Subset  X^2_\ll$, and for every $\mu_0 \in \Dom(\Ent(\cdot | \mm))$ with $\supp \, \mu_0 \subset B^\sfd_{r^4}(x)$ there exists an $\ell_p$-geodesic $(\mu_t)_{t \in [-1, 1]}$ satisfying
    \begin{itemize}
        \item $\supp \, \mu_1 \subset B^\sfd_{r^2}(y)$,
        \item $\supp \, \mu_{-1} \times \supp\, \mu_1 \subset X^2_{\leq}$,
        \item $\bigcup_{t \in [-1, 1]} \supp \, \mu_t \subset B^\sfd_{10r_0}(x)$,
        \item $\Ent(\mu_{-1}| \mm) -2\Ent(\mu_{0}| \mm) + \Ent(\mu_{1}| \mm) \leq (K + \omega(r))r^2$.
    \end{itemize}
\end{definition}

The following result is inspired by \cite[Thm.~B.6]{mondino2022optimal}, where the stability of synthetic timelike Ricci upper bounds was established under a Lorentzian version of pointed measured Gromov convergence (see \cite{GMS-PLMS} for the corresponding metric framework), relying on isometric embeddings into a common Lorentzian pre-length space.  In contrast, our approach replaces the assumption of such isometric embeddings by requiring locally uniform convergence of the Lorentzian metrics. This perspective will be crucial in the subsequent sections, where we show that impulsive gravitational waves satisfy synthetic timelike Ricci upper bounds.
\smallskip
 
\begin{theorem}\label{stability_upper_bounds}
    Let $M$ be a smooth manifold endowed with  a continuous Lorentzian metric $g_\infty$  that turns $(M,g_\infty)$ into a globally hyperbolic geodesic space such that
 every length-maximising causal curve has a causal character.    For all $j \in \N$, let $g_j$ be smooth Lorentzian metrics on $M$ such that $(M, g_j)$ are uniformly globally hyperbolic (see Definition \ref{uniform_global_hyperbolicity_def}) and $g_j \to g_\infty$ locally uniformly.  
    Let $p \in (0,1)$, $K \in \R$, $r_0 > 0$ and $\omega: [0, r_0) \to [0, \infty)$ with $\lim_{r \to 0} \omega(r)=0$ and assume that for all $j \in \N$, the measured Lorentzian length spaces induced by $(M, g_j)$ have timelike Ricci curvature bounded above by $K$ in the synthetic sense, with respect to $p$, $r_0$ and remainder function $\omega$.
    
    Then $(M, g_\infty)$ has timelike Ricci curvature bounded above by $K$ in the synthetic sense, with respect to $p$, $r_0$ and remainder function $\omega$.
\end{theorem}

\begin{proof}
    Fix $x, y \in M$ with $\sfd_h(x, y) = r >0$ and {$B_{r^4}(x) \times B_{r^2}(y)  \Subset M^2_{\ll_\infty}$.} Then there exist  $\e > 0$ such that 
    \begin{align}
        \overline{B_{(r+\e)^4}(x)} \times \overline{B_{(r+\e)^2}(y)} \Subset M^2_{\ll_\infty} = \{\tau_\infty > 0\}.
    \end{align}
   Since by Proposition \ref{locally_uniform_convergence_tau} $\tau_\infty$ is continuous, by compactness there exists  $\delta> 0$ such that 
    \begin{align*}
        \overline{B_{(r+\e)^4}(x)} \times \overline{B_{(r+\e)^2}(y)} \subset \{\tau_\infty > 2\delta\}.
    \end{align*}
Using again  Proposition \ref{locally_uniform_convergence_tau}, we infer there exists  $J\in \N$ such that  
    \begin{align}\label{strict_tau_j_positive_support}
         \overline{B_{(r+\e)^4}(x)} \times  \overline{B_{(r+\e)^2}(y)} \subset \bigcap_{j \geq J} \{\tau_j > \delta\}.
    \end{align}

    \textbf{Step 1}. Fix $\mu_0^\infty \in \Dom(\Ent(\cdot | \mm_\infty))$  with $\supp \, \mu_0^\infty \subset B_{r^4}(x)$. For $j \geq J$, we will  construct measures $\mu_0^j$ such that $\supp \, \mu_0^j \subset B_{r^4}(x)$, $\mu_0^j \in \Dom(\Ent(\cdot | \mm_j))$ and
    \begin{align}\label{middle_measure_entropy}
        \Ent(\mu_0^\infty | \mm_\infty) \geq \limsup_{j \to \infty} \Ent(\mu^j_0| \mm_j). 
    \end{align}
    Since $\mm_j = \vol_{g_j} \to \vol_{g_\infty} = \mm_\infty$ narrowly, the fact that there exists  $R \in (2r^4, 3r^4)$ such that $B_R(x)$ is a continuity set of $\mm_\infty$ yields
    \begin{align}\label{eq:zj-1UB}
        z_j^{-1}:= \mm_j(B_R(x)) \to \mm_\infty(B_R(x)) =: z_\infty^{-1}.
    \end{align}
    For $j \in \N \cup \{\infty\}$, denote 
    \begin{equation}\label{eq:tildemjUB}
    \Tilde{\mm}_j := z_j \mm_j \LL B_R(x).
    \end{equation}
    As in the proof of Lemma \ref{restriction_of_plans_converge}, we get that $\Tilde{\mm}_j \to \Tilde{\mm}_\infty$ narrowly. As their supports are uniformly bounded in $M$, we get that $\Tilde{\mm}_j \to \Tilde{\mm}_\infty$ in the 2-Wasserstein distance $W^{(X, \sfd_h)}_2$. Let 
    \begin{align}
        \boldsymbol{\gamma}_j \in \Pi(\Tilde{\mm}_\infty, \Tilde{\mm}_j)\ \mathrm{be\ a\ } W^{(X, \sfd_h)}_2\mathrm{-optimal\ coupling}.
    \end{align}
Let us first consider the case $\mu^\infty_0 = \rho \Tilde{\mm}_\infty$ with $\rho \in L^\infty(\mm_\infty)$. Define $\boldsymbol{\gamma}'_j \in \mathcal{P}(M^2)$ as $\di\boldsymbol{\gamma}'_j(x, y) = \rho(x) \di\boldsymbol{\gamma}_j(x, y)$ and define $\mu^j_0:= (P_2)_\# \boldsymbol{\gamma}_j' \in \mathcal{P}(M)$. By construction, we have that $\boldsymbol{\gamma}_j' \ll \boldsymbol{\gamma}_j$, hence $\mu^j_0 \ll (P_2)_\# \boldsymbol{\gamma}_j = \Tilde{\mm}_j$. Let $\mu^j_0 = \rho^j_0 \Tilde{\mm}_j$. It follows from the definition that 
\begin{align*}
    \rho^j_0(y) = \int \rho(x) \di (\boldsymbol{\gamma}_j)_y(x), 
\end{align*}
where $\{(\boldsymbol{\gamma}_j)_y\}$ is the disintegration of $\gamma_j$ with respect to the projection on the second marginal. By Jensen's inequality applied to the convex function $u(s) = s \log s$, we have 
\begin{align}
    \Ent(\mu^j_0 | \Tilde{\mm}_j) &= \int u(\rho^j_0) \di \Tilde{\mm}_j = \int u \Bigg(  \int \rho(x) \di (\boldsymbol{\gamma}_j)_y(x), \Bigg) \di \Tilde{\mm}_j \nonumber \\
    & \leq \int \int  u(\rho(x)) \di (\boldsymbol{\gamma}_j)_y(x)\di \Tilde{\mm}_j = \int \int u(\rho(x)) \di\boldsymbol{\gamma}_j(x, y) \nonumber \\
    &= \int u(\rho) \di (P_1)_\#\boldsymbol{\gamma}_j = \int u(\rho) \di \Tilde{\mm}_\infty = \Ent(\mu^\infty_0 | \Tilde{\mm}_\infty).
\end{align}
Since by construction  $\boldsymbol{\gamma}_j' \in \Pi(\mu_0^\infty, \mu_0^j)$, it follows that 
\begin{align*}
     W^{(X, \sfd_h)}_2(\mu_0^\infty, \mu_0^j)^2 &\leq \int \sfd_h^2(x, y) \di \boldsymbol{\gamma}'_j(x, y) = \int \rho(x) \sfd_h^2(x, y) \di \boldsymbol{\gamma}_j(x, y) \leq \norm{\rho}_{L^\infty(\Tilde{\mm}_\infty)} W^{(X, \sfd_h)}_2(\Tilde{\mm}_\infty, \Tilde{\mm}_j)^2,
\end{align*}
hence $ W^{(X, \sfd_h)}_2(\mu_0^\infty, \mu^j_\infty) \to 0$, which implies that $\mu^j_0 \to \mu^\infty_0$ narrowly. \\
If $\rho$ is not essentially bounded, we proceed as follows: For $k \in \N$, define $\rho^k:= C_k \min\{\rho, k\}$, where 
\begin{align*}
    C_k^{-1} = \int_M \min\{\rho, k\}(x) \di \Tilde{\mm}_\infty(x),
\end{align*}
and set $$\mu_{k, \infty}:= \rho^k \Tilde{\mm}_\infty \in \mathcal{P}(M).$$ Then $C_k \to 1$ as $k \to \infty$ and $\rho_k \leq C_k \rho$, so we can proceed as in \eqref{rho_entropy_comparison} and \eqref{entropy_under_constant_factor} to get that
\begin{align*}
    \limsup_{k \to \infty} \Ent(\mu_{k, \infty}| \Tilde{\mm}_\infty) \leq \Ent( \mu_\infty| \Tilde{\mm}_\infty),
\end{align*}
and $\mu_{k, \infty} \to \mu^\infty_0$ narrowly. 
Applying the previous procedure to $\mu_{k, \infty}$ together with a diagonal argument yields \eqref{middle_measure_entropy}.

\smallskip

    \textbf{Step 2}. By assumption,  for every $j \geq J$ there exists  an $\ell_p^j$-geodesic $(\mu^j_t)_{t \in [-1, 1]}$ with $\supp\, \mu^j_1 \subset B_{r^2}(y)$, $\bigcup_{t \in [-1, 1]} \supp \, \mu^j_t \subset B_{10r_0}(x)$ such that 
\begin{align}\label{upper_bound_for_g_j}
    \Ent(\mu^j_{-1}| {\mm}_j) -2\Ent(\mu^j_{0}|  {\mm}_j) + \Ent(\mu^j_{1}| {\mm}_j) \leq (K + \omega(r))r^2.
\end{align}
Let $\eta_j \in \mathcal{P}([-1, 1], M)$ be the optimal dynamical plan associated to the geodesic $(\mu^j_t)_t$. Then $\eta_j(\mathrm{Geo}^j(X)) =~1$. Using  \eqref{strict_tau_j_positive_support}, we infer that 
\begin{align*}
    \supp \, (\ee_0)_\# \eta_j \times  \supp \, (\ee_1)_\# \eta_j \subset \{ \tau_j > \delta\}.
\end{align*}
As in \eqref{projection_support_inclusion}, it holds that 
\begin{align*}
    \ee_0(\supp\, \eta_j) \subset \supp\, ((\ee_0)_\# \eta_j).
\end{align*}
Hence, for each geodesic $\lambda \in \supp\, \eta_j$, we get that $\tau_j(\lambda(0), \lambda(1)) > \delta$, yielding that $$\tau_j(\lambda(-1), \lambda(1)) \geq 2\delta.$$ 
Consider the set $$T_{\delta}:= \bigcup_j T_\delta^j \subset C([-1,1], \overline{B_{10r_0}(x)}),$$ where $T_\delta^j$ is defined as in \eqref{def:eps_steep_geodesics} with $K = \overline{B_{10r_0}(x)}$. The above construction ensures that $\eta_j(T_\delta) = 1$. 
Analogously to the second step in the proof of Proposition \ref{dynamical_plans_precompactness}, we get that $T_\delta$ is precompact in $C([-1,1], M)$ and hence, by Prokhorov's Theorem, we infer that there exist  $\eta_\infty \in \mathcal{P}(C([-1,1], \overline{B_{10r_0}(x)}))$ such that 
\begin{equation}\label{eq:etajtoetainfNar}
\eta_j \to \eta_\infty \quad \text{narrowly}. 
\end{equation}
The continuity of the evaluation map and the uniqueness of narrow limits ensure that
\begin{align*}
    (\ee_0)_\# \eta_\infty = \mu_0^\infty.
\end{align*}
Moreover, setting $\pi_\infty:= (\ee_{-1}, \ee_1)_\# \eta_\infty$, it holds that $\pi_\infty (\{\tau_\infty \geq \delta\}) = 1$. Hence Lemma \ref{limit_plan_is_dynamical} implies that $\eta_\infty(\mathrm{TGeo}^\infty(M))=1$. 
Moreover, Lemma \ref{time_order_support_preserved} gives that 
\begin{align*}
         \supp \, \mu^\infty_0 \times \supp \, \mu^\infty_1 \subset M^2_{\leq_\infty}.
    \end{align*}
Denote $\pi_j:= (\ee_{-1}, \ee_1)_\# \eta_j$. From \eqref{eq:etajtoetainfNar}, we infer that  $\pi_j \to \pi_\infty$ narrowly. Lemma \ref{cyclical_monotonicity_limit}, ensures that $\pi_\infty$ is $\tau^p_\infty$-cyclically monotone. By  \cite[Remark 2.7]{cavalletti2020optimal}, it follows that $\pi_\infty$ is  $\ell^p_\infty$-cyclically monotone. By \cite[Prop.\ 2.8]{cavalletti2020optimal}, using that $\pi_\infty(M^2_{\ll_\infty}) = 1$, it follows that $\pi_\infty$ is optimal, and thus $\eta_\infty$ indeed represents an $\ell^p_\infty$-geodesic. The lower semicontinuity of $\Ent$ under joint narrow convergence (see, e.g.,  \cite[Eq.\ (1.15)]{cavalletti2020optimal}), gives that 
\begin{align*}
    \Ent(\mu_t^\infty| \Tilde{\mm}_\infty) \leq \liminf_{j \in \N} \Ent(\mu_t^j| \Tilde{\mm}_j).
\end{align*}
Hence, recalling \eqref{eq:zj-1UB} and \eqref{eq:tildemjUB},
\begin{align}\label{behaviour_of_entropy_whole_geodesic}
    \Ent(\mu_t^\infty| \mm_\infty) \leq \liminf_{j \in \N} \Ent(\mu_t^j| \mm_j).
\end{align}
Finally, combining \eqref{middle_measure_entropy}, \eqref{upper_bound_for_g_j}, and \eqref{behaviour_of_entropy_whole_geodesic}  we conclude that
\begin{align*}
     \Ent(\mu^\infty_{-1}| \mm_\infty) -2\Ent(\mu^\infty_{0}|  \mm_\infty) + \Ent(\mu^\infty_{1}| \mm_\infty) \leq (K + \omega(r))r^2.
\end{align*}
\end{proof}

\subsection{A sufficient condition for uniform global hyperbolicity}\label{SS:SuffUGH}

Uniform global hyperbolicity plays a crucial role in Theorems \ref{stability_tcd} and \ref{stability_upper_bounds}. However, this assumption is not automatic: there exist sequences of smooth globally hyperbolic Lorentzian metrics $g_j$ on $\mathbb{R}^{d+1}$, $d \ge 2$, converging locally uniformly to a smooth globally hyperbolic limit $g_\infty$, which fail to be uniformly globally hyperbolic (see Appendix \ref{subsec:counterexample_uniform_hyperbolicity}). In this subsection, we provide a sufficient condition ensuring uniform global hyperbolicity.
\smallskip

\begin{lemma}\label{stability_of_past_and_future_for_surfaces}
     Let $M$ be a smooth manifold and $g_j \in C^\infty, g_\infty \in C^0$  Lorentzian metrics on $M$ such that $g_j \to g_\infty$ locally uniformly.
     Assume that:
     \begin{itemize}
     \item there exists a continuous vector field $X$, serving as a time orientation for all $g_j$, $j \in \N \cup \{\infty\}$;
     \item all $g_\infty$-maximising curves have a causal character and are Lipschitz continuous when parametrized with respect to arclength;
     \item there exists a closed set $S\subset M$ such that $S$ is a Cauchy hypersurface for all spacetimes $(M, g_j)$, $j \in \N \cup \{\infty\}$.
     \end{itemize}
       Let $p \in J^-_{g_\infty}(S)\setminus S$ (or, respectively, $J^+_{g_\infty}(S)\setminus S$). Then there exists  $J \in \N, \theta >0$ such that $B_\theta^h(p) \subset J^-_{g_j}(S)\setminus S$ (resp.\ $J^+_{g_j}(S)\setminus S$), for all $j \geq J$.
\end{lemma}
\begin{proof}
    Since $S$ is a Cauchy hypersurface in $(M,g_\infty)$,  \cite[Lemma 5.4]{samann2016global} implies that $J^\pm_{g_\infty}(S)\setminus S = I^\pm_{g_\infty}(S)$.
    Consider the case $p \in J^-_{g_\infty}(S)\setminus S$. Take a vector $v \in T_pM$ such that $v$ is future directed and timelike. Since the statement is local, we may choose local  coordinates based  and work in a ball in $\R^{d+1}$. By the continuity of $g_\infty$, we get that for $\e>0$ small enough, the curve $\Tilde{\gamma}:[0, \e] \to M, t \mapsto p+ tv$ is timelike. We may then extend $\Tilde{\gamma}$ to a causal curve $\gamma:[0,1] \to M$ that satisfies $\gamma|_{[0, \frac{\e}{2}]} = \Tilde{\gamma}|_{[0, \frac{\e}{2}]}$ and such that $q:= \gamma(1) \in S$.  
    Notice that $\tau_{\infty}(p, q) >0$. Since the existence of a Cauchy hypersurface implies that $g_\infty$ is globally hyperbolic and geodesic \cite[Thm.\ 5.7, Prop.\ 6.4]{samann2016global}, the second assumption implies that there exists a Lipschitz continuous $g_\infty$-geodesic $\lambda:[0,1] \to M$ with $\lambda(0)=p$ and $\lambda(1) = q$. In particular, $|\dot{\lambda}|_h$ is bounded from above and $g_\infty(\dot{\lambda}, \dot{\lambda})$ is bounded away from zero. With a similar argument as in the second step of the proof of Lemma \ref{tau_uniform_equicontinuity}, one can find  $\theta, L, \kappa, R >0$ such that around for each $p' \in B_\theta^h(p)$ there exists an $L$-Lipschitz $g_\infty$-causal curve $\lambda_{p'}:[0,1]\to M$ such that $\lambda_{p'}(0) = p'$, $\lambda_{p'}(1) = q$, $\sfd_h(\lambda, \lambda_{p'}) \leq R$ and $|g_\infty(\dot{\lambda}_{p'}, \dot{\lambda}_{p'})| \geq \kappa$. 
    The locally uniform convergence $g_j\to g_\infty$ implies that there exists  $J \in \N$ such that 
    $$\norm{g_\infty-g_j}_{C^0(B_{2R}(\lambda([0,1])))} < \frac{\kappa}{4(1+L)^3},\quad \text{for all } j \geq J.$$ Then all $\lambda_{p'}$ are $g_j$-causal, which proves the lemma for $p \in J^-_{g_\infty}(S)\setminus S$. The argument for $p\in J^+_{g_\infty}(S)\setminus S$  is analogous.
\end{proof}

\begin{proposition}\label{common_cauchy_surface_uniform_global_hyperbolicity}
    Let $M$ be a smooth manifold and $g_j, g_\infty$ be as in Lemma \ref{stability_of_past_and_future_for_surfaces}. Then $(g_j)_{j\in \N}$ are uniformly globally hyperbolic. 
\end{proposition}
\begin{proof}
\textbf{Step 1.}
    Take two points $p, q \in M$. Up to reversing the time orientation, we may  assume that $p \in J^-_{g_\infty}(S)$. Pick a sequence $(x_k)_{k\in \N}$ such that 
    $$x_k \in \bigcup_{j\in \N} (J^+_{g_j}(p) \cap J^-_{g_j}(q)), \quad \text{ for all } k\in \N.$$
    We will prove that there exists a converging subsequence. 
    If there exists  $j_0 \in \N$ such that infinitely many  $x_k$'s lie in $J^+_{g_{j_0}}(p) \cap J^-_{g_{j_0}}(q)$, we get a convergent subsequence, as we already know that the existence of a Cauchy hypersurface for $(M, g_j)$ implies the compactness of causal diamonds and hence the existence of an accumulation point.
    Hence, relabeling and taking a subsequence, we may assume that  
    $$x_j \in J^+_{g_j}(p) \cap J^-_{g_j}(q), \quad \text{for all }j\in \N.$$  Assume by contradiction that the sequence does not have an accumulation point. Hence, we can pick another subsequence and relabel such that $x_j \in B_j^h(\{p,q\})^c$. 
    Pick $g_j$-causal curves $\gamma_j:[0, c_j]\to M$ from $p$ to $q$ that pass through $x_j$ and assume them to be parametrised by unit $h$-speed. Notice that $$\gamma_j([0,k])\cup \gamma_j([c_j-k, c_j]) \subset B_k^h(\{p,q\}), \quad \text{for each $k \in \N$}.$$
    Define $\gamma^-_j:[0, c_j] \to M, t \mapsto \gamma_j(c_j-t)$. 
    By the limit curve theorem (see e.g.\ \cite{minguzzi2008limit}) and a diagonal argument to get two limit curves $\gamma_p$, $\gamma_q$ such that, passing to a subsequence, $\gamma_j \to \gamma$ and $\gamma^-_j \to \gamma_q$ locally uniformly. Moreover, the same argument as in the proof of Proposition \ref{Lipschitz_bd_causal_curves} shows that $\gamma_p, \gamma_q$ are $g_\infty$-causal. By the assumption on the sequence $x_j$, we get that $c_j \to \infty$ as $j \to \infty$. 
    \smallskip
    
   \textbf{Step 2.} We next show  that both curves $\gamma_p, \gamma_q$ need to be inextendible. If $\gamma_p$ were extendible, then the limit $$\lim_{t \to \infty} \gamma_p(t) := p_\infty \in M$$ exists and hence there exists a compact set $K$ such that $\gamma([0, \infty)) \subset K$. Consider the compact set $$B:=\overline{B_1^h(K)}\subset M.$$ By  \cite[Thm.\ 4.5]{samann2016global} there exists a globally hyperbolic metric $g' \succ g_\infty$ on $B$. Then, for $J_1 \in \N$ large enough, we have that $g_j \prec g'$ on $B$, for all $j\geq J_1$. 
    By Proposition \ref{Lipschitz_bd_causal_curves},  there exists  $C >0$ such that all $g'$-causal curves in $B$ have  $h$-length at most $C$. Let $j^* \geq 2C + J_1$ be such that  $$\sfd_h(\gamma_{j^*}(t), \gamma_p(t)) \leq \frac{1}{2},\quad \text{ for all $t \leq 2C$}.$$ Hence $\gamma_{j^*}([0, 2C]) \subset B$ and it is $g'$-causal, thus its  $h$-length is at most $C$. That gives a contradiction, as $\gamma_{j^*}$ is parametrised by $h$-arclength and hence has length $2C$. The proof of the inextendibility of  $\gamma_q$ is analogous.
\smallskip
    
   \textbf{Step 3.}
    Recall that $p \in J_{g_\infty}^-(S) $. Let  $t_p \geq0$ be such that $\gamma_p(t_p) \in S$, and set $y := \gamma_p(t_p+1) \in J_{g_\infty}^+(S) \setminus S$. By Lemma \ref{stability_of_past_and_future_for_surfaces}, there exist $\theta >0$ and $J_1 \in \N$ such that  $B_\theta^h(y) \subset J_{g_j}^+(S) \setminus S$, for all $j \geq J_1$.
    
    \textit{Case 1:} $q \in J_{g_\infty}^-(S)$. \\
    Let $J_2 > 0$ be such that  $c_j \geq t_p+2$ and $\gamma_j(t_p+1) \in B_{\theta/2}^h(y)$, for all  $j \geq J_2$. Then 
    $$\gamma_j(0) \in J_{g_j}^-(S),\quad  \gamma_j(t_p+1) \in J_{g_j}^+(S)\text{ and } \gamma_j(c_j) \in J_{g_j}^-(S),$$
    for all $j \geq J_1+J_2$. By continuity, this implies that $\gamma_j$ meets $S$ at least twice, contradicting that $S$ is a Cauchy hypersurface.

      \textit{Case 2:} $q \in J_{g_\infty}^+(S) \setminus S$. \\
      Let $t_q > 0$ be such that $\gamma_q(t_q) \in S$. Define $x := \gamma_q(t_q+1)$. By  Lemma \ref{stability_of_past_and_future_for_surfaces}, we may  increase $J_1$ and shrink $\theta$ such that $B_\theta^h(x) \subset J_{g_\infty}^-(S) \setminus S$,  for all $j \geq J_1$.
      Let  $J_2 > 0$ be such that  
      $$\gamma_j(t_p+1) \in B_{\theta/2}^h(y)\quad  \text{ and } \quad\gamma_j(c_j-t_q-1) \in B_{\theta/2}^h(x), \quad \text{for all }j \geq J_2.$$ 
      Recalling that $c_j \to \infty$ as $j \to \infty$, we get that there exists  $J_3>0$ such that  $c_j > t_p + t_q +3$, for all $j \geq J_3$. Now for  $j \geq J_1 + J_2 + J_3$, we infer that 
      $$\gamma_j(0) \in J_{g_j}^-(S),\quad  \gamma_j(t_p+1) \in J_{g_j}^+(S), \quad \text{and} \quad  \gamma_j(c_j-t_q-1) \in J_{g_j}^-(S).$$ By continuity, this implies that $\gamma_j$ meets $S$ at least twice, contradicting that $S$ is a Cauchy hypersurface.

      \textbf{Step 4}. So far, we gave the proof in the case where both the compact sets $K_1$ and $K_2$ are singletons, i.e., $K_1=\{p\}$, $K_2=\{q\}$. The general case is analogous, let us briefly discuss it.
Pick a sequence $(x_k)_{k\in \N}$ such that 
    $$x_k \in \bigcup_{j\in \N} (J^+_{g_j}(K_1) \cap J^-_{g_j}(K_2)), \quad \text{ for all } k\in \N.$$
    We want to show that  there exists a converging subsequence. Since $K_1$ and $K_2$ are compact, up to subsequences, we may assume that there exists $p_j, p\in K_1$ and $q_j, q\in K_2$ such that: 
    $$
    x_j \in  (J^+_{g_j}(p_j) \cap J^-_{g_j}(q_j)), \quad p_j\to p, \quad q_j\to q.
    $$
    Assume by contradiction that the sequence $(x_j)$ does not have an accumulation point.  Hence, we can pick another subsequence and relabel such that $x_j \in B_j^h(\{p,q\})^c$. 
    Pick $g_j$-causal curves $\gamma_j:[0, c_j]\to M$ from $p_j$ to $q_j$ that pass through $x_j$ and assume them to be parametrised by unit $h$-speed. Define $\gamma^-_j:[0, c_j] \to M, t \mapsto \gamma_j(c_j-t)$. As before, one can apply the limit curve theorem and obtain two limit curves $\gamma_p$, $\gamma_q$ such that, up to a subsequence, $\gamma_j\to \gamma$ and $\gamma^-_j \to \gamma_q$ locally uniformly. Moreover, both $\gamma_p, \gamma_q$ are $g_\infty$-causal. By the assumption on the sequence $x_j$, we get that $c_j \to \infty$ as $j \to \infty$. From here one can repeat verbatim the arguments in steps 2 and 3 to reach a contradiction.
\end{proof}

\section{Application to nonexpanding impulsive gravitational waves}\label{Sec:3}

The aim of this section is to show that Penrose’s impulsive gravitational waves fit into the framework of measured Lorentzian geodesic spaces with synthetic lower and upper bounds on timelike Ricci curvature, as stated in Theorems \ref{thm:3.10Intro} and \ref{thm:IPP-UB-Intro}. We briefly outline the strategy.

In \cite{podolsky2014global}, Podolsk{\`y}, Steinbauer, \v{S}varc, and the third author constructed nonexpanding impulsive gravitational waves in constant curvature spacetimes. We focus on the Lipschitz continuous Rosen form of the metric (see \cite[Section~2.2]{podolsky2014global}), for which we compute the distributional Ricci tensor and identify the conditions under which it admits lower and/or upper bounds.

In the bounded case, we approximate the Lipschitz metric locally uniformly by smooth metrics with corresponding curvature bounds. We then perform a further modification to ensure uniform global hyperbolicity, thereby enabling the application of the results from the previous section. The construction of these smooth approximations proceeds in several steps in order to maintain control of the Ricci curvature; we refer to the beginning of Section \ref{SS:SmoothApproxIPP} for an overview.

\subsection{Analysis of the non-smooth metric}
Let us start by recalling the explicit form of the continuous Lorentzian metric describing an IPP wave (see for instance \cite[Subsec.\ 2.2]{podolsky2014global}).  Consider the coordinates $U, V, Z, \Bar{Z}$, where $U, V \in \R$ and $Z, \cz \in \C$ are given by $Z=x +iy$ and $\Bar{Z} = x-iy$. Let $\hh = \hh(Z, \cz)$ be a real-valued function of regularity $C^{3,1}_{loc}(\C^2)$. Let $\Lambda \in \R$ be the cosmological constant. Define the function 
\begin{align*}
    U^+:= \max(0, U).
\end{align*}
Consider the bilinear form 
\begin{align}\label{nonsmooth_metric}
    ds_0^2 = \frac{2\big|dZ + U^+(\partial^2_{Z, \cz}\hh dZ + \partial^2_{\cz, \cz}\hh d\cz)\big|^2-2dUdV}{\big[1+\frac{1}{6}\Lambda(Z\cz -UV - U^+G)\big]^2},
\end{align}
where $G = G(Z, \cz):= \hh-Z \partial_Z\hh -\cz \partial_{\cz}\hh$. 
It can also be expressed in matrix form as
\begin{align*}
    g_{(\Lambda)} = \frac{1}{\big[1+\frac{1}{6}\Lambda(Z\cz -UV - U^+G)\big]^2}\begin{pmatrix}
    0 & -1 & 0 & 0 \\
    -1 & 0 & 0 & 0 \\
    0 & 0 & {g_{33}} & {g_{34}} \\
    0 & 0 &  {g_{34}}  & {g_{44}} 
    \end{pmatrix},
\end{align*}
with
\begin{align}\label{explicit_metric_coefficients}
    {g}_{33} &=  2\partial^2_{Z, Z}\hh U^+ +2\partial^2_{Z,\Bar{Z}}\hh\partial^2_{Z, Z}\hh (U^+)^2, \nonumber \\
     {g}_{44} &=  2\partial^2_{\cz, \cz}\hh  U^+ +2\partial^2_{Z,\Bar{Z}}\hh\partial^2_{\cz, \cz}\hh (U^+)^2, \nonumber \\
     {g}_{43}= {g}_{34} &= 1+2\partial^2_{Z,\Bar{Z}}\hh U^+ +((\partial^2_{Z,\Bar{Z}}\hh)^2+|\partial^2_{\cz, \cz}\hh|^2)(U^+)^2.
\end{align}

For any continuous symmetric $2$-form $\alpha$ on an open set $O \subset \R^4$ denote by $\mathrm{Dom}_{\mathrm{Lor}}(\alpha)$ the open subset of $O$, where $\alpha$ has Lorentzian signature. 
Whenever the value of $\Lambda$ is clear from context, we will simply write $g$ instead of $g_{(\Lambda)}$.
We denote by $M_\Lambda$ the open subset of $\R^4$ in which $g$ is well-defined and has Lorentzian signature, i.e $\domlor(g)= \domlor(g_{(\Lambda)})$. Moreover, we will sometimes write $(M^0_\Lambda, g_0)$ for the case $\hh \equiv 0$.

Before we continue with the analysis of the metric, we need some distributional calculus. Denote by $\mathcal{D}'=\mathcal{D}'(\R^4)$ the space of distributions in $\R^4$.
Define the function $\Theta: \R^n \to \R^n$ as $\Theta(x) = \mathbbm{1}_{x_1 \geq 0}$. Note that $\Theta \in L^\infty \hookrightarrow \mathcal{D}'$.
\medskip

\begin{proposition}\label{distributional_derivative_heavisidy}
    It holds that $\partial_1 \Theta = \mathcal{S}^{n-1}\LL \{x_1 = 0\} \in \mathcal{D}'$, where $\mathcal{S}^{n-1}$ denotes the $(n-1)$-dimensional Hausdorff measure, hence a distribution of order zero.  
\end{proposition}
\begin{proof}
    Let $\phi \in C_c^{0} \cap W^{1, 1}_{loc} (\R^n)$. Then 
    \begin{align*}
        \langle \partial_1 \Theta, \phi \rangle_{\mathcal{D}', \mathcal{D}} = - \int_{\R^n} \partial_1 \phi \Theta dx = -\int_{\R^{n-1}} \int_{0}^\infty \partial_1 \phi (x_1, x') dx_1 dx' = \int_{\R^{n-1}} \phi(0, x') dx'.
    \end{align*}
\end{proof}

\subsubsection{The distributional Ricci curvature tensor}\label{subsec:distri_ricci_curv}
We start by computing the distributional Ricci curvature of \eqref{nonsmooth_metric}.
As $\partial_U \Theta$ is a distribution of order $0$, we get that  
\begin{align}
    U^+ \partial_U \Theta = 0. 
\end{align}
Noticing that $g, g^{-1} \in C^{0,1}_{loc}$, we can apply Proposition \ref{distributional_derivative_heavisidy}, to compute the distributional Ricci curvature via  
\begin{align*}
    &\Gamma_{ij,k} = \frac{1}{2}(\partial_ig_{jk} + \partial_jg_{ik} - \partial_kg_{ij}),  \\
    &R_{ijkl} = \partial_i \Gamma_{jk, l} -\partial_j \Gamma_{ik, l} + g^{st}(\Gamma_{ik, s}\Gamma_{jl, t}-\Gamma_{il, s}\Gamma_{jk, t}),  \\
    &\Ric_{ij} = g^{kl}R_{kijl}.
\end{align*}
Observing that 
\begin{align*}
    & U^+ \partial_U\Theta = 0 \ \mathrm{in}\ \mathcal{D}', \\
    &\Theta^2 = \Theta \in L^\infty, \\
    &(\Theta)^k (U^+)^l = (U^+)^l,\ \in C^0\ \mathrm{for\ integers}\ k, l \geq 1, \\
    & (U^+)^k (U^+)^l = (U^+)^{k+l},\ \in C^0\ \mathrm{for\ integers}\ k, l \geq 1,
\end{align*}
allows us to simplify the connection symbols and curvature coefficients.
We get that the non-zero curvature coefficients are given by

\begin{align*}
    R_{1221} &= -\frac{\frac{1}{3} \Lambda}{\left(1+ \frac{\Lambda}{6}(|Z|^2-UV -U^+(\hh-Z\partial_Z \hh-\Bar{Z}\partial_{\cz}\hh))\right)^4}, \\
    R_{2331} &= -\frac{\frac{2}{3} {\partial^2_{Z,Z}\hh} \Lambda U^+ ({\partial^2_{Z, \Bar{Z}}\hh} U^+ +1)}{\left(1+ \frac{\Lambda}{6}(|Z|^2-UV -U^+(\hh-Z\partial_Z \hh-\Bar{Z}\partial_{\cz}\hh))\right)^4}, \\
    R_{2341} &= -\frac{\frac{1}{3} \Lambda \left({\partial^2_{Z,Z}\hh} (U^+)^2 {\partial^2_{\Bar{Z},\Bar{Z}}\hh}+({\partial^2_{Z, \Bar{Z}}\hh} U^+ +1)^2\right)}{\left(1+ \frac{\Lambda}{6}(|Z|^2-UV -U^+(\hh-Z\partial_Z \hh-\Bar{Z}\partial_{\cz}\hh))\right)^4}, \\
    R_{2431} &= -\frac{\frac{1}{3} \Lambda \left({\partial^2_{Z,Z}\hh} (U^+)^2 {\partial^2_{\Bar{Z},\Bar{Z}}\hh}+({\partial^2_{Z, \Bar{Z}}\hh} U^+ +1)^2\right)}{\left(1+ \frac{\Lambda}{6}(|Z|^2-UV -U^+(\hh-Z\partial_Z \hh-\Bar{Z}\partial_{\cz}\hh))\right)^4},
\end{align*}
\begin{align*}
    R_{2441} &= -\frac{\frac{2}{3} \Lambda U^+ {\partial^2_{\Bar{Z},\Bar{Z}}\hh} ({\partial^2_{Z, \Bar{Z}}\hh} U^++1)}{\left(1+ \frac{\Lambda}{6}(|Z|^2-UV -U^+(\hh-Z\partial_Z \hh-\Bar{Z}\partial_{\cz}\hh))\right)^4}, \\
    R_{3443} &= -\frac{\frac{1}{3} \Lambda \left(({\partial^2_{Z, \Bar{Z}}\hh} U^+ +1)^2-{\partial^2_{Z,Z}\hh} (U^+)^2 {\partial^2_{\Bar{Z},\Bar{Z}}\hh}\right)^2}{\left(1+ \frac{\Lambda}{6}(|Z|^2-UV -U^+(\hh-Z\partial_Z \hh-\Bar{Z}\partial_{\cz}\hh))\right)^4}, \\
    R_{3114} &= -\frac{\partial^2_{Z, \cz} \hh \mathcal{S}^3\LL\{U =0\}}{(1+\frac{\Lambda}{6}|Z|^2)^2} -\frac{\Lambda(\hh - Z\partial_Z \hh -\cz \partial_{\cz}\hh) \mathcal{S}^3\LL\{U =0\}}{6(1+\frac{\Lambda}{6}|Z|^2)^3},\\
    R_{3113} &=-\frac{\partial^2_{Z, Z} \hh \mathcal{S}^3\LL\{U =0\}}{(1+\frac{\Lambda}{6}|Z|^2)^2}, \\
    R_{4114} &=-\frac{\partial^2_{\cz, \cz} \hh \mathcal{S}^3\LL\{U =0\}}{(1+\frac{\Lambda}{6}|Z|^2)^2}. 
\end{align*}
Finally, the Ricci curvature coefficients are given by 
\begin{align*}
    \Ric_{11} &= -2\partial^2_{Z, \cz} \hh \mathcal{S}^3\LL\{U =0\} -\frac{\Lambda(\hh - Z\partial_Z \hh -\cz \partial_{\cz}\hh) \mathcal{S}^3\LL\{U =0\}}{3(1+\frac{\Lambda}{6}|Z|^2)},\\
    \Ric_{12} &= -\frac{\Lambda}{\left(1+ \frac{\Lambda}{6}(|Z|^2-UV -U^+(\hh-Z\partial_Z \hh-\Bar{Z}\partial_{\cz}\hh))\right)^2},\\
    \Ric_{13} &= \Ric_{14} = \Ric_{22}= \Ric_{23}=  \Ric_{24}= 0,\\
    \Ric_{33} &= \frac{2 {\partial^2_{Z,Z}\hh} \Lambda U ({\partial^2_{Z, \Bar{Z}}\hh} U+1)}{\left(1+ \frac{\Lambda}{6}(|Z|^2-UV -U^+(\hh-Z\partial_Z \hh-\Bar{Z}\partial_{\cz}\hh))\right)^2},\\
    \Ric_{34} &= \frac{ \Lambda \left({\partial^2_{Z,Z}\hh} U^2 {\partial^2_{\Bar{Z},\Bar{Z}}\hh}+({\partial^2_{Z, \Bar{Z}}\hh} U+1)^2\right)}{\left(1+ \frac{\Lambda}{6}(|Z|^2-UV -U^+(\hh-Z\partial_Z \hh-\Bar{Z}\partial_{\cz}\hh))\right)^2},\\
    \Ric_{44} &= \frac{2 \Lambda U {\partial^2_{\Bar{Z},\Bar{Z}}\hh} ({\partial^2_{Z, \Bar{Z}}\hh} U+1)}{\left(1+ \frac{\Lambda}{6}(|Z|^2-UV -U^+(\hh-Z\partial_Z \hh-\Bar{Z}\partial_{\cz}\hh))\right)^2}.
\end{align*}
\subsubsection{Time orientation and timelike geodesics}
We will prove a general fact about Lorentzian metrics on an open set $\Omega \subset \R^4$,  that are of the form
\begin{equation}\label{eq:gpGeneral}
    g = \frac{1}{P^2}\begin{pmatrix}
    0 & -1 & 0 & 0 \\
    -1 & 0 & 0 & 0 \\
    0 & 0 & \Tilde{g}_{33} & \Tilde{g}_{34} \\
    0 & 0 &  \Tilde{g}_{43} & \Tilde{g}_{44}
    \end{pmatrix},
\end{equation}
where $\begin{pmatrix}
   \Tilde{g}_{33} & \Tilde{g}_{34} \\
    \Tilde{g}_{43} & \Tilde{g}_{44}
    \end{pmatrix}$
    is symmetric and positive definite on $\Omega$ and $P$ is a nowhere zero function on $\Omega$. The set of coordinates will be denoted by  $U, V, x, y$. Notice that $g$ defines a Lorentzian metric and that $\partial_U, \partial_V$ are null vectors. 

We choose $T=\partial_U + \partial_V$ as our \textbf{time orientation}.
\smallskip

\begin{proposition}\label{uv_monotone_along_causal_curves}
Let $g$ be a Lorentzian metric of the form \eqref{eq:gpGeneral}, defined on an open set $\Omega \subset \R^4$.
    For any causal, future directed vector $w \in T_p\Omega, p\in \Omega$, it holds
    \begin{align*}
        &\langle w, \partial_U  \rangle_g \leq 0, \\
        &\langle w, \partial_V  \rangle_g \leq 0.
    \end{align*}
\end{proposition}
\begin{proof}
    Notice that $w = (w_1, w_2, w_3, w_4)\in T_p\Omega$ is future directed if and only if
    \begin{align}\label{consequence_of_future_dir}
        \langle T, w \rangle_g = \frac{-w_1-w_2}{P^2} < 0.
    \end{align}
    The condition that $w$ is causal implies that $w \neq 0$ and 
    \begin{align*}
        \frac{-2w_1w_2 + \Tilde{g}_{33}w_3^2 + 2\Tilde{g}_{34}w_3w_4 + \Tilde{g}_{44}w_4^2}{P^2} \leq 0.
    \end{align*}
    Since $\Tilde{g}_{33}w_3^2 + 2\Tilde{g}_{34}w_3w_4 + \Tilde{g}_{44}w_4^2 \geq 0$, we get that 
    \begin{align}\label{consequence_causal}
        -2w_1w_2 \leq 0.
    \end{align}
    Now note that 
    \begin{align*}
        \langle w, \partial_U  \rangle_g = -\frac{w_2}{P^2}.
    \end{align*}
    Suppose $w_2 <0$. Then  \eqref{consequence_of_future_dir} yields that $w_1 >-w_2 >0$, forcing $-2w_1w_2 > 0$, thus contradicting \eqref{consequence_causal}. Hence $w_2 \geq 0$, showing $\langle w, \partial_U  \rangle_g \leq 0$. The argument for $\langle w, \partial_V \rangle_g$ is completely similar.
\end{proof}

\begin{proposition}[Impulsive gravitational waves are timelike non-branching]\label{prop-non-bra}
Every impulsive gravitational wave in the continuous form \eqref{nonsmooth_metric} (with arbitrary cosmological constant $\Lambda\in\R$)  is timelike non-branching. 
\end{proposition}
\begin{proof}
By \cite[Thm.\ 1.1]{LLS:21} we know that every maximizer has a parametrization as a (Filippov) geodesic and in that parametrization is $\mathcal{C}^{1,1}$. Moreover, we only need to consider timelike maximizers that cross the impulse (otherwise we are in the smooth background spacetime of constant curvature, which is non-branching). By \cite[Thm.\ 3.3]{podolsky2014global} all (causal) geodesics uniquely cross the impulse, hence are non-branching.
\end{proof}

\subsubsection{Global hyperbolicity}

In the results of this subsection it will be useful to use the ``cut-and-paste" formulation of the IPP-wave, see \cite[Ch.\ 20]{GP:09}. This form of the IPP-wave metric is related to the continuous form \eqref{nonsmooth_metric} by a so-called “discontinuous coordinate transformation” \cite{PSSS:19}. Nevertheless, this equivalence can be established in a mathematically rigorous way; see \cite{SSSS:24}.

Recall that in the ``cut-and-paste" form,   the IPP-wave metric on each side of the impulse is given by
\begin{align}\label{eq:IPP-C&P}
        d\Tilde{s}^2 = \frac{2 d \zz d\ccz -2d\uu d \vv}{\big[1+ \frac{1}{6}\Lambda(\zz \ccz - \uu \vv)\big]^2}.
    \end{align}
    We slice the space open along the hypersurface $\{\uu =0\}=: \mathcal{N}$ and identify the point $(\uu = 0^-, \vv, \zz, \ccz)$ with $(\uu = 0^+, \vv-\hh(\zz, \ccz), \zz, \ccz)$.

It is a classical fact that the constant curvature space for $\Lambda \geq 0$ is globally hyperbolic (see for instance \cite{harris1982triangle}). In the case $\Lambda <0$ one only gets compact diamonds of points $p, q$ if $\tau(p, q) < \sqrt{\frac{\pi}{-\Lambda}}$. When turning to the IPP-wave, the property of global hyperbolicity depends on the profile function $\hh$, cf.\ \cite{flores2006geometry} for a corresponding result in the smooth setting. In particular, a key assumption is that $\hh$ is subquadratic, i.e., 
\begin{equation}\label{eq:hhSubQuad}
\lim_{|Z|\to \infty} \frac{|\hh(Z, \cz)|}{|Z|^2} = 0.
\end{equation}

\begin{proposition}\label{non-smooth_metric_globally_hyperbolic}
    Let $\Lambda \geq 0$ and $\hh$ be subquadratic, i.e., \eqref{eq:hhSubQuad} holds. Then $(M_\Lambda, g)$ is globally hyperbolic. 
\end{proposition}
\begin{proof}
 Recall the ``cut-and-paste" form of the IPP-wave metric \eqref{eq:IPP-C&P}, with the identification of the point $(\uu = 0^-, \vv, \zz, \ccz)$ with $(\uu = 0^+, \vv-\hh(\zz, \ccz), \zz, \ccz)$.   For any point $p=(\uu_p, \vv_p, \zz_p, \ccz_p)$ with $\uu_p < 0$, we have that 
    \begin{align*}
        J^+(p) \cap \mathcal{N} = \{(0^-, V, Z, \cz): -\uu_p (V-\vv_p) \geq |Z-\zz_p|^2\}.
    \end{align*}
    Then for a point $q=(\uu_q, \vv_q, \zz_q, \ccz_q)$ with $\uu_q > 0$, we similarly get that
    \begin{align*}
        J^-(q) \cap \mathcal{N} = \{(0^+, V, Z, \cz): \uu_q (V-\vv_q) \leq -|Z-\zz_q|^2\}.
    \end{align*}
    Considering the warp, in the Cut-and-Paste-model one gets 
    \begin{align*}
        J^+(p) \cap  J^-(q) \cap \mathcal{N} =  \{(0^+, V, Z, \cz): \uu_q (V-\vv_q) \leq -|Z-\zz_q|^2, -\uu_p (V-\hh(Z, \cz) -\vv_p) \geq |Z-\zz_p|^2\}.
    \end{align*}
    More explicitly, this means 
    \begin{align}\label{interval_for_v}
        &V \leq \frac{-|Z-\zz_q|^2}{\uu_q}+\vv_q \\
        & V \geq \frac{|Z-\zz_p|^2}{-\uu_p} + \vv_p + \hh(Z, \cz), 
    \end{align}
    hence 
    \begin{align}\label{eq:IntervalV}
        V \in \Bigg[\frac{|Z-\zz_p|^2}{-\uu_p} + \vv_p + \hh(Z, \cz),\frac{-|Z-\zz_q|^2}{\uu_q}+\vv_q\Bigg].
    \end{align}
    \textbf{Claim}. The interval in \eqref{eq:IntervalV} is empty for $|Z|$ large enough.\\
    \emph{Proof of the claim.} We will proceed by contradiction and assume that there exist $Z$ with $|Z|$ arbitrarily large, such that the interval is non-empty. This implies that for such $Z$, it holds
    \begin{align*}
        |\hh(Z, \cz)| \geq -\hh(Z, \cz) \geq \frac{|Z-\zz_p|^2}{-\uu_p} + \vv_p-\vv_q +\frac{|Z-\zz_q|^2}{\uu_q}.
    \end{align*}
    For fixed $\uu_p, \uu_q, \vv_p, \vv_q, \zz_p, \zz_q$ as above, and $|Z|$ large enough, there is a constant $\alpha>0$ such that
    \begin{align*}
         \alpha|Z|^2 \leq \frac{|Z-\zz_p|^2}{-\uu_p} + \vv_p-\vv_q +\frac{|Z-\zz_q|^2}{\uu_q}  \leq |\hh(Z, \cz)|. 
    \end{align*}
    But this gives $\limsup_{|Z| \to \infty} \frac{|\hh(Z, \cz)|}{|Z|^2} \geq \alpha > 0$, which contradicts the assumption that $\hh$ is subquadratic. This proves the claim. 
    \smallskip

\textbf{Case $\Lambda=0$}. Thanks to the claim, there exists  $R >0$ such that 
   \begin{align*}
        J^+(p) \cap  J^-(q) \cap \mathcal{N} \subset \{ |Z| \leq R\}.
   \end{align*}
   Hence, as $\hh$ is locally bounded, we get that also that the values in the interval in \eqref{interval_for_v} are bounded. As $U \subset [\uu_p, \uu_q]$, it follows that all coordinates are bounded, and thus the causal diamond is contained in a compact set in $\R^4$.

   In the case $\uu_p=0$ (or $\uu_q=0$), we can  find points $p' \in I^-(p)$ such that $\uu_{p'} <0$ (or $q' \in I^+(q)$ such that $\uu_{q'} >0$) and argue with them instead. This finishes the proof in the case $\Lambda =0$. 
\smallskip

\textbf{Case $\Lambda>0$.}
In the case $\Lambda >0$, note that the zero set of the denominator $[1-\frac{\Lambda}{6}[\uu\vv-\zz \ccz]]^2$ consists of two acausal hypersurfaces $S^- \subset\{\uu < 0\}$ and $S^+ \subset \{\uu >0\}$ on which $\uu$ is bounded away from $0$ and that split the manifold into three open connected components $O_-, O_0, O_+$ given by
   \begin{align*}
       O_- = J^-(S^-),\; O_0 = J^+(S^-) \cap J^-(S^+),\; O_+ = J^+(S^+).
   \end{align*}
   For points $p, q$ that lie in the same connected component, the causal diamond has to stay in that component and hence equals the one from the flat case, which is compact. If $p, q$ lie in different connected components, then no causal curve can connect them, as it would need to cross $S^-$ or $S^+$.
\end{proof}

With similar arguments as in the previous proof, one can show the following result for the $\Lambda<0$ case.
\smallskip

\begin{proposition}\label{small_global_hyperbolicity_negative_cosmo}
    Let $\Lambda <0$ and $\hh$ be subquadratic, i.e.\ \eqref{eq:hhSubQuad} holds. Then for points $p, q\in M$, the causal diamond is compact if
     \begin{itemize}
         \item[$(i)$] $U(p), U(q)\leq 0$ and $\tau(p, q) < \sqrt{\frac{\pi}{-\Lambda}}$,
         \item[$(ii)$] $U(p), U(q)\geq 0$ and $\tau(p, q) < \sqrt{\frac{\pi}{-\Lambda}}$,
         \item[$(iii)$] $U(p)<0<U(q)$ and there exist $\alpha_1, \alpha_2 \geq 0$ such that $\alpha_1+\alpha_2 < \sqrt{\frac{\pi}{-\Lambda}}$ and 
         \begin{align*}
             \{p\}\times (J^+(p) \cap  J^-(q) \cap \mathcal{N})_- \subset \tau_\Lambda^{-1}((0, \alpha_1)), \\
            (J^+(p) \cap  J^-(q) \cap \mathcal{N})_+ \times \{q\} \subset \tau_\Lambda^{-1}((0, \alpha_2)).
         \end{align*}
     \end{itemize} 
\end{proposition}

\subsection{Smooth approximations of the continuous IPP-wave metric}\label{SS:SmoothApproxIPP}

The goal of this subsection is to construct smooth metrics $g_\e$ that approximate locally uniformly the continuous IPP-wave $g$ as in \eqref{nonsmooth_metric},  so that the Ricci curvature converges locally uniformly in the following sense.
Recall that by section \ref{subsec:distri_ricci_curv}, $\Ric[g]$ is a measure. 
Write $\Ric[g] = \Ric[g]^{a} + \Ric[g]^s$, where $\Ric[g]^{a}$ denotes the absolutely continuous part with respect to the volume measure and $\Ric[g]^{s}$ denotes the singular part. We construct $g_\e$ such that the Ricci curvature tensor can be decomposed into $\Ric[g_\e] = \Ric[g_\e]^{(a)} + \Ric[g_\e]^{(s)}$ such that 
\begin{itemize}
    \item[(i)]  $\Ric[g_\e]^{(a)} \to \Ric[g]^{a}$ locally uniformly (when seen as $L^\infty$-functions).
    \item[(ii)]  $\Ric[g_\e]^{(s)} \to \Ric[g]^{s}$ distributionally. 
    \item[(iii)] If $\Ric[g]^{s}$ is non-negative/non-positive, then for all $\e >0$, $\Ric[g_\e]^{(s)}$ is non-negative/non-positive.
\end{itemize}
 We quickly explain the strategy to achieve that. In a first step,  we will naively replace every $U^+$ with a $\rho_\e * U^+$ for a symmetric sequence of non-negative standard mollifiers $\rho_\e = \frac{1}{\e}(\rho(\frac{U}{\e}))$ and  compute the Ricci curvature tensor, using the classical equations:
\begin{align}
    &\Gamma_{ij,k} = \frac{1}{2}(\partial_ig_{jk} + \partial_jg_{ik} - \partial_kg_{ij}), \label{connection_first_kind} \\
    &R_{ijkl} = \partial_i \Gamma_{jk, l} -\partial_j \Gamma_{ik, l} + g^{st}(\Gamma_{ik, s}\Gamma_{jl, t}-\Gamma_{il, s}\Gamma_{jk, t}), \label{low_index_riem_curvature} \\
    &\Ric_{ij} = g^{kl}R_{kijl}. \label{ricci_coefficients_formula}
\end{align}

In order to discuss the behavior of the Levi--Civita connection symbols and the Riemannian and Ricci curvature coefficients, we need to examine products of the derivatives of $\rho_\e * (U^+)$ and compare them to the ones of $U^+$ on some fixed compact set $\mathcal{K}$. Recall that  $\partial_U U^+ = \Theta$ and $\partial_U \Theta = \delta_0$ in $\mathcal{D}'(\R)$. More precisely, if we consider $U^+, \Theta$ as a function of $U, V, Z, \cz$, then $\partial_U \Theta = \mathcal{S}^3\LL \{U=0\} = \delta_0 \times \mathcal{L}^3 \in \mathcal{D}'(\R^4)$. It follows that 
\begin{align*}
    & U^+ \partial_U\Theta = 0 \ \mathrm{in}\ \mathcal{D}' \\
    &\Theta^2 = \Theta \in L^\infty, \\
    &(\Theta)^k (U^+)^l = (U^+)^l\ \in C^0\ \mathrm{for\ integers}\ k, l \geq 1, \\
    & (U^+)^k (U^+)^l = (U^+)^{k+l}\ \in C^0\ \mathrm{for\ integers}\ k, l \geq 1 .
\end{align*}
These equations are essential in the computation of the distributional Ricci curvature of $g$, as they ensure that various terms cancel out. 
Note that $\partial_U \rho_\e * U^+ = \rho_\e * \Theta$ and $\partial_U \rho_\e * \Theta = \rho_\e$. We now compare the regularized and non-smooth objects. While some quantities exhibit good convergence properties, others require additional care. We begin with those that behave well.
For all integer $k, l \geq 1$ and compact subsets  $\kk$, there exists a constant $C=C(\kk, k,l)>0$ independent of $\e$ such that 
\begin{align}\label{good_convergence}
    &|\rho_\e \cdot (\rho_\e * U^+)^{k+1}| \leq C\e^k \nonumber\\
    &|\rho_\e \cdot (\rho_\e * (U^+)^{k+1})| \leq C\e^k  \nonumber \\
   &|(\rho_\e*\Theta)^k (\rho_\e * U^+)^l - (U^+)^l| + |\rho_\e *(U^+)^l-(U^+)^l| \leq C\e, \nonumber\\
    & |(\rho_\e *U^+)^k (\rho_\e * U^+)^l - (U^+)^{k+l}|+ |\rho_\e *(U^+)^k (\rho_\e * U^+)^l -  (U^+)^{k+l}| + |\rho_\e * (U^+)^{k+l} - (U^+)^{k+l}| \leq C\e,\nonumber\\
    &|\rho_\e * (U^+)^k - (\rho_\e * U^+)^k| \leq C\e,
\end{align}
uniformly as $\e \to 0$ . Instead, $(\rho_\e*U^+)\cdot \rho_\e$ and $(\rho_\e *\Theta)^2 -(\rho_\e *\Theta)$  do not uniformly converge to $0$, since
\begin{align}\label{problematic_convergence}
    & \norm{(\rho_\e*U^+)\cdot \rho_\e} + \norm{U\rho_\e}  = O(1), \nonumber\\
    &\norm{(\rho_\e *\Theta)^2 -(\rho_\e *\Theta)}_{L^\infty} = \frac{1}{4}.
\end{align}
Since terms of the form \eqref{problematic_convergence} require special care, we single them out whenever they arise.

Denote $$\Tilde{g} := P^2 g,$$ and let $\hat{g}_\varepsilon$ be the metric obtained by replacing, coefficient-wise, each occurrence of $U^+$ with its regularization $\rho_\varepsilon * (U^+)$. We then define $\Tilde{g}_\e$ by adding suitable correction terms to the components $\hat{g}^\varepsilon_{34}$ and $\hat{g}^\varepsilon_{43}$, and finally set
\[
g_\varepsilon := \frac{\Tilde{g}_\e}{P_\varepsilon^2}.
\]

In computing the Ricci curvature, following \eqref{connection_first_kind}, \eqref{low_index_riem_curvature}, and \eqref{ricci_coefficients_formula}, we rewrite all expressions over a common denominator and split the analysis into two parts:
\begin{align}
    \Tilde{g}_\e=\underbrace{\hat{g}_\e}_{(i)}+\underbrace{(\Tilde{g}_\e-\hat{g}_\e)}_{(ii)}.
\end{align}

(i) The first part concerns the terms in the numerator arising solely from $\hat{g}_\varepsilon$ and its derivatives. These terms resemble the Ricci tensor of the non-smooth metric, with all potentially distributional factors replaced by their mollified counterparts. They are close to the constant curvature expression, except for a small number of terms of the form \eqref{problematic_convergence}, which do not converge uniformly to zero.

(ii) The second part consists of terms involving $\Tilde{g}_\e - \hat{g}_\varepsilon$ and its derivatives. Most of these terms converge uniformly to zero as $\varepsilon \to 0$, and the remaining ones precisely cancel the non-vanishing contributions of type \eqref{problematic_convergence} identified in (i).

\subsubsection{The case of flat background, $\Lambda =0$}
In this subsection, we consider the case $\Lambda = 0$, for which $P = 1$. The non-smooth metric $g$ writes as
\begin{align*}
     \begin{pmatrix}
    0 & -1 & 0 & 0 \\
    -1 & 0 & 0 & 0 \\
    0 & 0 & 2\partial^2_{Z, Z}\hh U^+ +2\partial^2_{Z,\Bar{Z}}\hh\partial^2_{Z, Z}\hh (U^+)^2 & 1+2\partial^2_{Z,\Bar{Z}}\hh U^+ +((\partial^2_{Z,\Bar{Z}}\hh)^2+|\partial^2_{\cz, \cz}\hh|^2)(U^+)^2 \\
    0 & 0 &  1+2\partial^2_{Z,\Bar{Z}}\hh U^+ +((\partial^2_{Z,\Bar{Z}}\hh)^2+|\partial^2_{\cz, \cz}\hh|^2)(U^+)^2 & 2\partial^2_{\cz, \cz}\hh U^+ +2\partial^2_{Z,\Bar{Z}}\hh \partial^2_{\cz, \cz}\hh(U^+)^2 
    \end{pmatrix}.
\end{align*}
Recall that $\hh = \hh(Z, \cz)$ is a real valued function.  From now on require that 
\begin{equation}\label{eq:hC31}
\hh \in C^{3,1}_{loc}(\C^2).
\end{equation}
Writing $d = \det g[3,4] = g_{33}g_{44}-g_{34}^2$, we get that 
\begin{align*}
    g^{-1} = \begin{pmatrix}
    0 & -1 & 0 & 0 \\
    -1 & 0 & 0 & 0 \\
    0 & 0 & \frac{g_{44}}{d} & -\frac{g_{34}}{d} \\
    0 & 0 &  -\frac{g_{34}}{d}  & \frac{g_{33}}{d} 
    \end{pmatrix}.
\end{align*}
In order to construct a smooth approximation, fix $\rho \in C_c^\infty((-1, 1); [0, \infty))$ such that $\rho(x) = \rho(-x)$, $\int_\R \rho = 1$. For each $\e \in (0, 1)$,  define $\rho_\e: \R \to [0, \infty)$ by $\rho_\e (x) = \frac{1}{\e}\rho(\frac{x}{\e})$. In the following, we will always evaluate $\rho_\e$ at the variable $U$. Recall that $\Theta = \Theta(U)$ is the Heaviside function of $U$. Let us also define some further functions, used later on. Set
\begin{align*}
    \sigma_\e &:= (\rho_\e * \Theta)^2- \rho_\e * \Theta, \\
    \xi_\e(x) &:= \int_{-\infty}^x \sigma_\e(U)dU, \\
    \zeta_\e(x) &:= \int_{-\infty}^x \xi_\e(U) dU.
\end{align*}
We collect some properties of these functions. First, $\sigma_\e \in C_c^\infty((-\e, \e))$, $\sigma_\e \leq 0$, and $\norm{\sigma_\e}_{C^0} = \frac{1}{4}$. 
It follows that $\supp \,\xi_\e \subset (-\e, \infty)$, $\xi_\e \leq 0$ is non-increasing and constant outside $(-\e, \e)$, with $\norm{\xi_\e}_{L^\infty} = \norm{\sigma_\e}_{L^1} \leq \frac{\e}{2}$. 
Moreover, it follows that $\supp\, \zeta_\e \subset (-\e, \infty)$, $\zeta_\e \leq 0$ is non-increasing and $|\zeta_\e(U)|\leq \frac{\e}{2}U^++\frac{\e^2}{2}$.

Define also
\begin{align*}
    \eta_\e(x) &:= \int_{-\infty}^x (\rho_\e * U^+) \rho_\e(U)dU, \\
    \kappa_\e(x) &= \int_{-\infty}^x \eta_\e(U)dU.
\end{align*}
Note that  $|\rho_\e * U^+| \leq 2\e$, for $U \in (-\e, \e)$. It follows that 
$$\supp\, (\rho_\e * U^+)\rho_\e \subset (-\e, \e),\quad  (\rho_\e * U^+)\rho_\e \geq 0,\quad \text{ and } \norm{(\rho_\e * U^+)\rho_\e}_{C^0} \leq 2 \norm{\rho}_{C^0}.$$ 
Therefore, $\supp \, \eta_\e \subset (-\e, \infty)$, $\eta_\e \geq 0$, $\eta_\e$ is non-decreasing, and $\norm{\eta_\e}_{L^\infty} \leq C\e$ for some $C = C(\rho) > 0$. We then get that $\supp \, \kappa_\e \subset (-\e, \infty)$, $\kappa_\e \geq 0$, $\kappa_\e$ is non-decreasing, and $|\kappa_\e(x)| \leq C\e(\e+|U|)$.

Define $$H_\e:= \{\eta_\e, \zeta_\e, \kappa_\e, \xi_\e\},$$ and note that all functions in $H_\e$, as well as $\sigma_\e$ and $\rho_\e (\rho_\e * U^+)$ converge to $0$ in $\mathcal{D}'$ as $\e \to 0$. 

We are now ready to present a smooth approximation of the metric $g$:
\begin{align}\label{flat_apx}
    g_\e = \begin{pmatrix}
    g[1,2] & 0_2 \\
    0_2 & g_\e[3,4]
    \end{pmatrix} \in \C^{4 \times 4}, 
\end{align}
where $0_2 = \begin{pmatrix}
    0 & 0 \\
    0 & 0
    \end{pmatrix} \in \C^{2 \times 2}$, $g[1,2] = \begin{pmatrix}
    0 & -1 \\
    -1 & 0
    \end{pmatrix}\in \C^{2 \times 2}$ and $g_\e[3,4] = \begin{pmatrix}
    {g^{\e}}_{33} & {g^{\e}}_{34} \\
    {g^{\e}}_{43} & {g^{\e}}_{44}
    \end{pmatrix}\in \C^{2 \times 2}$. Here, 
\begin{align*}
    {g^{\e}}_{33} &=  2\partial^2_{Z, Z}\hh \rho_\e * U^+ +2\partial^2_{Z,\Bar{Z}}\hh\partial^2_{Z, Z}\hh \rho_\e *(U^+)^2, \\
     {g^{\e}}_{44} &=  2\partial^2_{\cz, \cz}\hh \rho_\e * U^+ +2\partial^2_{Z,\Bar{Z}}\hh\partial^2_{\cz, \cz}\hh \rho_\e *(U^+)^2, \\
     {g^{\e}}_{43}= {g^{\e}}_{34} &= 1+2\partial^2_{Z,\Bar{Z}}\hh\rho_\e *U^+ +((\partial^2_{Z,\Bar{Z}}\hh)^2+|\partial^2_{\cz, \cz}\hh|^2)\rho_\e *(U^+)^2 \\
     &+4((\partial^2_{Z, \cz}\hh)^2 + |\partial^2_{\cz, \cz}\hh|^2)\kappa_\e +2((\partial^2_{Z,\Bar{Z}}\hh)^2+|\partial^2_{\cz, \cz}\hh|^2)\zeta_\e.
\end{align*}
\begin{proposition}\label{prop:ApproxFlat}
    Assume that $\hh \in C^{3,1}_{loc}(\C^2)$ satisfies $\partial^2_{Z, \cz}\hh \leq 0$. Consider the  the metrics $g_\e$  defined in \eqref{flat_apx}. Fix a compact set $\mathcal{K} \subset M_0:=\domlor{g}$.
    Then
    \begin{itemize}
    \item $g_\e$ converge uniformly to $g$ on $\mathcal{K}$;
    \item there exist constants $c_\e>0$, with $c_\e \to 0$ as $\e \to 0$, such that  $\Ric_{g_\e} \geq -c_\e$.
    \end{itemize}
\end{proposition}
\begin{proof}
{The uniform convergence of $g_\e$ to $g$ on $\mathcal{K}$ follows from the properties of the convolution and the properties of $\zeta_\e, \kappa_\e$.}
Writing 
$$d_\e := \det g_\e[3,4] = g^\e_{33}g^\e_{44}-(g^\e_{34})^2,$$ we get that 
\begin{align}\label{inverse_mollified_metric}
    g_\e^{-1} = \begin{pmatrix}
    0 & -1 & 0 & 0 \\
    -1 & 0 & 0 & 0 \\
    0 & 0 & \frac{g^\e_{44}}{d_\e} & -\frac{g^\e_{34}}{d_\e} \\
    0 & 0 &  -\frac{g^\e_{34}}{d_\e}  & \frac{g^\e_{33}}{d_\e} 
    \end{pmatrix}.
\end{align}
We also define 
\begin{align*}
    {\hat{g}^{\e}}_{33} &=  g^\e_{33} \\
     {\hat{g}^{\e}}_{44} &=  g^\e_{44} \\
     {\hat{g}^{\e}}_{43}= {\hat{g}^{\e}}_{34} &= 1+2\partial^2_{Z,\Bar{Z}}\hh\rho_\e *U^+ +((\partial^2_{Z,\Bar{Z}}\hh)^2+|\partial^2_{\cz, \cz}\hh|^2)\rho_\e *(U^+)^2.
\end{align*}

Note that the matrix of $g_\e$ is symmetric and each entry converges locally uniformly to the corresponding entry of $g$ as $\e \to 0$. Hence $\bigcup_{\e > 0}  \domlor(g_\e) = \domlor{g}= M_0$. For the same arguments as in Proposition \ref{uv_monotone_along_causal_curves}, we get that $g_\e$ is strongly causal. 
The metrics $g, g_\e$ are bounded on $\kk$, as $D^2\hh, D^3\hh$, and $D^4\hh$ are.

Note that 
\begin{align*}
    &\kappa_\e \cdot \rho_\e \to 0,\ \zeta_\e \cdot \rho_\e \to 0,
\end{align*}
uniformly on $\kk$ as $\e \to 0$. Moreover,  for all non-negative integers $k, l, m, p$, it holds
\begin{align*}
    (|\eta_\e|+|\zeta_\e|+|\xi_\e|+|\kappa_\e|)((\rho_\e * \Theta)^l + (\rho_\e * U^+)^k + \rho_\e * (U^+)^m)^p \to 0
\end{align*}
and
\begin{align*}
    (|\eta_\e|+|\zeta_\e|+|\xi_\e|+|\kappa_\e|)^{k+1} \to 0
\end{align*}
uniformly on $\kk$ as $\e \to 0$. 
By the properties of $\kappa_\e$ and $\zeta_\e$, as well as the assumption that $\hh \in C^{3,1}_{loc}$, we get that
\begin{align}\label{bound_on_g_and_gamma_eps_to_hat}
    |\hat{g}^\e_{34} -g^\e_{34}| \leq C\e^2\ \mathrm{for}\ U \leq \e, \ \mathrm{and} \  |\hat{g}^\e_{34} -g^\e_{34}| \leq C(\e^2 + U\e) \leq C\e  \ \mathrm{for}\ U > \e.
\end{align}
Define $\hat{d}_\e := (\hat{g}^\e_{33}\hat{g}^\e_{44} -(\hat{g}^\e_{34})^2)$. Then
\begin{align}\label{bound_on_det_difference}
    |d_\e - \hat{d}_\e| \leq C \e^2 \ \mathrm{for}\ U \leq \e, \ \mathrm{and} \ |d_\e - \hat{d}_\e| \leq C\e \ \mathrm{for}\ U > \e.
\end{align}
We denote by $\Gamma^\e, R^\e$, and $\Ric^\e$ the Levi--Civita connection symbols, curvature coefficients, Ricci curvature coefficients with respect to the metric $g_\e$ and by $\Hat{\Gamma}^\e, \hat{R}^\e, \hat{\Ric}^\e$ the Levi--Civita connection symbols, curvature coefficients, Ricci curvature coefficients with respect to the metric $\hat{g}_\e$.

By definition, we get that for all indices $i, j, k, l$, it holds
\begin{align}\label{decomposition_in_metric_and_balance_terms}
    &g^\e_{ij} = \hat{g}^\e_{ij} + P^\zeta_{ij}(Z, \cz)\zeta_\e(U) + P^\kappa_{ij}(Z, \cz)\kappa_\e(U), \nonumber \\
    &\Gamma_{ij, k}^\e = \hat{\Gamma}^\e_{ij, k} + Q^\xi_{ij,k}(Z, \cz)\xi_\e(U) + Q^\eta_{ij,k}(Z, \cz) \eta_\e(U) + Q^\zeta_{ij,k}(Z, \cz)\zeta_\e(U) + Q^\kappa_{ij,k}(Z, \cz)\kappa_\e(U). 
\end{align}
Here, the functions $P^\zeta_{ij}(Z, \cz), P^\kappa_{ij}(Z, \cz), Q^\zeta_{ij,k}(Z, \cz), Q^\kappa_{ij,k}(Z, \cz), Q^\xi_{ij,k}(Z, \cz),  Q^\eta_{ij,k}(Z, \cz)$ are polynomials in $D^2\hh$ and $D^3\hh$.

As in the non-smooth case, we compute the Ricci curvature coefficients of the metric $g_\varepsilon$ using \eqref{connection_first_kind}, \eqref{low_index_riem_curvature}, and \eqref{ricci_coefficients_formula}. Since the latter two expressions involve the inverse metric, we first isolate the $\rho_\varepsilon$-terms, in view of \eqref{inverse_mollified_metric}. These terms will converge distributionally to the singular part of $\Ric(g)$.
We then rewrite the remaining terms --- which converge to the absolutely continuous part of $\Ric(g)$ --- over the common denominator $d_\varepsilon^2$, allowing for a systematic comparison and simplification. In this way, the numerators of the Ricci coefficients can be expressed as sums of products involving $g_\varepsilon$, the Christoffel symbols $\Gamma^\varepsilon$, and their derivatives.

We will use the decomposition from \eqref{decomposition_in_metric_and_balance_terms} and split our analysis into two parts:
\begin{itemize}
    \item[(i)] The contributions in the numerator that only contain terms from $\hat{g}_\e, \hat{\Gamma}^\e$, and $D\hat{\Gamma}^\e$.
    \item[(ii)] The remaining contributions in the numerator.
\end{itemize}
We will start with (i).  Restricting to the terms in the numerator arising from $\hat{g}_\varepsilon$ and its derivatives, one obtains expressions analogous to those in the non-smooth case, the only difference being that all functions of $U$ are replaced by their mollifications with $\rho_\varepsilon$. All terms of the form \eqref{good_convergence} converge uniformly and therefore require no further attention. It remains to identify and analyze the terms of type \eqref{problematic_convergence}.

\textbf{Term 1:} $\rho_\e$ and $\rho_\e \cdot (\rho_\e * U^+)$. \\
We get the term $\rho_\e$ if we differentiate $\hat{g}^\e_{33}, \hat{g}^\e_{34}$, or $\hat{g}^\e_{44}$ twice in the direction of $U$. These derivatives occur in the coefficients $R^\e_{ijkl}$ if two of the indices equal one and the two other indices are in $\{3, 4\}$. 
It follows that only $R^\e_{1313}, R^\e_{1314}, R^\e_{1414}$ (and the coefficients that arise from symmetries) may carry a $\rho_\e$ term. 
Recalling the formula \eqref{ricci_coefficients_formula}, we get that the above Riemannian curvature coefficients are only needed to compute $\Ric^\e_{11}, \Ric_{13}^\e, \Ric_{33}^\e, \Ric_{14}^\e, \Ric_{44}^\e$, and $\Ric_{34}^\e$. Recalling \eqref{inverse_mollified_metric}, we get that the only potentially non-zero contributions are in $\Ric^\e_{11}$ through $g^{34}_\e R^\e_{3114}$, $g^{33}_\e R^\e_{3113}$, and $g^{44}_\e R^\e_{4114}$. Using \eqref{low_index_riem_curvature}, \eqref{ricci_coefficients_formula}, and \eqref{inverse_mollified_metric}, we can extract all terms of the form $f_\e \cdot \rho_\e$, where $\norm{f_\e}_{L^\infty((-\e, \e))} \notin o(\e)$ and denote the remaining terms by $T^\e_{11}$. 
\begin{align*}
    \Ric^\e_{11} &= -2 \partial^2_{Z, \cz}\hh \rho_\e +\frac{1}{(d_\e)^2}(4(\partial^2_{Z, \cz}\hh)^2 + 4|\partial^2_{Z,Z}\hh|^2)\cdot (\rho_\e * U^+)\rho_\e  + T^\e_{11}.
    \end{align*}

\textbf{Term 2:} $\rho_\e* \Theta - (\rho_\e * \Theta)^2$. \\
We get $(\rho_\e * \Theta)$-terms in terms of the form $\partial_{1} g^\e_{ij}$ or $\partial^2_{1,k}g^\e_{ij}$ for $k \in \{1,3,4\}$ and $i,j \in \{3,4\}$. The first ones of these terms can be seen in $\Gamma_{[ij,1]}^\e$, and the second kind of terms can only appear in $R_{[1ijk]}$, where $k \in \{1,3,4\}$ and  $i,j \in \{3,4\}$.
Recalling \eqref{low_index_riem_curvature}, it follows that the only way to get a $(\rho_\e * \Theta)^2$-term is when two $\Gamma^\e$ symbols as above are multiplied. Hence, we consider terms of the form $g_\e^{st}\Gamma^\e_{ij,s}\Gamma^\e_{kl,t}$. 
Notice that we need $s \neq t$, otherwise 
$$|g^{st}| \leq C \e \text{ on } (-\e, \e) \supset \supp\, (\rho_\e* \Theta - (\rho_\e * \Theta)^2),$$ yielding uniform convergence to $0$ uniformly as $\e \to 0$. Since 
$$\Gamma^\e_{ij,2} = 0, \quad  \text{for all $i,j$},$$
we are left to consider the case  
$$\{s,t\}=\{3,4\}.$$ 
Then  $\{\{i,j\}, \{k,l\}\} \subset \{\{1,3\}, \{1,4\}\}$, which can only occur in $R^\e_{[1313]}, R^\e_{[1314]}, R^\e_{[1414]}$. Let us discuss the term $R^\e_{[1313]}$. Using \eqref{ricci_coefficients_formula}, we see that in order to appear in $\Ric^\e$, $R^\e_{[1313]}$ needs to be multiplied with $g^{11}_\e, g^{13}_\e$, or $g^{33}_\e$, which are all $L^\infty$-bounded by $C\e$ in $\supp\, (\rho_\e* \Theta - (\rho_\e * \Theta)^2)$; therefore the contribution stemming from $R^\e_{[1313]}$ will always uniformly converge to $0$ as $\e \to 0$. This applies to $R^\e_{[1414]}$ with the same reasoning. 
For $R^\varepsilon_{[1314]}$, it suffices to consider its contribution when contracted with $g^{34}_\varepsilon$, as occurs in the computation of $\Ric^\varepsilon_{11}$.
Recalling the formulas for $g^{33}_\e$ and $g^{44}_\e$, as well as the previous observations,  it suffices to consider the contributions from $2g^{34}_\e R^\e_{3114}$. Extracting the terms of the form $(\rho_\e * \Theta)^n$ for $n \geq 1$ and denoting the other terms by $Q_\e$, we get that 
\begin{align*}
    2g^{34}_\e R^\e_{3114} &= \frac{-2d_\e g^\e_{34}}{(d_\e)^2} (\partial_3 \Gamma^\e_{11,4}-\partial_1 \Gamma^\e_{13,4}) + 2\frac{(g^\e_{34})^2}{(d_\e)^2}(\Gamma^\e_{31,3}\Gamma^\e_{14,4}+\Gamma^\e_{31,4}\Gamma^\e_{14,3}- 0) \\
    &= \frac{2}{(d_\e)^2}(|\partial^2_{Z, Z}\hh|^2+(\partial^2_{Z, \cz}\hh)^2)((\rho_\e * \Theta)^2 - \rho_\e * \Theta) + Q_\e.
\end{align*}
This finishes the analysis of all terms in the numerator of the Ricci curvature, that only involve $\hat{g_\e}$ and its derivatives. 

We now turn to the remaining terms, namely those in (ii). We again distinguish between contributions involving second derivatives in the $U$-direction, i.e.\ terms of the form $\partial^2_{1,1}(g_\varepsilon - \hat{g}_\varepsilon)$, and those involving at most one derivative in the $U$-direction of $(g_\varepsilon - \hat{g}_\varepsilon)$.

\textbf{Terms involving  a factor with two derivatives in the $U$-direction:} \\
The only relevant term in this case is $\partial^2_{1,1} g^\varepsilon_{34}$. This term appears only in $R^\varepsilon_{[1314]}$ and therefore contributes solely to $\Ric^\varepsilon_{11}$. It thus suffices to analyze
\begin{align*}
    \frac{-2d_\e g^\e_{34}}{(d_\e)^2}(-\partial_1 (\Gamma^\e_{13,4}-\hat{\Gamma}^\e_{13,4})) = \frac{2d_\e g^\e_{34}}{(d_\e)^2} (2((\partial^2_{Z, \cz}\hh)^2 + |\partial^2_{Z,Z}\hh|^2)\rho_\e(\rho_\e * U^+) + ((\partial^2_{Z, \cz}\hh)^2 + |\partial^2_{Z,Z}\hh|^2)\sigma_\e).
\end{align*}
We note that $ \supp\, \sigma_\e \cup \supp \, \rho_\e(\rho_\e * U^+) \subset (-\e, \e)$ and both functions are bounded in $L^\infty$ by a constant $C>0$ independent of $\e$. Moreover, we can write
\begin{align*}
    d_\e g^\e_{34} = -1 + (\rho_\e * U^+) F^\e_1 + (\rho_\e * (U^+)^2) F^\e_2 + \zeta_\e F^\e_3 + \kappa_\e F^\e_4,
\end{align*}
where for all $i$, $ F^\e_i = F^\e_i(U, V, Z, \cz)$ is uniformly bounded in $L^\infty(\kk)$ by a constant $C>0$ independent of $\e$. Now,
$$
\sup_{U\in (-\e, \e)} (\rho_\e * U^+) \leq C\e \quad \text{and} \quad  \sup_{U\in (-\e, \e)} (\rho_\e * (U^+)^2) \leq C\e,
$$
where $C>0$ does not depend on $\e$. Then, the above term can be written as 
\begin{align}\label{relaxation_terms}
     \frac{-2d_\e g^\e_{34}}{(d_\e)^2}(-\partial_1 (\Gamma^\e_{13,4}-\hat{\Gamma}^\e_{13,4})) = & \frac{1}{(d_\e)^2} (-4((\partial^2_{Z, \cz}\hh)^2 + |\partial^2_{Z,Z}\hh|^2)\rho_\e(\rho_\e * U^+) \nonumber  \\
     &-2 ((\partial^2_{Z, \cz}\hh)^2 + |\partial^2_{Z,Z}\hh|^2)\sigma_\e) + F_\e,
\end{align}
where $\supp\, F_\e \in (-\e, \e)$ and $F_\e \to 0$ uniformly as $\e \to 0$.

\textbf{Terms involving at most first derivatives in the $U$-direction:} \\
We first note that the $\Gamma^\e$ and $g_\e$ are locally uniformly bounded independently of $\e$. 
Moreover, 
\begin{align}\label{bounds_on_dgamma_eps}
    &|D\Gamma^\e| + |D\hat{\Gamma}^\e| \leq C\e^{-1}\ \mathrm{in}\ \{-\e \leq U \leq \e \} \nonumber \\
     &|D\Gamma^\e| + |D\hat{\Gamma}^\e| \leq C\ \mathrm{in}\ \{|U| >\e \},
\end{align}
where $C> 0$ that does not depend on $\e$. 
We note that when computing the Ricci curvature coefficients along the lines of \eqref{connection_first_kind}, \eqref{low_index_riem_curvature}, and \eqref{ricci_coefficients_formula}, all terms involving a multiplication with $g_\e^{21}$ or $g_\e^{12}$ vanish, as $\Gamma^\e_{[2i,j]} = 0$ for all $i,j$. Then
\begin{align*}
    \Ric_{jk}^\e &= g_\e^{il}R^\e_{ijkl} = g_\e^{il}(\partial_i \Gamma^\e_{jk, l} -\partial_j \Gamma^\e_{ik, l} + g_\e^{st}(\Gamma^\e_{ik, s}\Gamma^\e_{jl, t}-\Gamma^\e_{il, s}\Gamma^\e_{jk, t})) \\
    &= \frac{1}{(d_\e)^2}g^\e_{i'l'}(d_\e(\partial_i \Gamma^\e_{jk, l} -\partial_j \Gamma^\e_{ik, l}) + g_{s't'}(\Gamma^\e_{ik, s}\Gamma^\e_{jl, t}-\Gamma^\e_{il, s}\Gamma^\e_{jk, t})),
\end{align*}
where we define $3'=4$, $4'=3$, $1'=2$, and $2'=1$ (the formula  holds because the $12,21$-terms are irrelevant). Using \eqref{bound_on_g_and_gamma_eps_to_hat}, \eqref{bound_on_det_difference}, \eqref{bounds_on_dgamma_eps}, and \eqref{relaxation_terms} it follows that, for $j=k=1$:
\begin{align*}
    \Ric_{jk}^\e &- \frac{1}{(d_\e)^2}\hat{g}^\e_{i'l'}(\hat{d}_\e(\partial_i \hat{\Gamma}^\e_{jk, l} -\partial_j \hat{\Gamma}^\e_{ik, l}) + \hat{g}^\e_{s't'}(\hat{\Gamma}^\e_{ik, s}\hat{\Gamma}^\e_{jl, t}-\hat{\Gamma}^\e_{il, s}\hat{\Gamma}^\e_{jk, t})) \\
    &= \frac{1}{(d_\e)^2} (-4((\partial^2_{Z, \cz}\hh)^2 + |\partial^2_{Z,Z}\hh|^2)\rho_\e(\rho_\e * U^+) -2 ((\partial^2_{Z, \cz}\hh)^2 + |\partial^2_{Z,Z}\hh|^2)\sigma_\e) + F_\e,
\end{align*}
for some function $F_\e$ such that $F_\e \to 0$ uniformly as $\e \to 0$. For any $(j,k) \neq (1,1)$, we get that 
\begin{align*}
    \Ric_{jk}^\e &- \frac{1}{(d_\e)^2}\hat{g}^\e_{i'l'}(\hat{d}_\e(\partial_i \hat{\Gamma}^\e_{jk, l} -\partial_j \hat{\Gamma}^\e_{ik, l}) + \hat{g}^\e_{s't'}(\hat{\Gamma}^\e_{ik, s}\hat{\Gamma}^\e_{jl, t}-\hat{\Gamma}^\e_{il, s}\hat{\Gamma}^\e_{jk, t})) = F_\e,
\end{align*}
for some function $F_\e$ such that $F_\e \to 0$ uniformly as $\e \to 0$.

We conclude that 
\begin{align*}
    \Ric^\e_{ij} = -2\partial^2_{Z, \cz}\hh \rho_\e \delta^1_i\delta^1_j + \frac{1}{(d_\e)^2}\Tilde{\Ric}^\e_{ij} + r_{ij}^\e,
\end{align*}
where $r^\e_{ij}$ converges to $0$ uniformly on $\mathcal{K}$ as $\e \to 0$ and $\Tilde{\Ric}^\e_{ij}$ consists of the terms of \newline
$(\hat{d}_\e)^2(\hat{\Ric}^\e_{ij} +2\delta^1_i\delta^1_j\partial^2_{Z, \cz}\hh \rho_\e)$ that occur in \eqref{good_convergence}. From \eqref{good_convergence} and the non-smooth case, we know that $\Tilde{\Ric}^\e_{ij}$ uniformly converges to $0$ as $\e \to 0$.
Finally,  the uniform convergence of $g_\e$ to $g$ on $\mathcal{K}$ implies that $d_\e$ is bounded away from zero, and $d_\e(0, V, Z, \cz) \to 1$ uniformly as $\e \to 0$. Combining the above, we conclude that 
$$\Ric^\e_{ij} +2\partial^2_{Z, \cz}\hh \rho_\e \delta^1_i\delta^1_j\to 0\quad \text{uniformly on $\mathcal{K}$, as $\e\to 0$.}$$
\end{proof}

\subsubsection{The case of non-zero cosmological constant}\label{subsubsec:general_case}
We now turn to the case in which the cosmological constant $\Lambda$ is non-zero. As in the flat case, we begin by naively mollifying the metric and identifying the resulting obstructions. We then refine the approximation by introducing suitable correction terms to compensate for these effects. We define
\begin{align*}
    &\tau_\e(x) := \int_{-\infty}^x U \rho_\e(U) \, dU, \\
    &\iota_\e(x) := \int_{-\infty}^x \tau_\e(U) \, dU.
\end{align*}
As $\rho(U) = \rho(-U)$, we get that $\supp \, \tau_\e \subset (-\e, \e)$ and $\norm{\tau_\e}_{C^0} \leq C\e$, for a constant $C>0$ that does not depend on $\e$. It follows that  $\supp\, \iota_\e \subset (-\e, \infty)$ and $\norm{\iota_\e}_{C^0} \leq C\e^2$.
Define 
\begin{align}\label{def:P}
    P = 1 + \frac{1}{6} \Lambda(Z \cz - UV - U^+G),
\end{align}
where $G$ is given by $G(Z, \cz) = H-Z \partial_ZH - \cz \partial_{\cz} H$. With such notations, the IPP-wave metric $g$ writes as
\begin{align*}
     \frac{1}{P^2}\begin{pmatrix}
    0 & -1 & 0 & 0 \\
    -1 & 0 & 0 & 0 \\
    0 & 0 & 2\partial^2_{Z, Z}\hh U^+ +2\partial^2_{Z,\Bar{Z}}\hh\partial^2_{Z, Z}\hh (U^+)^2 & 1+2\partial^2_{Z,\Bar{Z}}\hh U^+ +((\partial^2_{Z,\Bar{Z}}\hh)^2+|\partial^2_{\cz, \cz}\hh|^2)(U^+)^2 \\
    0 & 0 &  1+2\partial^2_{Z,\Bar{Z}}\hh U^+ +((\partial^2_{Z,\Bar{Z}}\hh)^2+|\partial^2_{\cz, \cz}\hh|^2)(U^+)^2 & 2\partial^2_{\cz, \cz}\hh U^+ +2\partial^2_{Z,\Bar{Z}}\hh \partial^2_{\cz, \cz}\hh(U^+)^2 
    \end{pmatrix}.
\end{align*}
Define 
\begin{align*}
    \Tilde{g} := P^2 g.
\end{align*}
The inverse metric is given by
\begin{align*}
    g^{-1} = P^2 \begin{pmatrix}
    0 & -1 & 0 & 0 \\
    -1 & 0 & 0 & 0 \\
    0 & 0 & \frac{\Tilde{g}_{44}}{\Tilde{d}} & -\frac{\Tilde{g}_{34}}{\Tilde{d}} \\
    0 & 0 &  -\frac{\Tilde{g}_{34}}{\Tilde{d}}  & \frac{\Tilde{g}_{33}}{\Tilde{d}} 
    \end{pmatrix},
\end{align*}
where $$\Tilde{d} := \Tilde{g}_{33}\Tilde{g}_{44}-(\Tilde{g}_{34})^2.$$ 

 Denote $M_\Lambda:=\domlor{g}$ and fix a compact set $\mathcal{K} \Subset M_\Lambda$. By continuity, on $\kk$,  both $P$ and $\Tilde{d}$ are bounded away from $0$ by a constant $\lambda >0$ and $|Z|$ is bounded by a constant $C>0$. 

Set 
\begin{align}
    P_\e = 1 + \frac{1}{6} \Lambda(Z \cz - UV - (\rho_\e* U^+)\cdot G),
\end{align}
and define the smooth approximations $\Tilde{g}_\e$ as
\begin{align*}
    \Tilde{g}_\e := \begin{pmatrix}
    g[1,2] & 0_2 \\
    0_2 & \Tilde{g}_\e[3,4]
    \end{pmatrix} \in \C^{4 \times 4}, 
\end{align*}
where $0_2 = \begin{pmatrix}
    0 & 0 \\
    0 & 0
    \end{pmatrix} \in \C^{2 \times 2}$, $g[1,2] = \begin{pmatrix}
    0 & -1 \\
    -1 & 0
    \end{pmatrix}\in \C^{2 \times 2}$ and $\Tilde{g}_\e[3,4] = \begin{pmatrix}
    {\Tilde{g}^{\e}}_{33} & {\Tilde{g}^{\e}}_{34} \\
    {\Tilde{g}^{\e}}_{43} & {\Tilde{g}^{\e}}_{44}
    \end{pmatrix}\in \C^{2 \times 2}$. Here, 
\begin{align}\label{definition_g_eps_tilde_noflat}
    {\Tilde{g}^{\e}}_{33} &=  2\partial^2_{Z, Z}\hh \rho_\e * U^+ +2\partial^2_{Z,\Bar{Z}}\hh\partial^2_{Z, Z}\hh \rho_\e *(U^+)^2, \nonumber  \\
     {\Tilde{g}^{\e}}_{44} &=  2\partial^2_{\cz, \cz}\hh \rho_\e * U^+ +2\partial^2_{Z,\Bar{Z}}\hh\partial^2_{\cz, \cz}\hh \rho_\e *(U^+)^2, \nonumber \\
     {\Tilde{g}^{\e}}_{43}= {g^{\e}}_{34} &= 1+2\partial^2_{Z,\Bar{Z}}\hh\rho_\e *U^+ +((\partial^2_{Z,\Bar{Z}})^2\hh+|\partial^2_{\cz, \cz}\hh|^2)\rho_\e *(U^+)^2 +A\kappa_\e +B\zeta_\e + E \iota_\e,
\end{align}
where $A, B$, and $E$ are functions of $Z, \cz$, and $V$ that are bounded on $\kk$.
Finally, define
\begin{align}\label{non-flat_apx}
    g_\e := \frac{1}{P_\e^2} \Tilde{g}_\e.
\end{align}

\begin{proposition}\label{prop:ApproxNonFlat}
     Assume that $\hh \in C^{3,1}_{loc}(\C^2)$ satisfies 
     \begin{equation}\label{eq:de2HLambda}
     -2\partial^2_{Z, \cz} \hh  -\frac{\Lambda(\hh - Z\partial_Z \hh -\cz \partial_{\cz}\hh)}{3(1+\frac{\Lambda}{6}|Z|^2)} \geq 0.
     \end{equation}
     Consider the  metrics $g^{\e}$  defined  as in \eqref{non-flat_apx}. Fix a compact set $\mathcal{K} \subset M_\Lambda:=\domlor{g}$.
    Then
    \begin{itemize}
    \item $g_\e$ converge uniformly to $g$ on $\mathcal{K}$;
    \item there exist constants $c_\e>0$, with $c_\e \to 0$ as $\e \to 0$, such that  $\Ric_{g_\e} \geq \Lambda -c_\e$.
    \end{itemize}
\end{proposition}
\begin{proof}
Define $\Tilde{d}_\e := \det \Tilde{g}_\e[3,4] = \Tilde{g}^\e_{33}\Tilde{g}^\e_{44}-(\Tilde{g}^\e_{34})^2$. Then the inverse metric can be written as
\begin{align}\label{inverse_mollified_metric_noflat}
    g_\e^{-1} = P_\e^2\begin{pmatrix}
    0 & -1 & 0 & 0 \\
    -1 & 0 & 0 & 0 \\
    0 & 0 & \frac{\Tilde{g}^\e_{44}}{\Tilde{d}_\e} & -\frac{\Tilde{g}^\e_{34}}{\Tilde{d}_\e} \\
    0 & 0 &  -\frac{\Tilde{g}^\e_{34}}{\Tilde{d}_\e}  & \frac{\Tilde{g}^\e_{33}}{\Tilde{d}_\e} 
    \end{pmatrix}.
\end{align}
We also define 
\begin{align*}
    {\hat{g}^{\e}}_{33} &=  2\partial^2_{Z, Z}\hh \rho_\e * U^+ +2\partial^2_{Z,\Bar{Z}}\hh\partial^2_{Z, Z}\hh \rho_\e *(U^+)^2, \\
     {\hat{g}^{\e}}_{44} &=  2\partial^2_{\cz, \cz}\hh \rho_\e * U^+ +2\partial^2_{Z,\Bar{Z}}\hh\partial^2_{\cz, \cz}\hh \rho_\e *(U^+)^2, \\
     {\hat{g}^{\e}}_{43}= {\hat{g}^{\e}}_{34} &= 1+2\partial^2_{Z,\Bar{Z}}\hh\rho_\e *U^+ +((\partial^2_{Z,\Bar{Z}}\hh)^2+|\partial^2_{\cz, \cz}\hh|^2)\rho_\e *(U^+)^2,
\end{align*}
and
\begin{align*}
    \hat{g}^\e_{ij} = \Tilde{g}_{ij}^\e \ \mathrm{if} \ \{i, j\} \not\subset \{3, 4\}.
\end{align*}
From the properties of $\kappa_\e$, $\iota_\e$, and $\zeta_\e$, it follows that 
\begin{align}\label{compare_g_hat_tilde}
  |\Tilde{g}_{ij}^\e- \hat{g}_{ij}^\e| \leq C \e^2 \ \text{for }U \leq \e, \ \mathrm{and} \quad
   |\Tilde{g}_{ij}^\e- \hat{g}_{ij}^\e| \leq C \e \ \text{for }U>\e.
\end{align}
Denoting $\hat{d}_\e := (\hat{g}^\e_{33}\hat{g}^\e_{44} -(\hat{g}^\e_{34})^2)$, we get that  
\begin{align}\label{bound_on_det_difference_noflat}
    |\Tilde{d}_\e - \hat{d}_\e| \leq C\lambda^2 \e^2 \ \mathrm{for}\ U \leq \e, \, \mathrm{and} \;
    |\Tilde{d}_\e - \hat{d}_\e| \leq C\lambda^2\e \ \mathrm{for}\ U > \e.
\end{align}
By definition,  for every $i, j, k, l$, it holds
\begin{align}\label{decomposition_in_metric_and_balance_terms_noflat}
    \Tilde{g}^\e_{ij} =&\; \hat{g}^\e_{ij} + P^\zeta_{ij}(Z, \cz)\zeta_\e(U) + P^\kappa_{ij}(Z, \cz)\kappa_\e(U) + P^\iota_{ij}(V, Z, \cz)\iota_\e(U),\nonumber \\
    \partial_k\Tilde{g}^\e_{ij} =& \;\partial_k\hat{g}^\e_{ij} + Q^\xi_{ij,k}(Z, \cz)\xi_\e(U) + Q^\eta_{ij,k}(Z, \cz) \eta_\e(U) + Q^\zeta_{ij,k}(Z, \cz)\zeta_\e(U) + Q^\kappa_{ij,k}(Z, \cz)\kappa_\e(U) \nonumber \\
    &+ Q^\iota_{ij,k}(V,Z, \cz)\iota_\e(U) + Q^\tau_{ij,k}(V, Z, \cz)\tau_\e(U).
\end{align}
Here, the functions $P^\zeta_{ij}, P^\kappa_{ij}, Q^\kappa_{ij}, Q^\xi_{ij,k}, Q^\eta_{ij,k}$ are polynomials in $Z, \cz,$ and $D^i \hh$ for $0 \le i \le 3$, while $P^\iota_{ij}, Q^\iota_{ij},$ and $Q^\tau_{ij,k}$ are polynomials in $V, Z, \cz,$ and $D^i \hh$ for $0 \le i \le 3$.

Using that $g_\e = \frac{\Tilde{g}_\e}{P_\e^2}$, we get that 
\begin{align}
    \partial_{i}g_\e &= \frac{P_\e \partial_i \Tilde{g_\e}-2\Tilde{g}_\e \partial_i P_\e}{P_\e^3},\label{quotient_derivatives} \\
    \partial^2_{i,j}g_\e &= \frac{P_\e^2 \partial^2_{i,j}\Tilde{g}_\e-2P_\e \partial_i\Tilde{g}_\e \partial_jP_\e -2P_\e \partial_j\Tilde{g}_\e \partial_iP_\e -2P_\e\Tilde{g}_\e \partial^2_{i,j}P_\e + 6\Tilde{g}_\e \partial_iP_\e \partial_j P_\e}{P_\e^4}. \label{second_quotient_derivatives}
\end{align}

As in the flat case, we compute the Ricci curvature using \eqref{connection_first_kind}, \eqref{low_index_riem_curvature}, and \eqref{ricci_coefficients_formula}. We again isolate the $\rho_\varepsilon$-terms and rewrite the remaining expressions over the common denominator $P_\varepsilon^2 \widetilde{d}_\varepsilon^2$. Using the decomposition \eqref{decomposition_in_metric_and_balance_terms_noflat}, we split the analysis into two parts:
\begin{itemize}
    \item[(i)] Terms in the numerator involving only $\hat{g}_\varepsilon$, its derivatives $D\hat{g}_\varepsilon$, $D^2\hat{g}_\varepsilon$, and $P_\varepsilon$;
    \item[(ii)] The remaining terms.
\end{itemize}

We begin with (i). The only terms that converge to zero merely distributionally, but not uniformly, are $(\rho_\varepsilon * U^+)\cdot \rho_\varepsilon$, $U \cdot \rho_\varepsilon$, and $(\rho_\varepsilon * \Theta)^2 - (\rho_\varepsilon * \Theta)$. We proceed to analyze the occurrence of these terms.

\textbf{Step 1.} Products involving $\rho_\varepsilon$ whose $L^\infty$-norm is bounded below by a constant $C>0$.\\
We first identify where terms involving $\rho_\varepsilon$ may arise. This occurs when computing $\partial^2_{1,1} P_\varepsilon$ and $\partial^2_{1,1} \hat{g}^\varepsilon_{ij}$ for $i,j \in \{3,4\}$, that is, when differentiating a non-vanishing component of $g_\varepsilon$ twice in the $U$-direction.

Now, $\partial^2_{1,1} g^\varepsilon_{12} = \partial^2_{1,1} g^\varepsilon_{21}$ can only appear in $R^\varepsilon_{1211}$ or $R^\varepsilon_{1121}$, which vanish by the symmetries of the curvature tensor. Consequently, terms involving $\rho_\varepsilon$ can arise only through curvature components containing $\partial^2_{1,1} g^\varepsilon_{ij}$ with $i,j \in \{3,4\}$.  These contribute exclusively to components of the form $R^\varepsilon_{[1i1j]}$ for $i,j \in \{3,4\}$. For such terms to appear in the Ricci tensor, they must be contracted with a non-zero component of the inverse metric, which occurs only when $g_\varepsilon^{ij}$ multiplies $R^\varepsilon_{i11j}$ for $i,j \in \{3,4\}$. Hence, these contributions arise solely in $\Ric^\varepsilon_{11}$, which write as 
\begin{align*}
    \Ric^\e_{11} &= 2g_\e^{12}R^\e_{1112} + 2g_\e^{34}R^\e_{3114} +g_\e^{33}R^\e_{3113}+g_\e^{44}R^\e_{4114} = \frac{P_\e^2}{\Tilde{d_\e}}(-2\Tilde{g}^\e_{34}R^\e_{3114} +\Tilde{g}^\e_{44}R^\e_{3113}+\Tilde{g}^\e_{33}R^\e_{4114}).
\end{align*}
Next, we extract the $\rho_\e$-contributions from the curvature coefficients. Denoting $Q^\e_{ijkl}$ the terms of order at most $O(1)$, one gets
\begin{align*}
    R^\e_{3114} &= \frac{-  \partial^2_{Z, \cz}\hh \rho_\e }{P_\e^2} - \frac{ \hat{g}_{34}^\e \frac{\Lambda}{6}G \rho_\e}{P_\e^3} + Q^\e_{3114}, \\
    R^\e_{3113} &= \frac{-  \partial^2_{Z, Z}\hh \rho_\e }{P_\e^2} - \frac{ \hat{g}_{33}^\e \frac{\Lambda}{6}G \rho_\e}{P_\e^3} + Q^\e_{3113} , \\
    R^\e_{4114} &= \frac{-  \partial^2_{\cz, \cz}\hh \rho_\e }{P_\e^2} - \frac{ \hat{g}_{44}^\e \frac{\Lambda}{6}G \rho_\e}{P_\e^3} + Q^\e_{4114}.
\end{align*}
Recalling that on $\{U \in (-\varepsilon,\varepsilon)\} \supset \supp \rho_\varepsilon$ one has $\hat{g}^\varepsilon_{44}, \hat{g}^\varepsilon_{33} = O(\varepsilon)$, we may extract all terms of the form $\rho_\varepsilon f_\varepsilon$ with $f_\varepsilon \notin o(\varepsilon)$, and obtain
\begin{align*}
    \Ric^\e_{11} =&\; 2\hat{g}^\e_{34}\cdot \Big(\frac{\partial^2_{Z, \cz}\hh \rho_\e }{\Tilde{d_\e}} +\frac{ \hat{g}_{34}^\e \frac{\Lambda}{6}G \rho_\e}{\Tilde{d_\e}P_\e}\Big) -\hat{g}^\e_{44}\frac{\partial^2_{Z, Z}\hh \rho_\e }{\Tilde{d_\e}} -\hat{g}^\e_{33} \frac{\partial^2_{\cz, \cz}\hh \rho_\e }{\Tilde{d_\e}}  + (Q^\e_{11})'' \\
    &+ \frac{\frac{2\Lambda}{3}G \partial_{Z, \cz}^2\hh \rho_\e(\rho_\e * U^+)}{\Tilde{d_\e}P_\e} -\frac{4|\partial^2_{\cz, \cz}\hh|^2 \rho_\e(\rho_\e *U^+) }{\Tilde{d_\e}}  + (Q^\e_{11})' ,
\end{align*}
where $(Q^\e_{11})', (Q^\e_{11})''$ denote terms that are of a different form. 
Next, we will use that for any $x_0>-1$ and $y> -1-x_0$, it holds
\begin{align}\label{useful_taylor}
    \frac{1}{1+x_0+y} = \frac{1}{1+x_0} - \frac{y}{(1+x_0)^2} + \frac{y^2}{(1+x_0+ty)^3}
\end{align}
for some $t \in (0, 1)$. Moreover, we note that in $(-\e, \e)$, it holds
\begin{align*}
    \Tilde{d}_\e &= -1 -4\partial^2_{Z, \cz}\hh \rho_\e*U^+ + O(\e^2), \\
    \Tilde{d}_\e P_\e &= -(1 + \frac{1}{6} \Lambda(Z \cz - UV - (\rho_\e* U^+)\cdot G)) -4\partial^2_{Z, \cz}\hh \rho_\e*U^+(1 + \frac{1}{6} \Lambda |Z|^2 ) + O_{Z, \hh}(\e^2) \\
     &= -((1 + \frac{1}{6} \Lambda|Z|^2) -\frac{1}{6} \Lambda( UV + (\rho_\e* U^+)\cdot G)) -4\partial^2_{Z, \cz}\hh \rho_\e*U^+(1 + \frac{1}{6} \Lambda |Z|^2 ) + O_{Z, \hh}(\e^2).
\end{align*}
Here $O_{Z,\hh}(\varepsilon^2)$ denotes terms depending on $Z$, $\hh$, and its derivatives $D^k \hh$ for $k \le 4$, multiplied by $\varepsilon^2$. Since we work on a fixed compact set $\mathcal{K}$, these terms are uniformly bounded by $C \varepsilon^2$ for some constant $C>0$.
We now rewrite all expressions over the desired denominator and extract the terms of the form $\rho_\varepsilon f_\varepsilon$ with $f_\varepsilon \notin o(\varepsilon)$, collecting the remaining contributions in $Q^\varepsilon_{11}$. This yields
\begin{align}\label{rho_e_contributions}
    \Ric^\e_{11} =&\; -2\partial^2_{Z, \cz}\hh \rho_\e  -\frac{\Lambda}{3+ \frac{\Lambda |Z|^2}{2}}G \rho_\e - \frac{\frac{\Lambda^2}{18}GUV\rho_\e}{P_\e^2 \Tilde{d}_\e^2} \nonumber \\
    &+\frac{\frac{\Lambda}{3}G \rho_\e(-\frac{\Lambda}{6} (\rho_\e*U^+)G + 4 \partial^2_{Z, \cz}\hh\rho_\e * U^+(1+\frac{\Lambda |Z|^2}{6}))}{P_\e^2\Tilde{d}_\e^2} \nonumber  \\
    & -\frac{\frac{4\Lambda}{3}G(1+\frac{\Lambda}{6}|Z|^2) \partial_{Z, \cz}^2\hh \rho_\e(\rho_\e*U^+)}{\Tilde{d_\e}^2P^2_\e} \nonumber\\
    & +\frac{4(1+\frac{\Lambda}{6}|Z|^2)^2((\partial^2_{Z, \cz}\hh)^2 +|\partial^2_{\cz, \cz}\hh|^2) \rho_\e(\rho_\e *U^+) }{P_\e^2\Tilde{d_\e}^2}  + Q^\e_{11}.
\end{align}
Define 
\begin{align}\label{def_E}
    E_0(V,Z, \Bar{Z}):= -\frac{\Lambda^2}{18}G(Z, \cz)V,
\end{align}
and 
\begin{align}\label{def_A}
    A_0(Z, \cz) := &\; \frac{\Lambda}{3}G (-\frac{\Lambda}{6} G + 4 \partial^2_{Z, \cz}\hh(1+\frac{\Lambda |Z|^2}{6})) - \frac{4\Lambda}{3}G(1+\frac{\Lambda}{6}|Z|^2) \partial_{Z, \cz}^2\hh \nonumber \\
    &+4(1+\frac{\Lambda}{6}|Z|^2)^2((\partial^2_{Z, \cz}\hh)^2 +|\partial^2_{\cz, \cz}\hh|^2).
\end{align}
\textbf{Step 2.} Estimate of $(\rho_\e *\Theta)^2 -(\rho_\e *\Theta)$.\\
We first analyze the terms of constant size (with respect to $U$) arising at the level of the connection coefficients. These consist either of contributions of the form $(\rho_\varepsilon * \Theta)$ or of terms independent of $U$, such as constants or products involving the remaining coordinates. To identify the $(\rho_\varepsilon * \Theta)$-terms at the level of the connection, it suffices to consider $\partial_U g_\varepsilon$. For $i,j$ such that $g^\e_{ij} \neq 0$, we have that 
\begin{align*}
    \partial_{1}g^\e_{ij} &= \frac{P_\e \partial_1 \Tilde{g}^\e_{ij}-2\Tilde{g}^\e_{ij} \partial_1 P_\e}{P_\e^3}.
\end{align*}
Hence, denoting terms of a different form by $T^\e_{i,jk}$, we get
\begin{align}\label{relevant_u_derivatives_of_metric}
    \partial_{1}g^\e_{12} &= \frac{2\partial_1 P_\e}{P_\e^3}= \frac{-\frac{\Lambda}{3}(V+ G(\rho_\e*\Theta))}{P_\e^3} + T^\e_{1,12}, \nonumber \\
    \partial_{1}g^\e_{34} &= \frac{P_\e \partial_1 \hat{g}^\e_{34}-2\hat{g}^\e_{34} \partial_1 P_\e}{P_\e^3} + T^\e_{1,34} = \frac{2(1+ \frac{\Lambda}{6}|Z|^2)\partial^2_{Z, \cz}\hh \rho_\e * \Theta+\frac{\Lambda}{3}(V+ G(\rho_\e*\Theta))}{P_\e^3} + T^\e_{1,34}, 
 \nonumber\\
      \partial_{1}g^\e_{33} &= \frac{P_\e \partial_1 \hat{g}^\e_{33}-2\hat{g}^\e_{33} \partial_1 P_\e}{P_\e^3} + T^\e_{1,33} =\frac{2(1+ \frac{\Lambda}{6}|Z|^2)\partial^2_{Z, Z}\hh \rho_\e * \Theta}{P_\e^3} + T^\e_{1,33}, \nonumber \\
     \partial_{1}g^\e_{44} &=\frac{2(1+ \frac{\Lambda}{6}|Z|^2)\partial^2_{\cz, \cz}\hh \rho_\e * \Theta}{P_\e^3} + T^\e_{1,44}.
\end{align}
We note that the constant-size term $V$ also appears in $\partial_1 g^\varepsilon_{12}$ and $\partial_1 g^\varepsilon_{34}$. For $k>1$, however, $\partial_k \hat{g}_\varepsilon$ contains no terms of constant size in $U$. Moreover,
\begin{align*}
    \partial_2 P_\e &= -\frac{\Lambda}{6}U, \\
    \partial_3 P_\e&=  \frac{\Lambda}{6} (\cz + (Z\partial^2_{Z,Z}\hh + \cz \partial^2_{Z, \cz}\hh)(\rho_\e *U^+)), \\
     \partial_4 P_\e&=  \frac{\Lambda}{6}( Z + (Z\partial^2_{Z,\cz}\hh + \cz \partial^2_{\cz, \cz}\hh)(\rho_\e *U^+)).
\end{align*}
Using \eqref{quotient_derivatives}, the only constant-size terms (in $U$) that can arise at the level of the connection coefficients, besides $(\rho_\varepsilon * \Theta)$, are $V$, $Z$, and $\cz$. From \eqref{relevant_u_derivatives_of_metric}, the only connection coefficients that may contain a $(\rho_\varepsilon * \Theta)$-term are $\Gamma^\varepsilon_{[11,2]}$, $\Gamma^\varepsilon_{[13,4]}$, $\Gamma^\varepsilon_{[13,3]}$, and $\Gamma^\varepsilon_{[14,4]}$. Consequently, terms of the form $(\rho_\varepsilon * \Theta)^2$ can arise only through products of such coefficients in the curvature expressions, namely in terms of the form $g_\varepsilon^{st}\Gamma^\varepsilon_{ik,s}\Gamma^\varepsilon_{jl,t}$.
\\In order for such contributions to be of constant size on $(-\varepsilon,\varepsilon)$, it is necessary that $\{s,t\} \in \{\{1,2\}, \{3,4\}\}$. We first consider $\{s,t\} = \{1,2\}$ and, by symmetry, restrict to $s=1$, $t=2$, yielding terms of the form $\Gamma^\varepsilon_{ik,1}\Gamma^\varepsilon_{jl,2}$. A contribution of type $(\rho_\varepsilon * \Theta)^2$ then requires $j=l=1$ and either $\{i,k\}=\{1,2\}$ or $i,k \in \{3,4\}$. In the former case, the resulting term appears only in $R^\varepsilon_{[1112]}$, which vanishes by the symmetries of the curvature tensor. In the latter case, such terms contribute to $R^\varepsilon_{[1313]}$, $R^\varepsilon_{[1414]}$, or $R^\varepsilon_{[1314]}$. Among these, only $R^\varepsilon_{[1314]}$ yields a contribution of constant size to $\Ric^\varepsilon_{11}$, since $g_\varepsilon^{11}=g_\varepsilon^{13}=g_\varepsilon^{14}=0$ and $g_\varepsilon^{33}, g_\varepsilon^{44} \le C\varepsilon$ on $(-\varepsilon,\varepsilon)$.
\\It remains to consider $\{s,t\} = \{3,4\}$, which by symmetry reduces to $s=3$, $t=4$, leading to terms $\Gamma^\varepsilon_{ik,3}\Gamma^\varepsilon_{jl,4}$. In this case, a $(\rho_\varepsilon * \Theta)^2$-term can arise only if $\{i,k\}, \{j,l\} \in \{\{1,3\}, \{1,4\}\}$, corresponding again to curvature components $R^\varepsilon_{[1314]}$, $R^\varepsilon_{[1313]}$, and $R^\varepsilon_{[1414]}$. As before, the only contribution of constant size to the Ricci tensor occurs in $\Ric^\varepsilon_{11}$ via $g_\varepsilon^{34} R^\varepsilon_{3114}$.
\\We now compute all such contributions involving $(\rho_\varepsilon * \Theta)^2$. Since these arise solely from products of connection coefficients, we denote the remaining terms by $T^\varepsilon_{ijkl}$ and obtain
\begin{align*}
    R^\e_{3114} &= g_\e^{st}(\Gamma^\e_{13,s}\Gamma^\e_{14,t}-\Gamma^\e_{11,s}\Gamma^\e_{34,t}) + T^\e_{3114} \\
    &= g_\e^{34}(\Gamma^\e_{13,3}\Gamma^\e_{14,4} + \Gamma^\e_{13,4}\Gamma^\e_{14,3}) -g_\e^{12}\Gamma^\e_{11,2}\Gamma^\e_{34,1} + T^\e_{3114} .
\end{align*}
Using that 
\begin{align*}
    \Gamma^\e_{11,2} = \partial_1g^\e_{12},\  \Gamma^\e_{13,3} = \frac{1}{2}\partial_1g^\e_{33},\  \Gamma^\e_{14,4} = \frac{1}{2}\partial_1g^\e_{44},\ \Gamma^\e_{13,4} = \frac{1}{2}\partial_1g^\e_{34} = -\Gamma^\e_{34,1}, 
\end{align*}
this gives
\begin{align}\label{theta_contribution_constant_size}
    R^\e_{3114}=&  \frac{1}{P_\e^4}\Big(-(1+ \frac{\Lambda}{6}|Z|^2)^2|\partial^2_{Z, Z}\hh|^2 (\rho_\e * \Theta)^2 - \Big((1+ \frac{\Lambda}{6}|Z|^2)\partial^2_{Z, \cz}\hh +\frac{\Lambda}{6}G\Big)^2(\rho_\e*\Theta)^2 \nonumber \\
    &+\Big((1+ \frac{\Lambda}{6}|Z|^2)\partial^2_{Z, \cz}\hh +\frac{\Lambda}{6}G\Big) \cdot \frac{\Lambda}{3} G (\rho_\e*\Theta)^2 \Big) + T^\e_{3114}.
\end{align}
Define 
\begin{align} \label{def_B}
    B_0(Z, \cz) :=& 2\Big((1+ \frac{\Lambda}{6}|Z|^2)^2|\partial^2_{Z, Z}\hh|^2 + \Big((1+ \frac{\Lambda}{6}|Z|^2)\partial^2_{Z, \cz}\hh +\frac{\Lambda}{6}G\Big)^2 \nonumber \\
    &+\frac{\Lambda}{3} G \Big((1+ \frac{\Lambda}{6}|Z|^2)\partial^2_{Z, \cz}\hh +\frac{\Lambda}{6}G\Big)\Big).
\end{align}

We now turn to the remaining terms, namely case (ii). We again distinguish between contributions involving second derivatives in the $U$-direction, i.e.\ terms of the form $\partial^2_{1,1}(\Tilde{g}_\e - \hat{g}_\e)$, and those in $(\Tilde{g}_\e - \hat{g}_\e)$ involving at most first derivatives in the $U$-direction.

\textbf{Terms involving two derivatives in the $U$-direction.} \\
The only relevant term in this case is $\partial^2_{1,1} g^\varepsilon_{34}$. This term appears only in $R^\varepsilon_{[1314]}$ and therefore contributes solely to $\Ric^\varepsilon_{11}$. It thus suffices to analyze
\begin{align*}
    2g_\e^{34}\Big(\frac{-\frac{1}{2}\partial^2_{1,1}(\Tilde{g}^\e_{34}-\hat{g}^\e_{34})}{P_\e^2}\Big) &= \frac{-P_\e^2\Tilde{g_{34}^\e}}{\Tilde{d}_\e} \cdot \frac{P_\e^2(A \rho_\e (\rho_\e * U^+) +B\sigma_\e + E U\rho_\e)}{P_\e^4}  \\
    &= \frac{-\Tilde{g_{34}^\e}P_\e^2}{\Tilde{d}^2_\e} \cdot \frac{\Tilde{d_\e}(A \rho_\e (\rho_\e * U^+) +B\sigma_\e + E U\rho_\e)}{P_\e^2}, 
\end{align*}
where, recalling the definitions \eqref{def_A}, \eqref{def_B},\eqref{def_E} of $A_0, B_0,E_0$, we set:
\begin{align}\label{final_relaxation_coefficiens_noflat}
    A &:= \frac{1}{(1+\frac{\Lambda|Z|^2}{6})^2}A_0,\quad B := \frac{1}{(1+\frac{\Lambda|Z|^2}{6})^2}B_0,\quad E := \frac{1}{(1+\frac{\Lambda|Z|^2}{6})^2}E_0. 
\end{align}

Note that 
$$ \supp\, \sigma_\e \cup \supp \, \rho_\e(\rho_\e * U^+) \cup \supp \, (\rho_\e \cdot U) \subset \{U\in (-\e, \e)\}$$
and all three functions are bounded in $L^\infty$ by a constant $C>0$ independent of $\e$. In the region $\{U\in (-\e, \e)\}$, it holds that
\begin{align*}
    &\Tilde{g}^\e_{34} = 1 + O(\e), \\
    &\Tilde{d}_\e = -1 + O(\e), \\
    & P_\e = 1+\frac{\Lambda}{6}|Z|^2 + O(\e).
\end{align*}
It follows that 
\begin{align}\label{relaxation_terms_noflat}
     2g_\e^{34}\Big( \frac{-\frac{1}{2}\partial^2_{1,1}(\Tilde{g}^\e_{34}-\hat{g}^\e_{34})}{P_\e^2}\Big) &=  \frac{-(1+\frac{\Lambda}{6}|Z|^2)^2(A \rho_\e (\rho_\e * U^+) +B\sigma_\e + E U\rho_\e)}{\Tilde{d}^2_\e P_\e^2} + O(\e).
\end{align}

\textbf{Terms involving at most first derivatives in the $U$-direction.} \\
We first note that  $Dg_\e$, $D\Tilde{g}_\e$, $\Tilde{g}_\e$, $g_\e$ and $P_\e$ are bounded on $\mathcal{K}$ independently of $\e$. Recall that $P_\e$ is also bounded away from $0$. 
Moreover, \begin{align}\label{boundedness_second_deri}
    &|D^2g_\e| + |D^2\Tilde{g}_\e|+ |D^2\hat{g}_\e| \leq C\e^{-1}\ \mathrm{in}\ \{-\e \leq U \leq \e \} \nonumber \\
     &|D^2g_\e| + |D^2\Tilde{g}_\e|+ |D^2\hat{g}_\e| \leq C\ \mathrm{in}\ \{|U| >\e \},
\end{align}
where $C> 0$ that does not depend on $\e$. 
From \eqref{compare_g_hat_tilde}, we get that 
\begin{align}\label{appx_inverse_metric}
  |\Tilde{g}^{ij}_\e- \hat{g}^{ij}_\e| \leq C \e^2 \ \mathrm{in}\ \{U< \e\}  \ \mathrm{and} \   |\Tilde{g}^{ij}_\e- \hat{g}^{ij}_\e| \leq C \e \ \mathrm{in}\ \{U\geq \e\}  \quad \text{for $i, j = 3,4$}.
\end{align}
Combining \eqref{good_convergence}, \eqref{appx_inverse_metric}, \eqref{boundedness_second_deri}, \eqref{rho_e_contributions}, \eqref{def_A}, \eqref{def_E}, \eqref{theta_contribution_constant_size}, \eqref{def_B}, and arguing as in the $\Lambda=0$ case (i.e., see the proof of Proposition \ref{prop:ApproxFlat}), we infer that
\begin{align*}
    \Ric_{11}^\e &= -2\partial^2_{Z, \cz}\hh \rho_\e  -\frac{\Lambda}{3+ \frac{\Lambda |Z|^2}{2}}G \rho_\e+ \Lambda g^\e_{11} + F_\e,\\
    \Ric_{jk}^\e &= \Lambda g_{jk}^\e + F_\e, \quad \text{for all } (j,k) \neq (1,1),
\end{align*}
for some function $F_\e$, with  $F_\e \to 0$ uniformly as $\e \to 0$. 
We conclude that
\begin{align}\label{approximation_ricci}
    \Ric^\e_{ij} = \left(-2\partial^2_{Z, \cz}\hh  -\frac{\Lambda}{3+ \frac{\Lambda |Z|^2}{2}}G \right) \rho_\e \delta^1_i\delta^1_j + \Lambda g^\e_{ij} + r_{ij}^\e,
\end{align}
where $r^\e_{ij}$ converges to $0$ uniformly on $\kk$ as $\e \to 0$.
\end{proof}

\subsection{Uniform global hyperbolicity of the smooth approximations}

To prove that IPP-waves satisfy the $\tcd$ condition, we will apply the stability Theorem \ref{stability_tcd} to suitable smooth approximations whose timelike Ricci curvature admits uniform lower bounds. A key assumption in Theorem \ref{stability_tcd} is that the approximating metrics are uniformly globally hyperbolic (see Definition \ref{uniform_global_hyperbolicity_def}). The purpose of this section is to modify the smooth approximations constructed in Propositions \ref{prop:ApproxFlat} and \ref{prop:ApproxNonFlat} so that this condition is satisfied.

Firstly, note that 
\begin{align}
    &\rho_\e *(U^+) = U^+,\ \mathrm{for\ } |U| \geq \e, \\
    &|\rho_\e *(U^+) - U^+| \leq \e,\ \mathrm{for\ } |U| \leq \e, \\
    & |\rho_\e *(U^+)^2 - (U^+)^2| \leq C(\e^2 + U^+\e),  
\end{align}
where the first identity follows from the rotational symmetry of  $\rho$. 
Recalling \eqref{definition_g_eps_tilde_noflat}, \eqref{decomposition_in_metric_and_balance_terms_noflat}, \eqref{def_A}, \eqref{def_B}, and \eqref{def_E} as well as 
\begin{align}
    |\iota_\e| + |\kappa_\e| + |\zeta_\e| \leq C(\e^2 + \e U^+),
\end{align}
we get that
\begin{align}\label{difference_to_g_tilde}
    |\Tilde{g}^\e_{ij}- \Tilde{g}_{ij}| \leq Q(\Lambda, Z, \cz, V, \hh, D\hh, D^2\hh)(\e^2 + U\e),\ \mathrm{for}\ |U| \geq \e,
\end{align}
where $Q$ is a non-negative smooth function and $i,j \in \{3,4\}$.
\\Define the $(0,2)$-tensor   $(\overset{\circ}{g}^\e_{ij})_{1 \leq i,j, \leq 4} \in C^\infty(\R^4, \R^{4 \times 4})$, by 
\begin{align}
    &\overset{\circ}{g}^\e_{34} = \overset{\circ}{g}^\e_{43} := \Tilde{g}^\e_{34} +4(\e^2 +U\e)Q, \label{eq:gcircQ} \\
    &\overset{\circ}{g}^\e_{ij} := \Tilde{g}^\e_{ij}\ \mathrm{for} \ (i,j) \neq (3,4), (4,3), \nonumber
\end{align}
and define the $(0,2)$-tensor $\check{g}_\e$ by
\begin{align}\label{eq:defgepscheck}
    \check{g}_\e = \frac{1}{P_\e^2} \begin{pmatrix}
        0 & -1 & 0 & 0 \\
        -1 & 0 & 0 & 0 \\
        0 & 0 & \overset{\circ}{g}^\e_{33} & \overset{\circ}{g}^\e_{34} \\
         0 & 0 & \overset{\circ}{g}^\e_{43} & \overset{\circ}{g}^\e_{44}
    \end{pmatrix}.
\end{align}
Note that $\check{g}_\e \to g$ locally uniformly and $\liminf_{\e\to 0} \domlor(\check{g}_\e)\subset \domlor(g)$, where the left hand side denotes the Kuratowski $\liminf$ of a family of sets.
The upper bound \eqref{difference_to_g_tilde} implies that
\begin{align}
    \check{g}_\e \prec g \ \mathrm{on} \ \{U \geq \e\} \cap \domlor(\check{g}_\e) \cap \domlor(g).
\end{align}
Moreover, combining \eqref{approximation_ricci} with  the bounds
\begin{align*}
    &|\overset{\circ}{g}_\e -\Tilde{g}_\e|\leq C|\e^2 + U\e| \leq C\e^2,\ \mathrm{for\ } |U| \leq \e, \\
    &|\overset{\circ}{g}_\e -\Tilde{g}_\e|\leq C\e,\ \mathrm{for\ } |U| \geq \e, \\ 
    &|\partial_U(\overset{\circ}{g}_\e -\Tilde{g}_\e)| \leq C\e, \ \mathrm{and} \\
    &\partial^2_{U,U}(\overset{\circ}{g}_\e -\Tilde{g}_\e)=0,
\end{align*}
and arguing as in section \ref{subsubsec:general_case}, 
we get  that
\begin{align}\label{normal_vs_check_g}
    |\Ric[{g}_\e]- \Ric[\check{g}_\e]| \leq C(K)\e, 
\end{align}
for every compact set $K \Subset \domlor(\hat{g}_\e) \cap \domlor(\check{g}_\e)$.

From the above discussion, we conclude that $\check{g}_\varepsilon \to g$ locally uniformly, and that the Ricci curvature of $\check{g}_\varepsilon$ converges to that of $g$ both distributionally (globally), and locally uniformly away from $\{U=0\}$. We will henceforth restrict $\check{g}_\varepsilon$ to $\domlor(\check{g}_\varepsilon) \cap \domlor(g) \subset \mathbb{R}^4$.

We next recall some basic properties of Lorentzian spaces of constant curvature. Let $m$ denote the Minkowski metric on $\mathbb{R}^4$, i.e.\ $\mathbb{R}^4$ endowed with coordinates $(t,w,x,y)$ and metric $m=\mathrm{diag}(-1,1,1,1)$. For $\Lambda \in \mathbb{R}$, one obtains a Lorentzian metric of constant sectional curvature $\Lambda/3$ (and hence constant Ricci curvature $\Lambda$) via the conformal transformation
\[
m_\Lambda := \frac{m}{\mathcal{P}_\Lambda^2},
\]
where
\[
\mathcal{P}_\Lambda(t,w,x,y) := 1 - \frac{\Lambda}{6}(x^2+y^2+w^2 - t^2).
\]
If $\Lambda \neq 0$, the function $\mathcal{P}_\Lambda$ may vanish; we denote its zero set by
\[
N_\Lambda := \mathcal{P}_\Lambda^{-1}(0),
\]
and set $O_\Lambda := \mathbb{R}^4 \setminus N_\Lambda$. A direct computation shows that $(O_\Lambda, m_\Lambda)$ has constant sectional curvature $\Lambda/3$, and hence constant Ricci curvature $\Lambda$.
Notice that if $\Lambda >0$, $N_\Lambda$ consists of two distinct connected hypersurfaces that separate $O_\Lambda$ into three connected components and for $\Lambda <0$, $N_\Lambda$ consists of one connected hypersurface that splits $O_\Lambda$ into two connected components.
\\The following lemma shows that if two points $p,q \in \mathbb{R}^4$ have a compact causal diamond with respect to $m_\Lambda$, i.e.\ if $J^+_{m_\Lambda}(p) \cap J^-_{m_\Lambda}(q)$ is compactly contained in $\{\mathcal{P}_\Lambda > 0\}$, then the corresponding Minkowski diamond remains disjoint from $\{\mathcal{P}_\Lambda \le 0\}$ and
\(
J^+_{m_\Lambda}(p) \cap J^-_{m_\Lambda}(q) = J^+_{m}(p) \cap J^-_{m}(q).
\)
\medskip

\begin{lemma}\label{lem:DiamDS}
   Consider the Lorentzian manifold $(O_\Lambda =\R^4 \setminus N_\Lambda, m_\Lambda=\frac{m}{\mathcal{P}_\Lambda^2})$ defined above. Let  $p, q \in O_\Lambda$ lie in the same connected component of $O_\Lambda$ and such that $ \emptyset \neq J_{m_\Lambda}^+(p) \cap J_{m_\Lambda}^-(q) \Subset  O_\Lambda$. Then  
   \begin{equation}\label{eq:JmLJm}
   J_{m_\Lambda}^+(p) \cap J_{m_\Lambda}^-(q) = J_{m}^+(p) \cap J_{m}^-(q).
   \end{equation}
\end{lemma}
\begin{proof}
    Fix  $\Lambda \neq 0$ and let $p, q$ be as in the statement of the lemma. Denote by $\Omega$ the open connected component of $O_\Lambda$ that contains $p$ and $q$. Since $m_\Lambda$ and $m$ are conformally equivalent on $\Omega$, and the causal relations are conformally invariant,   in order to obtain \eqref{eq:JmLJm}, it suffices to show that  
 \begin{equation}\label{eq:JmOmega}
 J_{m}^+(p) \cap J_{m}^-(q) \subset \Omega.
 \end{equation}
     By assumption, there exists a $m_\Lambda$-causal, future directed curve $\gamma_0:[0, 1] \to \Omega$ such that $\gamma_0(0) =p$, $\gamma_0(1)=q$. Denote 
    $$K:= J_{m_\Lambda}^+(p) \cap J_{m_\Lambda}^-(q)$$ 
    and observe that there exists  $\delta>0$ such that 
    \begin{equation}\label{eq:dKdeOmega}
    \sfd_{euc}(K, \partial \Omega) = \delta>0.
    \end{equation}
    Suppose by contradiction that there exists a $m$-causal, future directed curve $\gamma_1:[0,1] \to \R^4$ such that $\gamma_1(0) = p$ and $\gamma_1(1) = q$ and $\gamma_1([0,1]) \not \subset \Omega$. For $\nu \in [0,1]$, define 
    $$\gamma_\nu:[0,1] \to \R^4, \quad \gamma_\nu(t) := (1-\nu) \gamma_0(t) + \nu \gamma_1(t).$$
    Since every $m_\Lambda$-causal curve is $m$-causal, the convexity of the future cone in Minkowski spacetime implies that $\gamma_\nu$ is $m$-causal for all $\nu \in [0,1]$.  By continuity, $\gamma_0([0,1]),\gamma_1([0,1]) \subset \R^4$ are compact and hence bounded in $\R^4$. Moreover, there exists $\delta>0$ such that $B^{euc}_\delta(\gamma_0([0,1])) \subset \Omega$. Therefore,
    \begin{align}
        \nu^* := \inf\{\nu \in [0,1] \colon \gamma_\nu([0,1]) \cap \Omega^c \neq \emptyset\} \in (0,1].
    \end{align}
    As $\Omega$ is open,  there exists  $t^* \in (0,1)$ such that $\gamma_{\nu^*}(t^*) \notin \Omega$. By definition,  for $\nu < \nu^*$ we have that $\gamma_\nu([0,1]) \subset \Omega$ and hence $\gamma_\nu(t^*) \in J_{m_\Lambda}^+(p) \cap J_{m_\Lambda}^-(q)$. Then, for an arbitrary increasing sequence $\nu _n \in (0, \nu)$ such that $\lim_{n \to \infty} \nu_n = \nu$, the sequence $(\gamma_{\nu_n}(t^*))_{n}$ converges to $\gamma_{\nu}(t^*) \in \partial \Omega$.
     But then there exists  $n \in \N$ such that $\sfd_{euc}(\gamma_{\nu_n}(t^*), \partial\Omega)< \delta$. This contradicts \eqref{eq:dKdeOmega}, since  $\gamma_{\nu_n}(t^*) \in K =J_{m_\Lambda}^+(p) \cap J_{m_\Lambda}^-(q).$ 
\end{proof}
\smallskip

\begin{lemma}\label{rescaled_metric_emeralds_wave}
   Let $g$ be the IPP-wave metric as in \eqref{nonsmooth_metric}, $P$ be as in \eqref{def:P}, and  $M_\Lambda=\domlor{g}\subset \R^4$. Denote by $\Tilde{g}$ the unique continuous extension of $P^2g$ to $\R^4$. Let  $K_1, K_2 \subset M_\Lambda$ be compact sets  contained in the same connected component of $M_\Lambda$.
   Assume that either
        \begin{itemize}
\item   $\Lambda \geq 0$, or  
   \item  $\Lambda <0$ and $J^+_g(K_1) \cap J^-_g(K_2) \subset  M_\Lambda$ is compact. 
    \end{itemize}
    Then 
    \begin{equation}\label{eq:DiamCmpML}
    J^+_{\Tilde{g}}(K_1) \cap J^-_{\Tilde{g}}(K_2)  \Subset M_\Lambda.
    \end{equation}
\end{lemma}
\begin{proof}
    Under the coordinate change 
    $$U+V =t,\; U-V=w,\; Z=x+iy,$$ 
    we get that 
    $$g|_{\{ U \leq 0\}} = m_\Lambda.$$
    Analogously, under the coordinate change 
    $$\mathcal{U} + \mathcal{V} = t, \; \mathcal{U} - \mathcal{V} = w,\; \eta = x+iy, \text{ and } \mathcal{U} = U, \;\mathcal{V} = V +\hh+U\partial_Z \hh \partial_{\cz}\hh, \; \eta = Z + U\partial_{\cz}\hh,$$
    we get that 
    $$g|_{\{ U \geq 0\}} = m_\Lambda.$$ 
    The claim \eqref{eq:DiamCmpML} is thus a direct consequence of Lemma \ref{lem:DiamDS}.
\end{proof}

\begin{proposition}[Uniform global hyperbolicity of approximations. Case $\Lambda\geq 0$]\label{uniform_global_hyper_positive}
 Let $g$ be the IPP-wave metric as in \eqref{nonsmooth_metric} with $\Lambda\geq 0$, and  $M_\Lambda=\domlor{g}\subset \R^4$.
    Assume that $|D^2\hh|$ is uniformly bounded and $\hh$ has subquadratic growth, i.e.\ \eqref{eq:hhSubQuad} holds.
    \\Then the approximating family  $(\check{g}_\e)$ defined in \eqref{eq:defgepscheck} is uniformly globally hyperbolic in $M_\Lambda$, namely:
    for all compact sets $K_1, K_2 \subset M_\Lambda$, there exist $\bar{\e}>0$ and a compact set $\mathcal{K} \subset M_\Lambda$, such that
    \begin{equation}\label{eq:UGHLageq0}
        J^+_{\check{g}_\e}(K_1) \cap J^-_{\check{g}_\e}(K_2) \subset \mathcal{K},\quad  \text{for all } \e\in (0,\bar{\e}).
    \end{equation}
\end{proposition}
\begin{proof}
    Let $a_2 \in \R$ be such that $|D^2\hh| \leq a_2$ everywhere. Then  there exist  $a_0, a_1\in \R$ such that
    $$|D\hh| \leq a_1 + a_2|Z| \text{ and }|\hh| \leq a_0 + a_1|Z| + a_2|Z|^2.$$
    It follows that there exist $\alpha_0, \alpha_2\in \R$ such that
    \begin{align*}
        |G(Z, \cz)| + |\hh| \leq \alpha_0 + \alpha_2 |Z|^2.
    \end{align*}
     Recalling the expressions of $A_0, B_0$, and $E_0$ (see \eqref{def_A}, \eqref{def_B}, and \eqref{def_E}), \eqref{final_relaxation_coefficiens_noflat}, and using the assumption $\Lambda \geq 0$, we get that there exists a constant $C = C(a_2, \alpha_0, \alpha_2, \Lambda)>0$ such that 
     \begin{align*}
         |A| + |B| + |E| + (1+|D^2\hh|)^2 \leq C(1+V).
     \end{align*}
     Recall the definition  \eqref{eq:gcircQ} of $\overset{\circ}{g}^\e$ in terms of $Q$, as  in \eqref{difference_to_g_tilde}.
     The estimates above imply that there exists a constant $C' = C'(a_2, \alpha_0, \alpha_2, \Lambda) > 0$, such that   $Q \geq  C'(1+V)$. That gives that, for all $\e > 0$:
     \begin{align*}
         &|\overset{\circ}{g}^\e_{34}-1| \leq Q (U^++\e)(1+|U|), \nonumber \\
          &|\overset{\circ}{g}^\e_{33}| + |\overset{\circ}{g}^\e_{44}|  \leq Q(U^++ \e)(1+U^+).
     \end{align*}
     Now, fix  $R>0$ and consider the set 
     $$\Omega_R:= \{-R \leq U, V \leq R\}.$$ 
     By continuity, up to replacing $Q$ with its maximum, we may choose $Q$ to be constant on $\Omega_R$.
     By Lemma \ref{uv_monotone_along_causal_curves}, for every two compact sets $K_1, K_2\subset M_\Lambda$ there exists $R>0$ such that  
     $$J^+_{\Tilde{g}}(K_1) \cap J^-_{\Tilde{g}}(K_2)\subset \Omega_R.$$
     It follows that  there exists $\e_0>0$ such that 
     \begin{align*}
         \overset{\circ}{g}_\e[3,4] \geq \frac{1}{2}\mathrm{Id}_2, \ \mathrm{on\ } \{|U| \leq \e_0\}\cap \Omega_R, \quad \text{for all }\e \in (0, \e_0),
     \end{align*}
     in particular, 
     $$\{|U| \leq \e_0\} \cap \Omega_R \subset \domlor(\overset{\circ}{g}^\e),\quad  \text{ for } \e \in (0, \e_0).
     $$
      Furthermore,  for $U \leq -\e_0$ and $\e \in (0, \e_0)$, it holds that
      \begin{align*}
           &\overset{\circ}{g}^\e_{34}-1 = 4Q\e(\e + U), \nonumber \\
          &\overset{\circ}{g}^\e_{33}=\overset{\circ}{g}^\e_{44} =0.
      \end{align*}
     Hence, there exists  $\e_1 >0$ such that
     $$\Omega_R \cap \{U \leq \e_0\} \subset \domlor(\overset{\circ}{g}_\e), \quad \text{for all } \e \in (0, \e_1).$$
     Fix two compact sets $K_1, K_2 \subset \domlor(g) \cap \Omega_R$. 
      From the explicit expression of $\overset{\circ}{g}_\e$,  there exists  $\e_2 \in (0, \min(\e_1, \e_0))$ such that 
      \begin{align*}
          \overset{\circ}{g}_\e \prec m_{1/2} := \begin{pmatrix}
              0 & -1 & 0 & 0 \\
               -1 &  0 & 0 & 0\\
              0 & 0 & 0 & \frac{1}{2} \\
              0&0& \frac{1}{2} & 0
          \end{pmatrix}, 
          \quad \text{on } \{U \leq \e_0\}, \; \text{ for all } \e \in (0, \e_2).
      \end{align*}
      By the properties of the Minkowski metric, the set
      \begin{align*}
          K:=J^+_{m_{1/2}}(K_1) \cap \{U\leq \e_0, V \leq R\} \subset \R^4
     \end{align*}
      is compact in $\R^4$ with the Euclidean topology and so is $K_0 := K \cap \{U = \e_0\} \subset \R^4$. Note that,  for all $ \e \in (0, \e_2)$, it holds 
     \begin{align*}
          J^+_{\overset{\circ}{g}_\e}(K_1) \cap J^-_{\overset{\circ}{g}_\e}(K_0) \cap \{U\leq \e_0\} \subset J^+_{m_{1/2}}(K_1)\cap J^-_{m_{1/2}}(K_0) \cap \{U\leq \e_0\}.
      \end{align*}
    Since $m_{1/2}$ is a rescaled Minkowski metric, we have that 
    \begin{align*}
        J^+_{m_{1/2}}(K_1)\cap J^-_{m_{1/2}}(K_0) \subset \R^4
    \end{align*}
    is compact.
    Recalling that $\overset{\circ}{g}_\e \prec \Tilde{g}$ on $\{U \geq \e_0\}$, we get that 
    \begin{align*}
         J^+_{\overset{\circ}{g}_\e}(K_0) \cap J^-_{\overset{\circ}{g}_\e}(K_2) \subset  J^+_{\Tilde{g}}(K_0) \cap J^-_{\Tilde{g}}(K_2),
    \end{align*}
    where the right-hand side again is compact in $\domlor(\Tilde{g}) \subset \R^4$ as $\Tilde{g}$ is the Minkowski metric under a coordinate change. Hence,  the compact set $K'' \subset \R^4$ defined by 
    \begin{align*}
        K'' := (J^+_{m_{1/2}}(K_1)\cap J^-_{m_{1/2}}(K_0)) \cup (J^+_{\Tilde{g}}(K_0) \cap J^-_{\Tilde{g}}(K_2))
    \end{align*}
    satisfies that 
    $$J^+_{\check{g}_\e}(K_1) \cap J^-_{\check{g}_\e}(K_2) \subset K'', \quad \text{for all } \e\in (0,\e_2).$$
    By  \cite[Thm.\ 4.5]{samann2016global}, there exists a Lorentzian metric $g'$ on $\R^4$ such that $g'$ is globally hyperbolic and $g' \succ \Tilde{g}$. By Lemma \ref{rescaled_metric_emeralds_wave} and the global hyperbolicity (see Proposition \ref{non-smooth_metric_globally_hyperbolic}) we have that 
    $$S:= J^+_{{g}}(B_\delta(K_1)) \cap J^-_{{g}}(B_\delta(K_2)) = J^+_{\Tilde{g}}(B_\delta(K_1)) \cap J^-_{\Tilde{g}}(B_\delta(K_2)) \subset \domlor(g)$$ is compact. Thus  there exists  $\theta > 0$ such that $$B^h_\theta(S) \subset \domlor(g).$$ Denote $$\mathcal{K}:= \overline{B^h_{\theta/2}}(S).$$ 
    Since $\mathcal{K}\subset M_\Lambda$ is compact, the desired \eqref{eq:UGHLageq0} follows by the following claim.
    
    \textbf{Claim:} There exists  $\bar{\e} > 0$ such that 
    \begin{align*}
         J^+_{\check{g}_\e}(K_1) \cap J^-_{\check{g}_\e}(K_2) \subset \mathcal{K}, \quad \text{for all }\e \in (0, \bar{\e}).
    \end{align*}

    \emph{Proof of the claim.} If that was not the case, one could find a sequence $\e^k>0$ such that $\e^k \to 0$, $x_{\e^k} \in \mathcal{K}^c$ and $\check{g}_{\e^k}$-causal curves $\gamma_{\e^k}$ such that $\gamma_{\e^k}(0) \in K_1$, $\gamma_{\e^k}(1) \in K_2$ and $\gamma_{\e^k}((0,1)) \ni x_{\e^k}$. It follows that $\gamma_{\e^k}$ is $\overset{\circ}{g}_{\e^k}$-causal and hence it is contained in $K''$. For $\e^k$ small enough we have that $g' \succ \overset{\circ}{g}_{\e^k}$ on $K''$. 
    We may now reparametrise the curves to unit Euclidean speed and get a limit curve $\gamma$ by Arzel\'a--Ascoli's theorem and the fact that $g'$ is non-totally imprisoning. 
    Note that there exists a $t_0>0$ such that $\gamma_{\e^k}^{-1}(\mathcal{K}^c) \subset (t_0, 1-t_0)$.
   By the limit curve theorem (see \cite[Thm.\ 3.1 (1)]{minguzzi2008limit}),  the limit curve $\gamma$ is $\Tilde{g}$-causal and passes through a point $x \in \mathcal{K}^c$. This gives a contradiction, thus proving the claim.  
\end{proof}

The case of negative cosmological constant requires additional care. An important step in the proof is a coordinate change, see \eqref{eq-neg-coord-shift-def}.
\smallskip
\begin{proposition}[Uniform global hyperbolicity of approximations. Case $\Lambda<0$]\label{lambda_negative_coord_shift}
 Let $g$ be the IPP-wave metric as in \eqref{nonsmooth_metric} with $\Lambda< 0$, and  $M_\Lambda=\domlor{g}\subset \R^4$.
    Assume that $|D^2\hh|$ is uniformly bounded and $\hh$ has subquadratic growth, i.e.\ \eqref{eq:hhSubQuad} holds. Let $K_1, K_2 \subset M_\Lambda$ be compact sets  such that $J^+_g(K_1) \cap J^-_g(K_2) \subset  M_\Lambda$ is compact. Then the approximating family  $(\check{g}_\e)$, defined in \eqref{eq:defgepscheck}, is uniformly globally hyperbolic in a neighbourhood of $J^{+}_{g}(K_1) \cap J^-_{g}(K_2)$.   
\end{proposition}

\begin{proof}
    For any $w \in \C$,  consider the function 
    \begin{equation}\label{eq-neg-coord-shift-def}
    \hh_w: (Z, \cz) \mapsto \hh(Z-w, \cz-\Bar{w}).
    \end{equation}
    Let  $R>0$ be such that 
    $$K_1, K_2 \subset \Omega_R:=\{-R \leq U, V \leq R\}.$$
    For $|Z| \geq (1+ \frac{|\Lambda|}{6} +\frac{6}{|\Lambda|} +R)^4$ and $|U|, |V| \leq R$, we get that
    \begin{align*}
        \left|1+\frac{\Lambda}{6}(|Z|^2-UV)\right| \leq \frac{|\Lambda|}{6}(|Z|^2 - |Z|).  
    \end{align*}
    Then there exists  $Z_0 \geq (1+ \frac{|\Lambda|}{6} +\frac{6}{|\Lambda|} +R)^4$ and  $c>0$ such that 
    \begin{align}
         \left|1+\frac{\Lambda}{6}(|Z|^2-UV)\right| \geq c(1+ |Z|^2)\ \mathrm{on}\ \{|Z| \geq Z_0\} \cap \Omega_R.
    \end{align}
    Since $|D^2\hh|$ is  bounded, we can argue as in the proof of Proposition \ref{uniform_global_hyper_positive} and infer that
    \begin{itemize}
 \item   we may choose $Q$ constant on $\Omega_R \cap \{|Z| \geq Z_0\}$;
 \item there exists $\e_0 >0$ such that, for all $\e \in (0, \e_0)$:
 \begin{enumerate}
\item[(i).] $\Tilde{g}, \overset{\circ}{g}_\e$ are well defined and Lorentzian on $\{|Z| \geq Z_0\} \cap \Omega_R \cap \{U \leq \e_0\}$;
\item[(ii).] $\overset{\circ}{g}_\e \prec \Tilde{g}$ on $\Omega_R \cap \{|Z| \geq Z_0\} \cap \{U \geq \e_0\} \cap \domlor{\overset{\circ}{g}_\e}$. 
\end{enumerate}
\end{itemize}
It follows that there exists  $\e_1 \in (0, \e_0)$ such that 
    \begin{align}\label{control_by_rescaled_minkowski_large_Z}
        \overset{\circ}{g}_\e \prec m_{1/2}, \ \mathrm{on} \ \{|Z| \geq Z_0\} \cap \Omega_R \cap \{U \leq \e_0\}, \quad \text{for all } \e \in (0, \e_1).
    \end{align}
  From the properties of the Minkowski metric it follows that for each $Z_1>0$ and each set $\mathcal{K} \subset \Omega_R \cap \{|Z| \geq Z_1\}$: 
  \begin{align}\label{minkowski_growth}
      (J^+_{m_{1/2}}(\mathcal{K}) \cap \Omega_R), (J^-_{m_{1/2}}(\mathcal{K}) \cap \Omega_R) \subset \Omega_R \cap \{|Z| \geq Z_1-4R\}.
  \end{align}
 Moreover,  
 \begin{align*}
     K:= J^+_{m_{1/2}}(K_1) \cap \{U \leq \e_0, V \leq R\}
 \end{align*}
 is compact.
  Define $K_0 := K \cap \{U= \e_0\} \subset \R^4$ and note that it is compact. 
  Define the compact set
 \begin{align}\label{compactness_after_eps_0}
       K':=J^-_{\Tilde{g}} (K_2) \cap J^+_{\Tilde{g}}(K_0) \subset \{U \geq \e_0\}.
  \end{align}
  Choosing a suitable  $w \in \C$ and using $\hh_w$ instead of $\hh$, we get that 
  $$K_1 \cap \{U \leq \e_0\}, K_2 \cap \{U \leq \e_0\},\quad  K', K \subset \{|Z| \geq Z_0 + 8R +1\}.$$
  Moreover,  for all $\e \in (0, \e_1) \subset (0, \e_0)$, it holds that 
  \begin{align}\label{compactness_after_eps_0}
      J^-_{\overset{\circ}{g}_\e} (K_2) \cap J^+_{\overset{\circ}{g}_\e}(K_0) \subset J^-_{\Tilde{g}} (K_2) \cap J^+_{\Tilde{g}}(K_0) \subset \{U \geq \e_0\}.
  \end{align}
  By \eqref{control_by_rescaled_minkowski_large_Z}, \eqref{minkowski_growth}, and our choice of $w$, we have that for $\e \in (0, \e_1)$, it holds
   \begin{align*}
      J^+_{\overset{\circ}{g}_\e}(K_1 \cap \{U \leq \e_0\}) \cap \Omega_R \subset J^+_{m_{1/2}}(K_1 \cap \{U \leq \e_0\}) \cap \{|Z| \geq Z_0 +1\} \cap \Omega_R \cap \{U \leq \e_0\}.
  \end{align*}
  This gives that 
  \begin{align*}
      J^+_{\overset{\circ}{g}_\e}(K_1) \cap J^-_{\overset{\circ}{g}_\e}(K_2) &\subset (J^-_{\Tilde{g}} (K_2) \cap J^+_{\Tilde{g}}(K_0)) \cup  J^+_{\overset{\circ}{g}_\e}(K_1 \cap \{U \leq \e_0\})  \cap \Omega_R) \\
      & \subset K \cup K'=: K'', \quad \text{for all } \e \in (0 , \e_1),
  \end{align*}
  where $K''\subset \R^4$ is compact. We now complete the proof by following the arguments of Proposition \ref{uniform_global_hyper_positive}.
  \end{proof}
\subsection{Synthetic timelike lower Ricci curvature bounds for IPP-waves}

We now combine the results of the previous subsections to establish the $\tcd$-condition for impulsive gravitational waves.
\smallskip

\begin{theorem}\label{thm:tcd_for_waves}
Let $\hh \in C^{3,1}_{\mathrm{loc}}(\mathbb{C}^2;\mathbb{R})$ be of subquadratic growth, i.e.\ \eqref{eq:hhSubQuad} holds, and assume that $D^2 \hh$ is bounded. Moreover, suppose that
\begin{equation*}
     -2\partial^2_{Z, \cz} \hh  -\frac{\Lambda(\hh - Z\partial_Z \hh -\cz \partial_{\cz}\hh)}{3(1+\frac{\Lambda}{6}|Z|^2)} \geq 0 \quad \text{on } \left\{1+\frac{\Lambda}{6}|Z|^2 \neq 0 \right\}.
     \end{equation*}
Let $g=g_{\hh, \Lambda}$ be the IPP-wave metric as in \eqref{nonsmooth_metric} with profile function $\hh$ and cosmological constant $\Lambda$, defined on   $M_\Lambda=\domlor{g}\subset \R^4\simeq\C^2$.
    \begin{itemize}
        \item Case $\Lambda\geq 0$. The impulsive gravitational wave $(M_\Lambda, g_{\hh, \Lambda})$ satisfies the $\tcd^e_p(\Lambda, 4)$-condition, for all $p\in (0,1)$;
        \item Case $\Lambda<0$. The impulsive gravitational wave $(M_\Lambda, g_{\hh, \Lambda})$ satisfies the $\tcd^e_p(\Lambda, 4)$-condition locally, in the following sense. Let  $K_1, K_2 \subset M_\Lambda$ be compact sets  such that $J^+_g(K_1) \cap J^-_g(K_2) \subset  M_\Lambda$ is compact; then $(J^+_g(K_1) \cap J^-_g(K_2), g_{\hh, \Lambda})$ satisfies the $\tcd^e_p(\Lambda, 4)$-condition, for all $p\in (0,1)$.
    \end{itemize}
\end{theorem}
\begin{proof}
From \cite[Thm.\ 1.1]{LLS:21}, we know that every length-maximizing curve in $(M_\Lambda, g)$ can be parametrized as a (Filippov) geodesic  of class $\mathcal{C}^{1,1}$. Moreover, by \cite{GL:18}, every length-maximizing curve is either timelike or null, i.e.\ it has a causal character. 
    \smallskip

Recall that the smooth approximations $\check{g}_\e$  defined in \eqref{eq:defgepscheck} converge to $g=g_{\hh,\Lambda}$ locally uniformly, as $\e\to 0$.
    
   \textbf{Case $\Lambda\geq 0$.}   By Proposition \ref{uniform_global_hyper_positive}, $(\check{g}_\e)_{\e\in (0,\bar{\e})}$  are uniformly globally hyperbolic in $M_\Lambda$.  Moreover, \eqref{normal_vs_check_g} combined with Proposition \ref{prop:ApproxFlat} (in case $\Lambda=0$) or Proposition \ref{prop:ApproxNonFlat} (in case $\Lambda>0$; see in particular  \eqref{approximation_ricci}) yields that 
    $$\Ric[\check{g}_\e] \geq (\Lambda-r_\e) \check{g}_\e,$$ where $r_\e \to 0$ locally uniformly as $\e \to 0$. Applying Theorem \ref{stability_tcd} to the sequence $\check{g}_\e$, we conclude that the impulsive gravitational wave $(M_\Lambda, g_{\hh, \Lambda})$  satisfies the $\wtcd(\Lambda,4)$-condition. 
\smallskip

 \textbf{Case $\Lambda< 0$.} Given two compact sets $K_1, K_2 \subset M_\Lambda$, Proposition \ref{lambda_negative_coord_shift} gives a coordinate shift such that  the resulting smooth approximations $(\check{g}_\e)_{\e\in (0,\bar{\e})}$  are uniformly globally hyperbolic in a neighborhood of $J^+_g(K_1) \cap J^-_g(K_2)$. Then,  combining \eqref{normal_vs_check_g}, Proposition \ref{prop:ApproxNonFlat} and Theorem \ref{stability_tcd} similarly to the case $\Lambda \geq 0$, we infer that $(J^+_g(K_1) \cap J^-_g(K_2), g_{\hh, \Lambda})$ satisfies the $\wtcd(\Lambda,4)$-condition. 

To upgrade from $\wtcd(\Lambda,4)$ to $\tcd(\Lambda,4)$, we recall from Proposition \ref{prop-non-bra} that impulsive gravitational waves are timelike non-branching. It then follows from \cite[Thm.\ 1.5]{Bra:23} that, in both cases above, the $\wtcd(\Lambda,4)$ condition can be strengthened to $\tcd(\Lambda,4)$.
\end{proof}
\subsection{Boosted continuous coordinates and synthetic upper Ricci curvature bounds for impulsive gravitational waves}

For simplicity, throughout this subsection we restrict to the case of vanishing cosmological constant, i.e.\ $\Lambda = 0$. In contrast to the previous sections, we adopt a different coordinate system for impulsive gravitational waves, which is better suited for the analysis of optimal transport along arbitrary timelike directions. As in \cite[Sec.\ 2.1]{podolsky2014global}, we consider the flat Minkowski space with the coordinates $\uu, \vv, \eta, \Bar{\eta}$, and the metric
$$\di s^2 = -2\, \di \mathcal{U} \di \mathcal{V} + 2\, \di \eta\, \di \Bar{\eta}.$$
The hypersurface $\{\mathcal{U}=0\}$ is removed, and the two halves $\{\mathcal{U}>0\}$ and $\{\mathcal{U}<0\}$ are re-glued via a warp determined by the function $\mathcal{H}$; that is, we identify
\begin{align}\label{warp_upper_bound}
    (\uu = 0^-, \vv, \eta, \Bar{\eta}) = (\uu = 0^+, \vv-\hh(\eta, \Bar{\eta}), \eta, \Bar{\eta}).
\end{align}
Defining the real valued coordinates $x, y$ via  $\eta = \frac{x+iy}{\sqrt{2}}$ yields that
\begin{align}\label{cut_and_paste_xy_coords}
    \di s^2 &=  -2 \di \mathcal{U} \di \mathcal{V} + (\di x)^2 + (\di y)^2.
\end{align}
Given a timelike direction $w$, we introduce a coordinate system adapted to $w$ in which we construct a family of optimal transport plans moving along $w$ and satisfying the synthetic upper Ricci curvature bounds. For arbitrary $a, b \in \R$, consider the coordinate transform 
\begin{equation}\label{coord_chage_ab}
\begin{cases}
   &\mathcal{U} = u, \\
    &\mathcal{V} = v + \Theta(u) \hh + aX + bY + u^+(a\partial_X \hh + b\partial_Y \hh + \frac{1}{2}((\partial_X \hh)^2+ (\partial_Y \hh)^2)), \\
    & x = X + au + u^+ \partial_X \hh, \\
    & y = Y + bu + u^+ \partial_Y \hh.
\end{cases}
\end{equation}
Computing the metric tensor separately in $\{\mathcal{U}>0\}$ and  $\{\mathcal{U}<0\}$, and patching together yields 
\begin{align*}
  \di s^2 =&  -2 \di \mathcal{U} \di \mathcal{V} + (\di x)^2 + (\di y)^2 \nonumber \\
   = & (a^2 + b^2) (\di u)^2 - 2 \di u \di v + (1+2u^+\partial^2_{XX} \hh + (u^+\partial^2_{XX} \hh)^2 +(u^+\partial^2_{XY} \hh)^2) (\di X)^2 \nonumber \\
  &+ (1+2u^+\partial^2_{YY} \hh + (u^+\partial^2_{YY} \hh)^2 +(u^+\partial^2_{XY} \hh)^2) (\di Y)^2\nonumber \\ 
  &+ (4u^+\partial^2_{XY} \hh + 2(u^+)^2(\partial^2_{XX} \hh+\partial^2_{YY} \hh)\partial^2_{XY} \hh) \di X \di Y.
\end{align*}
Note that $\mathrm{span}(\partial_u, \partial_v) \perp \mathrm{span}(\partial_X, \partial_Y)$. Moreover,  in the region $\{\mathcal{U} <0\}$, it holds that
\begin{align*}
    \frac{\partial}{\partial u} =& \frac{\partial \mathcal{U}}{\partial u} \frac{\partial}{\partial \mathcal{U}} +  \frac{\partial \mathcal{V}}{\partial u} \frac{\partial}{\partial \mathcal{V}}+ \frac{\partial x}{\partial u} \frac{\partial}{\partial x}+ \frac{\partial y}{\partial u} \frac{\partial}{\partial y} = \frac{\partial}{\partial \mathcal{U}} + a\frac{\partial}{\partial x} + b\frac{\partial}{\partial y},
\end{align*}
and, in the region $\{\mathcal{U}>0\}$, we have that 
\begin{align*}
    \frac{\partial}{\partial u} &= \frac{\partial}{\partial \mathcal{U}} + (a\partial_X \hh +b\partial_Y \hh + \frac{1}{2}((\partial_X \hh)^2+ (\partial_Y \hh)^2)) \frac{\partial}{\partial \mathcal{V}}  + (a + \partial_X \hh)\frac{\partial}{\partial x} + (b+ \mathcal{\hh}_Y)\frac{\partial}{\partial y}.
\end{align*}
Moreover, 
\begin{align*}
    \frac{\partial}{\partial v}= \frac{\partial}{\partial \mathcal{V}}.
\end{align*}
Using the identification \eqref{warp_upper_bound} and the metric tensor \eqref{cut_and_paste_xy_coords}, we can go back to the flat Minkowski space and consider the point $p$ given by 
\begin{align*}
\mathcal{U}^-(p) = 0,\; \mathcal{V}^-(p)=\mathcal{V},\; x^-(p)=x, y^-(p)=y, \text{ or, equivalently, } \\
\mathcal{U}^+(p) = 0,\; \mathcal{V}^+(p)=\mathcal{V}-\hh(x, y),\; x^+(p)=x, y^+(p)=y.
\end{align*}
Let $q \in I^+(p)$ and notice that  $\mathcal{U}(q)>0$. Define 
\begin{align*}
    \Tilde{w}^+:= (q-p^+) \in \R^4
\end{align*}
and denote 
\begin{align*}
    w^+:= \frac{1}{\mathcal{U}(q)} \Tilde{w}^+ = (w_\mathcal{U}^+, w_\mathcal{V}^+, w_x^+, w_y^+). 
\end{align*}
Then 
\begin{align*}
    (-2 \di \mathcal{U} \di \mathcal{V} + (\di x)^2 + (\di y)^2)(w^+,w^+) <0.
\end{align*}
Now define 
\begin{align}\label{def_coord_change_w-}
    w^- := (w_\mathcal{U}^-, w_\mathcal{V}^-, w_x^-, w_y^-),
\end{align}
where 
\begin{equation}\label{eq-boost}
\begin{cases}
    &w_\mathcal{U}^- = w_\mathcal{U}^+ =1,\\
    &w_\mathcal{V}^- = w_\mathcal{V}^+ - w_x^+ \partial_x \hh - w_y^+ \partial_y \hh + \frac{1}{2}((\partial_x \hh)^2 + (\partial_y \hh)^2), \\
    & w_x^- = w_x^+ -\partial_x \hh, \\
    & w_y^- = w_y^+ -\partial_y \hh.
\end{cases}
\end{equation}
In order to keep notation short, we will simply write $w^-=(1, s, a, b)$, some $a, b, s \in \R$. Note that $w^-$ being timelike and future directed in $\{\uu <0\}$ amounts to the condition 
\begin{align*}
    2s> a^2+b^2.
\end{align*}
Then, in the above coordinate transform, we get that 
\begin{align*}
    \partial_u + s \partial_v &= w^- \ \mathrm{for}\ u <0, \ \mathrm{and}\\
    \partial_u + s \partial_v &= w^+ \ \mathrm{for}\ u >0. 
\end{align*}
Finally, define the quantities $T, W$ via
\begin{align}\label{coord_change_s}
    u=T+W, \ v=sT+(a^2+b^2-s)W.
\end{align}
Then,  
\begin{align*}
    \partial_T= \partial_u + s \partial_v \perp \partial_u +(a^2+b^2-s)\partial_v= \partial_W. 
\end{align*}
As $\mathrm{span}(\partial_u, \partial_v) \perp \mathrm{span}(\partial_X, \partial_Y)$, we get that $\mathrm{span}(\partial_T, \partial_W) \perp \mathrm{span}(\partial_X, \partial_Y)$. Hence, the metric in $T, W, X, Y$-coordinates is of the form 
\begin{align}\label{boosted_wave}
    (a^2+b^2-2s)\di T^2 + g[2,3,4],
\end{align}
where $g[2,3,4]$ is a Riemannian metric in the $W, X, Y$-coordinates, given by 
\begin{align}\label{metric_tensor_in_TW}
  \di s^2 =&  (2s-a^2 - b^2) (\di W)^2  + (1+2(T+W)^+\partial^2_{XX} \hh  \nonumber \\
  &+ ((T+W)^+\partial^2_{XX} \hh)^2 +((T+W)^+\partial^2_{XY} \hh)^2) (\di X)^2 \nonumber \\
  &+ (1+2(T+W)^+\partial^2_{YY} \hh + ((T+W)^+\partial^2_{YY} \hh)^2 +((T+W)^+\partial^2_{XY} \hh)^2) (\di Y)^2 \nonumber \\
  &+ (4(T+W)^+\partial^2_{XY} \hh + 2((T+W)^+)^2(\partial^2_{XX} \hh+\partial^2_{YY} \hh)\partial^2_{XY} \hh) \di X \di Y.
\end{align}
\begin{proposition}\label{tau_for_straight_line}
    Let $p = (T, W, X, Y) \in \R^4$. Then $\tau(p, p+ (r, 0, 0, 0))= \sqrt{(2s-a^2+b^2)}r$. 
\end{proposition}
\begin{proof}
    First, we prove that $\tau(p, p+ (r, 0, 0, 0))\geq \sqrt{(2s-a^2-b^2)r}$. Consider the timelike curve $\gamma:[0,1] \to \R^4, t \mapsto p+ t(r,0,0,0)$.
    Then
    \begin{align*}
        L_g(\gamma) = \int_0^1\sqrt{(2s-a^2-b^2)}r\, \di t.
    \end{align*}
    Next, we prove that $\tau(p, p+ (r, 0, 0, 0))\leq \sqrt{(2s-a^2-b^2)}r$. Let $\gamma:[0,1] \to \R^4$ be any causal curve from $p$ to $p+(r,0,0,0)$. Write $\gamma = (\gamma^1, \gamma^2, \gamma^3, \gamma^4)$. Then as $\gamma$ is causal and future directed, we get that $\dot{\gamma}^1_t \geq 0$ for almost all $t \in [0,1]$, hence 
    \begin{align*}
         \int_0^1 |\dot{\gamma}^1_t| \di t=\int_0^1 \dot{\gamma}^1_t \di t = r.
    \end{align*}
    Moreover, \eqref{boosted_wave} implies that for any timelike vector $v\in \R^4$, it holds $|g(v,v)| \leq (2s-a^2-b^2)v_1^2$. Then, 
    \begin{align*}
        \int_0^1 \sqrt{-g(\dot{\gamma}_t,\dot{\gamma}_t)} \di t\leq \int_0^1 \sqrt{(2s-a^2-b^2)} \sqrt{(\dot{\gamma}^1_t)^2} \di t =  \sqrt{(2s-a^2-b^2)}r.
    \end{align*}
    This proves the proposition.
\end{proof}

\begin{lemma}\label{r_0_for_upper_bounds}
    There exists $r_0 >0$ with the following property. Let $p \in M$ satisfy that 
    $$u(p) = T(p)+W(p) \leq r_0,$$
    and let  $q \in B^{euc}_{r^2}(p)$, for some $r<r_0$.
    Then $p\ll q+(r,0, 0, 0)$ and $q \ll p+(r,0, 0, 0)$.
\end{lemma}
\begin{proof}
Set 
\begin{align*}
    r_0 := \frac{\min(1, 2s-a^2-b^2)}{5000(1+\norm{D^2\hh}_{L^\infty(B_1((X(p), Y(p))))})^2}.
\end{align*}
Then, on the region  $B^{euc}_{2r}(p)+[0, 5r](1,0,0,0)$, it holds that
    \begin{align*}
        \Big|g[2,3,4] - \mathrm{diag}(2s-a^2-b^2, 1, 1)\Big| \leq \frac{1}{1000}.
    \end{align*}
    Write $p':=p+5r\partial_T, q':=q+5r\partial_T$. Note that 
    \begin{align}\label{straight_lines_stay_in_causal_region}
        p+ [0,1](q'-p), q+ [0,1](p'-q) \subset B_{2r}(p)+[0, 5r]\partial_T.
    \end{align}
    Define the curve $\gamma:[0,1]\to M, t \mapsto p+t(q'-p)$. Then $\gamma$ has a constant derivative and we can use \eqref{straight_lines_stay_in_causal_region}
    to estimate  
    \begin{align*}
        g(\gamma_t)(q'-p, q'-p) &= 25r^2g(\gamma_t)(\partial_T, \partial_T) +10g(\gamma_t)(r\partial_T, q-p) + g(\gamma_t)(q-p, q-p) \\
        &\leq (a^2+b^2-2s)r^2 + \Big(1+ \frac{1}{1000}\Big)(2r^3+r^4)<0.
    \end{align*}
    It follows that $\gamma$ is a timelike and future directed curve from $p$ to $q'$, yielding that $q' \gg p$. The proof  that $p' \gg q$ is analogous. 
\end{proof}

\begin{lemma}\label{cyclical monotonicity_in_coords}
Let $r_0>0$ be given by Lemma \ref{r_0_for_upper_bounds}, and let $x \in M$ be such that $\mathcal{U}(x)\leq r_0$. Pick 
$$x_1, \ldots, x_n \in B^{euc}_{r^2}(x),$$ for some $n\in \N$ and $r\in (0, r_0)$. For $i=1, \ldots, n$, define 
$$x_i^-:= x_i - (r, 0, 0,0) \quad \text{and}\quad  x_i^+:= x_i + (r, 0, 0, 0),$$
and denote  $x_{n+1}:= x_1, x^{\pm}_{n+1}:= x^{\pm}_1$.
Then, for all $\mathfrak{p} \in (0,1]$,  it holds that
    \begin{align}
        \sum_{i=1}^n \tau(x_i^-, x_{i+1}^+)^\mathfrak{p} \leq \sum_{i=1}^n \tau(x_i^-, x_{i}^+)^\mathfrak{p} = (2r\sqrt{2s-a^2-b^2})^\mathfrak{p}n.
    \end{align}
\end{lemma}
\begin{proof}
    For each $i=1, \ldots, n$, Proposition \ref{tau_for_straight_line} yields 
    \begin{align}
        \tau(x_i^-, x_{i}^+) = 2r\sqrt{2s-a^2-b^2} = \sqrt{2s-a^2-b^2} (T(x_i^+)-T(x_i^-)).
    \end{align}
    Let $\gamma$ be a causal and future directed curve from $x_i^-$ to $x_{i+1}^+$. Write $\gamma = (\gamma^1, \gamma^2, \gamma^3, \gamma^4)$. Since $\gamma$ is causal and future directed, we get that $\dot{\gamma}^1_t\geq 0$ for almost every $t\in [0,1]$. Hence
    \begin{align*}
        \int_0^1 |\dot{\gamma}^1_t| \di t = \int_0^1 \dot{\gamma}^1_t \di t = T(x_{i+1}^+)-T(x_i^-).
    \end{align*}
    It follows that
    \begin{align*}
        L_g(\gamma) = \int_0^1 \sqrt{-g(\dot{\gamma}_t, \dot{\gamma}_t)} \di t \leq \int_0^1 \sqrt{2s-a^2-b^2} |\dot{\gamma}^1_t| \di t = \sqrt{2s-a^2-b^2} (T(x_{i+1}^+)-T(x_i^-)).  
    \end{align*}
    Therefore, 
    \begin{align}
        \sum_{i=1}^n \tau(x_i^-, x_{i+1}^+)^\mathfrak{p} \leq \sqrt{2s-a^2-b^2}^\mathfrak{p}\sum_{i=1}^n(T(x_{i+1}^+)-T(x_i^-))^\mathfrak{p} \leq \sqrt{2s-a^2-b^2}^\mathfrak{p}(2r)^\mathfrak{p}n,
    \end{align}
    where in the last step, we used that $z \mapsto z^\mathfrak{p}$ is concave to apply Jensen's inequality.
\end{proof}
Denoting $w^+ = (1,\sigma,\alpha,\beta)$, it is convenient to also consider the coordinate transformation
\begin{equation}\label{coord_chage_aphabeta}
\begin{cases}
    &\mathcal{U} = u,  \\
    &\mathcal{V} = v - \Theta(-u) \hh + \alpha X + \beta Y + \min(u, 0)(-\alpha \partial_X \hh - \beta \partial_Y \hh + \frac{1}{2}((\partial_X \hh)^2+ (\partial_Y \hh)^2)),  \\
    & x = X + \alpha u -\min(u, 0) \partial_X \hh,  \\
    & y = Y + \beta u -\min(u, 0) \partial_Y \hh.
\end{cases}
\end{equation}
In such a coordinate system, the Minkowski metric writes as 
\begin{align}
    -2\di \mathcal{U} \di \mathcal{V} + \di x^2 + \di y^2 =& (\alpha^2+\beta^2)\di u^2 -2 \di u \di v + ((1-\min(0, u)\partial^2_{XX} \hh)^2 + (\min(0, u)\partial^2_{XY} \hh)^2)\di X^2 \nonumber \\
    &+ ((1-\min(0, u)\partial^2_{YY} \hh)^2 + (\min(0, u)\partial^2_{XY} \hh)^2)\di Y^2 \nonumber \\
    &-2 \min(0, u)\partial^2_{XY} \hh(2-\min(0, u)\partial^2_{XX} \hh-\min(0, u)\partial^2_{YY} \hh) \di X \di Y.
\end{align}
Then, applying the transformation
\begin{equation}\label{coord_change_sigma}
\begin{cases}
    u&=T+W, \\
    v&=\sigma T+(\alpha^2+\beta^2-\sigma)W,
\end{cases}
\end{equation}
yields a coordinate system in which $\partial_T$ equals $w^+$ in $\{u >0\}$ and $\partial_T= w^-$ in $\{u <0\}$. Then statements of Proposition \ref{tau_for_straight_line}, Lemma \ref{r_0_for_upper_bounds} and Lemma \ref{cyclical monotonicity_in_coords} hold analogously for these coordinates. We now have all the ingredients to establish the synthetic upper bounds on timelike Ricci curvature for IPP-waves. 
\smallskip

\begin{theorem}[Synthetic upper bounds on timelike Ricci curvature for IPP-waves]\label{upper_bounds_waves}
Let $\hh \in C^{3,1}_{\mathrm{loc}}(\mathbb{C}^2;\mathbb{R})$ and assume that $D^2 \hh$ is bounded.
Let $g=g_{\hh}$ be the IPP-wave metric as in \eqref{nonsmooth_metric} with profile function $\hh$ and zero cosmological constant, i.e., $\Lambda=0$.  Assume there exists a constant $c>0$ such that
\begin{equation}\label{eq:AssURB}
\partial^2_{XX} \hh+ \partial^2_{YY} \hh \geq c\left( (\partial^2_{XY} \hh)^2+\frac{1}{2}((\partial^2_{XX} \hh)^2+(\partial^2_{YY} \hh)^2)\right).
\end{equation}
Then the impulsive gravitational wave $(\R^4, g_{\hh})$ has synthetic timelike Ricci curvature bounded  above by $0$. 
\end{theorem}
\begin{proof}
    Fix a point $p\in \R^4$. First assume that $|\mathcal{U}(p)| \leq r_0$, where $r_0$ is given by Lemma \ref{r_0_for_upper_bounds}. 
    Choose $r \in (0, r_0/2)$ and a Borel probability measure $\mu_0$ that is absolutely continuous with respect to the Lebesgue measure and such that $\supp\, \mu_0 \subset B^{euc}_{r^4}(p)$. 
\smallskip

\textbf{Case $2r \leq |\mathcal{U}(p)|$}. In this case,  $B_{2r}^{euc}(p)$ is contained in either $\{\mathcal{U}<0\}$ or $\{\mathcal{U}>0\}$. Then, by the properties of the Minkowski space, we get that the volume measure is given by the Lebesgue measure and for any $y \in B_r(p)$, there exists a geodesic $(\mu_t)_{t \in [-1, 1]}$ such that $\supp\, \mu_1 \subset B_{r^2}(y)$ and 
    \begin{align}
        \mathrm{Ent}(\mu_0| \mathcal{L}^4) = \mathrm{Ent}(\mu_1| \mathcal{L}^4) = \mathrm{Ent}(\mu_{-1}| \mathcal{L}^4).
    \end{align}

\textbf{Case $|\mathcal{U}(p)|<2r$}.
Let $y \in I^+(p)\cap B_r(p)$. 
    If $\mathcal{U}(p) \geq 0$, set $\Tilde{w}^+:= y-p \in \R^4$ and compute $w^-$ as in \eqref{def_coord_change_w-}. If $\mathcal{U}(p), \mathcal{U}(y) < 0$, denote $w^-= \frac{y-p}{\mathcal{U}(y)-\mathcal{U}(p)} \in \R^4$. If $\mathcal{U}(p) <0$ and $\mathcal{U}(y)\geq 0$ denote by $p_0$ the point in $\{\mathcal{U}=0\}^+$ such that the geodesic from $p$ to the $y$ passes through $p_0$. 
    Then denote $\Tilde{w}^+:= y-p_0 \in \R^4$ and compute $w^-$ as in \eqref{def_coord_change_w-}. Now let $s, a, b \in \R$ such that  $w^-=(1, s, a, b)$ and apply the coordinate change \eqref{coord_chage_ab} and \eqref{coord_change_s}.
    
    In the coordinates $(T, W, X, Y)$, we define the probability measures $\mu_1, \mu_{-1}$ via 
    \begin{align*}
        \di\mu_{-1}(q) = \di \mu_0(q+(r,0,0,0)), \ \di\mu_{1}(q) = \di \mu_0(q-(r,0,0,0)).
    \end{align*}
    Lemma \ref{cyclical monotonicity_in_coords} shows that $$(B_r(p)-r\partial_T)\times (B_r(p)+r\partial_T) \supset \supp\, \mu_{-1} \times \supp\, \mu_{1}$$ is $\tau^\mathfrak{p}$-cyclically monotone for any $\mathfrak{p} \in (0,1)$.

    Denote by $\rho$ the density of $\mu_{0}$ with respect to the rescaled Lebesgue measure $(2s-a^2-b^2)^2\mathcal{L}^4$. 
    From the expressions for the metric tensor in \eqref{boosted_wave} and \eqref{metric_tensor_in_TW}, it follows that its determinant is given by
    \begin{align}\label{taylor_det_g}
        -\det[g]=& (2s-a^2-b^2)^2 \Big(1+2(T+W)^+(\partial^2_{XX} \hh+\partial^2_{YY} \hh) \nonumber \\
        &+ (4\partial^2_{XX} \hh\partial^2_{YY} \hh-2(\partial^2_{XY} \hh)^2+(\partial^2_{XX} \hh)^2+(\partial^2_{YY} \hh)^2)((T+W)^+)^2 \nonumber \\
        &+ ((T+W)^+)^3F((T+ W)^+, D^2 \hh)\Big),
    \end{align}
    where $F$ is a polynomial. 
    Then the densities $\rho_{-1}, \rho_0, \rho_1$ of $\mu_{-1}, \mu_0, \mu_1$ are given by
    \begin{align*}
        \rho_i(T, W, X, Y) = (2s-a^2-b^2)^2\frac{\rho(T-ir, W, X, Y)}{\sqrt{-\det[g](T, W, X, Y)}}, 
    \end{align*}
    for $i \in \{-1, 0, 1\}$.
    It follows that 
    \begin{align*}
        \mathrm{Ent}(\mu_i|\vol_g) &= \int_{\R^4} (2s-a^2-b^2)^2\rho \log ((2s-a^2-b^2)^2\rho) \di \mathcal{L}^4 \nonumber \\
        &- (2s-a^2-b^2)^2\int_{\R^4} \rho \log(\sqrt{-\det[g](T+ir, W, X, Y)})\di \mathcal{L}^4(T, W, X, Y).
     \end{align*}
 Hence,  
 \begin{align}\label{entropy_comparison_upper_bound}
      &\mathrm{Ent}(\mu_{-1}|\vol_g) -2  \mathrm{Ent}(\mu_0|\vol_g) +  \mathrm{Ent}(\mu_1|\vol_g)=\nonumber \\
      & \qquad -(2s-a^2-b^2)\int_{\R^4} \rho \log(\sqrt{-\det[g](T-r, W, X, Y)}) -2 \rho \log(\sqrt{-\det[g](T, W, X, Y)}) \nonumber \\
      &\qquad   + \rho\log(\sqrt{\det[g](T+ir, W, X, Y)}) \di \mathcal{L}^4(T, W, X, Y).
 \end{align}

Using \eqref{taylor_det_g}, we get that 
\begin{align*}
    \log \sqrt{-\det g}(T, W, X, Y) =& \log\big((2s-a^2-b^2)\big)+ (T+W)^+(\partial^2_{XX} \hh+\partial^2_{YY} \hh) \nonumber \\
    & \quad  + \frac{1}{2}(4\partial^2_{XX} \hh\partial^2_{YY} \hh-2(\partial^2_{XY} \hh)^2+(\partial^2_{XX} \hh)^2+(\partial^2_{YY} \hh^2))((T+W)^+)^2  \nonumber\\
    &\quad -\Big((T+W)^+(\partial^2_{XX} \hh+\partial^2_{YY} \hh)  \Big)^2+ O(((T+W)^+)^3).
\end{align*}
Then 
\begin{align*}
    &\frac{\di^2}{\di r^2}\log \sqrt{-\det g} (T+r, W, X, Y)\Big|_{r=0} = ((\partial^2_{XX} \hh+\partial^2_{YY} \hh))\delta_0(T+W) \nonumber \\
    &\qquad \quad + (-(\partial^2_{XY} \hh)^2-\frac{1}{2}((\partial^2_{XX} \hh)^2(\partial^2_{YY} \hh)^2))\Theta(T+W) + O((T+W)).
\end{align*} 
Since \( r \leq 1 \), it follows that for any \( q \in B^{\mathrm{euc}}_{r^4}(p) \supset \supp \mu_0 \), one has \( |u(q)| \le 3r \). Consequently, for such \( q \), we obtain
\begin{align}
    &\Bigg(\frac{\log(\sqrt{-\det[g](T-r, W, X, Y)})-2\log(\sqrt{-\det[g](T, W, X, Y)})+\log(\sqrt{-\det[g](T+r, W, X, Y)})}{r^2}\Bigg) \nonumber \\
    &\qquad = \frac{\Delta\hh((u(q)-r)^+-2u(q)^++(r+u(q))^+)}{r^2} \label{singular_concavity_part}z \\
   &\qquad  \quad -((\partial^2_{XY} \hh)^2+\frac{1}{2}((\partial^2_{XX} \hh)^2)+(\partial^2_{YY} \hh)^2))\Theta(u(q)) \label{abs_cts_concavity_part} + O(r). 
\end{align}
Now as $3r \geq |u(q)|$, we get that  
\begin{align*}
    \frac{\Delta\hh((u(q)-r)^+-2u(q)^++(r+u(q))^+)}{r^2} = \frac{\Delta\hh(-2u(q)^++(r+u(q))^+)}{r^2} = \frac{\Delta\hh(r-|u(q)|)}{r^2}.
\end{align*}

The assumption \eqref{eq:AssURB} gives that, for $u(p) \leq 0$: 
\begin{align*}
    &\frac{\log(\det[g](T-r, W, X, Y))-2\log(\det[g](T, W, X, Y))+\log(\det[g](T+r, W, X, Y))}{r^2}\nonumber \\
    &= \frac{\Delta\hh(r-|u(p)|) - B(X, Y)(r-|u(p)|)^2}{r^2}.
\end{align*}
Then, 
\begin{align}\label{positive_laplace_estimate_error}
     &\frac{\log(\sqrt{-\det[g](T-r, W, X, Y)})-2\log(\sqrt{-\det[g](T, W, X, Y)})+\log(\sqrt{-\det[g](T+r, W, X, Y)})}{|2s-a^2-b^2|^2 r^2} \nonumber\\
     &\leq \frac{1}{r}(|B|_{L^\infty}r-c)^+,
\end{align}
Substituting this into \eqref{entropy_comparison_upper_bound}, we conclude that in this case the timelike synthetic Ricci curvature is bounded above by \( 0 \).
If $u(p)\geq 0$, we perform the change of coordinates as in \eqref{coord_change_sigma}. In this case the determinant can be expanded as 
\begin{align}\label{taylor_det_g_new}
        \frac{\det[g]}{(2\sigma-\alpha^2-\beta^2)^2}=&  \Big(1-2\min(0,(T+W))(\partial^2_{XX} \hh+\partial^2_{YY} \hh) + (4\partial^2_{XX} \hh\partial^2_{YY} \hh-2(\partial^2_{XY} \hh)^2+(\partial^2_{XX} \hh)^2 \nonumber \\
        &+(\partial^2_{YY} \hh^2))((T+W)^-)^2 + ((T+W)^-)^3F(T, W, D^2 \hh)\Big),
    \end{align}
yielding an analogous estimate as in \eqref{positive_laplace_estimate_error}.

If \( |\mathcal{U}(p)| \ge r_0 \), then either \( B^{\mathrm{euc}}_{r_0}(p) \subset \{\mathcal{U} > 0\} \) or \( B^{\mathrm{euc}}_{r_0}(p) \subset \{\mathcal{U} < 0\} \). 
In this case, we may work entirely within the cut-and-paste coordinate system, and the flatness of the Minkowski metric implies that the timelike synthetic Ricci curvature is bounded above by \( 0 \); see \cite[Thm.~4.7]{mondino2022optimal}.
\end{proof}

\begin{remark}
The proof (see in particular  \eqref{entropy_comparison_upper_bound} and \eqref{abs_cts_concavity_part}) furthermore shows that, if 
\[
\partial^2_{XX} \hh + \partial^2_{YY} \hh \ge 0,
\]
then the impulsive gravitational wave \( (\mathbb{R}^4, g_{\hh}) \)  as in Theorem \ref{upper_bounds_waves} has timelike Ricci curvature bounded above --- in the synthetic sense of \cite{mondino2022optimal} --- by the function
\begin{align*}
(\partial^2_{XY} \hh)^2 + \frac{1}{2}\bigl((\partial^2_{XX} \hh)^2 + (\partial^2_{YY} \hh)^2\bigr).
\end{align*}
\end{remark}

\section{Application to weak solutions of the Einstein equations and nonlinearly interacting impulsive gravitational waves}\label{Sec:4}
In \cite{LR:12}, Luk and Rodnianski study the characteristic initial value problem for impulsive gravitational waves and prove local existence results for the vacuum Einstein equations within certain classes of such waves, where the impulse propagates along compact two-surfaces --- in contrast to the classical impulsive waves considered here. In the subsequent work \cite{luk2013nonlinear}, the same authors analyze the nonlinear interaction of impulsive gravitational waves, again in the framework of a characteristic initial value problem with initial data prescribed on two null hypersurfaces. They establish the existence of weak solutions to the Einstein equations for a broad class of initial data posed on the null hypersurfaces \( H_0 \) and \( \underline{H}_0 \). Moreover, they show that smooth approximations \( \gamma_n \) of the initial data \( \gamma \) give rise to smooth solutions \( g_n \) of the Einstein equations that converge locally uniformly to the corresponding weak solution \( g \).
\\We apply our stability result for the \( \tcd \)-condition to the non-smooth spacetimes constructed in \cite{LR:12, luk2013nonlinear}, thereby showing that a large class of non-smooth Lorentzian metrics --- including models describing interacting impulsive gravitational waves --- satisfy the \( \wtcd(0,4) \)-condition.
\\Finally, we emphasize that the characteristic initial value problem is solved in the future of the impulse (or of the interacting impulses). By contrast, our approach in the previous sections captures the entire impulsive region and establishes the curvature bound \( \Ric \ge 0 \) across it -- in a synthetic sense, see Theorem \ref{thm:tcd_for_waves}.

\subsection{Setup}
We restrict attention to the setting of \cite{luk2013nonlinear}, as it provides a more general framework than \cite{LR:12}. 
In this work, the characteristic initial value problem is formulated with initial data prescribed on two null hypersurfaces \( H_0 \) and \( \underline{H}_0 \), intersecting along the sphere \( S_{0,0} \). 
The resulting spacetime is constructed as a solution to the Einstein equations in a neighbourhood of \( H_0 \) and \( \underline{H}_0 \) containing \( S_{0,0} \).

For a spacetime in a neighbourhood of $S_{0,0}$, \cite{luk2013nonlinear} defines a double null foliation as follows: consider  solutions $u, \underline{u}$ to the eikonal equation
\begin{align*}
    (g^{-1})^{\mu, \nu} \partial_\mu u \partial_\nu u =0, \ (g^{-1})^{\mu, \nu} \partial_\mu \underline{u} \partial_\nu \underline{u} =0,
\end{align*}
such that the initial conditions $u=0$ on $H_0$ and $\underline{u}=0$ on $\underline{H}_0$ are satisfied. Let 
\begin{align*}
    {L'}^\mu = -2(g^{-1})^{\mu, \nu} \partial_\nu u, \  \underline{L'}^\mu = -2(g^{-1})^{\mu, \nu} \partial_\nu \underline{u}.
\end{align*}
These are null  geodesic vector fields. Let
\begin{align*}
    2\Omega^{-2} = -g(L, L').
\end{align*}
Define
\begin{align*}
    e_3 = \Omega \underline{L'}, \ e_4 = \Omega L',
\end{align*}
and note that $g(e_3, e_4)=-2$. Furthermore, define
\begin{align*}
    \underline{L} = \Omega^2 \underline{L'}, \ L= \Omega^2 L'.
\end{align*}
Consider spacetime solutions to the vacuum Einstein equations in the gauge such that $\Omega=1$ on $H_0$ and $\underline{H}_0$. We denote the level sets of $u$ as $H_u$ and the level sets of $\underline{u}$ as $\underline{H}_{\underline{u}}$. 
The sets defined by the intersections of the hypersurfaces $H_u$ and $\underline{H}_{\underline{u}}$ are topologically $2$-spheres, which will be denoted by $S_{u, \underline{u}}$.

We now introduce a coordinate system \( (u,\underline{u},\theta^1,\theta^2) \) as follows. 
On the sphere \( S_{0,0} \), choose local coordinates \( (\theta^1,\theta^2) \) such that, on each coordinate patch, the metric \( \gamma \) is smooth, bounded, and positive definite. 
We then extend these coordinates to the initial hypersurfaces \( H_0 \) and \( \underline{H}_0 \) by requiring that \( \theta^A \) remain constant along the integral curves of \( L \) and \( \underline{L} \), respectively. 

More generally, the functions \( \theta^A \) are defined by the transport equations
\begin{align*}
\mathcal{L}^{\circ}_L \theta^A = 0,
\end{align*}
where \( \mathcal{L}^{\circ} \) denotes the restriction of the Lie derivative to the tangent spaces \( T S_{u,\underline{u}} \).
Relative to this coordinate system, the null pair $e_3$ and $e_4$ can be expressed as
\begin{align*}
    e_3 = \Omega^{-1}\left(\frac{\partial}{\partial u} + b^A \frac{\partial}{\partial \theta^A}\right), \ e_4 = \frac{\partial}{\partial \underline{u}},
\end{align*}
for some $b^A$ such that $b^A=0$ on $\underline{H}_0$. The metric $g$ takes the form 
\begin{align}\label{luk_rod_form_of_metric}
    g= -\Omega^2(du \otimes d\underline{u} + d\underline{u} \otimes du) + \gamma_{AB}(d\theta^A -b^A du) \otimes (d\theta^B -b^B du).
\end{align}
In \cite{luk2013nonlinear},  the Einstein equations are recasted as a system for Ricci coefficients and curvature components associated to the null frame $e_3, e_4$ defined above and an orthonormal frame $e_1, e_2$ tangent to the $2$-spheres $S_{u, \underline{u}}$. Using $A,B$ to denote $1,2$, we recall the definition of the Ricci coefficients relative to the null frame
\begin{align*}
    & \chi_{AB} = g(D_A e_4, e_B), \quad   \underline{\chi}_{AB} = g(D_A e_3, e_B),\\
    & \eta_A = -\frac{1}{2}g(D_3e_A, e_4), \quad  \underline{\eta}_A = -\frac{1}{2}g(D_4e_A, e_3), \\
    & \omega = -\frac{1}{4}g(D_4 e_3, e_4), \quad  \underline{\omega} = -\frac{1}{4}g(D_4 e_3, e_4), \\
    & \zeta_A = \frac{1}{2}g(D_A e_4, e_3) ,
\end{align*}
where $D_A = D_{e(A)}$. We denote by $\hat{\chi}, \underline{\hat{\chi}}$ the traceless parts of $\chi, \underline{\chi}$. Moreover, we adopt the notation 
\begin{align*}
    \psi \in \{\mathrm{tr} \chi, \mathrm{tr} \underline{\chi}, \eta, \underline{\eta}\}, \psi_H \in \{\hat{\chi}, \omega\}, \psi_{\underline{H}} \in \{\hat{\chi}, \underline{\omega}\}.
\end{align*}
The null curvature components are denoted by
\begin{align*}
    &\alpha_{AB} = R(e_A, e_4, e_B, e_4),\ \underline{\alpha}_{AB} = R(e_A, e_3, e_B, e_3),\\
    & \beta_A = \frac{1}{2}R(e_A, e_4, e_3, e_4), \ \underline{\beta}_A = \frac{1}{2}R(e_A, e_3, e_3, e_4), \\
    & \rho = \frac{1}{4}R(e_4, e_3, e_4, e_3), \ \sigma = \frac{1}{4}{}^*R(e_4, e_3, e_4, e_3),
\end{align*}
where ${}^*R$ denotes the Hodge dual of $R$. Finally define 
\begin{align*}
    \check{\rho} = \rho - \frac{1}{2}\hat{\chi} \cdot \underline{\hat{\chi}}, \ \check{\sigma} = \sigma + \frac{1}{2} \underline{\hat{\chi}} \wedge \hat{\chi}. 
\end{align*}
We refer to \cite{luk2013nonlinear}, identities (22)--(26),  for the expression of the vacuum Einstein equations in terms of the above quantities.
\subsection{The approximation result and consequences on the TCD condition for interacting waves}
Using the above notation, let us recall the following result from \cite{luk2013nonlinear}:
\begin{theorem}[\cite{luk2013nonlinear} Theorem 3]\label{luk_rod_main_thm}
    Suppose the initial data set for the characteristic initial value problem is given on $H_0$ for $0 \leq \underline{u} \leq \underline{u}_*$ and on $\underline{H}_0$ for $0 \leq u \leq  u_* \leq I$ such that 
    \begin{align*}
    &    c \leq |\det \gamma|_{S_{u,0}}|, |\det \gamma|_{S_{0,\underline{u}}}| \leq C, \\
   &     \sum_{i \leq 3} \Bigg( |(\frac{\partial}{\partial \theta})^{i} \gamma|_{S_{u,0}}|, |(\frac{\partial}{\partial \theta})^{i} \gamma|_{S_{S_{0,\underline{u}}}}| \Bigg) \leq C, \\
   &     \mathcal{O}_0 := \sum_{i \leq 3} \big( \norm{\nabla^{i} \psi}_{L^\infty_u L^2(S_{u, 0})} + \norm{\nabla^{i} \psi}_{L^\infty_{\underline{u}} L^2(S_{0, \underline{u}})} + \norm{\nabla^{i} \psi_H}_{L^2(H_0)} +\norm{\nabla^{i} \psi_{\underline{H}}}_{L^2(\underline{H}_0)} \big) \leq C, \\
  &      \mathcal{R}_0 := \sum_{i \leq 2} \Bigg( \norm{\nabla^{i} \beta}_{L^2(H_0)} +\norm{\nabla^{i} \underline{\beta}}_{L^2(\underline{H}_0)} + \sum_{\Psi \in \{ \check{\rho}, \check{\sigma}\}}(\norm{\nabla^{i} \Psi}_{L^\infty_u L^2(S_{u, 0})} + \norm{\nabla^{i} \Psi}_{L^\infty_{\underline{u}} L^2(S_{0, \underline{u}})}) \Bigg) \leq C.
    \end{align*}
    Then there exists $\e >0$ sufficiently small depending only on $C, c$ and $I$ such that if $\underline{u}_* \leq \e$, there exists a spacetime $(\mathcal{M}, g)$ that solves the characteristic initial value problem to the vacuum Einstein equations in the region $0 \leq u \leq u_*$, $0\leq \underline{u}, \leq \underline{u}_*$.
    Geometrically, this is the region to the future of the initial hypersurfaces $H_0$, $\underline{H}_0$, which is bounded in the future by the incoming null hypersurface from $S_{u_*, 0}$. Associated to the spacetime $(\mathcal{M}, g)$, there exists a system of null coordinates $(u, \underline{u}, \theta^1, \theta^2)$ in which the metric is continuous and takes the form \eqref{luk_rod_form_of_metric}.
    In addition, given a sequence of smooth initial data sets such that the metrics $\gamma_n$ approach $\gamma$ in $L^{\infty}_u W^{3, \infty}(S_{u,0}) \cap L^\infty_{\underline{u}}W^{3,\infty}(S_{0, \underline{u}})$, the Ricci coefficients $(\psi, \psi_H, \psi_{\underline{H}})_n$ approach $(\psi, \psi_H, \psi_{\underline{H}})$ in the norm given by $\mathcal{O}_0$ (where norms and connection symbols on the spheres $S_{u,0}, S_{0,\underline{u}}$ are defined with respect to $\gamma$) and the renormalised curvature components $(\beta, \check{\rho}, \check{\sigma}, \underline{\beta})_n$ approach $(\beta, \check{\rho}, \check{\sigma}, \underline{\beta})$ in the norm $\mathcal{R}_0$, the sequence of initial data give rise to a sequence of smooth spacetimes which approximate $(\mathcal{M}, g)$ in $C^0$. $(\mathcal{M}, g)$ is also the unique spacetime solving the characteristic initial value problem among all $C^0$ limits of smooth solutions.
\end{theorem}
\medskip

\begin{proposition}\label{prop:InteractiveWavesWTCD}
Let \( \gamma \) be an initial data set satisfying the assumptions of Theorem~\ref{luk_rod_main_thm}, and let \( (\gamma_n) \) be a sequence of smooth initial data sets approximating \( \gamma \) as in that theorem. Let \( \varepsilon > 0 \) be as in Theorem~\ref{luk_rod_main_thm}. 
\\Assume moreover that the resulting spacetime \( (\mathcal{M}, g) \) is geodesic and that every maximizing causal curve has a well-defined causal character. Then, for any \( p \in (0,1) \), the spacetime \( (\mathcal{M}, g) \), restricted to the region
\[
0 \le u \le u_*, \qquad 0 \le \underline{u} \le \underline{u}_*, \qquad \underline{u}_* \le \varepsilon,
\]
which is a weak solution of the vacuum Einstein equations in the PDE sense, also satisfies the \( \wtcd^{e}_{p}(0,4) \)-condition.
\end{proposition}
\begin{proof}
    From  Theorem \ref{luk_rod_main_thm}, we get that there exists an $\e >0$ and a sequence of smooth metrics $(\mathcal{M}, g_n)$ defined on $0 \leq u \leq u_*$, $0 \leq \underline{u} \leq \underline{u}_*$, $\underline{u}_*\leq \e$ such that they satisfy the vacuum Einstein equations and such that $g_n \to g$ locally in $C^0$. 
    Moreover,  Theorem \ref{luk_rod_main_thm} gives us that the domain of existence equals the future of the same hypersurfaces, hence Proposition \ref{common_cauchy_surface_uniform_global_hyperbolicity} gives us uniform global hyperbolicity of the metrics $(\mathcal{M}, g_n)$ and $(\mathcal{M}, g)$. Now all conditions to apply Theorem \ref{stability_tcd} are met, which proves the proposition. 
\end{proof}
\newpage
\appendix
\section{Locally uniform convergence does not imply uniform global hyperbolicity}\label{subsec:counterexample_uniform_hyperbolicity}

Recall that uniform global hyperbolicity (see Definition \ref{uniform_global_hyperbolicity_def})  was a key assumption to obtain the stability of lower bounds on the timelike Ricci curvature under $C^0_{loc}$-convergence of Lorentzian metrics (see Theorem \ref{thm:VECC0} and Theorem \ref{stability_tcd}). 

A natural question is whether such a notion is redundant.  In order to show it is not, in this appendix we  construct metrics $g, (g_j)_{j \in \N}$ such that:
\begin{enumerate}
\item  $g$ and each   $g_j$  is Lipschitz  and globally hyperbolic;
\item $g_j$ converges to $g$ locally uniformly;
\item the family  $(g_j)_{j \in \N}$ is not  uniformly globally hyperbolic.
\end{enumerate}
It will be clear from the construction that one can regularise $g_j$ and $g$ to smooth Lorentzian metrics still satisfying the properties 1,2,3 above.

We will work on $(-\frac{1}{3}, \frac{1}{3}) \times \R^3$ and denote the coordinates by $(t, x, y, z)$.
Let $\phi: \R \to [0,1]$ be a smooth and increasing function such that 
$$\phi(x) =0 \quad  \text{for $x \leq 0$}\quad  \text{and}\quad  \phi(x)=1\quad \text{ for $x \geq 1$}.$$  

Consider the Lorentzian metric $g$ defined by 
\begin{align*}
    g = \begin{pmatrix}
    -1 & 0 & 0 & 0 \\
    0 & \frac{1}{2x^4} & 0 & 0 \\
    0 & 0 & \frac{1}{2x^4} & 0 \\
    0 & 0 &  0  & \frac{1}{2x^4}
    \end{pmatrix} + \Big(1-\frac{1}{2x^4}\Big)\phi\left(4x^2\Big({\Big(t-\frac{1}{x}\Big)^2 + y^2 + z^2\Big)}\right) \begin{pmatrix}
    0 & 0 & 0 & 0 \\
    0 & 1 & 0 & 0 \\
    0 & 0 & 1 & 0 \\
    0 & 0 &  0  & 1 
    \end{pmatrix}
\end{align*}
 if $|x| \geq 2^{-1/4}$ and by 
\begin{align*}
   g = \begin{pmatrix}
     -1 & 0 & 0 & 0 \\
    0 & 1 & 0 & 0 \\
     0 & 0 & 1 & 0 \\
     0 & 0 &  0  & 1
     \end{pmatrix},
 \end{align*}    
 if $|x| \leq 2^{-1/4}$, see Figure \ref{fig}.

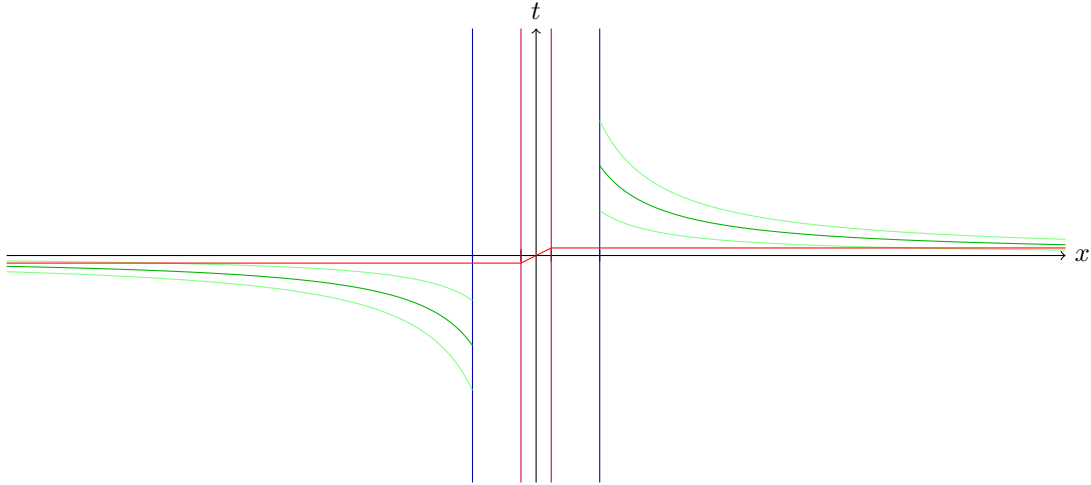
\begin{figure}[h!]
\caption{The different regions of definition of $g$ in the $t$--$x$-plane\label{fig}}
 \begin{tikzpicture}[scale=1]

\draw[->, thin] (-7,0) -- (7,0) node[right] {$x$};
\draw[->, thin] (0,-3) -- (0,3) node[above] {$t$};

\def\a{2^(-1/4)}
\def\b{1/5}

\def\tick{0.08}

\draw[thin] ({\a},-\tick) -- ({\a},\tick);
\draw[thin] ({-\a},-\tick) -- ({-\a},\tick);

\draw[thin] ({\b},-\tick) -- ({\b},\tick);
\draw[thin] ({-\b},-\tick) -- ({-\b},\tick);

\draw[blue!70!black, thin] ({\a}, -3) -- ({\a}, 3);
\draw[blue!70!black, thin] ({-\a}, -3) -- ({-\a}, 3);

\draw[purple, thin] ({\b}, -3) -- ({\b}, 3);
\draw[purple, thin] ({-\b}, -3) -- ({-\b}, 3);

\draw[green!70!black, thin, domain=\a:7, samples=200] plot (\x, {1/\x});
\draw[green!70!black, thin, domain=-7:-\a, samples=200] plot (\x, {1/\x});

\draw[green!50, thin, domain=\a:7, samples=200] plot (\x, {1/(2*\x)});
\draw[green!50, thin, domain=-7:-\a, samples=200] plot (\x, {1/(2*\x)});

\draw[green!50, thin, domain=\a:7, samples=200] plot (\x, {3/(2*\x)});
\draw[green!50, thin, domain=-7:-\a, samples=200] plot (\x, {3/(2*\x)});

\draw[red, thin] (\b, 0.1) -- (7, 0.1);
\draw[red, thin] (-7, -0.1) -- (-\b, -0.1);

\draw[red, thin, domain=-\b:\b] plot (\x, {\x/2});

\end{tikzpicture}
\end{figure}

Note that $g$ coincides with the Minkowski metric outside of the open set  
\begin{equation}\label{eq:defN}
N:= \bigcup_{ |x|> 2^{-1/4}} \left\{4x^2\Big(\Big(t-\frac{1}{x}\Big)^2 + y^2 + z^2\Big)< 1\right\}\supset \left\{\left(t, \frac{1}{t}, 0, 0\right)\colon |t|< 2^{1/4} \right\}.
\end{equation}
  Moreover, note that the curve $ s \mapsto (s, \frac{1}{s}, 0, 0)$ is $g$-timelike for $s \in (-\frac{1}{3}, \frac{1}{3}) \setminus \{0\}$. 
Define the achronal sets 

$$
S_1 := \left\{t=\frac{1}{2}x: |x| \leq \frac{1}{5}\right\} \quad \text{ and }\quad  S_2 = \left\{t= \mathrm{sgn}(x)\frac{1}{10}: |x|> \frac{1}{5}\right\}.
$$

\textbf{Claim 1.} $S = S_1 \cup S_2$ is acausal.

\emph{Proof of claim 1.}
Let $p \neq q\in S$ and suppose that $p < q$. The only case to consider is if $p,q$ are in different $S_i$s ($i=1,2$).
Then there exists a smooth $g$-causal curve $\gamma$ connecting them, which we assume to be parametrised by uniform $t$-speed, then it follows that $t(p) < t(q)$. 
Note that projecting \( \gamma \) onto the subspace \( \{y = z = 0\} \) yields a \( g \)-causal curve. Indeed, the time component of the velocity remains unchanged, while the metric can only decrease in the spacelike directions (as $\phi$ is increasing), and only one of the three spatial directions is retained. From now on, we work with this reduced curve.
From the definition of $S$, it follows that $t(p)<t(q)$ implies $x(p) < x(q)$. Furthermore, we may assume that $x(\gamma)$ is increasing. If not, we may consider the function $t \mapsto \min(x(q), \max\{ x(\gamma(s)),s \leq t\})$, which is piecewise smooth and $g$-causal.   
Hence, we constructed a piecewise smooth, causal curve  $\gamma:[t(p), t(q)] \to \R^4$ such that 
$$\gamma(t(p)) = (t(p), x(p), 0, 0), \quad \gamma(t(q)) = (t(q), x(q), 0, 0),\quad\text{ and } \quad P_2\circ \gamma\; \text{ is increasing}.$$

Note that $t(p)<t(q)$ implies that $t(p) < \frac{1}{10}$ or $t(q) > - \frac{1}{10}$. The definition of $S$ implies that $x(p) < \frac{1}{5}$ or $x(q)> -\frac{1}{5}$. Define furthermore
\begin{align*}
    t_p:= \max\big(\{t: x(\gamma(t)) \leq -2\} \cup \{t(p)\}\big), \ t_q:= \min\big(\{t: x(\gamma(t)) \geq 2\} \cup \{t(q)\}\big). 
\end{align*}
Note that by the definition of $S$ and the fact that $t \circ \gamma$ is strictly increasing, it holds that $t_p \geq t(p) \geq -\frac{1}{10}$ and $t_q \leq t(q) \leq \frac{1}{10}$, and 
\begin{align*}
    \frac{x\big(\gamma(t_q)\big)-x\big(\gamma(t_p)\big)}{t_q-t_p} \geq 2. 
\end{align*}
Moreover, along \( \gamma([t_p, t_q]) \subset \{|x|\leq 2, |t| \leq \frac{1}{10}\} \), the metric \( g \) coincides with the Minkowski metric. 
Since \( \gamma \) is piecewise smooth, there exists an interval \( (t_1,t_2) \subset [t_p, t_q] \) on which \( \gamma \) is smooth and such that 
\[
x(\gamma(t_2)) - x(\gamma(t_1)) \ge 2(t_2 - t_1).
\]
By the intermediate value theorem, there exists \( \tau \in (t_1,t_2) \) such that 
\[
\frac{d}{dt} \bigl(x \circ \gamma\bigr)(\tau) \ge 2.
\]
However, this contradicts the causality of \( \gamma \) at \( \tau \).
\hfill$\Box$
\medskip

\textbf{Claim 2:} $S$ is a Cauchy hypersurface for $g$.  

\emph{Proof of claim 2.}
We may assume that $\partial_t$ is future directed. 
Let \( \gamma \colon (a,b) \to \bigl(-\tfrac{1}{3},\tfrac{1}{3}\bigr) \times \mathbb{R}^3 \) be an inextendible causal and future directed curve. We may assume that \( a < 0 \) and \( b > 0 \). Since \( S \) is a Cauchy hypersurface for Minkowski space, it suffices to consider the case in which \( \gamma((a,b)) \cap N \neq \emptyset \).

We may therefore assume that \( \gamma(0) \in N \), and set \( t_0 := t(\gamma(0)) > 0 \) and \( x_0 := x(\gamma(0)) > 0 \). Suppose that \( \gamma \) does not intersect \( S \). 

\textbf{Case 1.} \( t_0 \leq \tfrac{1}{10} \). \\
We will show that $\gamma$ will meet $S$ in the future. Indeed, as $t_0>0$, the definition of the Lorentzian metric $g$ yields that there exists  $\lambda_0>0$ such that $g \prec \mathrm{diag}(-1, \lambda_0, \lambda_0, \lambda_0)$ in $\{t \geq t_0\}$. Now, as $t$ is increasing along future directed curves, we get that $\gamma([0,b)) \subset \{t \geq t_0\}$. Now, as $\{t=\frac{1}{10}\}$ is a Cauchy hypersurface for $\mathrm{diag}(-1, \lambda_0, \lambda_0, \lambda_0)$, we get that there exists a $\beta \in (0,b)$ such that $t(\gamma(\beta))> \frac{1}{10}$. However, the definition of $S$ together with the continuity of $\gamma$ implies that $\gamma$ meets $S$ at some time in $(0, \beta]$.

\textbf{Case 2.} \( t_0 > \tfrac{1}{10} \). \\
We will prove that $\gamma$ meets $S$ in the past.
Note that there exists  $\lambda_0>0$ such that $g \prec \mathrm{diag}(-1, \lambda_0, \lambda_0, \lambda_0)$ in $\{t \geq \frac{1}{20}\}$. As $\{t = \frac{1}{10}\}$ is a Cauchy hypersurface for $\mathrm{diag}(-1, \lambda_0, \lambda_0, \lambda_0)$, it follows that $t(\gamma(\alpha)) = \frac{1}{10}$ for some $\alpha \in (a, 0)$. If $\gamma(\alpha) \in S_2 \subset S$, we are done. Suppose $\gamma(\alpha) \notin S_2$. Then $x(\gamma(\alpha))< \frac{1}{5}$. If $\gamma((a, \alpha])$ intersects $\{x \geq \frac{1}{5}\}$, then it must intersect $S$. Suppose $\gamma((a, \alpha]) \subset \{x < \frac{1}{5}\}$. Now, there are two options. If $\gamma((a, \alpha])$ stays in the region where $g$ equals the Minkowski metric, then it must intersect $S$. If not, it must intersect $N \cap \{t <0\}$ at some time $\sigma \in (a, \alpha)$. If $t(\gamma(\sigma)) \leq -\frac{1}{10}$ then $\gamma$ must intersect $S$ at some time in $[\sigma, \alpha)$. 
If $t(\gamma(\sigma)) >-\frac{1}{10}$, we may argue as in case 1. 

This proves the second claim.

\hfill$\Box$

We next construct the sequence of Lorentzian metrics $(g_j)$, failing to be uniformly globally hyperbolic. On $B_j(0)$, we set $g_j$ to coincide with $g$. Let $\e_j \in (0,\frac{1}{4j^2})$ and $\delta_j \in (0, \frac{1}{2^j})$. Let $\rho_{j,1}, \rho_{j,2}$ be a partition of unity subordinate to the open cover $B_{j+\delta_j}(0), (\overline{B_j(0)})^c$ and define 
\begin{align*}
     g_j = \rho_{j,1}g + \rho_{j,2} \begin{pmatrix}
    -1 & 0 & 0 & 0 \\
    0 & \e_j & 0 & 0 \\
    0 & 0 & \e_j & 0 \\
    0 & 0 &  0  & \e_j 
    \end{pmatrix}.
\end{align*}
Then the curves $$\gamma:\left[-\frac{1}{4}, -\frac{1}{j+\delta_j}\right] \to M, \quad t \mapsto \left(t, \frac{1}{t}, 0, 0\right)$$ and $$\lambda:\left[\frac{1}{j+\delta_j}, \frac{1}{4}\right] \to M, \quad t \mapsto \left(t, \frac{1}{t}, 0, 0\right)$$ are $g_j$-causal. For $\e_j>0$ small enough, we can find a $g_j$-causal curve from $(-\frac{1}{j+\delta_j}, -(j+ \delta_j), 0, 0)$ to $(\frac{1}{j+\delta_j}, j+ \delta_j, 0, 0)$ that intersects in $B_{j}(0)^c$. 
As  all the positive eigenvalues of $g_j$ are bounded below by $\e_j$, we have that $g_j$ is globally hyperbolic. 
By construction $g_j\to g$  locally uniformly.  However,
\begin{align*}
    J^+_{g_j}\Big(\gamma\Big(-\frac{1}{4}\Big)\Big) \cap J^{-}_{g_j}\Big(\lambda\Big(\frac{1}{4}\Big)\Big) \cap (B_j(0)^c) \neq \emptyset, \quad \text{  for all $j>0$,}
\end{align*}
showing that the sequence $(g_j)$ is not uniformly globally hyperbolic.

\section{Details on the approximation of the metrics}
\subsection{The case $\Lambda=0$}\label{App:B1Flat}
In this appendix we give more details of the computations involved in the proof of Proposition \ref{prop:ApproxFlat}.

Recall the decomposition \eqref{decomposition_in_metric_and_balance_terms} and that the analysis in  the proof of Proposition \ref{prop:ApproxFlat} is split into two parts:
\begin{itemize}
    \item[(i)] The contributions in the numerator that only contain terms from $\hat{g}_\e, \hat{\Gamma}^\e$, and $D\hat{\Gamma}^\e$.
    \item[(ii)] The remaining contributions in the numerator.
\end{itemize}
For (i), we restricted our considerations to the terms in the numerator that arise from $\hat{g}_\e$ and its derivatives and analysed where terms like \eqref{problematic_convergence} occur:  

\textbf{Term 1:} $\rho_\e$ and $\rho_\e \cdot (\rho_\e * U^+)$. \\
We get the term $\rho_\e$ if we differentiate $\hat{g}^\e_{33}, \hat{g}^\e_{34}$, or $\hat{g}^\e_{44}$ twice in the direction of $U$. These derivatives occur in the coefficients $R^\e_{ijkl}$ if two of the indices is $1$ and the two other indices are in $\{3, 4\}$. 
It follows that only $R^\e_{1313}, R^\e_{1314}, R^\e_{1414}$ (and the coefficients that arise from symmetries) may carry a $\rho_\e$ term. 
Recalling the formula \eqref{ricci_coefficients_formula}, we get that the above Riemannian curvature coefficients are used to compute $\Ric^\e_{11}, \Ric_{13}^\e, \Ric_{33}^\e, \Ric_{14}^\e, \Ric_{44}^\e$, and $\Ric_{34}^\e$. Recalling \eqref{inverse_mollified_metric}, we get that the only potentially non-zero contributions are in $\Ric^\e_{11}$ through $g^{34}_\e R^\e_{3114}$, $g^{33}_\e R^\e_{3113}$, and $g^{44}_\e R^\e_{4114}$. Using \eqref{low_index_riem_curvature}, \eqref{ricci_coefficients_formula}, and \eqref{inverse_mollified_metric}, we can extract all terms of the form $f_\e \cdot \rho_\e$, where $\norm{f_\e}_{L^\infty((-\e, \e))} \notin o(\e)$. In this appendix, when extracting terms, we will denote the remaining terms that are not in the focus of our analysis simply with  ``\ldots" . 
\begin{align*}
    \Ric^\e_{11} &= 2g^{34}_\e R^\e_{3114} +g^{33}_\e R^\e_{3113} +g^{44}_\e R^\e_{4114} \\
    &= \frac{d_\e}{d_\e}\cdot \frac{-2g^\e_{34}}{d_\e} R^\e_{3114} + \frac{d_\e}{d_\e}\cdot\frac{g^\e_{44}}{d_\e} R^\e_{3113} + \frac{d_\e}{d_\e}\cdot\frac{g^\e_{33}}{d_\e} R^\e_{4114} \\
    &= \frac{-2d_\e g^\e_{34}}{(d_\e)^2} (\partial_3 \Gamma^\e_{11,4}-\partial_1 \Gamma_{13,4} + \ldots) + \frac{d_\e g^\e_{44}}{(d_\e)^2} (\partial_3 \Gamma^\e_{11,3}-\partial_1 \Gamma_{13,3} + \ldots) + \frac{d_\e g^\e_{33}}{(d_\e)^2} (\partial_4 \Gamma^\e_{11,4}-\partial_1 \Gamma_{14,4} + \ldots).
\end{align*}
We now restrict our considerations to the terms arising from $\hat{g}_\e$, and extract all terms of the form $f_\e \cdot \rho_\e$, where $\norm{f_\e}_{L^\infty((-\e, \e))} \notin o(\e)$. 
\begin{align*}
    \Ric^\e_{11} &= \frac{-2\hat{d}_\e \hat{g}^\e_{34}}{(d_\e)^2} (\partial_3 \hat{\Gamma}^\e_{11,4}-\partial_1 \hat{\Gamma}_{13,4} + \ldots) + \frac{\hat{d}_\e \hat{g}^\e_{44}}{(d_\e)^2} (\partial_3 \hat{\Gamma}^\e_{11,3}-\partial_1 \hat{\Gamma}_{13,3} + \ldots) + \frac{\hat{d}_\e \hat{g}^\e_{33}}{(d_\e)^2} (\partial_4 \hat{\Gamma}^\e_{11,4}-\partial_1 \hat{\Gamma}_{14,4} + \ldots) \\
    &= \frac{-2\hat{d}_\e \hat{g}^\e_{34}}{(d_\e)^2} ((-\partial^2_{Z, \cz}\hh) \rho_\e + ...) + \frac{\hat{d}_\e \hat{g}^\e_{44}}{(d_\e)^2}(- \partial^2_{Z,Z}\hh \rho_\e) +\frac{\hat{d}_\e \hat{g}^\e_{33}}{(d_\e)^2}(- \partial^2_{\cz,\cz}\hh \rho_\e) + \ldots \\
    &= \frac{1}{(d_\e)^2}(-2(\partial^2_{Z, \cz}\hh)+(-12(\partial^2_{Z, \cz}\hh)^2 + 4|\partial^2_{Z,Z}\hh|^2)\cdot (\rho_\e * U^+))\rho_\e  + \ldots. 
\end{align*}
On $(-\e, \e) \supset \supp \, \rho_\e$, we have that 
\begin{align*}
    (d_\e)^2 = 1+ 8\partial^2_{Z,\cz}\hh \rho_\e * U^+ + o(\e),
\end{align*}
hence
\begin{align*}
    \frac{1}{(d_\e)^2} = 1-8\partial^2_{Z,\cz}\hh \rho_\e * U^+ + o(\e).
\end{align*}
It follows that
\begin{align*}
    \Ric^\e_{11} &= \frac{1}{(d_\e)^2}(-2(\partial^2_{Z, \cz}\hh)+(-12(\partial^2_{Z, \cz}\hh)^2 + 4|\partial^2_{Z,Z}\hh|^2)\cdot (\rho_\e * U^+))\rho_\e  + \ldots\\
    &= -2 \partial^2_{Z, \cz}\hh \rho_\e +16(\partial^2_{Z, \cz}\hh)^2  (\rho_\e * U^+)\rho_\e + \frac{1}{(d_\e)^2}((-12(\partial^2_{Z, \cz}\hh)^2 + 4|\partial^2_{Z,Z}\hh|^2)\cdot (\rho_\e * U^+))\rho_\e  + \ldots\\
    &= -2 \partial^2_{Z, \cz}\hh \rho_\e +\frac{1}{(d_\e)^2}(4(\partial^2_{Z, \cz}\hh)^2 + 4|\partial^2_{Z,Z}\hh|^2)\cdot (\rho_\e * U^+)\rho_\e  + \ldots,
    \end{align*}
where in the last step, we used that $d_\e^2 = 1+ O(\e)$ in $(-\e, \e)$.

\textbf{Term 2:} $\rho_\e* \Theta - (\rho_\e * \Theta)^2$. \\
We now provide a precise argument showing that this term can only arise in the contribution of \( 2\,g^{34}_\varepsilon R^\varepsilon_{3114} \) to \( \Ric_{11} \).
We get $(\rho_\e * \Theta)$-terms from $\partial_{1} g^\e_{ij}$ or $\partial^2_{1,k}g^\e_{ij}$ for $k \in \{1,3,4\}$ and $i,j \in \{3,4\}$. The first ones of these terms appear in $\Gamma_{[ij,1]}^\e$, and the second kind of terms can only show up in $R_{[1ijk]}$, where $k \in \{1,3,4\}$ and  $i,j \in \{3,4\}$.
Recalling \eqref{low_index_riem_curvature}, it follows that the only way to get a $(\rho_\e * \Theta)^2$-term is when two $\Gamma^\e$ symbols as above are multiplied. Hence, we consider terms of the form $g_\e^{st}\Gamma^\e_{ij,s}\Gamma^\e_{kl,t}$. 
We first notice that we need $s \neq t$, otherwise $|g^{st}| \leq C \e$ on $(-\e, \e) \supset \supp\, (\rho_\e* \Theta - (\rho_\e * \Theta)^2)$, so the contribution converges to $0$ uniformly as $\e \to 0$.  Since  $\Gamma^\e_{ij,2} = 0$ for all $i,j$, it
 follows that $\{s,t\}=\{3,4\}$. Hence $\{\{i,j\}, \{k,l\}\} \subset \{\{1,3\}, \{1,4\}\}$, which can only occur in $R^\e_{[1313]}, R^\e_{[1314]}, R^\e_{[1414]}$. Using \eqref{ricci_coefficients_formula}, we observe that for \( R^\varepsilon_{[1313]} \) to contribute to \( \Ric^\varepsilon \), it must be contracted with one of the coefficients \( g^{11}_\varepsilon \), \( g^{13}_\varepsilon \), or \( g^{33}_\varepsilon \). 
All of these are \( L^\infty \)-bounded by \( C\varepsilon \) on \( \supp\bigl(\rho_\varepsilon * \Theta - (\rho_\varepsilon * \Theta)^2\bigr) \), and hence the corresponding contribution converges uniformly to \( 0 \) as \( \varepsilon \to 0 \). 
The same argument applies to \( R^\varepsilon_{[1414]} \).
For \( R^\varepsilon_{[1314]} \), it suffices to consider its contraction with \( g^{34}_\varepsilon \), which appears in the computation of \( \Ric^\varepsilon_{11} \). 
Recalling the expressions for \( g^{33}_\varepsilon \) and \( g^{44}_\varepsilon \), together with the observations above, we see that it is enough to focus on the contribution of \( 2\,g^{34}_\varepsilon R^\varepsilon_{3114} \). 
We have that
\begin{align*}
    2g^{34}_\e R^\e_{3114} &= \frac{-2d_\e g^\e_{34}}{(d_\e)^2} (\partial_3 \Gamma^\e_{11,4}-\partial_1 \Gamma^\e_{13,4}) + 2\frac{(g^\e_{34})^2}{(d_\e)^2}(\Gamma^\e_{31,3}\Gamma^\e_{14,4}+\Gamma^\e_{31,4}\Gamma^\e_{14,3}- 0) \\
    &= \frac{-2d_\e g^\e_{34}}{(d_\e)^2} (-(|\partial^2_{Z, Z}\hh|^2+(\partial^2_{Z, \cz}\hh)^2)(\rho_\e * \Theta)) +\frac{2(g^\e_{34})^2}{(d_\e)^2} ((|\partial^2_{Z, Z}\hh|^2+(\partial^2_{Z, \cz}\hh)^2)(\rho_\e * \Theta)^2) + \ldots \\
    &= \frac{2}{(d_\e)^2}(|\partial^2_{Z, Z}\hh|^2+(\partial^2_{Z, \cz}\hh)^2)((\rho_\e * \Theta)^2 - \rho_\e * \Theta) + \ldots.
\end{align*}
This concludes the analysis of those terms in the numerator of the Ricci curvature, only involving $\hat{g_\e}$ and its derivatives. 

For the terms as in (ii), we distinguished between contributions involving second derivatives of the metric in the $U$-direction , i.e., terms of the form $\partial^2_{1,1}(g_\e - \hat{g_\e})$ and those that only involve first derivatives in the $U$-direction of $(g_\e - \hat{g_\e})$.

\textbf{Terms involving  a factor with two derivatives in the $U$-direction:} \\
In this case, the only relevant term is \( \partial_{1,1}^2 g^\varepsilon_{34} \). 
This term appears exclusively in \( R^\varepsilon_{[1314]} \) and therefore contributes only to \( \Ric^\varepsilon_{11} \). 
Consequently, it suffices to analyze
\begin{align*}
    \frac{-2d_\e g^\e_{34}}{(d_\e)^2}(-\partial_1 (\Gamma^\e_{13,4}-\hat{\Gamma}^\e_{13,4})) = \frac{2d_\e g^\e_{34}}{(d_\e)^2} (2((\partial^2_{Z, \cz}\hh)^2 + |\partial^2_{Z,Z}\hh|^2)\rho_\e(\rho_\e * U^+) + ((\partial^2_{Z, \cz}\hh)^2 + |\partial^2_{Z,Z}\hh|^2)\sigma_\e).
\end{align*}
We note that $ \supp\, \sigma_\e \cup \supp \, \rho_\e(\rho_\e * U^+) \subset \{-\e<U<\e\}$ and both functions are bounded in $L^\infty$ by a constant $C>0$ independent of $\e$. Moreover, we can write
\begin{align*}
    d_\e g^\e_{34} = -1 + (\rho_\e * U^+) F^\e_1 + (\rho_\e * (U^+)^2) F^\e_2 + \zeta_\e F^\e_3 + \kappa_\e F^\e_4.
\end{align*}
where for all $i$, $ F^\e_i = F^\e_i(U, V, Z, \cz)$ is uniformly bounded in $L^\infty(\mathcal{K})$ by a constant $C>0$ independent of $\e$. 
Now, in $ \{-\e<U<\e\}$,  both $(\rho_\e * U^+)$ and $(\rho_\e * (U^+)^2)$ are bounded in $L^\infty$ by $C\e$, where $C>0$ does not depend on $\e$. Then, the above term can be written as 
\begin{align*}
     \frac{-2d_\e g^\e_{34}}{(d_\e)^2}(-\partial_1 (\Gamma^\e_{13,4}-\hat{\Gamma}^\e_{13,4})) = & \frac{1}{(d_\e)^2} (-4((\partial^2_{Z, \cz}\hh)^2 + |\partial^2_{Z,Z}\hh|^2)\rho_\e(\rho_\e * U^+) \\
     & -2 ((\partial^2_{Z, \cz}\hh)^2 + |\partial^2_{Z,Z}\hh|^2)\sigma_\e) + F_\e,
\end{align*}
where $\supp\, F_\e \in (-\e, \e)$ and $F_\e \to 0$ uniformly as $\e \to 0$.

\textbf{Terms involving at most first derivatives in the $U$-direction:} \\
We first note that the $\Gamma^\e$ and $g_\e$ are locally uniformly bounded independently of $\e$. 
Moreover, 
\begin{align}
    &|D\Gamma^\e| + |D\hat{\Gamma}^\e| \leq C\e^{-1}\ \mathrm{in}\; \{-\e \leq U \leq \e \} \nonumber \\
     &|D\Gamma^\e| + |D\hat{\Gamma}^\e| \leq C\ \mathrm{in}\; \{|U| >\e \},
\end{align}
where $C> 0$ that does not depend on $\e$. 
We note that when computing the Ricci curvature coefficients along the lines of \eqref{connection_first_kind}, \eqref{low_index_riem_curvature}, and \eqref{ricci_coefficients_formula}, all terms involving a multiplication with $g_\e^{21}$ or $g_\e^{12}$ vanish, as $\Gamma^\e_{[2i,j]} = 0$ for all $i,j$. Then
\begin{align*}
    \Ric_{jk}^\e &= g_\e^{il}R^\e_{ijkl} = g_\e^{il}(\partial_i \Gamma^\e_{jk, l} -\partial_j \Gamma^\e_{ik, l} + g_\e^{st}(\Gamma^\e_{ik, s}\Gamma^\e_{jl, t}-\Gamma^\e_{il, s}\Gamma^\e_{jk, t})) \\
    &= \frac{1}{(d_\e)^2}g^\e_{i'l'}(d_\e(\partial_i \Gamma^\e_{jk, l} -\partial_j \Gamma^\e_{ik, l}) + g_{s't'}(\Gamma^\e_{ik, s}\Gamma^\e_{jl, t}-\Gamma^\e_{il, s}\Gamma^\e_{jk, t})),
\end{align*}
where we define $3'=4$, $4'=3$, $1'=2$, and $2'=1$ (this formula is valid because the \(12\)- and \(21\)-terms do not contribute). Using \eqref{bound_on_g_and_gamma_eps_to_hat}, \eqref{bound_on_det_difference}, \eqref{bounds_on_dgamma_eps}, and \eqref{relaxation_terms}, we get that for $j=k=1$:
\begin{align*}
    \Ric_{jk}^\e &- \frac{1}{(d_\e)^2}\hat{g}^\e_{i'l'}(\hat{d}_\e(\partial_i \hat{\Gamma}^\e_{jk, l} -\partial_j \hat{\Gamma}^\e_{ik, l}) + \hat{g}^\e_{s't'}(\hat{\Gamma}^\e_{ik, s}\hat{\Gamma}^\e_{jl, t}-\hat{\Gamma}^\e_{il, s}\hat{\Gamma}^\e_{jk, t})) \\
    &= \frac{1}{(d_\e)^2} (-4((\partial^2_{Z, \cz}\hh)^2 + |\partial^2_{Z,Z}\hh|^2)\rho_\e(\rho_\e * U^+) -2 ((\partial^2_{Z, \cz}\hh)^2 + |\partial^2_{Z,Z}\hh|^2)\sigma_\e) + F_\e,
\end{align*}
for some function $F_\e$ such that $F_\e \to 0$ uniformly as $\e \to 0$. For any $(j,k) \neq (1,1)$, we get that 
\begin{align*}
    \Ric_{jk}^\e &- \frac{1}{(d_\e)^2}\hat{g}^\e_{i'l'}(\hat{d}_\e(\partial_i \hat{\Gamma}^\e_{jk, l} -\partial_j \hat{\Gamma}^\e_{ik, l}) + \hat{g}^\e_{s't'}(\hat{\Gamma}^\e_{ik, s}\hat{\Gamma}^\e_{jl, t}-\hat{\Gamma}^\e_{il, s}\hat{\Gamma}^\e_{jk, t})) = F_\e,
\end{align*}
for some function $F_\e$ such that $F_\e \to 0$ uniformly as $\e \to 0$.

We now get that 
\begin{align*}
    \Ric^\e_{ij} = -2\partial^2_{Z, \cz}\hh \rho_\e \delta^1_i\delta^1_j + \frac{1}{(d_\e)^2}\Tilde{\Ric}^\e_{ij} + r_{ij}^\e,
\end{align*}
where $r^\e_{ij}$ converges to $0$ uniformly on $\mathcal{K}$ as $\e \to 0$ and $\Tilde{\Ric}^\e_{ij}$ consists of those terms of $(\hat{d}_\e)^2(\hat{\Ric}^\e_{ij} +2\delta^1_i\delta^1_j\partial^2_{Z, \cz}\hh \rho_\e)$ that occur in \eqref{good_convergence}. From \eqref{good_convergence} and the non-smooth case, we know that $\Tilde{\Ric}^\e_{ij}$ uniformly converges to $0$ as $\e \to 0$.
We finally note that  the uniform convergence  $g_\e\to g$ on $\mathcal{K}$ implies that $d_\e$ is bounded away from zero, and $d_\e(0, V, Z, \cz) \to 1$ uniformly as $\e \to 0$, which proves that $\Ric^\e_{ij} +2\partial^2_{Z, \cz}\hh \rho_\e \delta^1_i\delta^1_j$ uniformly converges to $0$ as $\e \to 0$.

\subsection{Illustration of the approach on a special case}
In order to illustrate the main estimates, in this subsection we consider the special case 
$$\mathcal{H}(Z, \Bar{Z}) = a |Z|^2 + \frac{1}{2}b(Z^2 + \Bar{Z}^2),  \quad \text{for some } a, b \in \R.$$
 In this case, we have:
\begin{align*}
    \partial^2_{Z,\Bar{Z}}\mathcal{H} &= a \\
    \partial^2_{Z, Z}\hh = \partial^2_{\cz, \cz}\hh &= b.
\end{align*}
Then the metric is of the form
\begin{align*}
    g^\e_{34} = g^\e_{43} &= 1+ 2a \rho_\e*U^+ +(a^2+b^2)\rho_\e*(U^+)^2 + 4(a^2 + b^2)\kappa_\e + 2(a^2 + b^2)\zeta_\e \\
    g^\e_{33} = g^\e_{44} &= 2b \rho_\e * U^+ + 2ab \rho_\e * (U^+)^2.
\end{align*}
The only non-zero derivatives of $g^\e$ are in the direction of $U$. We get that 
\begin{align*}
    \partial_1 g^\e_{34} =\partial_1 g^\e_{43} &=  2a \rho_\e*\Theta +2(a^2+b^2)\rho_\e*(U^+) + 4(a^2 + b^2)\eta_\e + 2(a^2 + b^2)\xi_\e \\
    \partial_1g^\e_{33} = \partial_1g^\e_{44} &= 2b \rho_\e * \Theta + 4ab \rho_\e * U^+.
\end{align*}
Then 
\begin{align*}
    \Gamma^\e_{13,3} &= b \rho_\e * \Theta + 2ab \rho_\e * U^+ , \\
    \Gamma^\e_{13,4} &= a \rho_\e*\Theta +(a^2+b^2)\rho_\e*(U^+) + 2(a^2 + b^2)\eta_\e + (a^2 + b^2)\xi_\e, 
\end{align*}
and $\Gamma^\e_{13,3} =\Gamma^\e_{14,4} =-\Gamma^\e_{33,1} =-\Gamma^\e_{44,1}$ and  $\Gamma^\e_{13,4} =\Gamma^\e_{14,3} =-\Gamma^\e_{34,1} =-\Gamma^\e_{43,1}$. 

For any $i,j,k,l$, we have that if $R^\e_{ijkl} \neq 0$, then  $\partial_i \Gamma^\e_{jk,l} - \partial_j \Gamma^\e_{ik,l} \neq 0$ or $g^{st}(\Gamma^\e_{ik,s}\Gamma^\e_{jl,t}-\Gamma^\e_{il,s}\Gamma^\e_{jk,t}) \neq 0$. For the first term, we need two indices of $i,j,k,l$ to be $1$ and the other two to be in $\{3,4\}$.
In the second case we further get that $g^{st}(\Gamma^\e_{ik,s}\Gamma^\e_{jl,t}) \neq 0$ or $g^{st}(\Gamma^\e_{il,s}\Gamma^\e_{jk,t}) \neq 0$. As $\Gamma^\e_{ij,2} = 0$, for all $i,j$, it follows that in both cases $s,t \in \{3,4\}$ to get a non-zero contribution. 
For symmetry reasons it suffices to find all indices $i,j,k,l$ such that $\Gamma^\e_{ik,3}, \Gamma^\e_{jl,3}, \Gamma^\e_{ik,4}$, or $\Gamma^\e_{jl,4}$ are potentially non-zero. For that to hold, we need one out of $i,k$ and one out of $j,l$ to be $1$, as well as the other two indices to be on $\{3,4\}$. Hence, the only potentially non-zero curvature coefficients are $R^\e_{[1313]}, R^\e_{[1314]}$ and $R^\e_{[1414]}$. Since $g^{ij}_\e = 0$ unless $i,j \in \{3,4\}$ or $(i,j) = (1,2), (2,1)$, we get that the only way for these terms to contribute to a Ricci curvature coefficients (along \eqref{ricci_coefficients_formula}) is in 
\begin{align*}
    \Ric^\e_{11} &= 2g_\e^{34}R^\e_{3114} + g_\e^{33}R^\e_{3113} + g_\e^{44}R^\e_{4114}
     \\
     &= 2g_\e^{34}R^\e_{3114} + 2g_\e^{33}R^\e_{3113}.
\end{align*}
Observe that 
\begin{align*}
    d_\e = \det g_\e[3,4] &= g^\e_{33} g^\e_{44} - (g^\e_{34})^2. 
\end{align*}
Now we can compute 
\begin{align*}
    R^\e_{3114} &= \partial_3 \Gamma^\e_{11,4} - \partial_1 \Gamma^\e_{13,4} + g_\e^{st}(\Gamma^\e_{13,s}\Gamma^\e_{14,t}-\Gamma^\e_{11,s}\Gamma^\e_{34,t}) \\
    &= - \partial_1 \Gamma^\e_{13,4} + g_\e^{33}\Gamma^\e_{13,3}\Gamma^\e_{14,3} +g_\e^{44}\Gamma^\e_{13,4}\Gamma^\e_{14,4}+g_\e^{34}\Gamma^\e_{13,3}\Gamma^\e_{14,4}+g_\e^{43}\Gamma^\e_{13,4}\Gamma^\e_{14,3} \\
    &= - \partial_1 \Gamma^\e_{13,4} + 2g_\e^{33}\Gamma^\e_{13,3}\Gamma^\e_{14,3} + g^{34}_\e ((\Gamma^\e_{13,3})^2 + (\Gamma^\e_{13,4})^2) \\
    &= - \partial_1 \Gamma^\e_{13,4} + \frac{1}{d_\e}(2g^\e_{33}\Gamma^\e_{13,3}\Gamma^\e_{14,3} - g_{34}^\e ((\Gamma^\e_{13,3})^2 + (\Gamma^\e_{13,4})^2)).
\end{align*}
We get that 
\begin{align*}
    2g^\e_{33}&\Gamma^\e_{13,3}\Gamma^\e_{14,3} \\
    =& (4b \rho_\e * U^+ + 4ab \rho_\e * (U^+)^2)(b \rho_\e * \Theta + 2ab \rho_\e * U^+)(a \rho_\e*\Theta +(a^2+b^2)\rho_\e*(U^+) + (a^2 + b^2)(2\eta_\e +\xi_\e)) \\
    =& (4b^2 (\rho_\e * U^+)(\rho_\e * \Theta) + 8ab^2 (\rho_\e * U^+)^2 + 4ab^2(\rho_\e * \Theta)(\rho_\e * (U^+)^2) + 8a^2b^2 (\rho_\e * U^+)(\rho_\e * (U^+)^2)) \\
    & \cdot (a \rho_\e*\Theta +(a^2+b^2)\rho_\e*(U^+) + (a^2 + b^2)(2\eta_\e +\xi_\e)) \\
    =& 4ab^2 (\rho_\e * U^+)(\rho_\e * \Theta)^2 + 8a^2b^2 (\rho_\e*\Theta) (\rho_\e * U^+)^2 + 4a^2b^2(\rho_\e * \Theta)^2(\rho_\e * (U^+)^2) \\
    &+  8a^3b^2 (\rho_\e*\Theta)(\rho_\e * U^+)(\rho_\e * (U^+)^2) +   4a^2b^2 (\rho_\e * U^+)^2(\rho_\e * \Theta) +  8a^3b^2 (\rho_\e * U^+)^3 +  \\
    &+ 4a^3b^2(\rho_\e * \Theta)(\rho_\e*U^+)(\rho_\e * (U^+)^2) + 8a^4b^2 (\rho_\e * U^+)^2(\rho_\e * (U^+)^2) +  4b^4 (\rho_\e * U^+)^2(\rho_\e * \Theta) \\
    &+ 8ab^4 (\rho_\e * U^+)^3 + 4ab^4(\rho_\e * \Theta)(\rho_\e*U^+)(\rho_\e * (U^+)^2) + 8a^2b^4 (\rho_\e * U^+)^2(\rho_\e * (U^+)^2) \\
    &+ (a^2 + b^2)(2\eta_\e +\xi_\e)O(1) \\
    =& 4ab^2 (\rho_\e * U^+)(\rho_\e * \Theta)^2 \\
    &+ 8a^2b^2 (\rho_\e*\Theta) (\rho_\e * U^+)^2 + 4a^2b^2(\rho_\e * \Theta)^2(\rho_\e * (U^+)^2) +4a^2b^2 (\rho_\e * U^+)^2(\rho_\e * \Theta) +4b^4 (\rho_\e * U^+)^2(\rho_\e * \Theta) \\
    &+  8a^3b^2 (\rho_\e*\Theta)(\rho_\e * U^+)(\rho_\e * (U^+)^2) +  8a^3b^2 (\rho_\e * U^+)^3 +  + 4a^3b^2(\rho_\e * \Theta)(\rho_\e*U^+)(\rho_\e * (U^+)^2) \\
    & + 8ab^4 (\rho_\e * U^+)^3 + 4ab^4(\rho_\e * \Theta)(\rho_\e*U^+)(\rho_\e * (U^+)^2) \\
    &+ 8a^4b^2 (\rho_\e * U^+)^2(\rho_\e * (U^+)^2) + 8a^2b^4 (\rho_\e * U^+)^2(\rho_\e * (U^+)^2) + (a^2 + b^2)(2\eta_\e +\xi_\e)O(1) \\
    =& 4ab^2 (\rho_\e * U^+)(\rho_\e * \Theta)^2 \\
    &+ (8a^2b^2+4a^2b^2+4a^2b^2 +4b^4)(\rho_\e *U^+)^2 + O(\e) \\
    &+ (8a^3b^2+  8a^3b^2 + 4a^3b^2 + 8ab^4 + 4ab^4) (\rho_\e *U^+)^3 + O(\e) \\
    &+ (8a^4b^2 + 8a^2b^4) (\rho_\e *U^+)^4 + O(\e).
\end{align*}
In the last step we used \eqref{good_convergence}.
Similarly,
\begin{align*}
     - g_{34}^\e& ((\Gamma^\e_{13,3})^2 + (\Gamma^\e_{13,4})^2)) \\
     =& -(1+ 2a \rho_\e*U^+ +(a^2+b^2)\rho_\e*(U^+)^2 + O(\e^2))\\
     &\cdot ((b \rho_\e * \Theta + 2ab \rho_\e * U^+ )^2 + (a \rho_\e*\Theta +(a^2+b^2)\rho_\e*(U^+) + O(\e))^2)\\
     =& -(1+ 2a \rho_\e*U^+ +(a^2+b^2)\rho_\e*(U^+)^2 + O(\e^2)) \\
     &\cdot (b^2 (\rho_\e * \Theta)^2 + 4ab^2(\rho_\e * \Theta)(\rho_\e * U^+) + 4a^2b^2(\rho_\e * U^+)^2 \\
     &+ a^2 (\rho_\e * \Theta)^2 +(2a^3+2ab^2)  (\rho_\e * \Theta)(\rho_\e * U^+) +(a^2+b^2)^2 (\rho_\e * U^+)^2 + O(\e)) \\
     =&-(a^2+b^2)(\rho_\e * \Theta)^2 \\
     &- 4ab^2(\rho_\e * \Theta)(\rho_\e * U^+) - (2a^3+2ab^2)  (\rho_\e * \Theta)(\rho_\e * U^+) - 2a(a^2+b^2) (\rho_\e * \Theta)^2(\rho_\e * U^+)\\
     &-4a^2b^2(\rho_\e * U^+)^2 - (a^2+b^2)^2 (\rho_\e * U^+)^2 - 8a^2b^2(\rho_\e * \Theta)(\rho_\e * U^+)^2 \\
     &-4(a^4+a^2b^2)  (\rho_\e * \Theta)(\rho_\e * U^+)^2 - (a^2+b^2)^2(\rho_\e * \Theta)(\rho_\e*(U^+)^2) \\
     &- 8a^3b^2(\rho_\e * U^+)^3 -2a(a^2+b^2)^2 (\rho_\e * U^+)^3 \\
     &-  4(a^2 + b^2)ab^2(\rho_\e * \Theta)(\rho_\e * U^+)(\rho_\e * (U^+)^2) - (a^2 + b^2)(2a^3+2ab^2)  (\rho_\e * \Theta)(\rho_\e * U^+)(\rho_\e * (U^+)^2)\\
     &- 4a^2b^2(a^2 + b^2)(\rho_\e * U^+)^2(\rho_\e * (U^+)^2) - (a^2+b^2)^3 (\rho_\e * U^+)^2(\rho_\e * (U^+)^2) + O(\e) \\
     =& -(a^2+b^2)(\rho_\e * \Theta)^2 \\
     &- (4ab^2 +(2a^3+2ab^2) +2a(a^2+b^2))(\rho_\e * U^+)\\
     &- (4a^2b^2 + (a^2+b^2)^2 + 8a^2b^2 + 4(a^4+ a^2b^2) + (a^2+b^2)^2) (\rho_\e * U^+)^2\\
     &- (8a^3b^2 +2a(a^2+b^2)^2 + 4ab^2(a^2 + b^2) + 2a(a^2+b^2)^2) (\rho_\e * U^+)^3\\
     &- (4a^2b^2(a^2 + b^2) + (a^2+b^2)^3) (\rho_\e * U^+)^4 + O(\e).
\end{align*}
Finally,
\begin{align*}
    - \partial_1 \Gamma^\e_{13,4} =& -a \rho_\e - \frac{d_\e}{d_\e}((a^2+b^2)(\rho_\e * \Theta) + (a^2+b^2)(2 \rho_\e (\rho_\e *U^+) + \sigma_\e )) + O(\e) \\
    =& -a \rho_\e + O(\e) -\frac{1}{d_\e}((a^2+b^2)(\rho_\e * \Theta)+ (a^2+b^2)(2 \rho_\e (\rho_\e *U^+) + \sigma_\e ))\\
    & \cdot ((2b \rho_\e * U^+ + 2ab \rho_\e * (U^+)^2)^2 -( 1+ 2a \rho_\e*U^+ +(a^2+b^2)\rho_\e*(U^+)^2+ O(\e^2))^2).
\end{align*}
It follows that 
\begin{align*}
    - d_\e&(\partial_1 \Gamma^\e_{13,4} - a\rho_\e) \\
    =&(a^2+b^2)(\rho_\e * \Theta)+ (a^2+b^2)(2 \rho_\e (\rho_\e *U^+) + \sigma_\e ) \\
    & +4a(a^2+b^2)(\rho_\e * \Theta)(\rho_\e *U^+)\\
    & +(a^2 + b^2)(6a^2 -2b^2)(\rho_\e *U^+)^2 + O(\e) \\
    & +(a^2+b^2)(4a(a^2 + b^2)-8ab^2)(\rho_\e *U^+)^3 + O(\e) \\
    & + (a^2+b^2)((a^2 + b^2)^2-(2ab)^2) (\rho_\e *U^+)^4 + O(\e) \\
    & + (a^2+b^2)(2 \rho_\e (\rho_\e *U^+) + \sigma_\e)\cdot (1-d_\e).
\end{align*}
Collecting the above formulas and using that $|d_\e-1| + O(\e)$ in $\{-\e<U< \e\}$, gives that
\begin{align*}
    R^\e_{3114} &= -a\rho_\e + \frac{1}{d_\e}((a^2+b^2)(\rho_\e * \Theta)-(a^2+b^2)(\rho_\e * \Theta)^2 +(a^2+b^2)(2 \rho_\e (\rho_\e *U^+)) + \sigma_\e)  +O(\e)\\
    &= -a\rho_\e + \frac{1}{d_\e}(a^2+b^2)(2 \rho_\e (\rho_\e *U^+))  +O(\e).
\end{align*}
Following the same steps for $R^\e_{3113}$, we obtain that 
\begin{align*}
    R^\e_{3113} &= - \partial_1 \Gamma^\e_{13,3} + \frac{1}{d_\e}(-2g^\e_{34}\Gamma^\e_{13,3}\Gamma^\e_{13,4} + g_{33}^\e ((\Gamma^\e_{13,3})^2 + (\Gamma^\e_{13,4})^2)) \\
    &= -b\rho_\e + \frac{1}{d_\e}(2ab \rho_\e * \Theta -2ab (\rho_\e * \Theta)^2) + O(\e).
\end{align*}
Finally,
\begin{align*}
    \Ric_{11}^\e &= 2g_\e^{34}R^\e_{3114} + 2g_\e^{33}R^\e_{3113} \\
    &= \frac{1}{d_\e}(-2g^\e_{34}R^\e_{3114} + 2g^\e_{33}R^\e_{3113}).
\end{align*}
Note that 
\begin{align*}
    \frac{1}{d_\e}2g^\e_{33}R^\e_{3113}= -\frac{1}{d_\e}4b^2\rho_\e(\rho_\e * U^+) + O(\e),
\end{align*}
and
\begin{align*}
    \frac{-1}{d_\e}2g^\e_{34}R^\e_{3114}&= \frac{1+ 2a(\rho_\e *U^+)}{d_\e}2a\rho_\e - \frac{4}{(d_\e)^2}(a^2+b^2)(\rho_\e (\rho_\e *U^+)) + O(\e) \\
    &= -2a\rho_\e + 4a^2U\rho_\e(\rho_\e * U^+) + \frac{1}{d_\e}(a^2+b^2)(4\rho_\e (\rho_\e *U^+)) + O(\e).
\end{align*}
In the last step we used that $d_\e = -1 -4a\rho_\e*U^+ + O(\e^2)$ on $(-\e, \e)$. Combining the above estimates, we conclude that:
\begin{align*}
     \Ric_{11}^\e &= -2a\rho_\e +  \frac{1}{d_\e}(-4a^2U\rho_\e(\rho_\e * U^+) + (a^2+b^2)(4\rho_\e (\rho_\e *U^+))- 4b^2\rho_\e(\rho_\e * U^+)) + O(\e) \\
     &=-2a\rho_\e + O(\e).
\end{align*}
\subsection{The case of non-zero cosmological constant}
In this appendix we give more details of the computations involved in the proof of Proposition \ref{prop:ApproxNonFlat}. 
As in the $\Lambda=0$ case, we compute the Ricci curvature using the formulas \eqref{connection_first_kind}, \eqref{low_index_riem_curvature}, and \eqref{ricci_coefficients_formula}. The strategy is to isolate the \( \rho_\varepsilon \)-terms and rewrite the remaining expression as a quotient with denominator \( P_\varepsilon^2 \tilde{d}_\varepsilon^2 \). 
Once this is achieved, we employ the decomposition in \eqref{decomposition_in_metric_and_balance_terms_noflat} to split the analysis into two separate contributions:
\begin{itemize}
    \item [(i)] The terms in the numerator that only involve $\hat{g}_\e, D\hat{g}_\e, D^2\hat{g}_\e$ and $P_\e$. 
    \item [(ii)] The remaining terms in the numerator.
\end{itemize}

We start with (i). The terms that converge to \( 0 \) only in the distributional sense, but not uniformly, are 
\[
(\rho_\varepsilon * U^+)\,\rho_\varepsilon, \qquad U \cdot \rho_\varepsilon, \qquad (\rho_\varepsilon * \Theta)^2 - (\rho_\varepsilon * \Theta).
\]
We now analyze the circumstances under which these terms arise.

\textbf{Step 1.} Products with $\rho_\e$, whose $L^\infty$-norm is bounded below by a constant $C>0$. \\
We begin by identifying where terms involving \( \rho_\varepsilon \) can arise. 
Such contributions occur when computing \( \partial^2_{1,1} P_\varepsilon \) and \( \partial^2_{1,1} \hat{g}^\varepsilon_{ij} \) for \( i,j \in \{3,4\} \). 
In other words, they appear when a non-vanishing component of \( g_\varepsilon \) is differentiated twice in the \( U \)-direction.
Now $\partial^2_{1,1}g^\e_{12} = \partial^2_{1,1}g^\e_{21}$ can only occur in $R^\e_{1211}$ or $R^\e_{1121}$, which are zero by the symmetries of the curvature tensor. 
Hence, terms involving \( \rho_\varepsilon \) can only arise in those curvature components containing \( \partial^2_{1,1} g^\varepsilon_{ij} \) for \( i,j \in \{3,4\} \), namely in \( R^\varepsilon_{[1i1j]} \) with \( i,j \in \{3,4\} \). 
\\For such terms to contribute to the Ricci tensor, they must be contracted with non-vanishing components of the inverse metric. This occurs only when \( g^{ij}_\varepsilon \) is paired with \( R^\varepsilon_{i11j} \) for \( i,j \in \{3,4\} \). Consequently, these contributions appear exclusively in \( \Ric^\varepsilon_{11} \).
Recall that 
\begin{align*}
    \Ric^\e_{11} &= 2g_\e^{12}R^\e_{1112} + 2g_\e^{34}R^\e_{3114} +g_\e^{33}R^\e_{3113}+g_\e^{44}R^\e_{4114} \\
    &= 2g_\e^{34}R^\e_{3114} +g_\e^{33}R^\e_{3113}+g_\e^{44}R^\e_{4114} \\
    &= P_\e^2(2\Tilde{g}_\e^{34}R^\e_{3114} +\Tilde{g}_\e^{33}R^\e_{3113}+\Tilde{g}_\e^{44}R^\e_{4114})  \\
    &= \frac{P_\e^2}{\Tilde{d_\e}}(-2\Tilde{g}^\e_{34}R^\e_{3114} +\Tilde{g}^\e_{44}R^\e_{3113}+\Tilde{g}^\e_{33}R^\e_{4114}).
\end{align*}
Next, we  analyse the $\rho_\e$-contributions from the curvature coefficients. We have that 
\begin{align*}
    R^\e_{3114} &= \partial_3 \Gamma^\e_{11,4} - \partial_1 \Gamma^\e_{13,4} + \ldots \\
    &= -\frac{1}{2}\partial^2_{1,1}g^\e_{34}+ \ldots \\
    &= -\frac{P_\e^2 \partial^2_{1,1}\Tilde{g}^\e_{34}-4P_\e \partial_1\Tilde{g}^\e_{34} \partial_1P_\e -2P_\e\Tilde{g}_{34}^\e \partial^2_{1,1}P_\e + 6\Tilde{g}_{34}^\e \partial_1P_\e \partial_1 P_\e}{2P_\e^4} + \ldots \\
    &= \frac{-P_\e^2 \partial^2_{1,1}\hat{g}_{34}^\e +2P_\e\hat{g}_{34}^\e \partial^2_{1,1}P_\e}{2P_\e^4} + \ldots \\
    &= \frac{-P_\e^2 \cdot  \partial^2_{Z, \cz}\hh \rho_\e -P_\e\hat{g}_{34}^\e \cdot \frac{\Lambda}{6}G \rho_\e}{P_\e^4} + \ldots \\
    &= \frac{-  \partial^2_{Z, \cz}\hh \rho_\e }{P_\e^2} - \frac{ \hat{g}_{34}^\e \frac{\Lambda}{6}G \rho_\e}{P_\e^3} + \ldots
\end{align*}
Moreover,
\begin{align*}
    R^\e_{3113} &= \partial_3 \Gamma^\e_{11,3} - \partial_1 \Gamma^\e_{13,3} + \ldots \\
    &= -\frac{1}{2}\partial^2_{1,1}g^\e_{33}+ \ldots \\
    &= -\frac{P_\e^2 \partial^2_{1,1}\Tilde{g}^\e_{33}-4P_\e \partial_1\Tilde{g}^\e_{33} \partial_1P_\e -2P_\e\Tilde{g}_{33}^\e \partial^2_{1,1}P_\e + 6\Tilde{g}_{33}^\e \partial_1P_\e \partial_1 P_\e}{2P_\e^4} + \ldots \\
    &= \frac{-P_\e^2 \partial^2_{1,1}\hat{g}_{33}^\e +2P_\e\hat{g}_{33}^\e \partial^2_{1,1}P_\e}{2P_\e^4} + \ldots \\
    &= \frac{-P_\e^2 \cdot  \partial^2_{Z, Z}\hh \rho_\e -P_\e\hat{g}_{33}^\e \cdot \frac{\Lambda}{6}G \rho_\e}{P_\e^4} + \ldots \\
    &= \frac{-  \partial^2_{Z, Z}\hh \rho_\e }{P_\e^2} - \frac{ \hat{g}_{33}^\e \frac{\Lambda}{6}G \rho_\e}{P_\e^3} + \ldots , 
\end{align*}
and similarly
\begin{align*}
    R^\e_{4114} &= \frac{-  \partial^2_{\cz, \cz}\hh \rho_\e }{P_\e^2} - \frac{ \hat{g}_{44}^\e \frac{\Lambda}{6}G \rho_\e}{P_\e^3} + \ldots\,  .
\end{align*}
Hence, 
\begin{align*}
    \Ric^\e_{11} &=  \frac{P_\e^2}{\Tilde{d_\e}}\Big(-2\Tilde{g}^\e_{34}\cdot \Big(\frac{-  \partial^2_{Z, \cz}\hh \rho_\e }{P_\e^2} - \frac{ \hat{g}_{34}^\e \frac{\Lambda}{6}G \rho_\e}{P_\e^3}\Big) +\Tilde{g}^\e_{44}\cdot \Big(\frac{-  \partial^2_{Z, Z}\hh \rho_\e }{P_\e^2} - \frac{ \hat{g}_{33}^\e \frac{\Lambda}{6}G \rho_\e}{P_\e^3}\Big) \\
    &\quad +\Tilde{g}^\e_{33}\Big( \frac{-  \partial^2_{\cz, \cz}\hh \rho_\e }{P_\e^2} - \frac{ \hat{g}_{44}^\e \frac{\Lambda}{6}G \rho_\e}{P_\e^3}\Big)\Big) + \ldots\\
    &= -2\hat{g}^\e_{34}\cdot \Big(\frac{- \partial^2_{Z, \cz}\hh \rho_\e }{\Tilde{d_\e}} - \frac{ \hat{g}_{34}^\e \frac{\Lambda}{6}G \rho_\e}{\Tilde{d_\e}P_\e}\Big) +\hat{g}^\e_{44}\cdot \Big(\frac{-  \partial^2_{Z, Z}\hh \rho_\e }{\Tilde{d_\e}} - \frac{ \hat{g}_{33}^\e \frac{\Lambda}{6}G \rho_\e}{\Tilde{d_\e}P_\e}\Big)\\
    &\quad +\hat{g}^\e_{33}\Big( \frac{- \partial^2_{\cz, \cz}\hh \rho_\e }{\Tilde{d_\e}} - \frac{ \hat{g}_{44}^\e \frac{\Lambda}{6}G \rho_\e}{\Tilde{d_\e}P_\e}\Big) + \ldots
\end{align*}

Recall that on $\{-\e<U<\e\} \supset \supp \, \rho_\e$, it holds that $\hat{g}_{44}^\e, \hat{g}_{33}^\e = O(\e)$. Therefore:
\begin{align*}
    \Ric^\e_{11} &= 2\hat{g}^\e_{34}\cdot \Big(\frac{\partial^2_{Z, \cz}\hh \rho_\e }{\Tilde{d_\e}} +\frac{ \hat{g}_{34}^\e \frac{\Lambda}{6}G \rho_\e}{\Tilde{d_\e}P_\e}\Big) -\hat{g}^\e_{44}\frac{\partial^2_{Z, Z}\hh \rho_\e }{\Tilde{d_\e}} -\hat{g}^\e_{33} \frac{\partial^2_{\cz, \cz}\hh \rho_\e }{\Tilde{d_\e}}  + \ldots \\
    &= 2\hat{g}^\e_{34}\cdot \Big(\frac{\partial^2_{Z, \cz}\hh \rho_\e }{\Tilde{d_\e}} +\frac{\frac{\Lambda}{6}G \rho_\e + \frac{\Lambda}{3}G \partial_{Z, \cz}^2\hh \rho_\e(\rho_\e * U^+)}{\Tilde{d_\e}P_\e}\Big) -\frac{4|\partial^2_{\cz, \cz}\hh|^2 \rho_\e(\rho_\e *U^+) }{\Tilde{d_\e}}  + \ldots \\
    &= 2\hat{g}^\e_{34}\cdot \Big(\frac{\partial^2_{Z, \cz}\hh \rho_\e }{\Tilde{d_\e}} +\frac{\frac{\Lambda}{6}G \rho_\e}{\Tilde{d_\e}P_\e}\Big) +\frac{\frac{2\Lambda}{3}G \partial_{Z, \cz}^2\hh \rho_\e(\rho_\e * U^+)}{\Tilde{d_\e}P_\e} -\frac{4|\partial^2_{\cz, \cz}\hh|^2 \rho_\e(\rho_\e *U^+) }{\Tilde{d_\e}}  + \ldots \\
    &= \frac{2\partial^2_{Z, \cz}\hh \rho_\e }{\Tilde{d_\e}} +\frac{\frac{\Lambda}{3}G \rho_\e}{\Tilde{d_\e}P_\e} +  \frac{4(\partial^2_{Z, \cz}\hh)^2 \rho_\e(\rho_\e * U^+) }{\Tilde{d_\e}} +\frac{\frac{2\Lambda}{3}G \partial_{Z, \cz}^2\hh \rho_\e(\rho_\e*U^+)}{\Tilde{d_\e}P_\e}\\
    &\quad + \frac{\frac{2\Lambda}{3}G \partial_{Z, \cz}^2\hh \rho_\e(\rho_\e * U^+)}{\Tilde{d_\e}P_\e} -\frac{4|\partial^2_{\cz, \cz}\hh|^2 \rho_\e(\rho_\e *U^+) }{\Tilde{d_\e}}  + \ldots 
\end{align*}
Observe that, for any $x_0>-1$ and $y> -1-x_0$:
\begin{align}
    \frac{1}{1+x_0+y} = \frac{1}{1+x_0} - \frac{y}{(1+x_0)^2} + \frac{y^2}{(1+x_0+ty)^3}
\end{align}
for some $t \in (0, 1)$. Moreover, we note that in $\{-\e<U< \e\}$:
\begin{align*}
    \Tilde{d}_\e &= -1 -4\partial^2_{Z, \cz}\hh \rho_\e*U^+ + O(\e^2) \\
    \Tilde{d}_\e P_\e &= -(1 + \frac{1}{6} \Lambda(Z \cz - UV - (\rho_\e* U^+)\cdot G)) -4\partial^2_{Z, \cz}\hh \rho_\e*U^+(1 + \frac{1}{6} \Lambda |Z|^2 ) + O_{\mathcal{K}}(\e^2) \\
     &= -((1 + \frac{1}{6} \Lambda|Z|^2) -\frac{1}{6} \Lambda( UV + (\rho_\e* U^+)\cdot G)) -4\partial^2_{Z, \cz}\hh \rho_\e*U^+(1 + \frac{1}{6} \Lambda |Z|^2 ) + O_{\mathcal{K}}(\e^2).
\end{align*}
Combining the above considerations gives that 
\begin{align*}
    \Ric^\e_{11} =& -2\partial^2_{Z, \cz}\hh \rho_\e  -\frac{\Lambda}{3+ \frac{\Lambda |Z|^2}{2}}G \rho_\e+  8(\partial^2_{Z, \cz}\hh)^2 \rho_\e(\rho_\e * U^+) \\
    &+\frac{\frac{\Lambda}{3}G \rho_\e(-\frac{\Lambda}{6}(UV + (\rho_\e*U^+)G) + 4 \partial^2_{Z, \cz}\hh\rho_\e * U^+(1+\frac{\Lambda |Z|^2}{6}))}{(1+\frac{\Lambda |Z|^2}{6})^2}   \\
    &+\frac{4(\partial^2_{Z, \cz}\hh)^2 \rho_\e(\rho_\e * U^+) }{\Tilde{d_\e}} +\frac{\frac{4\Lambda}{3}G \partial_{Z, \cz}^2\hh \rho_\e(\rho_\e*U^+)}{\Tilde{d_\e}P_\e}\\
    &-\frac{4|\partial^2_{\cz, \cz}\hh|^2 \rho_\e(\rho_\e *U^+) }{\Tilde{d_\e}}  + \ldots 
\end{align*}
The right hand side can be rewritten as 
\begin{align}\label{a_rho_e_contributions}
    & -2\partial^2_{Z, \cz}\hh \rho_\e  -\frac{\Lambda}{3+ \frac{\Lambda |Z|^2}{2}}G \rho_\e - \frac{\frac{\Lambda^2}{18}GUV\rho_\e}{P_\e^2 \Tilde{d}_\e^2} \nonumber \\
    &+\frac{\frac{\Lambda}{3}G \rho_\e(-\frac{\Lambda}{6} (\rho_\e*U^+)G + 4 \partial^2_{Z, \cz}\hh\rho_\e * U^+(1+\frac{\Lambda |Z|^2}{6}))}{P_\e^2\Tilde{d}_\e^2} \nonumber  \\
    & -\frac{\frac{4\Lambda}{3}G(1+\frac{\Lambda}{6}|Z|^2) \partial_{Z, \cz}^2\hh \rho_\e(\rho_\e*U^+)}{\Tilde{d_\e}^2P^2_\e} \nonumber\\
    & +\frac{4(1+\frac{\Lambda}{6}|Z|^2)^2((\partial^2_{Z, \cz}\hh)^2 +|\partial^2_{\cz, \cz}\hh|^2) \rho_\e(\rho_\e *U^+) }{P_\e^2\Tilde{d_\e}^2}  + O(\e).
\end{align}
Now define 
\begin{align}\label{a_def_E}
    E_0(V,Z, \Bar{Z}):= -\frac{\Lambda^2}{18}G(Z, \cz)V,
\end{align}
and 
\begin{align}\label{a_def_A}
    A_0(Z, \cz) :=& \frac{\Lambda}{3}G (-\frac{\Lambda}{6} G + 4 \partial^2_{Z, \cz}\hh(1+\frac{\Lambda |Z|^2}{6})) - \frac{4\Lambda}{3}G(1+\frac{\Lambda}{6}|Z|^2) \partial_{Z, \cz}^2\hh \nonumber \\
    &+4(1+\frac{\Lambda}{6}|Z|^2)^2((\partial^2_{Z, \cz}\hh)^2 +|\partial^2_{\cz, \cz}\hh|^2).
\end{align}
\textbf{Step 2.} $(\rho_\e *\Theta)^2 -(\rho_\e *\Theta)$.\\
We begin by analyzing the terms of order one (with respect to \( U \)) that arise at the level of the connection coefficients. 
These may consist either of terms of the form \( (\rho_\varepsilon * \Theta) \), or of constant functions and products involving the remaining coordinates. 
\\To identify contributions involving \( (\rho_\varepsilon * \Theta) \) at the level of the connection coefficients, it suffices to examine derivatives of the form \( \partial_U g_\varepsilon \). For $i,j$ such that $g^\e_{ij} \neq 0$, we have that 
\begin{align*}
    \partial_{1}g^\e_{ij} &= \frac{P_\e \partial_1 \Tilde{g}^\e_{ij}-2\Tilde{g}^\e_{ij} \partial_1 P_\e}{P_\e^3}.
\end{align*}
Hence
\begin{align}\label{a_relevant_u_derivatives_of_metric}
    \partial_{1}g^\e_{12} &= \frac{2\partial_1 P_\e}{P_\e^3}= \frac{-\frac{\Lambda}{3}(V+ G(\rho_\e*\Theta))}{P_\e^3} + \ldots \nonumber \\
    \partial_{1}g^\e_{34} &= \frac{P_\e \partial_1 \hat{g}^\e_{34}-2\hat{g}^\e_{34} \partial_1 P_\e}{P_\e^3} + \ldots \nonumber \\
     &= \frac{2(1+ \frac{\Lambda}{6}|Z|^2)\partial^2_{Z, \cz}\hh \rho_\e * \Theta+\frac{\Lambda}{3}(V+ G(\rho_\e*\Theta))}{P_\e^3} + \ldots 
 \nonumber\\
      \partial_{1}g^\e_{33} &= \frac{P_\e \partial_1 \hat{g}^\e_{33}-2\hat{g}^\e_{33} \partial_1 P_\e}{P_\e^3} + \ldots \nonumber\\
     &=\frac{2(1+ \frac{\Lambda}{6}|Z|^2)\partial^2_{Z, Z}\hh \rho_\e * \Theta}{P_\e^3} + \ldots \nonumber \\
     \partial_{1}g^\e_{44} &=\frac{2(1+ \frac{\Lambda}{6}|Z|^2)\partial^2_{\cz, \cz}\hh \rho_\e * \Theta}{P_\e^3} + \ldots.
\end{align}
We note that the derivatives \( \partial_1 g^\varepsilon_{12} \) and \( \partial_1 g^\varepsilon_{34} \) also contain the term \( V \), which is of order one. 
For \( k > 1 \), however, \( \partial_k \hat{g}_\varepsilon \) does not involve any terms of order one in \( U \). 
Moreover, we observe that
\begin{align*}
    \partial_2 P_\e &= -\frac{\Lambda}{6}U, \\
    \partial_3 P_\e&=  \frac{\Lambda}{6} (\cz + (Z\partial^2_{Z,Z}\hh + \cz \partial^2_{Z, \cz}\hh)(\rho_\e *U^+)) \\
     \partial_4 P_\e&=  \frac{\Lambda}{6} (Z + (Z\partial^2_{Z,\cz}\hh + \cz \partial^2_{\cz, \cz}\hh)(\rho_\e *U^+)).
\end{align*}
Using \eqref{quotient_derivatives}, we deduce that the only terms of order one --- apart from \( \rho_\varepsilon * \Theta \) --- that can appear at the level of the connection coefficients are \( V \), \( Z \), and \( \cz \). 
From \eqref{relevant_u_derivatives_of_metric}, it follows that the only connection coefficients which may contain a \( \rho_\varepsilon * \Theta \)-term are 
\[
\Gamma^\varepsilon_{[11,2]}, \quad \Gamma^\varepsilon_{[13,4]}, \quad \Gamma^\varepsilon_{[13,3]}, \quad \Gamma^\varepsilon_{[14,4]}.
\]
A term of the form \( (\rho_\varepsilon * \Theta)^2 \) can therefore arise only when two such coefficients are multiplied in the computation of a curvature component. 
Such products occur in expressions of the form \( g_\varepsilon^{st} \Gamma^\varepsilon_{ik,s} \Gamma^\varepsilon_{jl,t} \). 
In order for these contributions to remain of order one on \( \{-\e<U<\e\} \), it is necessary that \( \{s,t\} \in \{\{1,2\}, \{3,4\}\} \).

First consider the case \( \{s,t\} = \{1,2\} \). By symmetry, it suffices to treat \( s = 1 \), \( t = 2 \), which yields terms of the form \( \Gamma^\varepsilon_{ik,1} \Gamma^\varepsilon_{jl,2} \). 
To obtain a \( (\rho_\varepsilon * \Theta)^2 \)-term, one must have \( j = l = 1 \) and either \( \{i,k\} = \{1,2\} \) or \( i,k \in \{3,4\} \). 
In the former case, the product appears only in \( R^\varepsilon_{[1112]} \), which vanishes by the symmetries of the curvature tensor. 
In the latter case, the terms contribute to \( R^\varepsilon_{[1313]} \), \( R^\varepsilon_{[1414]} \), or \( R^\varepsilon_{[1314]} \). 
Among these, only the contribution of \( R^\varepsilon_{[1314]} \) to \( \Ric^\varepsilon_{11} \) is of order one, since \( g^{11}_\varepsilon = g^{13}_\varepsilon = g^{14}_\varepsilon = 0 \) and \( g^{33}_\varepsilon, g^{44}_\varepsilon \le C\varepsilon \) on \( \{-\e<U<\e\} \).

Next consider the case \( \{s,t\} = \{3,4\} \), which by symmetry reduces to \( s = 3 \), \( t = 4 \). This leads to terms of the form \( \Gamma^\varepsilon_{ik,3} \Gamma^\varepsilon_{jl,4} \). 
A \( (\rho_\varepsilon * \Theta)^2 \)-term can arise only if \( \{i,k\}, \{j,l\} \in \{\{1,3\}, \{1,4\}\} \), which again contribute only to \( R^\varepsilon_{[1314]} \), \( R^\varepsilon_{[1313]} \), or \( R^\varepsilon_{[1414]} \). 
As before, the only contribution of order one on \( \{-\e<U<\e\} \) to the Ricci tensor occurs in \( \Ric^\varepsilon_{11} \) through the term \( g^{34}_\varepsilon R^\varepsilon_{3114} \).
\\We now proceed to compute all contributions involving a \( (\rho_\varepsilon * \Theta)^2 \)-term. 
As established above, such terms arise exclusively from products of connection coefficients, and hence
\begin{align*}
    R^\e_{3114} &= g_\e^{st}(\Gamma^\e_{13,s}\Gamma^\e_{14,t}-\Gamma^\e_{11,s}\Gamma^\e_{34,t}) + \ldots \\
    &= g_\e^{34}(\Gamma^\e_{13,3}\Gamma^\e_{14,4} + \Gamma^\e_{13,4}\Gamma^\e_{14,3}) -g_\e^{12}\Gamma^\e_{11,2}\Gamma^\e_{34,1} + \ldots .
\end{align*}
Using that 
\begin{align*}
    \Gamma^\e_{11,2} &= \partial_1g^\e_{12}, \\
    \Gamma^\e_{13,3} &= \frac{1}{2}\partial_1g^\e_{33}, \\
    \Gamma^\e_{14,4} &= \frac{1}{2}\partial_1g^\e_{44}, \\
    \Gamma^\e_{13,4} &= \frac{1}{2}\partial_1g^\e_{34} = -\Gamma^\e_{34,1}, 
\end{align*}
we infer that
\begin{align*}
    R^\e_{3114} &= -P_\e^2\frac{\hat{g}^\e_{34}}{\Tilde{d}_\e}(\Gamma^\e_{13,3}\Gamma^\e_{14,4} + \Gamma^\e_{13,4}\Gamma^\e_{14,3}) +P_\e^2\Gamma^\e_{11,2}\Gamma^\e_{34,1} + \ldots \nonumber \\
    &= \frac{1}{P_\e^4}\Big((1+ \frac{\Lambda}{6}|Z|^2)^2|\partial^2_{Z, Z}\hh|^2 (\rho_\e * \Theta)^2 + \Big((1+ \frac{\Lambda}{6}|Z|^2)\partial^2_{Z, \cz}\hh \rho_\e * \Theta+\frac{\Lambda}{6}(V+ G(\rho_\e*\Theta))\Big)^2 \nonumber \\
    &\quad +\Big((1+ \frac{\Lambda}{6}|Z|^2)\partial^2_{Z, \cz}\hh \rho_\e * \Theta+\frac{\Lambda}{6}(V+ G(\rho_\e*\Theta))\Big) \cdot \frac{\Lambda}{3}(V+ G(\rho_\e*\Theta)) \Big) + \ldots \nonumber \\
    &=  \frac{1}{P_\e^4}\Big((1+ \frac{\Lambda}{6}|Z|^2)^2|\partial^2_{Z, Z}\hh|^2 (\rho_\e * \Theta)^2 + \Big((1+ \frac{\Lambda}{6}|Z|^2)\partial^2_{Z, \cz}\hh +\frac{\Lambda}{6}G\Big)^2(\rho_\e*\Theta)^2 \nonumber \\
    &\quad +\Big((1+ \frac{\Lambda}{6}|Z|^2)\partial^2_{Z, \cz}\hh +\frac{\Lambda}{6}G\Big) \cdot \frac{\Lambda}{3} G (\rho_\e*\Theta)^2 \Big) + \ldots 
\end{align*}
It follows that
\begin{align*}
    \Ric^\e_{11} &= 2g_\e^{34}R^\e_{3114} + \ldots = -2P_\e^2\frac{\hat{g}^\e_{34}}{\Tilde{d}_\e}R^\e_{3114} + \ldots \nonumber \\
    &= \frac{2}{P_\e^2\Tilde{d}_\e^2} \Big((1+ \frac{\Lambda}{6}|Z|^2)^2|\partial^2_{Z, Z}\hh|^2 (\rho_\e * \Theta)^2 + \Big((1+ \frac{\Lambda}{6}|Z|^2)\partial^2_{Z, \cz}\hh +\frac{\Lambda}{6}G\Big)^2(\rho_\e*\Theta)^2 \nonumber \\
    &\quad +\Big((1+ \frac{\Lambda}{6}|Z|^2)\partial^2_{Z, \cz}\hh +\frac{\Lambda}{6}G\Big) \cdot \frac{\Lambda}{3} G (\rho_\e*\Theta)^2 \Big) + \ldots 
\end{align*}
Define 
\begin{align*} 
    B_0(Z, \cz) := 2\Big((1+ \frac{\Lambda}{6}|Z|^2)^2|\partial^2_{Z, Z}\hh|^2 + \Big((1+ \frac{\Lambda}{6}|Z|^2)\partial^2_{Z, \cz}\hh +\frac{\Lambda}{6}G\Big)^2 +\frac{\Lambda}{3} G \Big((1+ \frac{\Lambda}{6}|Z|^2)\partial^2_{Z, \cz}\hh +\frac{\Lambda}{6}G\Big)\Big).
\end{align*}

We can now turn to the remaining terms, i.e., (ii). We will again distinguish between contributions involving second derivatives of the metric in $U$-direction , i.e., terms of the form $\partial^2_{1,1}(\Tilde{g}_\e - \hat{g_\e})$ and those that only involve at most one derivative in $U$-direction of $(\Tilde{g}_\e - \hat{g_\e})$.

\textbf{Terms involving  second derivatives in the $U$-direction:} \\
In this case, the only relevant term is \( \partial^2_{1,1} g^\varepsilon_{34} \). 
It appears exclusively in \( R^\varepsilon_{[1314]} \) and therefore contributes only to \( \Ric^\varepsilon_{11} \). 
Accordingly, it suffices to analyze
\begin{align*}
    \Ric_{11}^\e &=2g_\e^{34} R^\e_{3114} + \ldots = 2g_\e^{34} (-\partial_1\Gamma^\e_{13,4})  + \ldots  =2g_\e^{34}\Big(\frac{-\frac{1}{2}\partial^2_{1,1}(\Tilde{g}^\e_{34}-\hat{g}^\e_{34})}{P_\e^2}\Big) + \ldots\\
    &= \frac{-P_\e^2\Tilde{g_{34}^\e}}{\Tilde{d}_\e} \cdot \frac{-P_\e^2(A \rho_\e (\rho_\e * U^+) +B\sigma_\e + E U\rho_\e)}{P_\e^4}  \\
    &= \frac{-\Tilde{g_{34}^\e}P_\e^2}{\Tilde{d}^2_\e} \cdot \frac{-\Tilde{d_\e}(A \rho_\e (\rho_\e * U^+) +B\sigma_\e + E U\rho_\e)}{P_\e^2}. 
\end{align*}

Note that 
$$ \supp\, \sigma_\e \cup \supp \, \rho_\e(\rho_\e * U^+) \cup \supp \, \rho_\e \cdot U \subset 
\{-\e<U< \e\}$$ and that all three functions are bounded in $L^\infty$ by a constant $C>0$ independent of $\e$. In the region $\{-\e<U< \e\}$, it holds that
\begin{align*}
    &\Tilde{g}^\e_{34} = 1 + O(\e), \\
    &\Tilde{d}_\e = -1 + O(\e), \\
    & P_\e = 1+\frac{\Lambda}{6}|Z|^2 + O(\e).
\end{align*}
It follows that 
\begin{align*}
     2g_\e^{34}\Big( \frac{-\frac{1}{2}\partial^2_{1,1}(\Tilde{g}^\e_{34}-\hat{g}^\e_{34})}{P_\e^2}\Big) &=  \frac{-(1+\frac{\Lambda}{6}|Z|^2)^2(A \rho_\e (\rho_\e * U^+) +B\sigma_\e + E U\rho_\e)}{\Tilde{d}^2_\e P_\e^2} + O(\e).
\end{align*}
We now define
\begin{align*}
    A &:= \frac{1}{(1+\frac{\Lambda|Z|^2}{6})^2}A_0, \\
     B &:= \frac{1}{(1+\frac{\Lambda|Z|^2}{6})^2}B_0, \\
      E &:= \frac{1}{(1+\frac{\Lambda|Z|^2}{6})^2}E_0. 
\end{align*}

\textbf{Terms involving at most first derivatives in the $U$-direction:} \\
We first observe that \( D g_\varepsilon \), \( D \tilde{g}_\varepsilon \), \( \tilde{g}_\varepsilon \), \( g_\varepsilon \), and \( P_\varepsilon \) are uniformly bounded on \( \mathcal{K} \), independently of \( \varepsilon \). 
Recall also that \( P_\varepsilon \) is bounded away from zero. 
Moreover, we obtain that
\begin{align*}
    &|D^2g_\e| + |D^2\Tilde{g}_\e|+ |D^2\hat{g}_\e| \leq C\e^{-1}\quad \mathrm{in}\; \{-\e \leq U \leq \e \} \nonumber \\
     &|D^2g_\e| + |D^2\Tilde{g}_\e|+ |D^2\hat{g}_\e| \leq C\quad \mathrm{in} \; \{|U| >\e \},
\end{align*}
where $C> 0$ that does not depend on $\e$. 
From \eqref{compare_g_hat_tilde}, we get that for $i, j = 3,4$, it holds
\begin{align*}
  |\Tilde{g}^{ij}_\e- \hat{g}^{ij}_\e| \leq C \e^2 \quad \mathrm{in}\ \{U<\e \}, \ \mathrm{and} \nonumber \\
   |\Tilde{g}^{ij}_\e- \hat{g}^{ij}_\e| \leq C \e \quad \mathrm{in}\ \{U\geq \e\}.
\end{align*}
Using \eqref{good_convergence}, \eqref{appx_inverse_metric}, \eqref{boundedness_second_deri}, \eqref{rho_e_contributions}, \eqref{def_A}, \eqref{def_E}, \eqref{theta_contribution_constant_size}, and \eqref{def_B}, and arguing analogously to the case \( \Lambda = 0 \), we obtain that
\begin{align*}
    \Ric_{11}^\e &= -2\partial^2_{Z, \cz}\hh \rho_\e  -\frac{\Lambda}{3+ \frac{\Lambda |Z|^2}{2}}G \rho_\e+ \Lambda g^\e_{11} + F_\e,
\end{align*}
and 
\begin{align*}
    \Ric_{jk}^\e &= \Lambda g_{jk}^\e + F_\e, \quad \text{for all }\ (j,k) \neq (1,1),
\end{align*}
for some function $F_\e$ such that $F_\e \to 0$ uniformly as $\e \to 0$.
Combining the above estimates, we conclude that
\begin{align}\label{a_approximation_ricci}
    \Ric^\e_{ij} = \left(-2\partial^2_{Z, \cz}\hh  -\frac{\Lambda}{3+ \frac{\Lambda |Z|^2}{2}}G \right) \rho_\e \delta^1_i\delta^1_j + \Lambda g^\e_{ij} + r_{ij}^\e,
\end{align}
where $r^\e_{ij}\to 0$  uniformly on $E$, as $\e \to 0$.
\phantomsection
\addcontentsline{toc}{section}{References}
\bibliography{bibliographie}
\bibliographystyle{plain}
\end{document}